\documentclass[10pt,reqno]{amsart}
\usepackage{latexsym,amsmath}
\usepackage{enumerate}
\usepackage{amsfonts}
\usepackage{amssymb}
\usepackage{latexsym}

\usepackage[T1]{fontenc}
\usepackage{fourier}
\usepackage{bbm}
\usepackage{verbatim}
\usepackage{graphicx}
\usepackage{float}
\usepackage{color}
\usepackage{fullpage}

\usepackage{hyperref}
\allowdisplaybreaks[1]

\newcommand{\cPcent}{\cP_*^{2}}
\newcommand{\disso}{disassortative}
\newcommand{\asso}{assortative}

\newcommand{\GG}{\mathbb G}
\newcommand{\Gsbm}{\G_{\mathrm{sbm}}}
\newcommand{\Gpotts}{\G_{\mathrm{Potts}}}
\newcommand{\hGpotts}{\hat\G_{\mathrm{Potts}}}

\newcommand\MU{\vec\mu}
\newcommand\NU{\vec\nu}
\newcommand\cMU{\check\MU}
\newcommand\cSIGMA{\check\SIGMA}

\newcommand\Bsbm{\cB_{\mathrm{Potts}}}

\newcommand\vU{\vec U}

\newcommand\PSI{\vec\psi}
\newcommand\nix{\,\cdot\,}

\newcommand\vS{\vec S}

\newcommand\dd{{\mathrm d}}

\newcommand\G{\vec G}

\numberwithin{equation}{section}

\renewcommand{\vec}[1]{\boldsymbol{#1}}

\newcommand\KL[2]{D_{\mathrm{KL}}\bc{{{#1}\|{#2}}}}

\newcommand\SIGMA{\vec\sigma}
\newcommand\TAU{\vec\tau}

\newtheorem{definition}{Definition}[section]
\newtheorem{claim}[definition]{Claim}

\newtheorem{remark}[definition]{Remark}
\newtheorem{theorem}[definition]{Theorem}
\newtheorem{lemma}[definition]{Lemma}
\newtheorem{proposition}[definition]{Proposition}
\newtheorem{corollary}[definition]{Corollary}

\newtheorem{fact}[definition]{Fact}

\newcommand\cA{\mathcal{A}}
\newcommand\cB{\mathcal{B}}

\newcommand\cG{\mathcal{G}}
\newcommand\cE{\mathcal{E}}
\newcommand\cU{\mathcal{U}}

\newcommand\cS{\mathcal{S}}
\newcommand\cT{\mathcal{T}}
\newcommand\cI{\mathcal{I}}

\newcommand\cP{\mathcal{P}}
\newcommand\cX{\mathcal{X}}
\newcommand\cY{\mathcal{Y}}
\newcommand\cV{\mathcal{V}}

\newcommand\cZ{\mathcal{Z}}

\def\cE{{\mathcal E}}

\newcommand{\cPfix}{\cP_{\mathrm{fix}}^2(d)}
\newcommand{\cPfixE}[1]{\cP_{\mathrm{fix}}^2(d,#1)}

\newcommand{\beq}{\begin{equation}} \newcommand{\eeq}{\end{equation}}
\newcommand{\dc}{d_{q,\mathrm{cond}}}
\newcommand{\betac}{\beta_{q,\mathrm{cond}}}
\newcommand{\dinf}{d_{\mathrm{inf}}}

\newcommand\eps{\varepsilon}
\newcommand\del{\delta}

\newcommand\Var{\mathrm{Var}}
\newcommand\Erw{\mathrm{E}}
\newcommand{\vecone}{\vec{1}}

\newcommand{\Po}{{\rm Po}}
\newcommand{\Bin}{{\rm Bin}}

\newcommand\TV[1]{\left\|{#1}\right\|_{\mathrm{TV}}}

\newcommand{\bink}[2] {{\binom{#1}{#2}}}

\newcommand\bc[1]{\left({#1}\right)}
\newcommand\cbc[1]{\left\{{#1}\right\}}
\newcommand\bcfr[2]{\bc{\frac{#1}{#2}}}
\newcommand{\bck}[1]{\left\langle{#1}\right\rangle}
\newcommand\brk[1]{\left\lbrack{#1}\right\rbrack}
\newcommand\scal[2]{\bck{{#1},{#2}}}
\newcommand\norm[1]{\left\|{#1}\right\|}
\newcommand\abs[1]{\left|{#1}\right|}

\newcommand\RR{\mathbb{R}}

\newcommand{\Whp}{W.h.p.}
\newcommand{\whp}{w.h.p.}

\newcommand{\stacksign}[2]{{\stackrel{\mbox{\scriptsize #1}}{#2}}}
\newcommand{\tensor}{\otimes}

\newcommand{\Erdos}{Erd\H{o}s}
\newcommand{\Renyi}{R\'enyi}

\newcommand\pr{\mathrm{P}} 
\renewcommand\Pr{\pr} 

\newcommand\Lem{Lemma}
\newcommand\Prop{Proposition}
\newcommand\Thm{Theorem}
\newcommand\Def{Definition}
\newcommand\Cor{Corollary}
\newcommand\Sec{Section}
\newcommand\Chap{Chapter}

\begin{document}

\title{Information-theoretic thresholds from the cavity method}

\author[Coja-Oghlan and Perkins]{Amin Coja-Oghlan$^{*}$, Florent Krzakala$^{**}$, Will Perkins$^{***}$, and Lenka Zdeborov\'a$^{****}$}
\thanks{$^{*}$The research leading to these results has received funding from the European Research Council under the European Union's Seventh 
Framework Programme (FP7/2007-2013) / ERC Grant Agreement n.\ 278857--PTCC\\
$^{**}$The research leading to these results has received funding from the European Research Council under the European Union's Seventh 
Framework Programme (FP7/2007-2013) / ERC Grant Agreement n.\ 307087--SPARCS\\
$^{***}$Supported in part by EPSRC grant EP/P009913/1. \\
$^{****}$LZ acknowledges funding from the European Research
Council (ERC) under the European Union's Horizon 2020 research and innovation program (grant agreement No 714608 - SMiLe).}

\address{Amin Coja-Oghlan, {\tt acoghlan@math.uni-frankfurt.de}, Goethe University, Mathematics Institute, 10 Robert Mayer St, Frankfurt 60325, Germany.}

\address{Florent Krzakala, {\tt florent.krzakala@ens.fr}, 
Laboratoire de Physique Statistique, CNRS, PSL Universit\'es \&\ Ecole Normale Sup\'erieure, Sorbonne Universit\'es et Universit\'e Pierre \&\ Marie Curie, 75005, Paris, France.}

\address{Will Perkins, {\tt math@willperkins.org}, School of Mathematics, University of Birmingham, Edgbaston, Birmingham, UK.}

\address{Lenka Zdeborov\'a, {\tt lenka.zdeborova@cea.fr}, 
	Institut de Physique Th\'eorique, CNRS, CEA, Universit\'e Paris-Saclay, F-91191, Gif-sur-Yvette, France}

\begin{abstract}
\noindent
Vindicating a sophisticated but non-rigorous physics approach called the cavity method,
we establish a formula for the mutual information in statistical inference problems induced by random graphs and
we show that the mutual information holds the key to understanding certain important phase transitions in random graph models.
We work out several concrete applications of these general results.
For instance, we pinpoint the exact condensation phase transition in the Potts antiferromagnet on the random graph,
thereby improving prior approximate results  [Contucci et al.: Communications in Mathematical Physics 2013].
Further,  we prove the conjecture from [Krzakala et al.: PNAS 2007] about the condensation phase transition in the random graph coloring problem for any number $q\geq3$ of colors.
Moreover, we prove the conjecture on the information-theoretic threshold in the \disso\ stochastic block model 
	[Decelle et al.: Phys.\ Rev.\ E 2011].
Additionally, our general result implies the conjectured formula for the mutual information in Low-Density Generator Matrix  codes
 [Montanari: IEEE Transactions on Information Theory 2005].
\end{abstract}

\maketitle

\section{Introduction}\label{Sec_intro}

\noindent
Since the late 1990's physicists have studied models of spin systems in which the geometry of interactions is determined by a sparse random graph
in order to better understand ``disordered'' physical systems such as glasses or spin glasses~\cite{MP1,MP2,Monasson}.
To the extent that the sparse random graph induces an actual geometry on the sites, 
such ``diluted mean-field models'' provide better approximations to physical reality than models on the complete graph
such as the Curie--Weiss or the Sherrington--Kirkpatrick model~\cite{MM}.
But in addition, and perhaps more importantly,
as random graph models occur in many branches of science, the physics ideas have since led to intriguing predictions
on an astounding variety of important problems in mathematics, computer science, information theory, and statistics.
Prominent examples include the phase transitions in the random $k$-SAT and random graph coloring problems~\cite{MPZ,LenkaFlorent},
	both very prominent problems in combinatorics,
	error correcting codes~\cite{MM}, compressed sensing~\cite{LF}, and
	the stochastic block model~\cite{Decelle}, a classical statistical inference problem.

The thrust of this work goes as follows.
In many problems random graphs are either endemic or can be introduced via probabilistic constructions.
As an example of the former think of the stochastic block model, where the aim is to recover a latent partition from a random graph.
For an example of the latter, think of low density generator matrix `LDGM' codes, where by design the generator matrix is  the adjacency matrix of a
random bipartite graph.
To models of either type physicists bring to bear the {\em cavity method}~\cite{mezard1990spin},
a comprehensive tool for studying random graph models, 
 to put forward predictions on phase transitions and
the values of key quantities.
The cavity method comes in two installments:
the replica symmetric version, whose mainstay is the Belief Propagation messages passing algorithm,
and the more intricate replica symmetry breaking version, but it has emerged that the replica symmetric version suffices to deal with many important models.

Yet the cavity method suffers an unfortunate drawback: it is utterly non-rigorous.
In effect, a substantial research effort in mathematics has been devoted to proving specific conjectures based on the physics calculations.
Success stories include the ferromagnetic Ising model and Potts models on the random graph~\cite{dembo, demboPotts},
the exact $k$-SAT threshold for large $k$~\cite{KostaSAT,DSS3},
the condensation phase transition in random graph coloring~\cite{Cond},
work on the stochastic block model~\cite{massoulie2014community,mossel2013proof,Mossel} and 
terrific results on error correcting codes~\cite{GMU}.
But while the cavity method can be applied mechanically to a wide variety of problems,
the current rigorous arguments are case-by-case.
For instance, the methods of \cite{Cond,KostaSAT,DSS3} depend on painstaking second moment calculations that take the physics intuition on board
but require 
extraneous assumptions
(e.g., that the clause length $k$ or the number of colors be very large).
Moreover, many proofs require lengthy detours or case analyses that ought to be expendable.
Hence, the {obvious question} is: can we vindicate the physics calculations wholesale?

The main result of this paper is that for a wide class of problems within the purview of the replica symmetric cavity method the answer is `yes'.
More specifically, the cavity method reduces a combinatorial problem on a random graph to an optimization problem
on the space of probability distributions on a simplex of bounded dimension.
We  prove that  this reduction is valid under a few easy-to-check conditions.
Furthermore, we verify that the stochastic optimization problem admits a combinatorial interpretation as the problem of finding an optimal set of Belief Propagation messages on a Galton-Watson tree.
Thus, we effectively reduce a problem on a random graph, a mesmerizing object characterized by expansion properties,
to a calculation on a random tree. 
This result reveals an intriguing connection between statistical inference problems and phase transitions in random graph models,
specifically a phase transition that we call the information-theoretic threshold, which  in many important models
 is identical to the so-called ``condensation phase transition'' predicted by physicists~\cite{pnas}.
Moreover, the proofs provide a direct rigorous basis for the physics calculations, and we therefore believe that our techniques
will find future applications.
To motivate the general results about the connection between statistical inference and phase transitions,
	which we state in \Sec~\ref{Sec_general}, we begin with four concrete applications
	that have each received considerable attention in their own right.

\subsection{The Potts antiferromagnet}\label{Sec_col}
As a first example we consider the antiferromagnetic Potts model on the \Erdos-\Renyi\ random graph $\GG=\GG(n,d/n)$
with $n$ vertices where any two vertices are connected by an edge with probability $d/n$ independently.
Let $\beta>0$ be a parameter that we call
`inverse temperature' and let $q\geq2$ be a fixed number of colors.
With $\sigma$ ranging over all color assignments $\{1,\ldots,n\}\to\{1,\ldots,q\}$
 the Potts model {\em partition function} is defined as
	\begin{align}\label{eqAntiPotts}
	Z_\beta(\GG)&=\sum_{\sigma}\exp\bc{-\beta\sum_{\{v,w\}\in E(\GG)}\vecone\{\sigma(v)=\sigma(w)\}}.
	\end{align}
Standard arguments show that the random variable $\ln Z_\beta(\GG)$ is concentrated about its expectation. Thus, the key quantity of interest is the function
	$$(d,\beta)\in(0,\infty)\times(0,\infty)\mapsto\lim_{n\to\infty}-\frac1n\Erw[\ln Z_\beta(\GG(n,d/n))],$$
the {\em free energy density} in physics jargon; the limit is known to exist for all $d,\beta$~\cite{bayati}.
In particular, for a given $d$ we say that a \emph{phase transition} occurs at $\beta_0\in(0,\infty)$ if the function
	$$\beta\mapsto\lim_{n\to\infty}-\frac1n\Erw[\ln Z_\beta(\GG(n,d/n))]$$
is non-analytic at $\beta_0$, i.e., there is no expansion as an absolutely convergent power series in a neighborhood of $\beta_0$.%
\footnote{This definition of `phase transition', which is standard in mathematical physics, is in line with the random graphs terminology.
For instance, 
the function that maps $d$ to the expected fraction of vertices in the largest connected component of $\GG(n,d/n)$ is non-analytic at $d=1$.}

According to the cavity method, for small values of $\beta$ (``high temperature'') the free energy is given by a simple explicit expression.
But as $\beta$ gets larger a phase transition occurs, called the {\em condensation phase transition}, 
provided that $d$ is sufficiently large.
Contucci, Dommers, Giardina and Starr~\cite{CDGS} derived upper and lower bounds on the critical value of $\beta$,
 later refined by Coja-Oghlan and Jaafari~\cite{Nor}.
The following theorem pinpoints the phase transition precisely for all $d,q$.
Indeed, the theorem shows that the exact phase transition is determined by the very stochastic optimization problem that the cavity method predicts~\cite{MM}.

To state the result we need a bit of notation.
For a finite set $\Omega$ we identify the set $\cP(\Omega)$ of probability measures on $\Omega$ 
 with the standard simplex in $\RR^\Omega$.
Let $\cP^2(\Omega)$ be the set of all probability measures on $\cP(\Omega)$
and write $\cPcent(\Omega)$ for the set of $\pi\in\cP^2(\Omega)$ 
whose mean $\int_{\cP(\Omega)}\mu d\pi(\mu)$ is the uniform distribution on $\Omega$.
Moreover, for $\pi\in\cPcent(\Omega)$ let $\vec\mu_1^{(\pi)},\vec\mu_2^{(\pi)},\ldots$ be a sequence of samples from $\pi$
and let $\vec\gamma=\Po(d)$, all mutually independent.
Further, let $\Lambda(x)=x\ln x$ for $x\in(0,\infty)$ and $\Lambda(0)=0$.
For an integer $k\geq1$ let $[k]=\{1,\ldots,k\}$.
Finally, we use the convention $\inf\emptyset=\infty$.

\begin{theorem}\label{Thm_Potts}
Let $q\geq2$ and $d>0$ and for $c\in[0,1]$ let
	\begin{align}
	\cB_{\mathrm{Potts}}(q,d,c)&=\sup_{\pi\in\cPcent([q])}\Erw\brk{\frac{\Lambda(\sum_{\sigma=1}^q\prod_{i=1}^{\vec\gamma}1-c\vec\mu_{i}^{(\pi)}(\sigma))}{q(1-c/q)^{\vec\gamma}}
	-\frac{ d\Lambda(1-\sum_{\tau=1}^qc\vec\mu_1^{(\pi)}(\tau)\vec\mu_2^{(\pi)}(\tau))}{2(1-c/q)}}
		\label{eqSBM},\\
	\betac(d)&=\inf\cbc{\beta>0:\cB_{\mathrm{Potts}}(q,d,1-\exp(-\beta))>\ln q+d\ln(1-(1-\exp(-\beta))/q)/2}.\label{eqSBM_2}
	\end{align}
Then  for all $\beta<\betac(d)$ we have 
	\begin{equation}\label{eqPottsRS}
	\lim_{n\to\infty}-\frac1n\Erw[\ln Z_\beta(\GG(n,d))]=-\ln q-d\ln(1-(1-\exp(-\beta))/q)/2
	\end{equation}
and if $\betac(d)<\infty$, then a phase transition occurs  at $\betac(d)$.
\end{theorem}

A simple first moment calculation shows that $\betac(d)<\infty$, and thus {\em that} a phase transition occurs, if $d>(2q-1)\ln q$~\cite{Nor}.
In fact, for any $\beta>0$ the formulas (\ref{eqSBM})--(\ref{eqSBM_2}) yield a finite maximum value
	\begin{equation}\label{eqSBM_3}
	\dc(\beta)=\inf\cbc{d>0:\cB_{\mathrm{Potts}}(q,d,1-\exp(-\beta))>\ln q+d\ln(1-(1-\exp(-\beta))/q)/2}
	\end{equation}
such that (\ref{eqPottsRS}) holds if and only if $d \le \dc(\beta)$.  (Strictly speaking, this last statement requires a monotonicity argument; see \Lem~\ref{Lemma_monotonicityFix}.)
 Thus, (\ref{eqSBM})--(\ref{eqSBM_2}) identify a line in the $(d,\beta)$-plane that marks the location of the condensation phase transition.

\subsection{Random graph coloring}
The random graph coloring problem  is one of the best-known problems in probabilistic combinatorics:
given a number $q\geq3$ of available ``colors'', for what values of $d$ is it typically possible to assign colors to the vertices of $\GG=\GG(n,d/n)$ 
such that no edge connects two vertices with the same color?
Since the problem was posed by \Erdos\ and \Renyi\ in their seminal paper that started the theory of random graphs~\cite{ER60},
the random graph coloring problem and its ramifications have received enormous attention 
	(e.g.,~\cite{AchMoore3col,AchNaor,AlonKriv,BBColor,Greenhill,KrivSud,LuczakColor,LenkaFlorent}).
Of course, an intimately related question is:
\emph{how many} ways are there to color the vertices of the random graph 
 $\GG$ with $q\geq3$ colors such that no edge is monochromatic?
In fact, for $q>3$ the best known lower bounds on the largest value of $d$ up to which $\GG$ remains $q$-colorable, the {\em $q$-colorability threshold},
	are derived by tackling this second question~\cite{AchNaor,Cond}.
If $d<1$, then the random graph $\GG$ does not have a `giant component'.
We therefore expect that the number  $Z_q(\GG)$ of $q$-colorings is about $q^n(1-1/q)^{dn/2}$,
	because a forest with $n$ vertices and average degree $d$ has that many $q$-colorings.
Indeed, for $d<1$ it is easy to prove that
	\begin{align}\label{eqFirstMmt}
	\frac1n\ln Z_q(\GG(n,d/n))\ \stacksign{$n\to\infty$}\to\ \ln q+\frac d2\ln(1-1/q)\qquad\mbox{in probability}
	\end{align}
and the largest degree $\dc$ up to which (\ref{eqFirstMmt}) holds is called the {\em condensation threshold}.
Perhaps surprisingly, the cavity method predicts that the condensation threshold is far greater than the giant component threshold.
Once more the predicted formula takes the form of a stochastic optimization problem~\cite{LenkaFlorent}.
Prior work based on the second moment method verified this under the assumption that $q$ exceeds some (undetermined but astronomical)
constant $q_0$~\cite{Cond}.
Here we prove the conjecture for all $q\geq3$.

\begin{theorem}\label{Thm_col}
For $q\geq3$ and $d>0$ and with $\cB_{\mathrm{Potts}}$ from (\ref{eqSBM}) let
	\begin{align}\label{eqcol}
	\dc&=\inf\cbc{d>0:\cB_{\mathrm{Potts}}(q,d,1)>\ln q+d\ln(1-1/q)/2}.
	\end{align}
Then (\ref{eqFirstMmt}) holds for all $d<\dc$.
By contrast, for every $d>\dc$ there exists $\eps>0$ such that \whp
	$$Z_q(\GG(n,d/n))<q^n(1-1/q)^{dn/2}\exp(-\eps n).$$
\end{theorem}

\noindent
It is conjectured that $d_{3,\mathrm{cond}}=4$~\cite{LenkaFlorent}, but we have no reason to believe $\dc$ admits a simple expression for $q>3$.
Asymptotically we know $\dc=(2q-1)\ln q-2\ln 2+\eps_q$ with $\lim_{q\to\infty}\eps_q=0$~\cite{Cond}.
By comparison, for $d>(2q-1)\ln q-1+\eps_q$ the random graph fails to be $q$-colorable with probability tending to $1$ as $n\to\infty$~\cite{Covers}.

Since~(\ref{eqFirstMmt}) cannot hold for $d$ beyond the $q$-colorability threshold,
$\dc$ provides a lower bound on that threshold.
In fact, $\dc$ is at least as large as the best prior lower bounds for $q>3$ from~\cite{AchNaor,Cond},
because their proofs imply (\ref{eqFirstMmt}).
But more importantly, \Thm~\ref{Thm_col} 
facilitates the study of the geometry of the set of $q$-colorings for small values of $q$.
Specifically, if $d,q$ are such that (\ref{eqFirstMmt}) is true, then the notoriously difficult experiment of sampling a random $q$-coloring of a random graph can be
studied indirectly by way of a simpler experiment called the planted model~\cite{Barriers,Silent,quiet}.
This approach has been vital to the analysis of, e.g., the geometry of the set of $q$-colorings or the emergence of
	``frozen variables''~\cite{Barriers,Molloy}.
Additionally,  
it can be derived from \Thm~\ref{Thm_col} and~\cite{montanari2011reconstruction} that for all $q\geq3$ the threshold for an important spatial mixing property called reconstruction on the random graph $\GG(n,d/n)$
equals the reconstruction threshold on the Galton-Watson tree with offspring distribution $\Po(d)$.

Finally, the formula (\ref{eqAntiPotts}) suggests to think of the inverse temperature parameter $\beta$ in the Potts antiferromagnet as a ``penalty''
imposed on monochromatic edges.
Then we can view the random graph coloring problem as the $\beta=\infty$ version of the Potts antiferromagnet.
Indeed, using the dominated convergence theorem, we easily verify that that the number $\dc$ from \Thm~\ref{Thm_col} is equal to the limit
$\lim_{\beta\to\infty}\dc(\beta)$ of the numbers from (\ref{eqSBM_3}).

\subsection{The stochastic block model}
We prove results such as \Thm~\ref{Thm_Potts} and~\ref{Thm_col} in an indirect and perhaps surprising way via statistical inference problems.
In fact, we will see that these provide the appropriate framework to investigate the replica symmetric cavity method.
Let us look at one well known example of such an inference problem, the {\em stochastic block model}, which can be viewed as
the statistical inference version of the Potts model.

Suppose we choose a random coloring $\SIGMA^*$ of $n$ vertices with $q\geq2$ colors,
 then generate a random graph by connecting any two vertices of the same color with probability $d_{\mathrm{in}}/n$ and
any two with distinct colors with probability $d_{\mathrm{out}}/n$ independently;
write $\G^*$ for the resulting random graph.
Specifically, set $d_{\mathrm{in}}=dq\exp(-\beta)/(q-1+\exp(-\beta))$ and 
$d_{\mathrm{out}}=dq/(q-1+\exp(-\beta))$ so that the expected degree of any vertex equals $d$.
Then bichromatic edges are preferred if $\beta>0$ (``\disso\ case''),  while monochromatic ones are preferred if $\beta<0$ (``\asso\ case''). 
The model was first introduced in machine learning by Holland, Laskey, and Leinhardt~\cite{Holland} as early as 1983, and has since attracted
rather considerable attention in probability, computer science, and combinatorics (e.g., \cite{AlonKahale,AKS,BJR,Boppana,Adaptive,McSherry}). 

The inference task associated with the model is to recover $\SIGMA^*$ given just $\G^*$.
When $d$ remains fixed as $n\to\infty$ then typically a constant fraction of vertices will have degree $0$, and so exact recovery of $\SIGMA^*$ is a hopeless task.
Instead we ask for a coloring that overlaps with $\SIGMA^*$ better than a mere random guess.  
Formally, define the {\em agreement} of two colorings $\sigma,\tau$ as
	\[ A(\sigma,\tau)= \frac{{-1+\max_{\kappa\in S_q} \frac{q}{n}\sum_{v \in V(G)} \vec 1 \{\sigma(v) = \kappa\circ\tau(v)   \}}}{q-1}.\]
Then  for all $\sigma,\tau$, $A(\sigma,\tau)\geq 0$, $A(\sigma, \sigma) =1$,  and  two independent random colorings $\sigma,\tau$ have expected agreement $o(1)$ as $n\to\infty$.
Hence, for what $d,\beta$  can we infer a coloring $\tau(\G^*)$ such that $A(\SIGMA^*,\tau(\G^*))$ is bounded away from $0$?

According to the cavity method, this question admits two possibly distinct answers~\cite{Decelle}.
First, for any given $q,\beta$ there exists an {\em information-theoretic threshold} $\dinf(q,\beta)$ such that {\em no} algorithm
produces a partition $\tau(\G^*)$ such that $A(\SIGMA^*,\tau(\G^*))\geq\Omega(1)$ with a non-vanishing probability if $d<d_{\mathrm{inf}}(q,\beta)$.
By contrast, for $d>\dinf(q,\beta)$ there is a (possibly exponential-time) algorithm that does.
The formula for $\dinf(q,\beta)$ comes as a stochastic optimization problem.
The second {\em algorithmic threshold} $d_{\mathrm{alg}}(q,\beta)$ marks the point from where the problem can be solved by an {efficient} 
	(i.e., polynomial time) algorithm.
The cavity method predicts the simple formula
	\begin{equation}\label{eqalg}
	d_{\mathrm{alg}}(q,\beta)=\bcfr{q-1+\exp(-\beta)}{1-\exp(-\beta)}^2.
	\end{equation}
While the information-theoretic threshold is predicted to coincide with the algorithmic threshold for {$q=2,3$} (and for $q=4$ and small $\beta$),
we do not expect that there is a simple expression for $\dinf(q,\beta)$ for other choices of parameters.

The physics conjectures have inspired quite a bit of rigorous work (e.g. \cite{deshpande15,Guedon,MontanariSen}). 
Mossel, Neeman and Sly~\cite{mossel2013proof,Mossel} and Massouli\'e~\cite{massoulie2014community} proved the conjectures for $q=2$.  Abbe and Sandon~\cite{abbe2015detection} proved the positive part of the algorithmic conjecture for all $q\geq3$;
	see also Bordenave, Lelarge, Massouli\'e~\cite{BLM} for a different but less general algorithm.
Moreover, independently of each other Abbe and Sandon~\cite{abbe2015detection} and Banks, Moore, Neeman and Netrapalli~\cite{Banks} derived upper bounds on the information-theoretic threshold that are strictly below $d_{\mathrm{alg}}(q,\beta)$ for $q\geq5$ by providing exponential-time algorithms to detect the planted partition. Banks, Moore, Neeman and Netrapalli additionally derived lower bounds on the information-theoretic threshold via a delicate second moment calculation in combination with small subgraph conditioning. Their lower bounds match the upper bounds up to a constant factor. 
The following theorem settles the exact information-theoretic threshold for all $q\geq3$, $\beta>0$.
Recall $\cB_{\mathrm{Potts}}$ from (\ref{eqSBM}).
    
\begin{theorem}\label{Thm_SBM}
Suppose $\beta>0$, $q\geq3$ and $d>0$.
Let
	$$\dinf(q,\beta)=\inf\cbc{d>0:\cB_{\mathrm{Potts}}(q,d,1-\exp(-\beta))>\ln q+d\ln(1-(1-\exp(-\beta))/q)/2}.$$
\begin{itemize}
\item If $d > \dinf(q,\beta)$, then there exists an algorithm (albeit not necessarily an efficient one) that 
	outputs a partition $\tau_{\mathrm{alg}}(\G^*)$ such that $\Erw [A(\SIGMA^*,\tau_{\mathrm{alg}}(\G^*))]\geq\Omega(1)$.
\item If $d < \dinf(q,\beta)$, then for any algorithm (efficient or not) 
		 we have $\Erw[A(\SIGMA^*,\tau_{\mathrm{alg}}(\G^*))]=o(1)$.
\end{itemize}
\end{theorem}  

While the claim that $d_{\mathrm{alg}}(q,\beta)=\dinf(q,\beta)$ for $q=3$ is not apparent from \Thm~\ref{Thm_SBM},
the theorem reduces this problem to a self-contained analytic question
that should be within the scope of known techniques (see \Sec~\ref{sec:discussionRelated}).
Furthermore, the proofs of \Thm s~\ref{Thm_Potts} and~\ref{Thm_col} are actually based on \Thm~\ref{Thm_SBM},
and we shall see that quite generally phase transitions in ``plain'' random graph models can be tackled by way
of a natural corresponding statistical inference problem.

\subsection{LDGM codes}
\label{sec:parity}
But before we come to that, let us consider a fourth application, namely {\em Low-Density Generator Matrix} codes~\cite{ChengMcEliece,KabashimaSaad}.
For a fixed $k\geq2$ form a bipartite graph $\G$ consisting of $n$ ``variable nodes'' and $m \sim \Po(dn/k)$ ``check nodes''.
Each check node $a$ gets attached to a random set $\partial a$ of $k$ variable nodes  independently.
Then select a {\em signal}  $\SIGMA^* \in  \mathbb \{\pm 1\}^n$ uniformly at random. 
An output message  $\vec y \in \mathbb \{\pm 1\}^m$ is obtained by setting 
$y_a = \prod_{i \in \partial a} \SIGMA^*_i$  with probability  $1- \eta$ resp.\ $y_a = -\prod_{i \in \partial a} \SIGMA^*_i$ with probability $\eta$
for each check node $a$ independently.
In other words, 
if we identify $(\{\pm1\},\nix)$ with $(\mathbb F_2,+)$,
 the signal $\SIGMA^*$ is  encoded by multiplication by the random biadjacency matrix of $\G$, then suffers from errors in transmission, each bit being flipped with probability $\eta$, to form the output message $\vec y$.  Now let $\G^*$ be the bipartite graph $\G$ decorated on each check node $a$ with the value $\vec y_a\in\{\pm1\}$.
The decoding task is to recover $\SIGMA^*$ given $\G^*$.

The appropriate measure to understand the information-theoretic limits of the decoding task
is the {\em mutual information } between $\vec\SIGMA^*$ and $\G^*$, which we recall is defined as%
\begin{align}\label{eqMutualInf}
I( \SIGMA^*, \G^*) &= \sum_{G,\sigma}\pr\brk{\G^*=G,\SIGMA^*=\sigma}\ln\frac{\pr\brk{\G^*=G,\SIGMA^*=\sigma}}
		{\pr\brk{\G^*=G}\pr\brk{\SIGMA^*=\sigma}},
\end{align}
with the sum ranging over all possible graphs $G$ and $\sigma\in\{\pm1\}^n$.
{Abbe and Montanari~\cite{Abbe} proved that for any $d,\eta$ and for even $k$
 the limit $\lim_{n\to\infty}\frac{1}{n}I(\SIGMA^*, \G^*)$ of the mutual information per bit {\em exists}.
The following theorem {\em determines} the limit for all $k\geq2$, even or odd.
Let $\cP_0([-1,1])$ be the set of all probability distributions on $[-1,1]$ with mean $0$.
Let  $\vec  J, (\vec J_b)_{b\geq1}$ be uniform $\pm 1 $ random variables, let $\vec \gamma=\Po(d)$, and let
	$(\vec \theta_j^{(\pi)})_{j\geq1}$ be samples from $\pi\in\cP_0([-1,1])$, all mutually independent.}

\begin{theorem}\label{thm:noisyXor}
For $k\geq2$, $\eta >0$, and $d>0$, let 
	\begin{align*}
	\cI(k,d, \eta)&=\sup_{\pi\in\cP_0([-1,1])} \Erw\brk{ \frac{1}{2} 
			\Lambda\bc{\sum_{\sigma\in \{ \pm1 \}}\prod_{b=1}^{\vec\gamma} 1+    \sigma \vec J_b (1-2 \eta)  \prod_{j =1}^{k-1}    \vec\theta_{kb+j}^{(\pi)} }-
		\frac{d(k-1)}{k}\Lambda\bc{ 1+  \vec J (1-2 \eta)    \prod_{j=1}^k \vec\theta_j^{(\pi)}}}.
		\end{align*}
Then
	$$\lim_{n \to \infty} \frac{1}{n}I (\SIGMA^*, \G^*)  =
	 \left(1 + d/{k} \right)  \ln 2 +\eta \ln \eta +(1-\eta)\ln(1-\eta)- \cI(k,d, \eta).$$
\end{theorem}

Kumar, Pakzad, Salavati, and Shokrollahi~\cite{kumar2010phase} conjectured the existence of a threshold density below which the normalized mutual information between $\SIGMA^*$ and $\vec y$ conditioned on $\G$, $\frac{1}{n} I(\SIGMA^* , \vec y | \G)$, is \whp\ strictly less than the capacity of the binary symmetric channel with error probability $\eta$.
Since a simple calculation shows that $I(\SIGMA^*, \G^*)$ coincides with the conditional mutual information $I(\SIGMA^* , \vec y | \G)$,
the result of  Abbe and Montanari~\cite{Abbe} that $\lim_{n \to \infty} \frac{1}{n}I (\SIGMA^*, \G^*)$ exists implies this conjecture for even $k$.
Theorem~\ref{thm:noisyXor} extends this result to all $k$. 
Moreover, Montanari~\cite{MontanariBounds} showed that for even $k$ the above formula gives an upper bound on the mutual information and extends to LDGM codes with given variable degrees.
He conjectured that this bound is tight. 
\Thm~\ref{thm:noisyXor} proves the conjecture for all~$k$ for the technically convenient case of Poisson variable degrees.
The LDGM coding model also appears in cryptography and hardness-of-approximation as the problem $k-\mathrm{LIN}(\eta)$ or planted noisy $k$-XOR-SAT (e.g., \cite{alekhnovich2003more,applebaum2010public,Feldman2015})
and the gap between the algorithmic and the information-theoretic threshold is closely related to deep questions in computational complexity~\cite{alekhnovich2003more,Feige}.

\section{The cavity method, statistical inference and the information-theoretic threshold}\label{Sec_general}

\noindent
In this section we state the main results of this paper about statistical inference problems and their connections to phase transitions.
\Thm s~\ref{Thm_stat} and~\ref{Cor_stat} below provide general exact formulas for the mutual information in
inference problem such as the stochastic block model or the LDGM model.
Then in \Thm s~\ref{Thm_G} and~\ref{Cor_cond} we establish the existence of an information-theoretic threshold that connects the statistical inference problem with
the condensation phase transition.
Let us begin with the general setup and the results for the mutual information.

\subsection{The mutual information}
The protagonist of this paper, the {\em teacher-student scheme}~\cite{LF}, can be viewed as a generalization of the LDGM problem from \Sec~\ref{sec:parity}.
We generalize the set $\{\pm 1\}$ to an arbitrary finite set $\Omega$ of possible values that we call {\em spins}
and the parity checks to an arbitrary finite collection $\Psi$ of {\em weight functions} $\Omega^k\to(0,2)$ of some fixed arity $k\geq2$.
The choice of the upper bound $2$ is convenient but somewhat arbitrary as $(0,\infty)$-functions could just be rescaled to $(0,2)$.
But the assumption that all weight functions are strictly positive is important to ensure that all the quantities that we introduce in the following are well-defined.
There is a fixed prior distribution $p$ on $\Psi$ and we write $\PSI$ for a random weight function chosen from $p$.
We have a {\em factor graph} $G=(V,F,(\partial a)_{a\in F},(\psi_a)_{a\in F})$ composed of a set 
$V=\{x_1,\ldots,x_n\}$ of {\em variable nodes}, a set $F=\{a_1,\ldots,a_m\}$ of {\em constraint nodes}, and for each $a\in F$,
an ordered $k$-tuple $\partial a=(\partial_1a,\ldots,\partial_ka)\in V^k$ of neighbors and a weight function $\psi_a\in\Psi$. 
We may visualize $G$ as a bipartite graph with edges going between variable and constraint nodes, although we keep in mind
that the neighborhoods of the constraint nodes are ordered.

\begin{definition}\label{Def_teacher}
Let $n$, $m$ be integers and set $V=\{x_1,\ldots,x_n\}$ and $F=\{a_1,\ldots,a_m\}$.
The \emph{teacher-student scheme} is the distribution on assignment/factor graph pairs induced by the following experiment.
\begin{description}
\item[TCH1] An assignment $\SIGMA^*_n\in\Omega^V$, the {\em ground truth}, is chosen uniformly at random.
\item[TCH2] Then
obtain the random factor graph $\G^*(n,m,p,\SIGMA_n^*)$ with variable nodes $V$ and constraint nodes $F$
by drawing independently for $j=1,\ldots,m$
the neighborhood and the weight function  from the joint distribution
	\begin{align}\label{eqTeacher}
	\pr\brk{\partial a_j=(y_1,\ldots,y_k),\psi_{a_j}=\psi}
			&=
		n^{-k} \xi^{-1}p(\psi)\psi(\SIGMA_n^*(y_1,\ldots,\SIGMA_n^*(y_k))
			\quad\mbox{for }y_1,\ldots,y_k\in V,\ \psi\in\Psi,\quad\mbox{where}\\
		\xi=\xi(p)&=|\Omega|^{-k}\sum_{\tau\in\Omega^k}\Erw[\PSI(\tau)].\label{eqxi}
	\end{align}
\end{description}
\end{definition}

\noindent
The idea is that a ``teacher'' chooses $\SIGMA_n^*$ and sets up a random factor $\G^*(n,m,p,\SIGMA_n^*)$
such that for each constraint node the weight function and the adjacent variable nodes are chosen from the joint distribution (\ref{eqTeacher}) induced by the ground truth.
Specifically, the probability of a weight function/variable node combination is proportional to the prior $p(\psi)$ times the weight
$\psi(\SIGMA_n^*(y_1),\ldots,\SIGMA_n^*(y_k))$ of the corresponding spin combination under the ground truth.
The teacher hands the random factor graph $\G^*(n,m,p,\SIGMA_n^*)$, but not the ground truth itself, to an imaginary ``student'', whose
task it is to infer as much information about $\SIGMA_n^*$ as possible.
Hence, the key quantity associated with the model is the mutual information of the ground truth and the random factor graph
defined as in (\ref{eqMutualInf}).
Let us briefly write $\SIGMA^*=\SIGMA_n^*$.
Moreover, letting $\vec m=\Po(dn/k)$ we use the shorthand $\G^*=\G^*(n,\vec m,p,\SIGMA^*)$.

\begin{figure}
\begin{description}
\item[SYM] For all $\sigma,\sigma'\in\Omega$, $i,i'\in[k]$ we have
	$\sum_{\tau\in\Omega^k}
		\Erw[\PSI(\tau_1,\ldots,\tau_k)]\cdot[\vecone\{\tau_i=\sigma\}-\vecone\{\tau_{i'}=\sigma'\}]=0.$
\item[BAL] 
	The function
	$\mu\in\cP(\Omega)\mapsto\sum_{\sigma\in\Omega^k}\Erw[\PSI(\sigma_1,\ldots,\sigma_k)]\prod_{i=1}^k\mu(\sigma_i)$
	is concave and attains its maximum at the uniform distribution.
\item[POS] For all $\pi,\pi'\in\cPcent(\Omega)$  and for every $l\geq2$ the following is true.
	With $\vec\mu_1^{(\pi)},\vec\mu_2^{(\pi)},\ldots$ chosen from $\pi$ and
	$\vec\mu_1^{(\pi')},\vec\mu_2^{(\pi')},\ldots$ from $\pi'$ and $\PSI\in\Psi$ chosen from $p$, all mutually independent,
	we have
	\begin{align*}\nonumber
	\Erw\bigg[\bigg(1-\sum_{\sigma\in\Omega^k}\PSI(\sigma)\prod_{j=1}^k\vec\mu_j^{(\pi)}(\sigma_j)\bigg)^l&
		+(k-1)\bigg(1-\sum_{\sigma\in\Omega^k}\PSI(\sigma)\prod_{j=1}^k\vec\mu_j^{(\pi')}(\sigma_j)\bigg)^l\\
		&-\sum_{i=1}^k\bigg(1-\sum_{\sigma\in\Omega^k}\PSI(\sigma)\vec\mu_i^{(\pi)}(\sigma_i)
			\prod_{j\in[k]\setminus\{ i\}}\vec\mu_j^{(\pi')}(\sigma_j)\bigg)^l\bigg]
		\geq0.
	\end{align*}
\end{description}
\caption{The assumptions {\bf SYM}, {\bf BAL} and {\bf POS}.}\label{Fig_assumptions}
\end{figure}

The cavity method predicts that the mutual information $\frac{1}{n}I(\SIGMA^*,\G^*)$  converges to the solution of a certain stochastic optimization problem.
We are going to prove this conjecture under the three general conditions shown in Figure~\ref{Fig_assumptions}.
The first condition {\bf SYM} requires that on the average the weight functions prefer all values $\sigma\in\Omega$ the same.
Condition {\bf BAL} requires that on average the weight functions do not prefer an imbalanced distribution of values
	(e.g., that $\sigma_1,\ldots,\sigma_k$ all take the same value).
The third condition {\bf POS} can be viewed as a convexity assumption.
Crucially, all three assumptions can be checked {\em solely} in terms of the prior distribution $p$ on weight functions. In \Sec~\ref{Sec_applications} we will see that the three assumptions hold in many important examples.
These include LDGM codes or variations thereof where the parity checks are replaced by
$k$-SAT clauses or by graph or hypergraph $q$-coloring  constraints for any $q\geq2$, and thus in particular the Potts antiferromagnet.

\begin{theorem}\label{Thm_stat}
Assume that {\bf SYM}, {\bf BAL} and {\bf POS} hold.
With
$\vec\gamma=\Po(d)$, $\PSI_1,\PSI_2,\ldots\in\Psi$ chosen from $p$,
$\vec\mu_1^{(\pi)},\vec\mu_2^{(\pi)},\ldots$ chosen from $\pi\in\cPcent(\Omega)$ and $\vec h_1,\vec h_2,\ldots\in[k]$ chosen uniformly,
 all mutually independent, let\/%
	\begin{align}\label{eqMyBethe}
	\cB(d,\pi)&=
	\Erw\brk{\frac{\xi^{-\vec\gamma}}{|\Omega|}
			\Lambda\bc{\sum_{\sigma\in\Omega}\prod_{i=1}^{\vec\gamma}\sum_{\tau\in\Omega^k}\vecone\{\tau_{\vec h_i}=\sigma\}\PSI_b(\tau)\prod_{j\neq {\vec h_i}}\vec\mu_{ki+j}^{(\pi)}(\tau_j)}
	-\frac{d(k-1)}{k\xi}\Lambda\bc{\sum_{\tau\in\Omega^k}\PSI(\tau)\prod_{j=1}^k\vec\mu_j^{(\pi)}(\tau_j)}}.
	\end{align}
Then for all $d>0$ we have
	\begin{equation*}
	\lim_{n\to\infty}\frac1n I(\SIGMA^*,\G^*)=\ln|\Omega|+
		\frac{ d}{k \xi |\Omega|^k}  \sum_{\tau \in \Omega^k} \Erw [\Lambda( \PSI(\tau)) ]
		-\sup_{\pi\in\cP^2_*(\Omega)}\cB(d,\pi).
	\end{equation*}
\end{theorem}

\noindent
\Thm~\ref{thm:noisyXor} follows immediately from \Thm~\ref{Thm_stat}
by verifying {\bf SYM}, {\bf BAL} and {\bf POS} for the LDGM setup (see \Sec~\ref{Sec_thm:noisyXor}).

\begin{remark}
The expression $\cB(d,\pi)$ is closely related to the ``Bethe free energy'' from physics~\cite{MM}, which
is usually written in terms of $|\Omega|$ different distributions $(\pi_\omega)_{\omega\in\Omega}$ on $\cP(\Omega)$ rather than just a single $\pi$.
But thanks to the `Nishimori property' (\Prop~\ref{Lemma_Nishi} below) we can rewrite the formula in the compact form displayed in \Thm~\ref{Thm_stat}.
\end{remark}

\subsection{Belief Propagation}\label{Sec_densityEvolution}
We proceed to establish that the stochastic optimization problem~(\ref{Thm_stat}) can be cast as the problem of finding
an optimal distribution of Belief Propagation messages on a random tree.
To be precise, let $\pi\in\cPcent(\Omega)$ and consider the following experiment that sets up a random tree of height two
and uses $\pi$ to calculate a ``message'' emanating from the root.
The construction ensures that the tree has asymptotically the same distribution as the depth-two neighborhood of a random variable node in $\G^*$.
\begin{description}
\item[BP1] The root is a variable node $r$ that receives a uniformly random spin $\SIGMA^\star(r)$.
\item[BP2] The root has a random number $\vec\gamma=\Po(d)$ of constraint nodes $a_1,\ldots,a_{\vec\gamma}$ as children, 
			and independently for each child $a_i$ the root picks a random index $h_i\in[k]$.
\item[BP3] Each $a_i$ has $k-1$ variable nodes $(x_{ij})_{j\in[k]\setminus\{h_i\}}$ as children and
		independently for each $a_i$ we choose a weight function $\PSI_{a_i}\in\Psi$ and spins $\SIGMA^\star(x_{ij})\in\Omega$ from the distribution
			$$\pr\brk{\PSI_{a_i}=\psi,\SIGMA^\star(x_{ij})=\sigma_{ij}}=\frac{
				p(\psi)\psi(\sigma_{i1},\ldots,\sigma_{ih_i-1},\SIGMA^\star(r),\sigma_{ih_i+1},\ldots,\sigma_{ik})}
					{\sum_{\psi'\in\Psi,\tau_{ij}\in\Omega}
					p(\psi')\psi'(\tau_{i1},\ldots,\tau_{ih_i-1},\SIGMA^\star(r),\tau_{ih_i+1},\ldots,\tau_{ik})}.$$
\item[BP4] For each $x_{ij}$ independently choose $\vec\mu_{x_{ij}}\in\cP(\Omega)$ from the distribution
	$|\Omega|\mu(\SIGMA^\star(x_{ij}))\dd\pi(\mu)$. 
\item[BP5] Finally, obtain $\vec\mu_r$  via the {\em Belief Propagation equations}:
	\begin{align*}
	\vec\mu_{a_i}(\sigma_{h_i})&=\sum_{\tau\in\Omega^k}\vecone\{\tau_{h_i}=\sigma_{h_i}\}\PSI_{a_i}(\tau)
			\prod_{j\neq h_i}\vec\mu_{x_{ij}}(\tau_{j}),&
	\vec\mu_{r}(\sigma)&=\frac{\prod_{i=1}^{\vec\gamma}\vec\mu_{a_i}(\sigma)}{\sum_{\tau\in\Omega}\prod_{i=1}^{\vec\gamma}\vec\mu_{a_i}(\tau)}.
	\end{align*}
\end{description}
Let $\cT_d(\pi)$ be the distribution (over all the random choices in {\bf BP1--BP4}) of $\vec\mu_r$  and let $$\cPfix=\{\pi\in\cPcent(\Omega):\cT_d(\pi)=\pi\}.$$
The stochastic fixed point problem $\cT_d(\pi)=\pi$ is known as the {\em density evolution equation} in physics~\cite{MM}.

\begin{theorem}\label{Cor_stat}
If {\bf SYM}, {\bf BAL} and {\bf POS} hold, then
	$
	\sup_{\pi\in\cP^2_*(\Omega)}\cB(d,\pi)=
		\sup_{\pi\in\cPfix}\cB(d,\pi).
	$
\end{theorem}

\noindent
\Thm~\ref{Thm_stat}  reduces a question about an infinite sequence of random factor graphs, one for each $n$, to a single stochastic optimization problem,
thereby verifying the key assertion of the replica symmetric cavity method.
Further, \Thm~\ref{Cor_stat} shows that this optimization problem can be viewed as the task of finding the dominant Belief Propagation fixed
	point on a Galton-Watson tree.
Extracting further explicit information (say, an approximation of the mutual information to seven decimal places or an asymptotic formula)
will require application-specific considerations.
But there are standard techniques available for studying stochastic fixed point equations analytically (such as the contraction method~\cite{Ralph})
as well as the numerical `population dynamics' heuristic~\cite{MM}.
Since $\cB(d,\pi)$ will occur in \Thm s~\ref{Thm_G} and~\ref{Cor_cond} as well, 
\Thm~\ref{Cor_stat} implies that those results can be phrased in terms of $\cPfix$.

\subsection{The information-theoretic threshold}
The teacher-student scheme immediately gives rise to the following question:
does the factor graph $\G^*$ reveal any discernible trace
of the ground truth at all?
To answer this question, we should compare $\G^*$ with a ``purely random'' null model.
This model is easily defined.

\begin{definition}\label{Def_null}
With $\Omega,p,V=\{x_1,\ldots,x_n\}$ and $F=\{a_1,\ldots,a_m\}$ 
as before, obtain 
$\G(n,m,p)$ by performing the following for every constraint $a_j$ independently:
	choose $\partial a_j\in V^k$ uniformly and independently sample $\psi_{a_j}\in\Psi$ from $p$.
With $\vec m=\Po(dn/k)$ we abbreviate $\G=\G(n,\vec m,p)$.
\end{definition}

But what corresponds to the ground truth in this null model?
Any factor graph $G$ induces a distribution on the set of assignments called the {\em Gibbs measure},
 defined by
	\begin{align}\label{eqGibbs}
	\mu_G(\sigma)&=\frac{\psi_G(\sigma)}{Z(G)}\quad\mbox{where}
		\quad\psi_G(\sigma)=\prod_{a\in F}\psi_a(\sigma(\partial_1 a),\ldots,\sigma(\partial_ka))
		\quad\mbox{for $\sigma\in \Omega^V$ and }
		Z(G)=\sum_{\tau\in\Omega^V}\psi_G(\tau).
	\end{align}
Thus, the probability of $\sigma$ is proportional to the product of the weights that the constraint nodes assign to $\sigma$.
Thinking of $\mu_G$ as the ``posterior distribution'' of the (actual or fictitious) ground truth given $G$ and
writing $\SIGMA=\SIGMA_G$ for a sample from $\mu_G$, we 
quantify the distance of the distributions $(\G^*, \SIGMA^*)$ and $(\G,\SIGMA_{\G})$ by the
	{\em Kullback-Leibler divergence}
	\begin{align*}
	\KL{\G^*,\SIGMA^*}{\G, \SIGMA_{\G}}&=\sum_{G,\sigma}{\pr\brk{\G^*=G,\SIGMA^*=\sigma}}
		\ln\frac{\pr\brk{\G^*=G,\SIGMA^*=\sigma}}{\pr\brk{\G=G, \SIGMA_{\G} =\sigma}}.
	\end{align*}
While it might be possible that $\KL{\G^*, \SIGMA^*}{\G,\SIGMA_{\G}}=o(n)$ for small $d$,  $\G^*$ should evince an imprint 
of $\SIGMA^*$ for large enough $d$, and thus we should have $\KL{\G^*, \SIGMA^*}{\G,\SIGMA_{\G}}=\Omega(n)$.
The following theorem pinpoints the precise \emph{information-theoretic threshold} at which this occurs.
Recall $\cB(d,\pi)$ from \Thm~\ref{Thm_stat}.

\begin{theorem}\label{Thm_G}
Suppose that $p,\Psi$ satisfy {\bf SYM}, {\bf BAL} and {\bf POS} and let
	\begin{align}\label{eqThm_G_infTh}
	\dinf&\textstyle=\inf\cbc{d>0:\sup_{\pi\in\cP_*^2(\Omega)}\cB(d,\pi)>
		(1-d)\ln|\Omega|+\frac{d}k\ln\sum_{\sigma\in\Omega^k}\Erw[\PSI(\sigma)]}.
	\end{align}
Then
	\begin{align}\label{eqThm_G}
	\lim_{n \to \infty} \frac{1}{n}\KL{\G^*, \SIGMA^*}{\G,\SIGMA_{\G}}&=0
		&&\mbox{if $d<\dinf$,}\\
	\liminf_{n \to \infty} \frac{1}{n}\KL{\G^*, \SIGMA^*}{\G,\SIGMA_{\G}}&>0 
		&&\mbox{if $d>\dinf$.}\nonumber
	\end{align}
\end{theorem}

\noindent
The first scenario (\ref{eqThm_G}) provides an extension of the ``quiet planting'' method from~\cite{Barriers,quiet}
to the maximum possible range of $d$.
This argument has been used in order to investigate aspects such as the spatial mixing properties
	of the ``plain'' random factor graph model $\G$ by way of the model $\G^*$.
Moreover, \Thm~\ref{Thm_G} casts light on statistical inference problems, and
in \Sec~\ref{Sec_sbm} we will see how \Thm~\ref{Thm_SBM} follows from \Thm~\ref{Thm_G}.

\subsection{The condensation phase transition}
The ``null model'' $\G$ from \Thm~\ref{Thm_G} is actually a fairly general version of random graph models
that have been studied extensively in their own right in physics (as ``diluted mean-field models'') as well as in combinatorics.
The key quantity associated with such a model is {$-\Erw[\ln Z(\G)]$}, the \emph{free energy}.
Unfortunately, computing the free energy can be fiendishly difficult due to the log inside the expectation.
By contrast, calculating $\Erw[Z(\G(n,m,p))]$ is straightforward: the assumption {\bf BAL} and a simple application of Stirling's formula yield
	$$\ln\Erw[Z(\G(n,m,p))]=n\ln|\Omega|+m\ln\sum_{\sigma\in\Omega^k}\frac{\Erw[\PSI(\sigma)]}{|\Omega|^k}+o(n+m).$$
As Jensen's inequality implies
$\Erw[\ln Z(\G(n,m,p))]\leq\ln\Erw[Z(\G(n,m,p))]$, we
obtain the {\em first moment bound}:
	\begin{equation}\label{eqJensen}
	{- \frac{1}{n} \Erw[\ln Z(\G)]\geq(d-1)\ln|\Omega|-\frac{d}k\ln\sum_{\sigma\in\Omega^k}\Erw[\PSI(\sigma)]+o(1)\qquad
		\mbox{for all $d>0$.}}
	\end{equation}
For many important examples (\ref{eqJensen}) is satisfied with equality for small enough $d>0$ (say, below the giant component threshold; cf.\ \Sec~\ref{Sec_col}).
Indeed, a great amount of rigorous work effectively deals with estimating the largest $d$ for which (\ref{eqJensen}) is tight in specific models
	(e.g., \cite{nae,AchMooreHyp2,AchNaor,ANP,Cond,Greenhill}).
The second moment method provides a sufficient condition:
if $d$ is such that $\Erw[Z(\G)^2]=O(\Erw[Z(\G)]^2)$, then (\ref{eqJensen}) holds with equality.
However, this condition is neither necessary nor easy to check.
But the precise answer follows from \Thm~\ref{Thm_G}. 

\begin{theorem}\label{Cor_cond}
Suppose that $p,\Psi$ satisfy {\bf SYM}, {\bf BAL} and {\bf POS}.
Then
	\begin{align*}
	\lim_{n\to\infty}-\frac1n\Erw[\ln Z(\G)]&=(d-1)\ln|\Omega|-\frac{d}k\ln\sum_{\sigma\in\Omega^k}\Erw[\PSI(\sigma)]
		&\mbox{for all $d<\dinf$},\\
	\limsup_{n\to\infty}-\frac1n\Erw[\ln Z(\G)]&<(d-1)\ln|\Omega|-\frac{d}k\ln\sum_{\sigma\in\Omega^k}\Erw[\PSI(\sigma)]
		&\mbox{for all $d>\dinf$}.
	\end{align*}
\end{theorem}

\noindent
Clearly, the function
	$$d\in(0,\infty)\mapsto(d-1)\ln|\Omega|-\frac{d}k\ln\sum_{\sigma\in\Omega^k}\Erw[\PSI(\sigma)]$$
is analytic.
Thus, if $\dinf>0$, then either 
	$\lim_{n\to\infty}-\frac1n\Erw[\ln Z(\G)]$ does not exist in a neighborhood of $\dinf$
	or the function $d\mapsto\lim_{n\to\infty}-\frac1n\Erw[\ln Z(\G)]$ is non-analytic at $\dinf$.
Hence, verifying an important prediction from~\cite{pnas},
 \Thm~\ref{Cor_cond} shows that if $\dinf>0$, then a phase transition occurs at $\dinf$, called the {\em condensation phase transition} in physics.

In \Sec s~\ref{Sec_Potts} and~\ref{Sec_graphcol} we will derive \Thm s~\ref{Thm_Potts} and~\ref{Thm_col} from \Thm~\ref{Cor_cond}.
While proving \Thm~\ref{Thm_Potts} from \Thm~\ref{Cor_cond} is fairly straightforward, \Thm~\ref{Thm_col} requires a bit of work.
This is because \Thm~\ref{Cor_cond} assumes that all weight functions $\psi\in\Psi$ are strictly positive, which precludes
hard constraints like in the graph coloring problem.
Nonetheless, in \Sec~\ref{Sec_graphcol} we show that these hard constraints, corresponding to $\beta=\infty$ in (\ref{eqAntiPotts}), can be dealt with by considering
the Potts antiferromagnet for finite values of $\beta$ and taking the limit $\beta\to\infty$.
We expect that this argument will find other applications.

\subsection{Discussion and related work}\label{sec:discussionRelated}
\Thm s~\ref{Thm_stat}, \ref{Thm_G} and \ref{Cor_cond} establish the physics predictions under modest
assumptions that only refer to the prior distribution of the weight functions, i.e., the `syntactic' definition of the model.
The proofs provide a conceptual vindication of the replica symmetric version of the cavity method.

Previously the validity of the physics formulas was known in any generality only under the assumption that the factor graph models satisfies
the Gibbs uniqueness condition, a very strong spatial mixing assumption~\cite{Victor,COPS,DM,dembo}.
Gibbs uniqueness typically only holds for very small values of $d$.
Additionally, under weaker spatial mixing conditions it was known that the free energy in random graph models is given by {\em some}
Belief Propagation fixed point~\cite{Will,dembo}.
However,  there may be infinitely many fixed points, and it was not generally known that the correct one is the maximizer of the functional $\cB(d,\nix)$.
In effect, it was not possible to derive the formula for the free energy or, equivalently, the mutual information, from such results.
Specifically, in the case of the teacher-student scheme  Montanari~\cite{Andrea} proved (under certain assumptions)
that the Gibbs marginals of $\G^*$ correspond to a Belief Propagation fixed point
as in \Sec~\ref{Sec_densityEvolution}, whereas \Thm~\ref{Cor_stat} identifies the particular fixed point that maximizes the functional $\cB(d,\nix)$
as the relevant one.

Yet the predictions of the replica symmetric cavity method have been verified in several specific examples.
The first ones were the ferromagnetic Ising/Potts model~\cite{demboPotts,dembo}, where the proofs exploit model-specific monotonicity/contraction properties.
More recently, the ingenious {\em spatial coupling} technique has been used to prove replica symmetric predictions
in several important cases, including low-density parity check codes~\cite{GMU}.
Indeed, spatial coupling provides an alternative probabilistic construction of, e.g., codes with excellent algorithmic properties~\cite{KRU}.
Yet the method falls short of providing a {wholesale justification} of the cavity method as
a potentially substantial amount of individual ingredients is required for each application (such as problem-specific algorithms~\cite{Dimitris}).

Subsequently to the posting of a first version of this paper on arXiv, and independently, Lelarge and Miolane~\cite{Lelarge} posted a paper on recovering a low rank matrix under a perturbation with Gaussian noise. They use some similar ingredients as we do to prove an upper bound on the mutual information matching the lower bound of~\cite{KXZ}. This setting is conceptually simpler as the infinite-dimensional stochastic optimization problem reduces to a one-dimensional optimization problem due to  central limit theorem-type behavior in the dense graph setting.

The random factor graph models that we consider in the present paper are of \Erdos-\Renyi\ type, i.e., the
constraint nodes choose their adjacent variable nodes independently.
In effect, the variable degrees are asymptotically Poisson with mean $d$.
While such models are very natural,
models with given variable degree distributions are of interest in some applications, such as error-correcting codes (e.g.~\cite{MontanariBounds}).
Although we expect that the present methods extend to models with (reasonable) given
 degree distributions,  here we confine ourselves to the Poisson case for the sake of clarity.
Similarly, the assumptions {\bf BAL}, {\bf SYM} and {\bf POS}, and the strict positivity of the constraint functions strike a balance between generality and convenience.
While these conditions hold in many cases of interest, {\bf BAL} fails for the ferromagnetic Potts model,
	which is why \Thm~\ref{Thm_SBM} does not cover the \asso\ block model.
Anyhow {\bf BAL}, {\bf SYM} and {\bf POS} are (probably) not strictly necessary for our results to hold and our methods to go through, 
a point that we leave to future work.

A further open problem is to provide a rigorous justification of
the more intricate `replica symmetry breaking' (1RSB) version of the cavity method.
The 1RSB version appears to be necessary to pinpoint, e.g., the $k$-SAT or $q$-colorability thresholds for $k\geq3$, $q\ge 3$ respectively.
Currently there are but a very few examples where predictions from the 1RSB cavity method have been established rigorously~\cite{DSS1,DSS2,SSZ}, 
the most prominent one being the proof of the $k$-SAT conjecture for large $k$~\cite{DSS3}.
That said, the upshot of the present paper is that for teacher-student-type  problems as well as for the purpose of finding the condensation
threshold, the replica symmetric cavity method is provably sufficient.

Additionally, the ``full replica symmetry breaking'' prediction has been established rigorously in the Sherrington-Kirkpatrick model
on the complete graph~\cite{Talagrand}.
Subsequently Panchenko~\cite{PanchenkoBook} proposed a different proof that combines the interpolation method with the so-called
`Aizenman-Sims-Starr' scheme, an approach that he attempted to extend to sparse random graph models~\cite{Panchenko}.
We will apply the interpolation method and the Aizenman-Sims-Starr scheme as well, but crucially exploit that the connection with the
statistical inference formulation of random factor graph models adds substantial power to these arguments.

\subsection{Preliminaries and notation}\label{Sec_prelims}
Throughout the paper we let $\Omega$ be a finite set of `spins' and fix an integer $k\geq2$.
Moreover, let $V=V_n=\{x_1,\ldots,x_n\}$ and $F_m=\{a_1,\ldots,a_m\}$ be sets of variable and constraint nodes
and we write $\SIGMA_n^*$ for a uniformly random map $V_n\to\Omega$.
Further, $\vec m=\vec m_d=\vec m_d(n)$ denotes a random variable with distribution $\Po(dn/k)$.

The $O(\nix)$-notation refers to the limit $n\to\infty$ by default.
In addition to the usual symbols $O(\nix)$, $o(\nix)$, $\Omega(\nix)$, $\Theta(\nix)$ we use $\tilde O(\nix)$ to hide logarithmic factors.
Thus, we write $f(n)=\tilde O(g(n))$ if there is $c>0$ such that for large enough $n$ we have $|f(n)|\leq g(n)\ln^cn$.
Furthermore, if $(E_n)_n$ is a sequence of events, then $(E_n)_n$ holds {\em with high probability} (`\whp') if $\lim_{n \to \infty} \pr[E_n] =1$.

Let $(\mu_n)_n,(\nu_n)_n$ be  sequences of probability distributions on measurable spaces $(\cX_n)_n$.
We call $(\mu_n)_n$ {\em contiguous} with respect to $(\nu_n)_n$ if for any $\eps>0$
there exist $\delta>0$ and $n_0>0$ such that for all $n>n_0$ for every event $\cE_n$ on $\Omega_n$ with $\nu_n(\cE_n)<\delta$ we have $\mu_n(\cE_n)<\eps$.
The sequences $(\mu_n)_n,(\nu_n)_n$ are {\em mutually contiguous} if $(\mu_n)_n$ is contiguous w.r.t.\ $(\nu_n)_n$ and 
$(\nu_n)_n$ is contiguous w.r.t.\ $(\mu_n)_n$.

If $X,Y$ are finite sets and $\sigma:X\to Y$ is a map, then we write $\lambda_\sigma\in\cP(Y)$ for the empirical distribution of $\sigma$.
That is, for any $y\in Y$ we let $\lambda_\sigma(y)=|\sigma^{-1}(y)|/|X|$.
Moreover, 
for assignments $\sigma,\tau:X\to Y$ we let $\sigma\triangle\tau=\{x\in X:\sigma(x)\neq\tau(x)\}.$

When defining probability distributions we use the $\propto$-symbol to signify the required normalization.
Thus, we use $\pr\brk{X=x}\propto q_x$ for all $x\in X$ as shorthand for
	$\pr\brk{X=x}=q_x/\sum_{y\in\cX}q_y$ for all $x\in\cX$, provided that $\sum_{y\in\cX}q_y>0$.
If $\sum_{y\in\cX}q_y=0$ the $\propto$-symbol defines the uniform distribution on $X$.

Suppose that $\cX$ is a finite set.
Given a probability distribution $\mu$ on $\cX^n$ we write
$\SIGMA_\mu,\SIGMA_{1,\mu},\SIGMA_{2,\mu},\ldots$ for independent samples from $\mu$.
Where $\mu$ is apparent from the context we drop it from the notation.
Further, we write $\bck{X(\SIGMA)}_\mu$ for the average of a random variable $X:\cX^n\to\RR$ with respect to $\mu$.
Thus, $\bck{X(\SIGMA)}_\mu=\sum_{\sigma\in\cX^n}X(\sigma)\mu(\sigma)$.
Similarly, if $X:(\cX^n)^l\to\RR$, then
	$$\bck{X(\SIGMA_1,\ldots,\SIGMA_l)}_\mu=\sum_{\sigma_1,\ldots,\sigma_l\in\cX^n}X(\sigma_1,\ldots,\sigma_l)\prod_{j=1}^l\mu(\sigma_j).$$
If $\mu=\mu_G$ is the Gibbs measure induced by a factor graph $G$, then we use the abbreviation $\bck{\nix}_G=\bck{\nix}_{\mu_G}$.

If $\cX,I$ are finite sets, $\mu\in\cP(\cX^I)$ is a probability measure and $i\in I$, then we write $\mu_i$ for the marginal distribution of the $i$-coordinate.
That is,
 	$\mu_i(\omega)=\sum_{\sigma:I\to\cX}\vecone\{\sigma(i)=\omega\}\mu(\sigma)$ for any $\omega\in\cX$.
Similarly, if $J\subset I$, then 
$\mu_J(\omega)=\sum_{\sigma:I\to\cX}\vecone\{\sigma|_J=\omega\}\mu(\sigma)$ for any $\omega:J\to\cX$
denotes the joint marginal distribution of the coordinates $J$.
If $J=\{i_1,\ldots,i_l\}$ we briefly write $\mu_{i_1,\ldots,i_l}$ rather than $\mu_{\{i_1,\ldots,i_l\}}$.
Further, a measure $\nu\in\cP(\cX^I)$ is {\em $\eps$-symmetric} if 
	\begin{align*}
	\sum_{i,j\in I}\TV{\nu_{i,j}-\nu_i\tensor\nu_j}&<\eps |I|^2.
	\end{align*}
More generally, $\nu$ is {\em $(\eps,l)$-symmetric} if 
\begin{align*}
	\sum_{i_1,\ldots,i_l\in I}\TV{\mu_{i_1,\ldots,i_l}-\mu_{i_1}\tensor\cdots\tensor\mu_{i_l}}<\eps |I|^l.
	\end{align*}
Crucially, in the following lemma $\eps$ depends on $\delta,l,\cX$ only, but not on $\mu$ or $I$.

\begin{lemma}[{\cite{Victor}}]\label{Lemma_lwise}
For any $\cX\neq\emptyset$, $l\geq3$, $\delta>0$ there is $\eps>0$ such that for all $I$ of size $|I|>1/\eps$ the following is true.
	\begin{quote}
	If $\mu\in\cP(\cX^I)$ is $\eps$-symmetric, then $\mu$ is $(\delta,l)$-symmetric.
	\end{quote}
\end{lemma}	

The total variation norm is denoted by $\TV\nix$.
Furthermore, for a finite set $\cX$ we identify the space $\cP(\cX)$ of probability distributions on $\cX$ with the standard simplex in $\RR^{\cX}$
and endow $\cP(\cX)$ with the induced topology and Borel algebra.
The space $\cP^2(\cX)$ of probability measures on $\cP(\cX)$ carries the topology of weak convergence.
Thus, $\cP^2(\cX)$ is a compact Polish space.
So is the closed subset $\cP^2_*(\cX)$ of measures $\pi\in\cP^2(\cX)$ whose mean $\int\mu\dd\pi(\mu)$ is the uniform distribution on $\cX$.
We use the $W_1$ Wasserstein distance, denoted by $W_1(\nix,\nix)$, to metrize the weak topology on $\cP^2_*(\cX)$~\cite{billingsley,villani}.
In particular, recalling $\cB(d,\nix)$ from~(\ref{eqMyBethe}) and $\cT_d(\nix)$ from \Sec~\ref{Sec_densityEvolution}, we observe

\begin{lemma}
\label{lem:BPcontinuous}
The map $\pi\in\cP^2(\Omega) \mapsto \cT_d(\pi)$ and the functional $\pi\in\cP^2(\Omega) \mapsto  \cB(d,\pi)$  are  continuous. 
\end{lemma}
\begin{proof}
We prove this for $\cT_d( \pi)$, the proof for $ \cB(d,\pi)$ is similar.  
We need to show that for every $\eps >0$, there is $\del >0$ so that if $W_1( \pi_1, \pi_2) < \del$, then $W_1(\cT_d(\pi_1) ,\cT_d( \pi_2) ) <\eps$. 
Let $\cT_{d,\vec \gamma \le M} ( \pi)$ be the output distribution of $\cT_d(\nix)$ conditioned on the event that $\vec \gamma \le M$.  For any fixed $M$,  $\cT_{d,\vec \gamma \le M} ( \pi)$ is a continuous function of $\pi$ in the weak topology as it is the composition of a continuous function and and a product distribution on at most $M$ independent samples from $\pi$. 
Now given $\eps$, choose $M$ large enough that $\Pr[ \vec \gamma > M ]<\eps/2$, and $\del$ small enough that $W_1( \pi_1, \pi_2) < \del$ implies $W_1(\cT_{d,\vec \gamma \le M}(\pi_1) ,\cT_{d,\vec \gamma \le M}( \pi_2) ) <\eps/2$. Then 
$W_1(\cT_d(\pi_1) ,\cT_d( \pi_2) ) \le W_1(\cT_{d,\vec \gamma \le M}(\pi_1) ,\cT_{d,\vec \gamma \le M}( \pi_2) ) + \Pr [\vec \gamma > M] < \eps.$
\end{proof}

\noindent
Furthermore, for a measure $\mu\in\cP(\cX)$ we denote by $\delta_\mu\in\cP^2(\cX)$ the Dirac measure on $\mu$.

\begin{proposition}[{Glivenko--Cantelli Theorem, e.g.~\cite[\Chap~11]{rachev}}]
\label{prop:sampling}
For any finite set $\Omega$, there is a sequence $\eps_K \to 0$ as $K \to \infty$ so that the following is true. Let $\mu_1,\mu_2, \dots \in \cP(\Omega)$ be independent samples from $\pi \in \cP^2(\Omega)$ and form the empirical marginal distribution
\[ \overline \mu_K = \frac{1}{K} \sum_{i=1}^K \del_{\mu_i} .\]
Then $\Erw[ W_1(\pi, \overline \mu_K)] \le \eps_K$.
\end{proposition}

Suppose that $(\cE,\mu)$ is a probability space and that $X,Y$ are random variables on $(\cE,\mu)$ with values in a finite set $\cX$.
We recall that the {\em mutual information} of $X,Y$ is
	\begin{align*}
	I(X,Y)&=\sum_{x,y\in\cX}\mu(X=x,Y=y)\ln\frac{\mu(X=x,Y=y)}{\mu(X=x)\mu(Y=y)},
	\end{align*}
with the usual convention that $0\ln\frac00=0$, $0\ln0=0$.
Moreover, the mutual information of $X,Y$ given a third $\cX$-valued random variable $W$ is defined as
	\begin{align*}
	I(X,Y|W)&=\sum_{x,y,w\in\cX}\mu(X=x,Y=y,W=w)\ln\frac{\mu(X=x,Y=y|W=w)}{\mu(X=x|W=w)\mu(Y=y|W=w)}.
	\end{align*}	
Furthermore, we recall the {\em entropy} and the {\em conditional entropy}:
	\begin{align*}
	H(X)&=-\sum_{x\in\cX}\mu(X=x)\ln\mu(X=x),&
	H(X|Y)&=-\sum_{x,y\in\cX}\mu(X=x,Y=y)\ln\mu(X=x|Y=y).
	\end{align*}
Viewing $(X,Y)$ as a $\cX\times\cX$-valued random variable, we have the {\em chain rule}
	\begin{align*}
	H(X,Y)&=H(X)+H(Y|X).
	\end{align*}
Analogously, for $\mu\in\cP(\cX)$ we write $H(\mu)=-\sum_{x\in\cX}\mu(x)\ln\mu(x)$.

The {\em Kullback-Leibler divergence} between two probability measures $\mu,\nu$ on a finite set $\cX$ is
	\begin{align*}
	\KL{\mu}{\nu}&=\sum_{\sigma \in \cX} \mu(\sigma)\ln\frac{\mu(\sigma)}{\nu(\sigma)}.
	\end{align*}
Finally, we recall {\em Pinsker's inequality}:
	for any two probability measures $\mu,\nu\in\cP(\cX)$ we have
	\begin{align}\label{eqPinsker}
	\TV{\mu-\nu}&\leq\sqrt{\KL\mu\nu/2}
	\end{align}

\section{The replica symmetric solution}\label{Sec_rss}

\noindent
In this section we prove \Thm s~\ref{Thm_stat}, \ref{Cor_stat}, \ref{Thm_G} and~\ref{Cor_cond}.
The proofs of \Thm s~\ref{Thm_Potts}--\ref{thm:noisyXor}
follow in \Sec~\ref{Sec_applications}, along with a few other applications.

\subsection{Overview}\label{Sec_outline}
To prove  \Thm~\ref{Thm_stat} we will  provide a rigorous foundation for the ``replica symmetric calculations'' that physicists wanted to do (and have been doing) all along.
To this end we adapt, extend and generalize various ideas from prior work, some of them relatively simple, some of them quite recent and not simple at all, and
 develop several new arguments.
But in a sense the main achievement lies in the interplay of these components, i.e., how the individual {cogs assemble into a functioning clockwork}.
Putting most  details off to the following subsections, here we outline the proof strategy.
We focus on \Thm~\ref{Thm_stat},
	from which we subsequently derive \Thm~\ref{Thm_G} and \Thm~\ref{Cor_cond} in \Sec~\ref{Sec_Thm_G}.
\Thm~\ref{Cor_stat} also follows from \Thm~\ref{Thm_stat} but the proof requires additional arguments, which can be found in \Sec~\ref{Sec_prop:betheUB}.
	
The first main {ingredient} to the proof of \Thm~\ref{Thm_stat}
	is a reweighted version of the teacher-student scheme that enables us to identify the ground truth with a sample
from the Gibbs measure of the factor graph; this identity is an exact version of the ``Nishimori property'' from physics.
The Nishimori property facilitates the use of a general lemma (\Lem~\ref{Lemma_pinning} below) that shows that a slight perturbation of the factor graph 
induces a correlation decay property called ``static replica symmetry'' in physics without significantly altering the mutual information;
	due to its great generality \Lem~\ref{Lemma_pinning} should be of independent interest.
Having thus paved the way, we derive a lower bound on the mutual information via the so-called `Aizenman-Sims-Starr' scheme.
This comes down to estimating the change in mutual information if we go from a model with $n$ variable nodes to one with $n+1$ variable nodes.
The proof of the matching upper bound is based on a delicate application of the interpolation method.

\subsubsection{The Nishimori property}\label{Sec_Outline1}
The Gibbs measure $\mu_G$ of the factor graph $G$ from  (\ref{eqGibbs}) provides a proxy for the ``posterior distribution'' of the ground truth given the graph $G$.
While we will see that this is accurate in the asymptotic sense of mutual contiguity, the assumptions {\bf BAL}, {\bf SYM} and {\bf POS} do not guarantee that the
Gibbs measure $\mu_{\G^*}$ is the {\em exact} posterior distribution of the ground truth.
This is an important point for us because the calculation of the mutual information relies on subtle coupling arguments.
Hence, in order to hit the nail on the head exactly, we introduce a reweighted version of the teacher-student scheme in which
the Gibbs measure coincides with the posterior distribution for all $n$.
Specifically, instead of the uniformly random ground truth $\SIGMA_n^*$ we consider a random assignment $\hat\SIGMA_{n,m,p}$ chosen from the distribution
	\begin{align}\label{eqNishi1}
	\pr\brk{\hat\SIGMA_{n,m,p}=\sigma}&=\frac{\Erw[\psi_{\G(n,m,p)}(\sigma)]}{\Erw[Z(\G(n,m,p))]}&&(\sigma\in\Omega^V).
	\end{align}
Thus, the probability of  an assignment is proportional to its {average} weight.
Further, any specific ``ground truth'' $\sigma$ induces a random factor graph $\G^*(n,m,p,\sigma)$ with distribution
	\begin{align}\label{eqNishi2}
	\pr\brk{\G^*(n,m,p,\sigma)\in\cA}&=
		\frac{\Erw\brk{\psi_{\G(n,m,p)}(\sigma)\vecone\{\G(n,m,p)\in\cA\}}}{\Erw[\psi_{\G(n,m,p)}(\sigma)]}&\mbox{for any event $\cA$}.
	\end{align}
In words, the probability that a specific graph $G$ comes up is proportional to $\psi_G(\sigma)$.

\begin{fact}\label{Fact_teacher}
For any $n,m,p,\sigma$ the distribution (\ref{eqNishi2}) coincides with the distribution from \Def~\ref{Def_teacher} given $\SIGMA^*=\sigma$.
\end{fact}
\begin{proof}
Consider a specific factor graph $G$ with constraint nodes $a_1,\ldots,a_m$.
Since the constraint nodes of the random factor graph $\G(n,m,p)$ are chosen independently (cf.\ \Def~\ref{Def_null}), we have
	\begin{align}\label{eqFact_teacher}
	\frac{\psi_G(\sigma)}{\Erw[\psi_{\G(n,m,p)}(\sigma)]}&=\prod_{j=1}^m\frac{\psi_{a_j}(\sigma(\partial_1a_j),\ldots,\sigma(\partial_ka_j))}
		{\sum_{\psi\in\Psi}\sum_{h_1,\ldots,h_k=1}^np(\psi)\psi(\sigma(x_{h_1}),\ldots,\sigma(x_{h_k}))}.
	\end{align}
Since the experiment from \Def~\ref{Def_teacher} generates the constraint nodes $a_1,\ldots,a_m$ independently, the probability of 
obtaining the specific graph $G$ equals the r.h.s.\ of (\ref{eqFact_teacher}).
\end{proof}

\noindent
Additionally, consider the random factor graph $\hat\G(n,m,p)$ defined by
	\begin{align}\label{eqNishi3}
	\pr\brk{\hat\G(n,m,p)\in\cA}&=
		\frac{\Erw[Z(\G(n,m,p))\vecone\{\G(n,m,p)\in\cA\}]}{\Erw[Z(\G(n,m,p))]}&\mbox{for any event $\cA$},
	\end{align}
which means that we reweigh $\G(n,m,p)$ according to the partition function.
Finally, recalling that $\vec m=\Po(dn/k)$, we introduce the shorthand $\hat\SIGMA=\hat\SIGMA_{n,\vec m,p}$,
$\G^*(\hat\SIGMA)=\G^*(n,\vec m,p,\hat\SIGMA_{n,\vec m,p})$ and $\hat\G=\hat\G(n,\vec m,p)$.

\begin{proposition}\label{Lemma_Nishi}
For all factor graph/assignment pairs $(G,\sigma)$ we have
	\begin{align}\label{eqLemma_Nishi}
	\pr\brk{\hat\SIGMA=\sigma,\G^*(\hat\SIGMA)=G}&=\pr\brk{\hat\G=G}\mu_G(\sigma).
	\end{align}
Moreover, {\bf BAL} and {\bf SYM} imply that
$\hat\SIGMA$ and the uniformly random assignment $\SIGMA^*$ 
are mutually contiguous.
\end{proposition}

In words, (\ref{eqLemma_Nishi}) provides that the distributions on assignment/factor graph pairs induced by the following two experiments are identical.
\begin{enumerate}[(i)]
\item Choose $\hat\SIGMA$, then choose $\G^*(\hat\SIGMA)$. 
\item Choose $\hat\G$, then choose $\SIGMA_{\hat\G}$ from $\mu_{\hat\G}$. 
\end{enumerate}
In particular, the conditional distribution of $\hat\SIGMA$ given just the factor graph $\G^*(\hat\SIGMA)$
coincides with the Gibbs measure of $\G^*(\hat\SIGMA)$.
This can be interpreted as an exact, non-asymptotic version of what physicists call the Nishimori property (cf.~\cite{LF}).
Although $(\hat\SIGMA,\G^*(\hat\SIGMA))$ and $(\SIGMA^*,\G^*)$ are not generally identical,
 the contiguity statement from \Prop~\ref{Lemma_Nishi} ensures that both are equivalent as far as ``with high probability''-statements are concerned.
The proof of \Prop~\ref{Lemma_Nishi} can be found in \Sec~\ref{Sec_Nishi}.

To proceed, we observe that the free energy of the random factor graph is tightly concentrated.

\begin{lemma}\label{Lemma_Azuma}
Under conditions {\bf BAL} and {\bf SYM} the following is true.
There is $C=C(d,\Psi)>0$ such that
	\begin{align}\label{eqLemma_Azuma1}
	\pr\brk{|\ln Z(\hat\G)-\Erw\ln Z(\hat\G)|>t n}\leq2\exp(-t^2n/C)\qquad\mbox{ for all }t>0.
	\end{align}
The same holds with $\hat\G$ replaced by $\G^*(\hat\SIGMA)$, $\G^*(\SIGMA^*)$ or $\G$.
Moreover,
	\begin{align}\label{eqLemma_Azuma2}
	\Erw[\ln Z(\hat\G)]=\Erw[\ln Z(\G^*(\SIGMA^*))]+o(n).
	\end{align}
\end{lemma}
\begin{proof}
Because all weight functions $\psi\in\Psi$ are strictly positive, (\ref{eqLemma_Azuma1}) is immediate from Azuma's inequality.
Moreover, since
$\hat\G$ and $\G^*(\hat\SIGMA)$ are identically distributed and
 $\hat\SIGMA$ and $\SIGMA^*$ are mutually contiguous by \Prop~\ref{Lemma_Nishi},
$\hat\G$ and $\G^*(\SIGMA^*)$ are mutually contiguous as well.
Therefore, (\ref{eqLemma_Azuma2}) follows from (\ref{eqLemma_Azuma1}).
\end{proof}

\noindent
The following statement, which is an easy consequence of \Prop~\ref{Lemma_Nishi},
 reduces the task of computing $I(\SIGMA^*,\G^*)$ to that of calculating the free energy $-\Erw[\ln Z(\hat\G)]$
of the reweighted model $\hat\G$.

\begin{lemma}
\label{lem:FEmutualInfo}
If {\bf BAL} and {\bf SYM} are satisfied, then
	\begin{align}\label{eq:MIbethe}
	I(\hat\SIGMA,\G^*(\hat\SIGMA))&=-\Erw[\ln Z(\hat\G)]+
			\frac{dn}{k \xi |\Omega|^k}  \sum_{\tau \in \Omega^k} \Erw [\Lambda(\PSI(\tau))]+n\ln|\Omega|+o(n),\\
	I(\SIGMA^*,\G^*(\SIGMA^*))&=-\Erw[\ln Z(\hat\G)]+
			\frac{dn}{k \xi |\Omega|^k}  \sum_{\tau \in \Omega^k} \Erw [\Lambda(\PSI(\tau))]+n\ln|\Omega|+o(n).	
			\label{eq:MIbethe2}
	\end{align}
\end{lemma}
\begin{proof}
\Prop~\ref{Lemma_Nishi} implies that
	\begin{align}
	I(\hat\SIGMA,\G^*(\hat\SIGMA))
		&=\sum_{\hat G}\pr\brk{\hat\G=\hat G}\sum_{\sigma}\mu_{\hat\G}(\sigma) 
			\ln\frac{\mu_{\hat\G}(\sigma)}{\pr\brk{\hat\SIGMA=\sigma}}
		=H(\hat\SIGMA)-\Erw[H(\mu_{\hat\G})]. 
			\label{eqProp_interpolation_2}
	\end{align}
{Further, since $\hat\SIGMA$ and the uniformly random $\SIGMA^*$ are mutually contiguous,} we have
	\begin{align}\label{eqProp_interpolation_2a}
	H(\hat\SIGMA)&=n\ln|\Omega|+o(n).
	\end{align}
Moreover, for any factor graph $G$ we have
	\begin{align}
	H(\mu_G)&=-\sum_\sigma\mu_G(\sigma)\ln\mu_G(\sigma)=-\sum_\sigma\frac{\psi_G(\sigma)}{Z(G)}\ln\frac{\psi_G(\sigma)}{Z(G)}
		=\ln Z(G)-\bck{\ln\psi_G(\SIGMA_G)}_G
			\label{eqProp_interpolation_3}
	\end{align}
and \Prop~\ref{Lemma_Nishi} shows that
 $\Erw\bck{\ln\psi_{\hat\G}(\SIGMA_{\hat\G})}_{\hat\G}=\Erw[\ln\psi_{\G^*\bc{\hat\SIGMA}}(\hat\SIGMA)]$.
Since $\hat\SIGMA$ and $\SIGMA^*$ are mutually contiguous by \Prop~\ref{Lemma_Nishi},
we see that $|\hat\SIGMA^{-1}(\omega)|\sim n/|\Omega|$ for all $\omega\in\Omega$ \whp\
In addition, the construction (\ref{eqTeacher}) of $\G^*(\hat\SIGMA)$ is such that the individual constraint nodes $a_1,\ldots,a_{\vec m}$
are chosen independently.
Therefore, (\ref{eqTeacher}) and the fact that the constraint functions are identically distributed yields
	\begin{align}\nonumber
	\Erw\bck{\ln\psi_{\hat\G}(\SIGMA_{\hat\G})}_{\hat\G}&=
		\frac{dn}k\Erw\brk{\ln \psi_{a_1}(\hat\SIGMA)}+o(n)\\
		&=  \frac{dn}{k \xi |\Omega|^k} \sum_{\tau \in \Omega^k, \psi \in \Psi}  p(\psi) \psi(\tau)   \ln \psi(\tau) +o(n)
		\label{eqProp_interpolation_4}
		= \frac{dn}{k \xi |\Omega|^k}  \sum_{\tau \in \Omega^k} \Erw [ \Lambda (\PSI(\tau))  ] +o(n).
	\end{align}
Combining (\ref{eqProp_interpolation_2})--(\ref{eqProp_interpolation_4}) completes the proof
of (\ref{eq:MIbethe}).
Applying the same steps to $(\SIGMA^*,\G^*(\SIGMA^*))$ yields (\ref{eq:MIbethe2}).
\end{proof}

\subsubsection{Symmetry and pinning}\label{Sec_Outline2}
Hence, we are left to calculate $-\Erw[\ln Z(\hat\G)]$.
Of course, computing $\ln Z(G)$ for a given $G$ is generally a daunting task.
The plain reason is the existence of correlations between the spins assigned to different variable nodes.
To see this, write $\SIGMA_G$ for a sample drawn from $\mu_G$.
If we fix two variable nodes  $x_h,x_i$ that are adjacent to the same constraint node $a_j$,
then in all but the very simplest examples the spins $\SIGMA_G(x_h),\SIGMA_G(x_i)$ will be correlated
because $\psi_{a_j}$ `prefers' certain spin combinations over others.
By extension, correlations persists if $x_h,x_i$ are at any bounded distance.
But what if we choose a pair of variable nodes $(\vec x,\vec y)\in V\times V$ uniformly at random?
If $G$ is of bounded average degree, then the distance of $\vec x,\vec y$ will typically be as large as $\Omega(\ln|V|)$.
Hence, we may hope that $\SIGMA_G(\vec x),\SIGMA_G(\vec y)$ are `asymptotically independent'.
Formally,  let  $\mu_{G,x}$ be the marginal distribution of  $\SIGMA_G(x)$ and
$\mu_{G,x,y}$ the distribution of  $(\SIGMA_G(x),\SIGMA_G(y))$.
Then we may hope that for a small $\eps>0$,
	\begin{align}\label{eqRS}
	\frac1{|V|^2}\sum_{x,y\in V}\TV{\mu_{G,x,y}-\mu_{G,x}\tensor\mu_{G,y}}<\eps.
	\end{align}
In the terminology from \Sec~\ref{Sec_prelims}, (\ref{eqRS}) expresses that $\mu_G$ is $\eps$-symmetric.

The replica symmetric cavity method {provides a heuristic} for calculating the free energy
of random factor graph where (\ref{eqRS}) is satisfied \whp\ for some $\eps=\eps(n)$ that tends to $0$ as $n\to\infty$.
But from a rigorous viewpoint two challenges arise.
First, for a given random factor graph model, how can we possibly verify that $\eps$-symmetry holds \whp?
Second, even granted $\eps$-symmetry, how are we going to beat a rigorous path from the innocent-looking condition (\ref{eqRS})
to the mildly awe-inspiring stochastic optimization problems predicted by the physics calculations?

The following very general lemma is going to resolve the first challenge for us.
Instead of providing a way of checking, the lemma shows that a slight random perturbation likely precipitates $\eps$-symmetry.

\begin{lemma}\label{Lemma_pinning}
For any $\eps>0$ there is $T=T(\eps,\Omega)>0$ such that for every $n>T$ and every probability measure $\mu\in\cP(\Omega^n)$ the following is true.
Obtain a random probability measure $\check\MU\in\cP(\Omega^n)$ as follows.
	\begin{quote}
	Draw a sample $\check\SIGMA\in\Omega^n$ from $\mu$,
	independently choose a number $\vec\theta\in(0,T)$ uniformly at random, then
	 obtain a random set $\vU\subset[n]$ by including each $i\in[n]$ with probability $\vec\theta/n$ independently and let
		\begin{align*}
		\check\MU(\sigma)&=\frac{\mu(\sigma)\vecone\{\forall i\in\vU:\sigma_i=\check\SIGMA_i\}}
			{\mu(\{\tau\in\Omega^n:\forall i\in\vU:\tau_i=\check\SIGMA_i\})}&&(\sigma\in\Omega^n).
		\end{align*}
	\end{quote}
Then $\cMU$ is $\eps$-symmetric with probability at least $1-\eps$.
\end{lemma}

\noindent
In words, take {\em any} distribution $\mu$ on $\Omega^n$ that may or may not be $\eps$-symmetric.
Then, draw one single sample $\check\SIGMA$ from $\mu$ and obtain $\check\MU$ by ``pinning'' a typically bounded
number of coordinates $\vU$ to the particular spin values observed under $\check\SIGMA$.
Then the perturbed measure $\check\MU$ is likely $\eps$-symmetric.
(Observe that $\check\MU$ is well-defined because
	$\mu(\{\tau\in\Omega^n:\forall i\in\vU:\tau_i=\check\SIGMA_i\})\geq\mu(\check\SIGMA)>0$.)
\Lem~\ref{Lemma_pinning} is a generalization of a result of Montanari~\cite[\Lem~3.1]{Andrea} and
the proof is by extension of the ingenious information-theoretic argument from~\cite{Andrea}, parts of which go back to~\cite{Macris,MMRU,MMU}.
The proof of \Lem~\ref{Lemma_pinning} can be found in \Sec~\ref{Sec_pinning}.  A closely related lemma is proved by Raghavendra and Tan~\cite{Raghavendra} and discussed in~\cite{Allen}. The random choice of $\vec\theta$ may be an artifact of the proof; we do not see a conceptual reason it must be random.

\Prop~\ref{Lemma_Nishi} and \Lem~\ref{Lemma_pinning} fit together marvelously.
Indeed, the apparent issue with  \Lem~\ref{Lemma_pinning} is that we need access to a pristine sample $\check\SIGMA$.
But \Prop~\ref{Lemma_Nishi} implies that we can replace $\check\SIGMA$ by the ``ground truth''~$\hat\SIGMA$.

\subsubsection{The free energy}\label{Sec_Outline3}
The computation of the free energy proceeds in two steps.
In \Sec~\ref{sec:UBproofOutline} we prove that the stochastic optimization problem yields a lower bound.

\begin{proposition}\label{Prop_Will}
If {\bf SYM} and {\bf BAL} hold, then
	$	\liminf_{n\to\infty}-\frac1n \Erw\ln Z(\hat\G)\geq-\sup_{\pi\in\cP^2_*(\Omega)}\cB(d,\pi).$
\end{proposition}

\noindent
To prove \Prop~\ref{Prop_Will} we use the Aizenman-Sims-Starr scheme~\cite{Aizenman}.
This is nothing but the elementary observation that we can  compute  $-\Erw[\ln Z(\hat\G)]$ by calculating 
the difference between the free energy of a random factor graph with $n+1$ variable nodes and one with $n$ variable nodes.
To this end we use a coupling argument.
Roughly speaking, the coupling is such that the bigger factor graph is obtained from the smaller one by adding one
variable node $x_{n+1}$ along with a few adjacent random constraint nodes $b_1,\ldots,b_\gamma$.
(Actually we also need to delete a few constraint nodes from the smaller graph, see \Sec~\ref{sec:UBproofOutline}.)
To track the impact of these changes, we apply pinning to the smaller factor graph to ensure $\eps$-symmetry.
The variable nodes adjacent to $b_1,\ldots,b_\gamma$ are ``sufficiently random'' and 
 $\gamma$ is typically bounded.
 Therefore, we can use $\eps$-symmetry in conjunction with \Lem~\ref{Lemma_lwise}
 to express the expected change in the free energy in terms of the empirical distribution $\rho_n$ of the Gibbs marginals of the smaller graph.
By comparison to prior work such as \cite{COPS,Panchenko} that also used the Aizenman-Sims-Starr scheme,
a delicate point here is that we need to verify that  $\rho_n$ satisfies an invariance property that results from the Nishimori property (\Lem~\ref{Lemma_reweight} below).
With \Lem s~\ref{Lemma_lwise} and \ref{Lemma_pinning} and the invariance property in place, we obtain the change in the free energy by following the steps of the
previously non-rigorous {Belief Propagation} computations, unabridged.
The result works out to be $-\cB(d,\rho_n)$, whence \Prop~\ref{Prop_Will} follows.
The details can be found in \Sec~\ref{sec:UBproofOutline}.

The third assumption {\bf POS} is needed in the proof of the upper bound only.

\begin{proposition}\label{Prop_interpolation}
If {\bf SYM}, {\bf BAL} and {\bf POS} hold, then
	$\limsup_{n\to\infty}-\frac1n \Erw\ln Z(\hat\G)\leq-\sup_{\pi\in\cP^2_*(\Omega)}\cB(d,\pi).$
\end{proposition}

\noindent
We prove \Prop~\ref{Prop_interpolation} via the interpolation method, originally developed by Guerra in order to investigate the Sherrington-Kirkpatrick model~\cite{Guerra}.
Given $\pi\in\cP^2_*(\Omega)$,
the basic idea is to set up a family of factor graphs $(\hat\G_t)_{t\in[0,1]}$ such
that $\hat\G=\hat\G_1$ is the original model and such that $\hat\G_0$ decomposes into connected components that each contain exactly one variable node.
In effect, the free energy of $\hat\G_0$ is computed easily.
The result is $-\cB(d,\pi)$.
Therefore, the key task is to show that the derivative of the free energy is non-positive for all $t\in(0,1)$.
The interpolation scheme that we use is an adaptation of the one of Panchenko and Talagrand~\cite{PanchenkoTalagrand}
to the teacher-student scheme.
A crucial feature of the construction is that the distributional identity from \Prop~\ref{Lemma_Nishi} remains valid for all $t\in[0,1]$.
Together with a coupling argument this enables us to apply pinning to the intermediate models for $t\in(0,1)$
and thus to deduce the negativity of the derivative from as modest an assumption as {\bf POS}.
The details are carried out in \Sec~\ref{Sec_interpolation}.

\Thm~\ref{Thm_stat} is immediate from \Prop s~\ref{Lemma_Nishi}, \ref{Prop_Will} and~\ref{Prop_interpolation} and~(\ref{eq:MIbethe2}).
We prove \Prop s~\ref{Lemma_Nishi}, \ref{Prop_Will} and~\ref{Prop_interpolation} in \Sec~\ref{Sec_Nishi}--\ref{Sec_interpolation}.
\Thm~\ref{Thm_G} follows from \Thm~\ref{Thm_stat} and a subtle (but brief) second moment argument that can be found in \Sec~\ref{Sec_Thm_G}.
The proof of \Thm~\ref{Cor_cond} is also contained in \Sec~\ref{Sec_Thm_G}.
Finally, the proof of \Thm~\ref{Cor_stat} comes in \Sec~\ref{Sec_prop:betheUB}.

\subsection{The Nishimori property}\label{Sec_Nishi}

\noindent
In this section we prove \Prop~\ref{Lemma_Nishi}.
Actually we will formulate and prove a generalized version to facilitate the interpolation argument in  \Sec~\ref{Sec_interpolation}.
To define the corresponding more general factor graph model,
let $k\geq2$ be an integer and let $\Psi$ be a (possibly infinite) set of weight functions $\psi:\Omega^{k_\psi}\to(0,2)$ where $k_\psi\in[k]$ is an integer.
Thus, the weight functions may have different arities, but all arities are bounded by $k$.
Since each function $\psi$ can be viewed as a point in  the $|\Omega|^{k_\psi}$-dimensional Euclidean space, the Borel algebra induces a $\sigma$-algebra on $\Psi$.
Let $p$ be a probability measure defined on this $\sigma$-algebra and let  $\PSI\in\Psi$ be a sample from $p$.
The conditions {\bf BAL} and {\bf SYM} extend without further ado.

Define the random factor graph model $\G(n,m,p)$ with variable nodes $V=\{x_1,\ldots,x_n\}$ and constraint nodes $F=\{a_1,\ldots,a_m\}$
by choosing for each $i\in[m]$ independently a weight function $\psi_{a_i}$ from $p$ and 
a neighborhood $\partial a_i$ consisting of $k_{\psi_{a_i}}$ variable nodes chosen uniformly, mutually independently and independently of $\psi_{a_i}$.
Formally, we view $\G(n,m,p)$ as consisting of a discrete neighborhood structure and an $m$-tuple of weight functions.
Let $\cG(n,m,p)$ be the measurable space consisting of all possible outcomes endowed with the corresponding product $\sigma$-algebra.

Any $G\in \cG(n,m,p)$ induces a Gibbs measure $\mu_G$ defined via (\ref{eqGibbs}).
Moreover, the model  $\G(n,m,p)$ induces a distribution $\hat\SIGMA_{n,m,p}$ on assignments,
a reweighted distribution $\hat\G(n,m,p)$ on factor graphs and for each assignment $\sigma$ a distribution $\G^*(n,m,p,\sigma)$
on factor graphs via the formulas (\ref{eqNishi1})--(\ref{eqNishi3}).
In particular, we have the following extension of Fact~\ref{Fact_teacher}.

\begin{fact}\label{Fact_teacher'}
The graph $\G^*(n,m,p,\sigma)$ is distributed as follows.
For all $j\in[m]$, $l\in[k]$, $i_1,\ldots,i_k\in[n]$ and any event $\cA\subset\Psi$ we have
	\begin{align*}
	\pr\brk{k_{\psi_{a_j}}=l,\psi_{a_j}\in\cA,\partial a_j=(x_{i_1},\ldots,x_{i_l}) }&=
		\frac{\Erw\brk{\vecone\{k_{\PSI}=l,\PSI\in\cA\}\PSI(\sigma(x_{i_1}),\ldots,\sigma(x_{i_l}))}}
			{\sum_{l=1}^k\sum_{h_1,\ldots,h_l=1}^n\Erw\brk{\vecone\{k_{\PSI}=l\}\PSI(\sigma(x_{h_1}),\ldots,\sigma(x_{h_l}))}}
	\end{align*}
and the $m$ pairs $(\psi_{a_j},\partial a_j)_{j\in[m]}$ are mutually independent.
\end{fact}

Additionally, we consider an enhanced version of these distributions where a few variables are pinned to specific spins.
More precisely, for a set $U\subset V=\{x_1,\ldots,x_n\}$, an assignment $\check\sigma\in\Omega^U$ and a factor graph $G$
let $G_{U,\check\sigma}$ be the factor graph obtained from $G$ by adding unary constraint nodes $\alpha_x$ with $\partial \alpha_x=x$ and 
	$\psi_{\alpha_x}(\sigma)=\vecone\{\sigma=\check\sigma(x)\}$ for all $x\in U$.
In contrast to all the weight functions from $\Psi$, the unary weight functions $\psi_{\alpha_x}$ are $\{0,1\}$-valued.
The total weight function, partition function and Gibbs measure of $G_{U,\check\sigma}$ relate to those of the underlying $G$ as follows:
	\begin{align}
	\psi_{G_{U,\check\sigma}}(\sigma)&=\psi_{G}(\sigma)\prod_{x\in U}\vecone\{\sigma(x)=\check\sigma(x)\},&
	Z(G_{U,\check\sigma})&=Z(G)\bck{\prod_{x\in U}\vecone\{\SIGMA(x)=\check\sigma(x)\}}_{G},\label{eqpinZ}\\
	\mu_{G_{U,\check\sigma}}(\sigma)&=\frac{\mu_{G}(\sigma)\prod_{x\in U}\vecone\{\sigma(x)=\check\sigma(x)\}}
		{\bck{\prod_{x\in U}\vecone\{\SIGMA(x)=\check\sigma(x)\}}_{G}}.\nonumber
	\end{align}
Thus, $\mu_{G_{U,\check\sigma}}$ is just the Gibbs measure of $G$ given that $\SIGMA(x)=\check\sigma(x)$ for all $x\in U$.
(Because all $\psi\in\Psi$ are strictly positive, we have $Z(G_{U,\check\sigma})>0$ and 
thus $\mu_{G_{U,\check\sigma}}$ is well-defined.)
Let $\check\cG(n,m,p)$ be the measurable space consisting of all $G_{U,\check\sigma}$
with $G\in\cG(n,m,p)$, $U\subset V$ and $\check\sigma:U\to\Omega$.

Further, let $\G_U(n,m,p)$ be the outcome of the following experiment.
\begin{description}
\item[PIN1] choose a spin $\check\SIGMA(x)\in\Omega$ uniformly and independently for each $x\in U$,
\item[PIN2] independently choose $\check\G=\G(n,m,p)$,
\item[PIN3] let $\G_U(n,m,p)=\check\G_{U,\check\SIGMA}$.
\end{description}
Thus,  $\G_U(n,m,p)$ is obtained from $\G(n,m,p)$ by pinning the variable nodes $x\in U$ to random spins $\check\SIGMA(x)$.
By extension of the formulas (\ref{eqNishi1})--(\ref{eqNishi3}) we obtain the following associated distributions  on assignments/factor graphs:
	\begin{align*}
	\pr\brk{\hat\SIGMA_{U,n,m,p}=\sigma}&=\frac{\Erw[\psi_{\G_{U}(n,m,p)}(\sigma)]}{\Erw[Z(\G_U(n,m,p))]}&&\mbox{for $\sigma\in\Omega^n$},\\
	\pr\brk{\hat\G_U(n,m,p)\in\cA}&=\frac{\Erw[Z(\G_U(n,m,p))\vecone\{\G_U(n,m,p)\in\cA\}]}{\Erw[Z(\G_U(n,m,p))]}&&
					\mbox{for an event }\cA,\\
		\pr\brk{\G_{U}^*(n,m,p,\sigma)\in\cA}&=\frac{\Erw[\psi_{\G_U(n,m,p)}(\sigma)\vecone\{\G_U(n,m,p)\in\cA\}]}
			{\Erw[\psi_{\G_U(n,m,p)}(\sigma)]}&&\mbox{for an event }\cA\mbox{ and }\sigma\in\Omega^n.
	\end{align*}
Finally, mimicking the construction from \Lem~\ref{Lemma_pinning} 
we introduce models where the set of pinned variables itself is random.

\begin{definition}\label{Def_Peg}
For $T\geq0$ let $\vec U=\vec U(T)\subset V$ be a random set generated via the following experiment.
\begin{description}
\item[U1] choose $\vec\theta\in[0,T]$  uniformly at random,
\item[U2] obtain $\vec U\subset V$ by including each variable node with probability $\vec\theta/n$ independently.
\end{description}
Then we let
	$$\G_T(n,m,p)=\G_{\vU}(n,m,p),\ 
		\hat\G_{T}(n,m,p)=\hat\G_{\vU}(n,m,p)\mbox{ and }\vec\G^*_T(n,m,p,\sigma)=\vec\G^*_{\vU}(n,m,p,\sigma).$$
Further, with $\vec m=\Po(dn/k)$ chosen independently of $\vU$, we define
	$$\G_T=\G_{\vU}(n,\vec m,p),\ 
		\hat\G_{T}=\hat\G_{\vU}(n,\vec m,p),\ \vec\G^*_T(\sigma)=\vec\G^*_{\vU}(n,\vec m,p,\sigma)\mbox{ and }
			\vec\G^*_T=\vec\G^*_{\vU}(n,\vec m,p,\SIGMA^*)$$
\end{definition}

\noindent
The following statement provides a Nishimori property for the models from \Def~\ref{Def_Peg}.

\begin{proposition}\label{Cor_NishimoriTilt}
The following two distributions on factor graph/assignment pairs are identical.
\begin{enumerate}[(i)]
\item Choose $\hat\SIGMA=\hat\SIGMA_{n,\vec m,p}$, then choose $\G^*_T\bc{\hat\SIGMA}$.
\item Choose $\hat\G_{T}$, then choose $\SIGMA_{\hat\G_{T}}$.
\end{enumerate}
Moreover, $(\SIGMA^*,\G_T^*(\SIGMA^*))$ and $(\hat\SIGMA,\G_T^*(\hat\SIGMA))$ are mutually contiguous. 
\end{proposition}

In formulas, (i), (ii) are the distributions defined by
	\begin{align*}
	\pr\brk{\hat\SIGMA=\sigma,\G^*_T\bc{\hat\SIGMA}\in\cA}&=\Erw\brk{\pr\brk{\hat\SIGMA=\sigma|\vec m}\cdot
		\pr\brk{\G^*_T\bc{\hat\SIGMA}\in\cA|\vec m}},&
		\pr\brk{\SIGMA_{\hat\G_T}=\sigma,\hat\G_T\in\cA}&=
		\Erw\brk{\mu_{\hat\G_T}(\sigma)\vecone\{\hat\G_T\in\cA\}}
	\end{align*}
respectively, for $\sigma\in\Omega^n$ and events $\cA$.  We prove \Prop~\ref{Cor_NishimoriTilt} by way of the following lemma regarding the model with a fixed pinned set $U$.
Observe that in the first two experiments we first choose an assignment/factor graph pair without paying heed to the set $U$ at all and subsequently pin
the variables in $U$.
By contrast, in the other two experiments we choose a pair that incorporates pinning from the outset.

\begin{lemma}\label{Prop_NishimoriTilt}
For any fixed set $U\subset V$ the distributions on assignment/factor graph pairs induced by the following four experiments are identical.
\begin{enumerate}[(1)]
\item Choose $\SIGMA^{(1)}=\hat\SIGMA_{n,m,p}$, then choose $\G^{(1)}=\G^*(n,m,p,\hat\SIGMA_{n,m,p})$ 
		and output $(\SIGMA^{(1)},\G^{(1)}_{U,\hat\SIGMA^{(1)}})$.
\item Choose $\G^{(2)}=\hat\G(n,m,p)$, then choose $\SIGMA^{(2)}=\SIGMA_{\G^{(2)}}$ and output
			$(\SIGMA^{(2)},\G^{(2)}_{U,\SIGMA^{(2)}})$.
\item Choose $\G^{(3)}=\hat\G_U(n,m,p)$, then choose $\SIGMA^{(3)}=\SIGMA_{\hat\G_U(n,m,p)}$ and output 
	$(\SIGMA^{(3)},\G^{(3)})$.
\item Choose $\SIGMA^{(4)}=\hat\SIGMA_{U,n,m,p}$, then choose $\G^{(4)}=\G_{U}^*(n,m,p,\hat\SIGMA^{(4)})$
	and output $(\SIGMA^{(4)},\G^{(4)})$.
\end{enumerate}
Moreover, the distributions of $\hat\SIGMA_{U,n,m,p}$ and $\hat\SIGMA_{n,m,p}$ are identical.
\end{lemma}
\begin{proof}
In order to show that (i) and (ii) are identical it suffices to prove that the pairs $(\SIGMA_{\hat\G(n,m,p)},\hat\G(n,m,p))$ and
	$(\hat\SIGMA_{n,m,p},\G^*(n,m,p,\hat\SIGMA_{n,m,p}))$ 
 are identically distributed.
Indeed, for any event $\cA$ and any $\sigma\in\Omega^n$,
	\begin{align*}
	\pr\brk{\hat\G(n,m,p)\in\cA,\SIGMA_{\hat\G}=\sigma}&=
		\frac{\Erw\brk{Z(\G(n,m,p))\vecone\{\G(n,m,p)\in\cA\}\mu_{\G(n,m,p)}(\sigma)}}{\Erw[Z(\G(n,m,p))]}\\
		&=\frac{\Erw[\psi_{\G(n,m,p)}(\sigma)\vecone\{\G(n,m,p)\in\cA\}]}{\Erw[Z(\G(n,m,p))]}\\
		&=\frac{\Erw[\psi_{\G(n,m,p)}(\sigma)]}{\Erw[Z(\G(n,m,p))]}\cdot\frac{\Erw[\psi_{\G(n,m,p)}(\sigma)\vecone\{\G(n,m,p)\in\cA\}]}
				{\Erw[\psi_{\G(n,m,p)}(\sigma)]}\\
		&=\pr\brk{\hat\SIGMA_{n,m,p}=\sigma}\pr\brk{\G^*(n,m,p,\hat\SIGMA_{n,m,p})\in\cA|\hat\SIGMA_{n,m,p}=\sigma}\\
		&=
			\pr\brk{\G^*(n,m,p,\hat\SIGMA_{n,m,p})\in\cA,\hat\SIGMA_{n,m,p}=\sigma}.
	\end{align*}
A very similar argument shows that (iii) and (iv) are identical: for any event $\cA$ and any $\sigma\in\Omega^n$,
	\begin{align*}
	\pr\brk{\hat\G_U(n,m,p)\in\cA,\SIGMA_{\hat\G_U(n,m,p)}=\sigma}&
		=\frac{\Erw\brk{Z(\G_U(n,m,p))\vecone\{\G_U(n,m,p)\in\cA\}\mu_{\G_U(n,m,p)}(\sigma)}}{\Erw[Z(\G_U(n,m,p))]}\\
		&=\frac{\Erw\brk{\vecone\{\G_U(n,m,p)\in\cA\}\psi_{\G_U(n,m,p)}(\sigma)}}{\Erw[Z(\G_U(n,m,p))]}\\
		&=\frac{\Erw[\psi_{\G_U(n,m,p)}(\sigma)]}{\Erw[Z(\G_U(n,m,p))]}\cdot\frac{\Erw[\psi_{\G_U(n,m,p)}(\sigma)\vecone\{\G_U(n,m,p)\in\cA\}]}
				{\Erw[\psi_{\G_U(n,m,p)}(\sigma)]}\\
		&=\pr\brk{\hat\SIGMA_{U,n,m,p}=\sigma,\G_{U}^*(n,m,p,\hat\SIGMA_{U,n,m,p})\in\cA}.
	\end{align*}
As a next step we show that $\hat\SIGMA_{n,m,p}$, $\hat\SIGMA_{U,n,m,p}$ are identically distributed.
Indeed, because the random choices performed in {\bf PIN1}, {\bf PIN2} are independent, (\ref{eqpinZ}) implies
	\begin{align}\label{eqWithOrWithoutU}
	\pr\brk{\hat\SIGMA_{U,n,m,p}=\sigma}&=\frac{\Erw[\psi_{\G_U(n,m,p)}(\sigma)]}{\Erw[Z(\G_U(n,m,p))]}
		=\frac{\Erw[\psi_{\G(n,m,p)}(\sigma)]\cdot|\Omega|^{-|U|}}{\Erw[Z(\G(n,m,p))]\cdot|\Omega|^{-|U|}}
			=\pr\brk{\hat\SIGMA_{n,m,p}=\sigma}.
	\end{align}
Finally, to prove that (i) and (iv) are identical,
consider the map $\check\cG(n,m,p)\to\cG(n,m,p)$, $G\mapsto G^\circ$, where $G^\circ$ is obtained from $G$ by deleting the unary factor nodes
$\alpha_x$, $x\in U$, that implement the pinning.
Then for any event $\cA\subset\cG(n,m,p)$ and any $\sigma\in\Omega^n$,
due to the independence of {\bf PIN1} and {\bf PIN2},
	\begin{align}\nonumber
	\pr\brk{\G^{(4)}{}^\circ\in\cA|\SIGMA^{(4)}=\sigma}&=
		\frac{\Erw[\psi_{\G_U(n,m,p)}(\sigma)\vecone\{(\G_U(n,m,p))^\circ\in\cA\}]}{\Erw[\psi_{\G_U(n,m,p)}(\sigma)]}\\
		&=\frac{\Erw[\psi_{\G(n,m,p)}(\sigma)\vecone\{\G(n,m,p)\in\cA\}]|\Omega|^{-|U|}}{\Erw[\psi_{\G(n,m,p)}(\sigma)]|\Omega|^{-|U|}}
		=\pr\brk{\G^{(1)}\in\cA|\SIGMA^{(1)}=\sigma}.\label{eqWithOrWithoutU2}
	\end{align}
Since $U$ is fixed and the unary weight functions $\psi_{\alpha_x}$, $x\in U$, are determined by  $\SIGMA^{(4)}$ resp.\ $\SIGMA^{(1)}$,
	(\ref{eqWithOrWithoutU}) and (\ref{eqWithOrWithoutU2}) imply that (ii) and (iii) are identical.
\end{proof}

\noindent
Next, we make the following simple observation.

\begin{lemma}\label{Lemma_contig}
Suppose that $m=O(n)$.
Under the assumption {\bf BAL} the distribution $\hat\SIGMA_{n,m,p}$ and the uniform distribution are mutually contiguous.
\end{lemma}
\begin{proof}
Recall that $\lambda_\sigma\in\cP(\Omega)$ denotes the empirical distribution of the spins under the assignment $\sigma\in\Omega^V$.
Since the constraint nodes of $\G(n,m,p)$ are chosen independently,
	\begin{align}\label{eqProofNishi}
	\Erw[\psi_{\G(n,m,p)}(\sigma)]&=\brk{\sum_{\tau\in\Omega^k}\Erw[\PSI(\tau_1,\ldots,\tau_{k_{\PSI}})]\prod_{j=1}^k\lambda_\sigma(\tau_j)}^m,\\
	\Erw[Z(\G(n,m,p))]&=\sum_{\sigma\in\Omega^n}\brk{\sum_{\tau\in\Omega^k}\Erw[\PSI(\tau_1,\ldots,\tau_{k_{\PSI}})]\prod_{j=1}^k\lambda_\sigma(\tau_j)}^m.
		\label{eqProofNishi2}
	\end{align}
Further, since the entropy function is concave, (\ref{eqProofNishi2}), Stirling's formula and {\bf BAL} ensure that there exists a number $C=C(\Psi,p)$ such that
	\begin{align}	\label{eqProofNishi10}
	|\Omega|^n\xi^m/C\leq\Erw[Z(\G(n,m,p))]&\leq 
		|\Omega|^n\xi^m.
	\end{align}
Further, let $u$ be the uniform distribution on $\Omega$ and let $\cS(L)$ be the set of all $\sigma\in\Omega^n$ such that 
	$\TV{\lambda_\sigma-u}\leq L/\sqrt n$.
Then {\bf BAL} guarantees that there exists $C'=C'(\Psi,p)>0$ such that for large enough $n$
	\begin{align*}
	\xi-C'L^2/n&\leq\sum_{\tau\in\Omega^k}\Erw[\PSI(\tau_1,\ldots,\tau_{k_{\PSI}})]\prod_{j=1}^k\lambda_\sigma(\tau_j)\leq\xi&\mbox{for all }
		\sigma\in\cS(L).
	\end{align*}
Therefore, (\ref{eqProofNishi}) shows that there exists $C''=C''(\Psi,p,L,m/n)$ such that
	\begin{align}\label{eqProofNishi11}
	C''\xi^m\leq\Erw[\psi_{\G(n,m,p)}(\sigma)]&\leq\xi^m
		&\mbox{for all }\sigma\in\cS(L).
	\end{align}
Since for any $\eps>0$ we can choose $L=L(\eps)$ large enough such that for a uniformly random $\SIGMA^*\in\Omega^n$ we have
$\pr\brk{\SIGMA^*\in\cS(L)}\geq1-\eps$, the assertion follows from (\ref{eqProofNishi10}) and (\ref{eqProofNishi11}).
\end{proof}

\begin{proof}[Proof of \Prop~\ref{Cor_NishimoriTilt}]
We couple the experiments (i) and (ii) such that both experiments pin the same set $\vU$ and use the same number $\vec m$ of constraint nodes.
Then \Lem~\ref{Prop_NishimoriTilt} directly implies that the two distributions are identical.
Analogously, couple $(\SIGMA^*,\G_T^*)$ and $(\hat\SIGMA,\G_T^*(\hat\SIGMA))$ such that both have the same $\vU,\vec m$.
Then the contiguity statement follows from \Lem~\ref{Lemma_contig} and the final assertion follows from \Lem~\ref{Prop_NishimoriTilt}.
\end{proof}

\begin{proof}[Proof of \Prop~\ref{Lemma_Nishi}]
The proposition follows from \Prop~\ref{Cor_NishimoriTilt} by setting $T=0$.
\end{proof}

\noindent
Finally, we highlight the following immediate consequence of \Prop~\ref{Cor_NishimoriTilt}.

\begin{corollary}\label{Cor_reweight}
For all $T\geq0$ and all $\omega\in\Omega$ we have
	$$\Erw\bck{||\SIGMA^{-1}(\omega)|- n/|\Omega||}_{\hat\G_T}=o(1)\quad\mbox{and}\quad
		\Erw\bck{||\SIGMA^{-1}(\omega)|- n/|\Omega||}_{\G_T^*}=o(1).$$
\end{corollary}
\begin{proof}
Since $\SIGMA^*$ assigns spins to vertices independently, Chebyshev's inequality shows that
	\begin{equation}\label{eqLemma_reweight0}
	\Erw\sum_{\omega\in\Omega}||\SIGMA^{*\,-1}(\omega)|- n/|\Omega||=o(1).
	\end{equation}
Because by \Prop~\ref{Lemma_Nishi} the distribution of $\hat\SIGMA$ is contiguous with respect to the uniform distribution,
(\ref{eqLemma_reweight0}) implies
	$
	\Erw\sum_{\omega\in\Omega}||\hat\SIGMA^{-1}(\omega)|- n/|\Omega||=o(1).
	$
\Prop~\ref{Cor_NishimoriTilt} therefore implies that
	\begin{equation}\label{eqLemma_reweight1}
	\Erw\sum_{\omega\in\Omega}\bck{||\SIGMA^{-1}(\omega)|- n/|\Omega||}_{\hat\G_T}=o(1).
	\end{equation}
Together with the contiguity statement from \Prop~\ref{Cor_NishimoriTilt} equation (\ref{eqLemma_reweight1}) yields
 the assertion.
\end{proof}

\subsection{The lower bound}\label{sec:UBproofOutline}

\noindent
In this section we prove \Prop~\ref{Prop_Will} regarding the lower bound on the free energy of $\hat\G$.
The following lemma shows that 
we can tackle this problem by way of lower-bounding the free energy of the random graph $\G_T^*$
from \Def~\ref{Def_Peg}.
Throughout this section we assume {\bf BAL} and {\bf SYM}.

\begin{lemma}\label{lem:FEpin}
For any $T>0$ we have $\Erw[\ln Z(\hat\G)]=\Erw[\ln Z(\G_T^*)]+o(n)$.
\end{lemma}
\begin{proof}
By \Prop~\ref{Lemma_Nishi} we have $\Erw[\ln Z(\hat\G)]=\Erw[\ln Z(\G^*(\hat\SIGMA))]$.
Moreover, since $\SIGMA^*$ and $\hat\SIGMA$ are mutually contiguous, so are $\G^*(\hat\SIGMA)$ and $\G^*(\SIGMA^*)$.
Since $\ln Z(\G^*)$ and $\ln Z(\G^*(\hat\SIGMA))$ are tightly concentrated around their expectations by \Lem~\ref{Lemma_Azuma},
we thus obtain	
	\begin{equation}\label{eqlem:FEpin}
	\Erw[\ln Z(\hat\G)] =\Erw[\ln Z(\G^*)] +o(n).
	\end{equation}
Further, a standard application of the Chernoff bound shows that with probability $1-O(n^{-2})$ the degrees of all variable
nodes of $\G^*$ are upper-bounded by $\ln^2n$.
If so, then pinning a single variable node to a specific spin can shift the free energy of $\G^*$ by no more than $O(\ln^2n)$,
because all weight functions $\psi\in\Psi$ are strictly positive.
Since the expected number of pinned variables is upper-bounded by $T$, we conclude that
	\begin{equation}\label{eqlem:FEpin2}
	\Erw[\ln Z(\G^*_T)] =\Erw[\ln Z(\G^*)]+O(\ln^2n).
	\end{equation}
The assertion follows from (\ref{eqlem:FEpin}) and (\ref{eqlem:FEpin2}).
\end{proof}

\noindent
Thus, we are left to calculate $\Erw[\ln Z(\G_T^*(\SIGMA^*))]$.
The key step is to establish the following estimate.

\begin{lemma}\label{Lemma_ASS}
Letting
	$$\Delta_T(n)=\Erw[\ln Z(\G^*_T(n+1,\vec m(n+1),p,\SIGMA_{n+1}^*)]-
		\Erw[\ln Z(\G^*_T(n,\vec m(n),p,\SIGMA_n^*)]$$
we have
	$$\limsup_{T\to\infty}\limsup_{n \to \infty}\Delta_T(n)
			\le \sup_{\pi\in\cP_*^2(\Omega)} \cB(d, \pi).$$
\end{lemma}

\noindent
Hence, we take a double limit, first taking $n$ to infinity and then $T$.
Let us write $f(n,T)=o_T(1)$ if $$\lim_{T\to\infty}\limsup_{n\to\infty}|f(n,T)|=0.$$
Then  \Lem~\ref{Lemma_ASS} yields
	\begin{align*}
	\frac1n{\Erw[\ln Z(\G^*)]}&=
	\frac1n\Erw[\ln Z(\G^*_T(1,\vec m(1),p,\SIGMA_{1}^*)]+
		\frac1n\sum_{N=1}^{n-1}\Delta_T(N)
		\leq
		\sup_{\pi\in\cP_*^2(\Omega)} \cB(d, \pi)+o_T(1).
	\end{align*}
Thus, applying \Lem s~\ref{lem:FEpin} and~\ref{Lemma_ASS} and taking the lim sup, we obtain \Prop~\ref{Prop_Will}.

Hence, we are left to prove \Lem~\ref{Lemma_ASS}.
To this end we highlight the following immediate consequence of  \Lem~\ref{Lemma_pinning}.

\begin{fact}\label{lem:tiltRScontig}
For any $\eps>0$ there is $T_0>0$ such that for all $T>T_0$ and all large enough $n$ the random factor graph
$\G_T^*$ is $\eps$-symmetric with probability at least $1-\eps$.
\end{fact}
\begin{proof}
 \Lem~\ref{Lemma_pinning} implies that
$\hat\G_T$ is $\eps$-symmetric with probability at least $1-\eps$, provided $T=T(\eps)$ is sufficiently large.
Therefore, the assertion follows from the contiguity statement from \Prop~\ref{Cor_NishimoriTilt}.
\end{proof}

\noindent
Additionally, 
we need to investigate the empirical distribution of the Gibbs marginals
	of the random factor graph $\G^*_T$.
Formally, for a factor graph $G$ we define the empirical marginal distribution $\rho_G$ as
	$$\rho_G=|V|^{-1}\sum_{x\in V}\delta_{\mu_{G,x}}\in\cP^2(\Omega).$$
Thus, $\rho_G$ is the distribution of the Gibbs marginal $\mu_{G,\vec x}$ of a uniformly random variable node $\vec x$ of $G$.
If we are also given an assignment $\sigma\in\Omega^V$, then we let
	\begin{align*}
	\rho_{G,\sigma,\omega}&=\frac1{|\sigma^{-1}(\omega)|}\sum_{x\in V}\vecone\{\sigma(x)=\omega\}\delta_{\mu_{G,x}},
	\end{align*}
unless $\sigma^{-1}(\omega)=\emptyset$ (in which case, say, $\rho_{G,\sigma,\omega}$ is the uniform distribution on $\cP(\Omega)$).
Thus, $\rho_{G,\sigma,\omega}$ is the empirical distribution of the Gibbs marginals of the variables with spin $\omega$  under $\sigma$.
Further, write $\hat\rho_{G,\omega}$ for the reweighted probability distribution
	\begin{align}\label{eqHatRho}
	\hat\rho_{G,\omega}(\mu)&=\frac{\mu(\omega)}{\int \nu(\omega)\dd\rho_G(\nu)}\dd\rho_G(\mu),
	\end{align}
unless $\int \mu(\omega)d\rho_G(\mu)=0$, in which case $\hat\rho_{G,\omega}$ is the uniform distribution.

\begin{lemma}\label{Lemma_reweight}
We have	$\sum_{\omega\in\Omega}\Erw|\int\mu(\omega)\dd\rho_{\G_T^*}(\mu)-|\Omega|^{-1}|=o(1)$.
\end{lemma}
\begin{proof}
\Cor~\ref{Cor_reweight} yields
	$\sum_{\omega\in\Omega}\Erw\bck{||\SIGMA^{-1}(\omega)|- n/|\Omega||}_{\G_T^*}=o(1).$
Hence, by the triangle inequality, for all $\omega\in\Omega$
	\begin{align*}
	\Erw\abs{\int\mu(\omega)\dd\rho_{\G_T^*}(\mu)-|\Omega|^{-1}}&=
		\Erw\abs{\frac1n\sum_{x\in V}\bck{\vecone\{\SIGMA(x)=\omega\}-|\Omega|^{-1}}_{\G_T^*}}
			\leq\Erw\bck{\abs{n^{-1}|\SIGMA^{-1}(\omega)|-|\Omega|^{-1}}}_{\G_T^*}=o(1),
	\end{align*}
as desired.
\end{proof}

\noindent
Recall that $W_1$ denotes the $L^1$-Wasserstein metric on $\cP^2(\Omega)$.

\begin{lemma}\label{Lemma_reweight2}
We have
	$\sum_{\omega\in\Omega}\Erw[W_1(\rho_{\G_T^*,\SIGMA^*,\omega},\hat\rho_{\G^*_T,\omega})]=o_T(1).$
\end{lemma}
\begin{proof}
By \Prop~\ref{Cor_NishimoriTilt} it suffices to prove that 
	\begin{align}\label{eqLemma_reweight2_0}
	\sum_{\omega\in\Omega}\Erw[W_1(\rho_{\hat\G_T,\SIGMA_{\hat\G_T},\omega},\hat\rho_{\hat\G_T,\omega})]=o_T(1).
	\end{align}
Let $\SIGMA=\SIGMA_{\hat\G_T}$ for brevity.
Since $W_1$ metrises weak convergence, in order to prove (\ref{eqLemma_reweight2_0}) it suffices to show that for any continuous function
$f:\cP(\Omega)\to[0,1]$ and for any $\eps>0$ for large enough $n,T$ we have
	\begin{align}\label{eqLemma_reweight2_1}
	\Erw\bck{\abs{
		\int_{\cP(\Omega)} f(\mu)\dd\hat\rho_{\hat\G_T,\omega}(\mu)-\int_{\cP(\Omega)}f(\mu)\dd\rho_{\hat\G_T,\SIGMA,\omega}(\mu)}}
			_{\hat\G_T}<3\eps
			\qquad\mbox{for all }\omega\in\Omega.
	\end{align}

To prove (\ref{eqLemma_reweight2_1}) pick $\delta=\delta(f,\eps)>0$ small enough.
The compact set $\cP(\Omega)$ admits a partition into pairwise disjoint measurable
subsets $S_1,\ldots,S_K$ such that any two
distributions that belong to the same set $S_i$ have total variation distance less than $\delta$ for
some $K=K(\delta,\Omega)>0$ that depends on $\delta,\Omega$ only.
Pick a small enough $\eta=\eta(\delta,K,\Omega)$.
Then by Fact~\ref{lem:tiltRScontig} there is $T_0(\eta,\Omega)$ such that for all $T>T_0$ for large enough $n$ we have
	\begin{align}\label{eqLemma_reweight2_2}
	\pr\brk{\mu_{\hat\G_T}\mbox{ is $\eta^4$-symmetric}}>1-\eta.
	\end{align}

Let $\cV_i=\cV_i(\hat\G_T)$ be the set of variable nodes of $\hat\G_T$ whose Gibbs marginal $\mu_{\hat\G_T,x}$ lies in $S_i$ and let $n_i=|\cV_i|$.
Let $\cX_{i,\omega}(\SIGMA)$ be the set of $x\in \cV_i$ such that $\SIGMA(x)=\omega$ and let $X_{i,\omega}(\SIGMA)=|\cX_{i,\omega}(\SIGMA)|$.
By the linearity of expectation we have
	\begin{align}\label{eqLemma_reweight2_3}
	\bck{X_{i,\omega}(\SIGMA)}_{\hat\G_T}&=\sum_{x\in \cV_i}\mu_{\hat\G_T, x}(\omega)\qquad\mbox{for all }\omega\in\Omega.
	\end{align}
Furthermore, if $\mu_{\hat\G_T}$ is $\eta^4$-symmetric, then the variance of $X_{i,\omega}(\SIGMA)$ works out to be
	\begin{align}\label{eqLemma_reweight2_4}
	{\bck{X_{i,\omega}^2(\SIGMA)}_{\hat\G_T}-\bck{X_{i,\omega}(\SIGMA)}_{\hat\G_T}^2}
		&=
		\sum_{x,y\in \cV_i}\bc{\mu_{\hat\G_T, x, y}(\omega,\omega)-\mu_{\hat\G_T,x}(\omega)\mu_{\hat\G_T,y}(\omega)}
		\leq2\eta^4n^2\qquad\mbox{for all }\omega\in\Omega.
	\end{align}
Combining (\ref{eqLemma_reweight2_3}) and (\ref{eqLemma_reweight2_4}) with Chebyshev's inequality, we obtain
	\begin{align*}
	\bck{\vecone\{|X_{i,\omega}(\SIGMA)-\sum_{x\in \cV_i}\mu_{\hat\G_T,\vec x}(\omega)|>\eta n\}}_{\hat\G_T}\leq2\eta^2
	\qquad\mbox{for all $i\in[K]$}.
	  \end{align*}
Hence, by the union bound and \Cor~\ref{Cor_reweight},
	\begin{align}\label{eqLemma_reweight10}
	\bck{\textstyle\vecone\{
		\sum_{i\in[K],\omega\in[\Omega]}
			|X_{i,\omega}(\SIGMA)-\sum_{x\in \cV_i}\mu_{\hat\G_T,\vec x}(\omega)|\leq\sqrt\eta n,\,
				\sum_{\omega\in\Omega}||\SIGMA^{-1}(\omega)|- n/|\Omega||\leq\eta n
				\}}_{\hat\G_T}\geq 1-\eta,
	\end{align}
provided $\eta$ was chosen small enough.

Now, suppose that $\hat\G_T$, $\SIGMA=\SIGMA_{\hat\G_T}$ are such that 
	\begin{align}\label{eqLemma_reweight11}
	\sum_{i\in[K],\omega\in[\Omega]}|X_{i,\omega}(\SIGMA)-\sum_{x\in \cV_i}\mu_{\hat\G_T,\vec x}(\omega)|&\leq\sqrt\eta n,\quad
		\sum_{\omega\in\Omega}||\SIGMA_{\hat\G_T}^{-1}(\omega)|- n/|\Omega||\leq\eta n,\quad
		\sum_{\omega\in\Omega}\abs{\int\mu(\omega)\dd\rho_{\hat\G_T}(\mu)-|\Omega|^{-1}}\leq\eta.
	\end{align}
Because $f:\cP(\Omega)\to[0,1]$ is uniformly continuous, we can pick $\delta,\eta$ small enough so that (\ref{eqLemma_reweight11}) implies that
	\begin{align*}
	\int_{\cP(\Omega)} f(\mu)\dd\hat\rho_{\hat\G_T,\omega}(\mu)&=
		\frac{\sum_{i=1}^K\sum_{x\in\cV_i}\mu_{\hat\G_T,x}(\omega)f(\mu_{\hat\G_T,x}(\omega))}
		{n\sum_{i=1}^K\sum_{x\in\cV_i}\mu_{\hat\G_T,x}(\omega)}\\
		&\leq\frac{\eps+\sqrt\eta+\sum_{i\in[K]}\sum_{x\in\cX_{i,\omega}(\SIGMA)}f(\mu_{\hat\G_T,x})}{|\Omega|^{-1}-\eta}
		\leq2\eps+\int_{\cP(\Omega)}f(\mu)\dd\rho_{\hat\G_T,\SIGMA,\omega}(\mu).
	\end{align*}
A similar chain of inequalities yields a corresponding lower bound.
Thus, 
	\begin{align}\label{eqLemma_reweight20}
	\mbox{(\ref{eqLemma_reweight11})}\ \Rightarrow\ 
	\abs{\int_{\cP(\Omega)} f(\mu)\dd\hat\rho_{\hat\G_T,\omega}(\mu)-\int_{\cP(\Omega)}f(\mu)\dd\rho_{\hat\G_T,\SIGMA,\omega}(\mu)}&\leq
		2\eps.
	\end{align}
Finally, since (\ref{eqLemma_reweight2_2}), (\ref{eqLemma_reweight10}) and \Lem~\ref{Lemma_reweight} show
that  (\ref{eqLemma_reweight11}) holds with probability at least $1-3\eta$ and since $f$ takes values in $[0,1]$, (\ref{eqLemma_reweight20}) implies (\ref{eqLemma_reweight2_1}).
\end{proof}

We proceed to prove \Lem~\ref{Lemma_ASS}.
To calculate 	$\Delta_T(n)$
we set up a coupling of $\G^*_T(n+1,\vec m(n+1),p,\SIGMA_{n+1}^*)$ and $\G^*_T(n,\vec m(n),p,\SIGMA_{n}^*)$.
Specifically, we are going to view both these factor graphs as supergraphs of one factor graph $\tilde\G$ on $n$ variable nodes.
To obtain $\tilde\G$ first choose a ground truth $\SIGMA^*_n:\{x_1,\ldots,x_n\}\to\Omega$ uniformly and let
$\SIGMA^*_{n+1}$ be a random extension obtained by choosing $\SIGMA^*_{n+1}(x_{n+1})$ uniformly.
Let
	\begin{align}\label{eqD}
	\vec D=\vec D(\SIGMA_{n+1}^*)&=\frac{\sum_{i_1,\ldots,i_k\in[n+1],\psi\in\Psi}\vecone\{n+1\in\{i_1,\ldots,i_k\}\}p(\psi)
		\psi(\SIGMA_{n+1}^*(x_{i_1}),\ldots,\SIGMA_{n+1}^*(x_{i_k}))}
		{\sum_{i_1,\ldots,i_k\in[n+1],\psi\in\Psi}p(\psi)
			\psi(\SIGMA_{n+1}^*(x_{i_1}),\ldots,\SIGMA_{n+1}^*(x_{i_k}))}\cdot\frac{d(n+1)}k
	\end{align}
Unravelling the construction (\ref{eqTeacher}), we see that $\vec D$ is the
expected degree of $x_{n+1}$ in $\G^*(n+1,\vec m(n+1),p,\SIGMA_{n+1}^*)$.
Additionally, let
	$$D=\Erw[\vec D|\SIGMA_n^*],\qquad D(\omega)=\Erw[\vec D|\SIGMA_n^*,\SIGMA_{n+1}^*(x_{n+1})=\omega],
		\qquad D_{\max}=\max\{D_\omega:\omega\in\Omega\}.$$
Further, define
	\begin{align*}
	\tilde\lambda&=\max\{0,\min\{d(n+1)/k-D_{\max},dn/k\}\},&\lambda'&=dn/k-\tilde\lambda,&
\lambda''&=\max\{0,d(n+1)/k-\tilde\lambda-\vec D\}.
	\end{align*}
Additionally, choose $\vec\theta\in[0,T]$ uniformly and suppose that $n>n_0(T)$ is sufficiently large.
Now, let $\tilde\G$ be the random factor graph with variable nodes $V_n=\{x_1,\ldots,x_n\}$ obtained by
	\begin{description}
	\item[CPL1] generating $\tilde{\vec m}=\Po(\tilde\lambda)$ independent random constraint nodes $a_1,\ldots,a_{\tilde{\vec m}}$
			according to the distribution (\ref{eqTeacher}) with respect to the ground truth ground truth $\SIGMA_n^*$, and
	\item[CPL2] inserting a unary constraint node that pins $x_i$ to $\SIGMA_n^*(x_i)$ with probability $\vec\theta/(n+1)$ for each $i\in[n]$ independently.
	\end{description}
Further, obtain $\G'$ from $\tilde\G$ by
	\begin{description}
	\item[CPL1$'$] adding $\vec m'=\Po(\lambda')$ independent random constraint nodes drawn according  to (\ref{eqTeacher}) w.r.t.\ $\SIGMA_n^*$, and
	\item[CPL2$'$] pinning each as yet unpinned variable node to $\SIGMA^*_n$ independently with probability $\vec \theta/(n(n+1-\vec \theta))$.
	\end{description}
Finally, obtain $\G''$ from $\tilde\G$ by adding the single variable node $x_{n+1}$ and
	\begin{description}
	\item[CPL1$''$] adding $\vec\gamma^*=\Po(\vec D)$ independent constraint nodes $b_1,\ldots,b_{\vec\gamma^*}$ such that for each $j\in[\vec\gamma^*]$,
					\begin{align*}
			\pr\brk{\psi_{b_j}=\psi,\partial b_j=(x_{i_1},\ldots,x_{i_k})}&\propto
				\vecone\{n+1\in\{i_1,\ldots,i_k\}\}p(\psi)
					\psi(\SIGMA_{n+1}^*(x_{i_1}),\ldots,\SIGMA_{n+1}^*(x_{i_k})).
			\end{align*}
			in words, $b_1,\ldots,b_{\vec\gamma^*}$ are chosen from (\ref{eqTeacher}) w.r.t.\ $\SIGMA_{n+1}^*$ subject to the condition
			that each is adjacent to $x_{n+1}$.
	\item[CPL2$''$] adding $\vec m''=\Po(\lambda'')$ independent random constraint nodes $c_1,\ldots,c_{\vec m''}$
		such that for each $j\in[\vec m'']$,
					\begin{align*}
			\pr\brk{\psi_{b_j}=\psi,\partial b_j=(x_{i_1},\ldots,x_{i_k})}&\propto
				\vecone\{n+1\not\in\{i_1,\ldots,i_k\}\}p(\psi)
					\psi(\SIGMA_{n}^*(x_{i_1}),\ldots,\SIGMA_{n}^*(x_{i_k}));
			\end{align*}
			thus, $b_1,\ldots,b_{\vec\gamma^*}$ are chosen from (\ref{eqTeacher}) subject to the condition
			that none is adjacent to $x_{n+1}$.
	\item[CPL3$''$] pinning $x_{n+1}$ to $\SIGMA^*(x_{n+1})$ with probability $\vec\theta/(n+1)$ independently of everything else.
	\end{description}
We observe that this construction produces the correct distribution.

\begin{fact}\label{Fact_coupling}
For sufficiently large $n$ the random factor graph
$\G'$ is distributed as $\G^*_T(n,\vec m(n),p,\SIGMA_{n}^*)$ and
$\G''$ is distributed as $\G^*_T(n+1,\vec m(n+1),p,\SIGMA_{n+1}^*)$.
\end{fact}
\begin{proof}
Because all $\psi\in\Psi$ are strictly positive $\vec D$ is bounded by some number depending on $\Psi,d$ only.
Therefore, $\tilde\lambda>0$ for large enough $n$ and $\tilde\lambda+\lambda'=dn/k$.
Consequently, since a sum of independent Poisson variables is Poisson, {\bf CPL1} and {\bf CPL1$'$} ensure that $\G'$ has $\vec m(n)=\Po(dn/k)$
independent constraint nodes drawn from (\ref{eqTeacher}).
Moreover, by {\bf CPL2} and {\bf CPL2$'$} each variable node of $\G'$ gets pinned with probability $\vec\theta/n$ independently.
Hence, $\G'$ has the desired distribution.

Analogously, by {\bf CPL2} and {\bf CPL3$''$} each variable node of $\G''$ gets pinned with probability $\vec\theta/(n+1)$ independently.
Further, by {\bf CPL1}, {\bf CPL1$''$} and {\bf CPL2$''$} the total expected number of constraint nodes of $\G''$ equals 
	$\tilde\lambda+\vec D+\lambda''=d(n+1)/k$ for large enough $n$.
Moreover, \Def~\ref{Def_teacher} and (\ref{eqD}) guarantee that $\vec D$ equals the expected number
of constraint nodes adjacent to $x_{n+1}$ in $\G^*_T(n+1,\vec m(n+1),p,\SIGMA_{n+1}^*)$.
Thus, $\G''$ has distribution $\G^*_T(n+1,\vec m(n+1),p,\SIGMA_{n+1}^*)$.
\end{proof}

Fact (\ref{Fact_coupling}) implies that for large enough $n$,
	\begin{align}\label{eqASS-1}
	\Delta_T(n)&=\Erw\brk{\ln\frac{Z(\G'')}{Z(\G')}}=\Erw\brk{\ln\frac{Z(\G'')}{Z(\tilde\G)}}-\Erw\brk{\ln\frac{Z(\G')}{Z(\tilde\G)}}.
	\end{align}
Actually the following slightly modified version of (\ref{eqASS-1}) is more convenient to work with.

\begin{claim}\label{Claim_E}
The event $$\cE=\{\forall\omega\in\Omega:|\sigma^{*\,-1}_n(\omega)-n/|\Omega||\leq\sqrt n\ln n\}$$
has probability $1-O(n^{-2})$ and
	\begin{align}\label{eqASS1}
	\Delta_T(n)&=\Erw\brk{\vecone\{\cE\}\ln\frac{Z(\G'')}{Z(\tilde\G)}}-\Erw\brk{\vecone\{\cE\}\ln\frac{Z(\G')}{Z(\tilde\G)}}+o(1).
	\end{align}
Moreover, on $\cE$ we have
	\begin{align}\label{eqGivenE}
	\vec D&=d+o(1),&\tilde\lambda&=d(n+1)/k-d+o(1),&\lambda'&=d(k-1)/k+o(1),&\lambda''&=o(1).
	\end{align}
\end{claim}
\begin{proof}
Because $\SIGMA^*$ is chosen uniformly, the Chernoff bound shows that   $\pr\brk{\cE}\geq1-O(n^{-2})$. 
Moreover, because all $\psi\in\Psi$ are strictly positive, there exists constant $C_{\Psi}>0$ depending on $\Psi$ only such that
 $\ln Z(G)\leq C_{\Psi}m$ for all factor graphs $G$ with $m$ constraint nodes.
Since the Poisson distribution has sub-exponential tails and  $\pr\brk{\cE}\geq1-O(n^{-2})$, (\ref{eqASS-1}) therefore yields (\ref{eqASS1}).
Further, {\bf SYM} guarantees that given $\cE$ we have $D_\omega=d+o(1)$ for all $\omega\in\Omega$, whence (\ref{eqGivenE}) follows.
\end{proof}

\begin{claim}\label{Claim_CPLTV}
The random factor graphs $\tilde\G$ and $\G_T^*$ have total variation distance $o(1)$.
\end{claim}
\begin{proof}
Let $\tilde\vU$ be the set of variables of $\tilde\G$ that got pinned.
Then {\bf CPL1--CPL2} ensures that given $\tilde{\vec m}=m$ and given $\tilde\vU=U$, $\tilde\G$
has distribution $\G_U^*(n,m,p,\SIGMA_n^*)$.
By comparison, $\G_T^*$ is defined as $\G_{\vU}^*(n,\vec m,p,\SIGMA_n^*)$,
where $\vec m=\Po(dn/k)$ and, as in \Def~\ref{Def_Peg}, $\vU$ is obtained by including every variable node with probability $\vec\theta/n$ independently.
Since $T/n-T/(n+1)=o(1)$ for every fixed $T$, the total variation distance of $\vU$ and $\tilde\vU$ is $o(1)$.
Similarly, since $\Erw[\tilde{\vec m}]-\Erw[{\vec m}]=\tilde\lambda-dn/k=O(1)$ while $\Var(\vec m)=\Theta(n)$,
the total variation distance of $\tilde{\vec m}$, $\vec m$ is $o(1)$.
\end{proof}

\noindent
Let $\pi=\rho_{\tilde\G}$ be the empirical distribution of the Gibbs marginals of $\tilde\G$ and recall the notation of \Thm~\ref{Thm_stat}.
We are going to show that the two expressions on the r.h.s.\ of (\ref{eqASS-1}) are equal to the formulas from \Thm~\ref{Thm_stat}, up to an $o_T(1)$
error term.

\begin{claim}\label{Claim_eqASS3}
With probability $1-o_T(1)$ over the choice of $\SIGMA_n^*$ and $\tilde\G$ we have
	\begin{align*}
	\vecone\{\cE\}\Erw[\ln(Z(\G')/Z(\tilde\G))|\tilde\G,\SIGMA_n^*]&=o_T(1)+
		\frac{d(k-1)}{k\xi}\Erw\brk{\Lambda\bc{\sum_{\tau\in\Omega^k}\PSI(\tau)\prod_{j=1}^k\vec\mu_j^{(\pi)}(\tau_j)}}.
	\end{align*}
\end{claim}
\begin{proof}
We may assume that $\SIGMA_n^*\in\cE$ and also,
 since $\tilde{\vec m}=\Po(\tilde\lambda)$ and the Poisson distribution has sub-exponential tails,
 that $\tilde{\vec m}\leq 2dn$.
Let $\cU$ be the event that {\bf CPL2$'$} did not pin any variable node at all.
Then for all $\tilde\G,\SIGMA_n^*$ for large enough $n$ we have $\pr[\cU|\tilde\G,\SIGMA_n^*]\geq1-2T/n$.
Consequently, since all weight functions are strictly positive and the average number of constraint nodes adjacent to any one variable
node is bounded by $k\tilde{\vec m}/n=O(1)$, we conclude that
	\begin{align}\label{eqeventU}
	\vecone\{\cE\}\Erw[\ln(Z(\G')/Z(\tilde\G))|\tilde\G,\SIGMA_n^*]&=o_T(1)+
	\Erw[\vecone\{\cU\}\ln(Z(\G')/Z(\tilde\G))|\tilde\G,\SIGMA_n^*].
	\end{align}
Moreover, let $b_1,\ldots,b_{\vec m'}$ be the constraint nodes added by {\bf CPL1$'$} and let  $Y$ be the set of adjacent variable nodes.
Because on the event $\cU$ the factor graph $\G'$ is obtained from $\tilde\G$ by just adding $b_1,\ldots,b_{\vec m'}$, (\ref{eqGibbs}) yields
	\begin{align}\label{eqClaim_eqASS3_1}
	\ln(Z(\G')/Z(\tilde\G))&=\ln\bck{\prod_{i=1}^{\vec m'}\psi_{b_i}(\SIGMA(\partial_1b_i),\ldots,\SIGMA(\partial_k b_i))}_{\tilde\G}
		=\ln\sum_{\tau\in\Omega^Y}\mu_{\tilde\G,Y}(\tau)\prod_{i=1}^{\vec m'}\psi_{b_i}(\tau|_{\partial b_i}).
	\end{align}

To make sense of the r.h.s.\ of (\ref{eqClaim_eqASS3_1}) we need to take a closer look at the distribution of $Y$.
Since $b_1,\ldots,b_{\vec m'}$ are chosen from  (\ref{eqTeacher}), $Y$ is not generally uniformly distributed.
Nonetheless, since all constraint functions $\psi\in\Psi$ are strictly positive and $\SIGMA_n^*\in\cE$,
there is a number $c=c(\Psi)>0$ such that for any set $Y_0\subset\{x_1,\ldots,x_n\}$ of size $|Y_0|=(k-1)\vec m'$ we have
	\begin{align}\label{eqcPsi}
	\pr\brk{|Y|=(k-1)\vec m'|\tilde\G,\SIGMA_n^*}=1-o(1)\mbox{ and }
		c^{\vec m'}\leq n^{(k-1)\vec m'}\pr\brk{Y=Y_0|\tilde\G,\SIGMA_n^*,\vec m'}\leq c^{-\vec m'}
	\end{align}
Hence, for any given value of $\vec m'$, $Y$ is contiguous with respect to a uniformly random set of size $(k-1)\vec m'$.
Consequently, because (\ref{eqGivenE}) shows that on $\cE$ the mean $\lambda'$ of the Poisson variable $\vec m'$ is bounded independently of $T$,
\Lem~\ref{Lemma_lwise}, Fact~\ref{lem:tiltRScontig}  and Claim~\ref{Claim_CPLTV} yield $\eps_T=o_T(1)$ such that
the event
	$$\cY=\cbc{\TV{\mu_{\tilde \G,Y}-\bigotimes_{y\in Y}\mu_{\tilde \G,y}}\leq\eps_T\mbox{ and }|Y|=k\vec m'}$$
satisfies
	\begin{align}\label{eqeventY}
	\pr\brk{\cY|\tilde\G,\SIGMA_n^*}&\geq1-\eps_T.
	\end{align}
Further, on the event $\cU\cap\cY$ equation (\ref{eqClaim_eqASS3_1}) becomes
	\begin{align}\label{eqClaim_eqASS3_2}
	\ln(Z(\G')/Z(\tilde\G))&
		=o_T(1)+\sum_{i=1}^{\vec m'}\ln\sum_{\tau\in\Omega^k}\psi_{b_i}(\tau)\prod_{h=1}^k\mu_{\tilde\G,\partial_hb_i}(\tau_h).
	\end{align}
Since the mean of the Poisson random variable $\vec m'$ is bounded independently of $T$, the Poisson distribution has
sub-exponential tails and all weight functions are strictly positive, (\ref{eqeventU}), (\ref{eqeventY}) and (\ref{eqClaim_eqASS3_2}) yield
	\begin{align}\label{eqClaim_eqASS3_3}
	\Erw[\ln(Z(\G')/Z(\tilde\G))|\tilde\G,\SIGMA_n^*]&=o_T(1)+
	\Erw\brk{\sum_{i=1}^{\vec m'}\ln\sum_{\tau\in\Omega^k}\psi_{b_i}(\tau)\prod_{h=1}^k\mu_{\tilde\G,\partial_hb_i}(\tau_h)\,\bigg|\,\tilde\G,\SIGMA_n^*}.
	\end{align}
Indeed, because the new constraint nodes $b_1,\ldots,b_{\vec m'}$ are chosen independently given $\tilde\G,\SIGMA_n^*$, (\ref{eqClaim_eqASS3_3}) yields
	\begin{align}\label{eqClaim_eqASS3_4}
	\Erw[\ln(Z(\G')/Z(\tilde\G))|\tilde\G,\SIGMA_n^*]&=o_T(1)+\lambda'
	\Erw\brk{\ln\sum_{\tau\in\Omega^k}\psi_{b_1}(\tau)\prod_{h=1}^k\mu_{\tilde\G,\partial_hb_1}(\tau_h)\,\bigg|\,\tilde\G,\SIGMA_n^*}.
	\end{align}
Let $\vec i_1,\ldots,\vec i_k\in[n]$ be chosen uniformly and independently and choose $\PSI$ from $p$ independently of everything else.
Since $|\SIGMA_n^{*\,-1}(\omega)|\sim n/|\Omega|$ for all $\omega\in\Omega$ we have
	$\Erw[\PSI(\SIGMA_n^*(x_{\vec i_1}),\ldots,\SIGMA_n^*(x_{\vec i_k}))]\sim \xi$.
Hence, recalling the distribution (\ref{eqTeacher}) from which $b_1$ is chosen, we can	write (\ref{eqClaim_eqASS3_4}) as
	\begin{align}\label{eqClaim_eqASS3_5}
	\Erw[\ln(Z(\G')/Z(\tilde\G))|\tilde\G,\SIGMA_n^*]&=o_T(1)+\frac{\lambda'}\xi
	\Erw\brk{\Lambda\bc{\sum_{\tau\in\Omega^k}\PSI(\tau)\prod_{h=1}^k\mu_{\tilde\G,x_{\vec i_h}}(\tau_h)}\,\bigg|\,\tilde\G,\SIGMA_n^*}.
	\end{align}
Since $\pi$ is the empirical distribution of the Gibbs marginals of $\tilde\G$, the assertion follows from (\ref{eqGivenE}) and (\ref{eqClaim_eqASS3_5}).
\end{proof}

\begin{claim}\label{Claim_eqASS2}
With probability $1-o_T(1)$ over the choice of $\SIGMA_n^*$ and $\tilde\G$ we have
	\begin{align*}
	\vecone\{\cE\}\Erw[\ln(Z(\G'')/Z(\tilde\G))|\tilde\G,\SIGMA_n^*]
	&=o_T(1)+\Erw\brk{\frac{\xi^{-\vec\gamma}}{|\Omega|}
			\Lambda\bc{\sum_{\sigma\in\Omega}\prod_{i=1}^{\vec\gamma}\sum_{\tau\in\Omega^k}\vecone\{\tau_{\vec h_i}=\sigma\}					\PSI_i(\tau)\prod_{j\neq\vec h_i}\vec\mu_{ki+j}^{(\pi)}(\tau_j)}}.
	\end{align*}	
\end{claim}
\begin{proof}
Once more we may assume that $\SIGMA_n^*\in\cE$ and $\tilde{\vec m}\leq 2dn$.
Additionally, by Claim~\ref{Claim_CPLTV}, \Lem~\ref{Lemma_reweight} and \Lem~\ref{Lemma_reweight2} we may assume that $\tilde\G,\SIGMA_n^*$ satisfy
	\begin{equation}\label{eqMyLemma_reweight2}
	\sum_{\omega\in\Omega}\abs{\int\mu(\omega)\dd\rho_{\tilde\G}(\mu)-|\Omega|^{-1}}=o(1)\quad\mbox{and}\quad
	\sum_{\omega\in\Omega}W_1(\rho_{\tilde\G,\SIGMA_n^*,\omega},\hat\rho_{\tilde\G,\omega})=o_T(1).
	\end{equation}
Moreover, let $\cU$ be the event that {\bf CPL3$''$} does not pin $x_{n+1}$ and that $\vec m''=0$.
Since $\pr\brk{\cU|\tilde\G,\SIGMA_n^*}=1-o(1)$, since by {\bf CPL1$''$} the expected number of constraint nodes adjacent to $x_{n+1}$ is bounded
and because $\lambda''=o(1)$ by (\ref{eqGivenE}), we have
	\begin{align}\label{eqASS2b}
	\Erw[\ln(Z(\G'')/Z(\tilde\G))|\tilde\G,\SIGMA_n^*]=o(1)+\Erw[\vecone\{\cU\}\ln(Z(\G'')/Z(\tilde\G))|\tilde\G,\SIGMA_n^*].
	\end{align}
Hence, with $b_1\ldots,b_{\vec\gamma^*}\in\partial x_{n+1}$ the new constraint nodes that {\bf CPL1$''$} attaches to $x_{n+1}$,
 on $\cU$ we have
	\begin{align}\label{eqASS2c}
	\ln\frac{Z(\G'')}{Z(\tilde\G)}&=
		\ln\sum_{\tau\in\Omega^{Y\cup\{x_{n+1}\}}}
		\mu_{\tilde\G,Y}(\tau|_{Y})
		\prod_{i=1}^{\vec\gamma^*}\psi_{b_i}(\tau(\partial_1b_i),\ldots,\tau(\partial_kb_i)).
				\hspace{-3mm}
	\end{align}

We need to get a handle on the distribution of $b_1,\ldots,b_{\vec\gamma^*}$.
With $\vec h_1,\vec h_2,\ldots$ mutually independent and uniformly distributed on $[k]$,
the assumptions {\bf SYM} and $\SIGMA_n^*\in\cE$ show that 
for every $j\in[\vec\gamma^*]$ and every $(i_1,\ldots,i_k)\in[n+1]^k$ with $i_{\vec h_j}=n+1$ we have
	\begin{align}\label{eqASS5}
	\pr\brk{\partial b_j=(x_{i_1},\ldots,x_{i_k}),\psi_{b_j}=\psi|\tilde\G,\SIGMA_{n+1}^*}=o(1)+
		\xi^{-1}p(\psi)\psi(\SIGMA_{n+1}^*(x_{i_1}),\ldots,\SIGMA_{n+1}^*(x_{i_k})).
	\end{align}
In particular, given their spins the variables $\partial b_j\setminus\{x_{n+1}\}$ are chosen asymptotically uniformly and independently.
Hence, we can characterize the distribution of $b_1,\ldots,b_{\vec\gamma^*}$ as follows.
Independently for each $b_j$,
\begin{enumerate}[(i)]
\item choose $\vec\omega_j=(\vec\omega_{j,1},\ldots,\vec\omega_{j,k})\in\Omega^k$ and $\hat{\vec\psi}_j$
from the distribution
	\begin{align*}
	\pr\brk{\vec\omega_j=(\omega_1,\ldots,\omega_k),\hat{\vec\psi}_j=\psi}&\propto\vecone\{\omega_{j,\vec h_j}=\SIGMA_{n+1}^*(x_{n+1})\}
		\xi^{-1}p(\psi)\psi(\omega_1,\ldots,\omega_k),
	\end{align*}
\item and subsequently choose variable nodes $\vec y_j=(\vec y_{j,1},\ldots,\vec y_{j,k})$ such that $\vec y_{j,\vec h_j}=x_{n+1}$ and
	$\vec y_{j,h}\in\{x_1,\ldots,x_n\}$ for all $h\neq\vec h_j$ such that $\SIGMA_{n+1}^*(y_{j,h})=\vec\omega_{j,h}$ for all $h\in[k]$
	uniformly at random.
\end{enumerate}
Then (\ref{eqASS5}) becomes
	\begin{align}\label{eqASS6}
	\pr\brk{\partial b_j=(x_{ i_1},\ldots,x_{i_k}),\psi_{b_j}=\psi|\tilde\G,\SIGMA_{n+1}^*}=o(1)+
		\pr\brk{\vec y_{j,1}=x_{i_1},\ldots,\vec y_{j,k}=x_{i_k},\hat\PSI_j=\psi}.
	\end{align}

Let $Y=\{\vec y_{j,h}:j\leq\vec\gamma^*,h\in[k]\}\setminus\{x_{n+1}\}$.
Since all weight functions $\psi\in\Psi$ are strictly positive and since $\SIGMA_n^*\in\cE$, the construction (i)--(ii) has the following property: we have
	\begin{align}\label{eqASS_Yfac1}
	\pr\brk{|Y|=(k-1)\vec\gamma^*|\tilde\G,\SIGMA_{n+1}^*}&=1-o(1)
	\end{align}
and there exists $c>0$ such that
	\begin{align}\label{eqASS_Yfac2}
	c^{\vec\gamma^*}&\leq n^{(k-1)\vec\gamma^*}\pr\brk{Y=Y_0|\tilde\G,\SIGMA_{n+1}^*,\vec\gamma^*}\leq c^{-\vec\gamma^*}
		\quad\mbox{for any $Y_0\subset\{x_1,\ldots,x_n\}$, $|Y_0|=(k-1)\vec\gamma^*$.}
	\end{align}
Hence, for any given value of $\vec\gamma^*$ the distribution of $Y$ and the uniform distribution are mutually contiguous.
Since  by (\ref{eqGivenE}) the mean $\vec D=d+o(1)$ of $\vec\gamma^*$ is bounded independently of $T$, (\ref{eqASS_Yfac1}), (\ref{eqASS_Yfac2}),
\Lem~\ref{Lemma_lwise}, Fact~\ref{lem:tiltRScontig}  and Claim~\ref{Claim_CPLTV} yield $\eps_T=o_T(1)$ such that
the event	$\cY=\{\|\mu_{\tilde \G,Y}-\bigotimes_{y\in Y}\mu_{\tilde \G,y}\|_{\mathrm{TV}}\leq\eps_T\mbox{ and }|Y|=(k-1)\vec \gamma^*\}$
satisfies
	\begin{align}\label{eqeventY''}
	\pr\brk{\cY|\tilde\G,\SIGMA_n^*}&\geq1-\eps_T.
	\end{align}

Thus, let
	$$E=		\Erw\brk{
		\ln\sum_{\sigma\in\Omega}\prod_{j=1}^{\vec\gamma^*}\sum_{\tau\in\Omega^{k}}\vecone\{\tau_{\vec h_j}=\sigma\}
			\hat\PSI_j(\tau)\hspace{-3mm}\prod_{h\in[k]\setminus\{\vec h_j\}}
				\mu_{\tilde\G,\vec y_{j,h}}(\tau_h)\bigg|\tilde\G,\SIGMA_n^*}.$$
Then (\ref{eqASS2b}), (\ref{eqASS2c}), (\ref{eqASS6}) and (\ref{eqeventY''}) yield
	\begin{align*}
	\Erw[\ln(Z(\G'')/Z(\tilde\G))|\tilde\G,\SIGMA_n^*]&=o_T(1)+\Erw[\vecone\{\cU\cap\cY\}\ln(Z(\G'')/Z(\tilde\G))|\tilde\G,\SIGMA_n^*]=E+o_T(1).
	\end{align*}
Further, let $(\hat\NU_{j,h,\omega})_{j \geq1,h\geq1,\omega\in\Omega}$ be a family of independent random distributions on $\Omega$
such that $\hat\NU_{j,h,\omega}$ has distribution $\hat\rho_{\tilde\G,\omega}$.
Since by (i)--(ii) above $\mu_{\tilde\G,\vec y_{j,h}}(\tau(\vec y_{j,h}))$ are independent samples from $\rho_{\tilde\G,\SIGMA_n^*,\omega}$,
(\ref{eqMyLemma_reweight2}) yields
	\begin{align}\label{eqASS7}
	E&=
		o_T(1)+
		\Erw\brk{
		\ln\sum_{\sigma\in\Omega}\prod_{j=1}^{\vec\gamma^*}\sum_{\tau\in\Omega^{k}}\vecone\{\tau_{\vec h_j}=\sigma\}
			\hat\PSI_j(\tau)\hspace{-3mm}\prod_{h\in[k]\setminus\{\vec h_j\}}
				\hat\NU_{j,h,\vec\omega_{j,h}}(\tau_h)\bigg|\tilde\G,\SIGMA_n^*}.
	\end{align}
As a next step we plug in the definition (\ref{eqHatRho}) of $\hat\rho_{\tilde\G,\omega}$.
Due to (\ref{eqMyLemma_reweight2}) the denominator of (\ref{eqHatRho})  is $|\Omega|+o(1)$.
Hence, (\ref{eqASS7}) becomes
	\begin{align}\label{eqASS8}
	E&=
		o_T(1)+
		\Erw\brk{|\Omega|^{\vec\gamma^*(k-1)}
			\brk{\prod_{j=1}^{\vec\gamma^*}\prod_{ h\neq\vec h_j}\MU^{(\pi)}_{h+jk}(\vec\omega_{j,h})}
		\ln\sum_{\sigma\in\Omega}\prod_{j=1}^{\vec\gamma^*}\sum_{\tau\in\Omega^{k}}\vecone\{\tau_{\vec h_j}=\sigma\}
			\hat\PSI_j(\tau)\prod_{h\neq\vec h_j}
				\MU^{(\pi)}_{h+jk}(\tau_h)\bigg|\tilde\G,\SIGMA_n^*}.
	\end{align}
Finally, writing out the distribution of $(\vec\omega_j,\hat\PSI_j)$ from (i) above, we obtain from (\ref{eqASS8}) that
	\begin{align*}
	E&=o_T(1)+
			\Erw\brk{\xi^{-\vec\gamma^*}\Lambda\bc{\sum_{\sigma\in\Omega}\prod_{j=1}^{\vec\gamma^*}\sum_{\tau\in\Omega^{k}}\vecone\{\tau_{\vec h_j}=\sigma\}
			\hat\PSI_j(\tau)\prod_{h\neq\vec h_j}
				\MU^{(\pi)}_{h+jk}(\tau_h)}\bigg|\tilde\G,\SIGMA_n^*}.
	\end{align*}
This last equation yields the assertion
because $\SIGMA_{n+1}^*(x_{n+1})$ is chosen uniformly and $\vec D=d+o(1)$ on $\cE$ by (\ref{eqGivenE}).
\end{proof}

\begin{proof}[Proof of \Lem~\ref{Lemma_ASS}]
The coupling {\bf CPL1--CPL2}, {\bf CPL1$'$--CPL2$'$}, {\bf CPL1$''$--CPL3$''$} is such that
$\G'$, $\G''$ are obtained from $\tilde\G$ by adding a Poisson number of constraint nodes such that the mean of the Poisson
distribution is bounded independently of $T$.
Therefore,  we obtain from Claims~\ref{Claim_eqASS3} and~\ref{Claim_eqASS2} that
	\begin{equation}\label{eqAlmostThere}
	\Delta_T(n)\leq o_T(1)+\Erw[\cB(d,\rho_{\tilde\G})].
	\end{equation}
The assertion would be immediate from (\ref{eqAlmostThere}) if $M(\tilde\G)=\int\mu\dd\rho_{\tilde\G}(\mu)$ were equal to the uniform distribution 
	$u=|\Omega|^{-1}\vecone$ on $\Omega$.
While this is generally not the case, \Lem~\ref{Lemma_reweight} shows that $\Erw\TV{M(\tilde\G)-u}=o(1)$.
Therefore, \whp\ there exists $\alpha(\tilde\G)\geq0$ and $\nu(\tilde\G)\in\cP(\Omega)$ such that
	\begin{equation}\label{eqAlmostThere2}
	\Erw[\alpha(\tilde\G)]=o(1)\quad\mbox{and}\quad(1-\alpha(\tilde\G))\rho_{\tilde\G}+\alpha(\tilde\G)\delta_{\nu(\tilde\G)}\in\cP^2_*(\Omega).
	\end{equation}
Finally, since \Lem~\ref{lem:BPcontinuous} shows that $\cB(d,\nix)$ is weakly continuous, the assertion follows from (\ref{eqAlmostThere}) and (\ref{eqAlmostThere2}).
\end{proof}

\subsection{The upper bound}\label{Sec_interpolation}
To prove \Prop~\ref{Prop_interpolation} we will show that for \emph{any} distribution $\pi\in\cP_*^2(\Omega)$,
	\begin{align}\label{eqInterpolation}
	-\frac1n\Erw[\ln Z(\hat\G)]&\leq o(1)- \cB(d,\pi).
	\end{align}
The proof of~(\ref{eqInterpolation}) is based on the interpolation method.
That is, 
for a given $\pi\in\cP_*^2(\Omega)$
we are going to set up a family of random factor graph models parametrized by $t\in[0,1]$ such that the free energy of the $t=0$ model is easily seen to
be $-n\cB(d,\pi)+o(n)$ and such that the $t=1$ model is identical to $\hat\G$.
Finally, we will show that the derivative of the free energy with respect to $t$ is non-positive, whence (\ref{eqInterpolation}) follows.
{\em Throughout this section we assume that {\bf BAL}, {\bf SYM} and {\bf POS} hold.}

\subsubsection{The interpolation scheme}\label{Sec_theIntScheme}
To construct the intermediate models let $\gamma=(\gamma_v)_{v\in[n]}$ be a sequence of integers.
Fix $\pi\in\cP_*^2(\Omega)$.
We define a random factor graph model $\G=\G(n,m,\gamma,\pi)$ as follows.
\begin{description}
\item[G1] the variable nodes are $V=\{x_1,\ldots,x_n\}$.
\item[G2] there are $k$-ary constraint nodes $a_{1},\ldots,a_{m}$; for each $i\in[m]$ independently choose $\partial a_i\in V^k$ uniformly and pick
	an independent $\psi_{a_i}\in\Psi$ from the prior $p$ (cf.\ \Def~\ref{Def_null}).
\item[G3] for each $x\in V$ there are unary constraint nodes $b_{x,1},\ldots,b_{x,\gamma_x}$ adjacent to 
	$x$ whose weight functions are generated as follows: for each $j\in[\gamma_x]$  independently,
		\begin{itemize}
		\item choose $\PSI_{x,j}\in\Psi$ from the prior distribution $p$,
		\item pick $i_{x,j}\in[k]$ uniformly,
		\item with $(\vec\mu_{x,j,h})_{h\in[k]}$ chosen independently from $\pi$, let
				$$\psi_{b_{x,j}}:\sigma\in\Omega\mapsto\sum_{\tau_1,\ldots,\tau_k\in\Omega}\PSI_{x,j}(\tau_1,\ldots,\tau_k)
				\vecone\{\tau_{i_{x,j}}=\sigma\}\prod_{h\neq i_{x,j}}\mu_{x,j,h}(\tau_h).$$
		\end{itemize}
\end{description}

Let $\cG(n,m,\gamma,\pi)$ be the set of all possible outcomes of this experiment.
Depending on $\pi$ the set $\Psi'$ of possible weight functions resulting from {\bf G3} may be infinite and thus we turn 
$\cG(n,m,\gamma,\pi)$ into a measurable space as in \Sec~\ref{Sec_Nishi}.
The fact that the given prior distribution $p$ on $\Psi$ satisfies {\bf SYM} immediately implies that the distribution $p'$ that {\bf G3} induces on $\Psi'$
satisfies {\bf BAL} and {\bf SYM}.
Therefore, so does any convex combination of $p$, $p'$.

We recall that the random factor graph model induces a few further distributions.
First, the Gibbs measure of $G\in\cG(n,m,\gamma,\pi)$ is
	\begin{align*}
	\mu_{G}(\sigma)&=\frac{\psi_{G}(\sigma)}{Z(G)}\quad\mbox{with}&
	\psi_{G}:&\sigma\in\Omega^V\mapsto\prod_{i=1}^m\psi_{a_i}(\sigma(\partial_1a_i,\ldots,\partial_ka_i))\prod_{x\in V}\prod_{j=1}^{\gamma_x}\psi_{b_{x,j}}(\sigma(x)),&
	Z(G)&=\sum_{\sigma\in\Omega^V}\psi_{G}(\sigma).
	\end{align*}
We also obtain a reweighted version  $\hat\G(n,m,\gamma,\pi)$ of the model by letting
	$$\pr\brk{\hat\G(n,m,\gamma,\pi)\in\cA}=\frac{\Erw[Z(\G(n,m,\gamma,\pi))\vecone\{\G(n,m,\gamma,\pi)\in\cA\}]}
		{\Erw[Z(\G(n,m,\gamma,\pi))]}\qquad\mbox{for any event }\cA.$$
Further, there is an induced distribution $\hat\SIGMA_{n,m,\gamma,\pi}$ on assignments defined by
	\begin{align}\label{eqIntHatSigma}
	\pr\brk{\hat\SIGMA_{n,m,\gamma,\pi}=\sigma}&=\Erw[\psi_{\G(n,m,\gamma,\pi)}(\sigma)]/\Erw[Z(\G(n,m,\gamma,\pi))].
	\end{align}
Finally, each assignment $\sigma$ induces a distribution $\G^*(n,m,\gamma,\pi,\sigma)$ on factor graphs by letting
	\begin{align*}
	\pr\brk{\G^*(n,m,\gamma,\pi,\sigma)\in\cA}&=
		\frac{\Erw\brk{\psi_{\G(n,m,\gamma,\pi)}(\sigma)\vecone\{\G(n,m,\gamma,\pi)\in\cA\}}}
			{\Erw[\psi_{\G(n,m,\gamma,\pi)}(\sigma)]}\qquad\mbox{for any event }\cA.
	\end{align*}

We are ready to set up the interpolation scheme.
Given $d>0$, $t\in[0,1]$ we let $\vec m_t=\Po(tdn/k)$.
Moreover, for each $x\in V$ independently we let $\vec\gamma_{t,x}=\Po((1-t)d)$.
Let $\vec\gamma_t=(\vec\gamma_{t,x})_{x\in V}$.
Finally, let
	$$\hat\G_t=\hat\G(n,\vec m_t,\vec\gamma_t,\pi).$$
Then $\hat\G_1$ is identical to our original factor graph model.
Moreover, 
all constraint nodes of $\hat\G_0$ are unary;
in other words, each connected component of $\hat\G_0$ contains just a single variable node.
Since $\vec\gamma_{t,x}$ and $\vec m_t$ are independent Poisson variables, the $\hat\G_t$ model fits
the general random factor graph model from \Sec~\ref{Sec_Nishi} with $\Po(dn(1-(1-1/k)t))$ random constraint nodes chosen with weight functions from $\Psi\cup\Psi'$ chosen from the prior distribution
	\begin{align*}
	p_t&=\frac{t}{k-t(k-1)}p+\frac{k(1-t)}{k-t(k-1)}p'.
	\end{align*}

The construction of $\hat\G_t$ is an adaptation of the interpolation schemes from~\cite{FranzLeone,PanchenkoTalagrand}.  But we need to apply one more twist.
Namely, we are going to use \Lem~\ref{Lemma_pinning} to perturb the intermediate factor graphs $\hat\G_t$ to make them `replica symmetric'.
Thus, for a number $T>0$ consider the following experiment.
	\begin{description}
	\item[INT1] choose an assignment $\check\SIGMA$ from the distribution $\hat\SIGMA_{n,\vec m_t,\vec\gamma_t,\pi}$.
	\item[INT2] generate a factor graph $\G^*(\check\SIGMA,n,\vec m_t,\vec\gamma_t,\pi).$
	\item[INT3] pick $\vec\theta\in[0,T]$ uniformly.
	\item[INT4] obtain $\vU$ by including each $x\in V$ independently with probability $\vec\theta/n$.
		For each $x\in\vU$ add a unary constraint node $\alpha_x$ with probability $\vec\theta/n$  whose 
		 sole adjacent variable node is $x$ and whose weight function is $\psi_{\alpha_x}(\sigma)=\vecone\{\sigma=\check\SIGMA(x)\}$.
	\end{description}
Write $\hat\G_{T,t}=\hat\G_{T,t}(n,\vec m_t,\vec\gamma_t,\pi)$ for  the resulting factor graph.
Then \Prop~\ref{Cor_NishimoriTilt} shows that $\hat\G_{T,t}$ is identical to the model from \Def~\ref{Def_Peg}.
Critically, the number $T>0$ in the following lemma is independent of $t$.

\begin{lemma}\label{Lemma_tpinning}
For any $\eps>0$ there is $T>0$ such that for all $t\in[0,1]$ the Gibbs measure of
$\hat\G_{T,t}$ is $\eps$-symmetric with probability at least $1-\eps$.
\end{lemma}
\begin{proof}
This is immediate from Fact~\ref{lem:tiltRScontig}, where $T$ depends on $\eps$ and $\Omega$ only.
\end{proof}

Finally, we need a correction term.  Let
	$$\Gamma_t=\frac{td(k-1)}{k\xi}\Erw\brk{\Lambda\bc{\sum_{\tau\in\Omega^k}\PSI(\tau)\prod_{j=1}^k\MU_j^{(\pi)}(\tau_j)}}.$$
The following is the centerpiece of the interpolation argument.

\begin{proposition}\label{Lemma_interpolation}
For every $\eps>0$ there is $T>0$ such that for all large enough $n$ the following is true.
Let
 	$$\phi_T:t\in[0,1]\mapsto(\Erw[\ln Z(\hat\G_{T,t})]+\Gamma_t)/n.$$
Then $\phi'_T(t)>-\eps$ for all $t\in[0,1]$.
\end{proposition}

\noindent
We prove \Prop~\ref{Lemma_interpolation} in \Sec~\ref{Sec_Lemma_interpolation}.
But in preparation we first need to construct couplings of the assignments $\hat\SIGMA_{n,m_t,\gamma,\pi}$ for different values of $m_t,\gamma$
in \Sec~\ref{Sec_Coupling_assignments}.
In \Sec~\ref{Sec_Prop_interpolation} we show how \Prop~\ref{Lemma_interpolation} implies \Prop~\ref{Prop_interpolation}.

\subsubsection{Coupling assignments}\label{Sec_Coupling_assignments}
As (\ref{eqIntHatSigma}) shows, 
to study the distribution of the assignment $\hat\SIGMA$ we need to get a handle on the expectations $\Erw[\psi_{\G(n,m,\gamma,\pi)}(\sigma)]$.
Recall 	that $\xi=|\Omega|^{-k}\sum_{\tau\in\Omega^k}\Erw[\PSI(\tau)].$

\begin{lemma}\label{Lemma_intFirstMoment}
For any $\sigma\in\Omega^V$ we have
	$\Erw[\psi_{\G(n,m,\gamma,\pi)}(\sigma)]=\xi^{\sum_{v\in V}\gamma_v}
		\brk{n^{-k}\sum_{\tau_1,\ldots,\tau_k}\Erw[\PSI(\tau_1,\ldots,\tau_k)]\prod_{j=1}^k|\sigma^{-1}(\tau_j)|}^m.$
\end{lemma}
\begin{proof}
In step {\bf G2} the weight functions of the $k$-ary constraint nodes $a_1,\ldots,a_m$ are chosen from $\PSI$
and the neighborhoods $\partial a_i$ are chosen uniformly.
Due to independence their overall contribution to the expectation is just the term in the square brackets.
Further, {\bf G3} ensures that the constraint nodes $b_{x,j}$ are set up independently by choosing a weight function $\PSI$
from the prior distribution and independent $\vec\mu_{x,j,h}$ from $\pi$.
Since $\pi\in\cP_*^2(\Omega)$, assumption {\bf SYM} implies that each $b_{x,j}$ contributes a factor $\xi$
to the expectation.
\end{proof}

\begin{corollary}\label{Cor_intContig}
For any $\gamma$ and $m=O(n)$ the distribution of $\hat\SIGMA_{n,m,\gamma,\pi}$ and the uniform distribution on $\Omega^V$ are mutually contiguous.
Moreover, 
	\begin{align*}
	\pr\brk{\norm{\lambda_{\hat\SIGMA_{n,m,\gamma,\pi}}-|\Omega|^{-1}\vecone}_2>\sqrt n\ln^{2/3}n}&\leq O(n^{-\ln\ln n}).
	\end{align*}
\end{corollary}
\begin{proof}
By \Lem~\ref{Lemma_intFirstMoment} we have
	\begin{align*}
	\pr\brk{\hat\SIGMA_{n,m,\gamma,\pi}=\sigma}&\propto
		\brk{\sum_{\tau_1,\ldots,\tau_k}\Erw[\PSI(\tau_1,\ldots,\tau_k)]\prod_{j=1}^k\lambda_\sigma\bc{\tau_j}}^m.
	\end{align*}
Moreover, by {\bf BAL} the expression on the r.h.s\ attains its maximum if $\lambda_\sigma$ is uniform.
At the same time, the uniform distribution maximizes the entropy $H(\lambda_\sigma)$.
Therefore, the assertion follows immediately from Stirling's formula and the fact that the entropy is strictly concave.
\end{proof}

\begin{corollary}\label{Cor_intFirstMoment1}
For any $\gamma,\gamma'$ the colorings $\hat\SIGMA_{n,m,\gamma,\pi}$, $\hat\SIGMA_{n,m,\gamma',\pi}$ are identically distributed.
\end{corollary}
\begin{proof}
This is immediate from \Lem~\ref{Lemma_intFirstMoment} and the definition of $\hat\SIGMA_{n,m,\gamma,\pi}$,
	$\hat\SIGMA_{n,m,\gamma',\pi}$.
\end{proof}

\begin{corollary}\label{Cor_GenCouple}
Suppose $m=O(n)$.
There is a coupling of $\hat\SIGMA_{n,m,\gamma,\pi}$, $\hat\SIGMA_{n,m+1,\gamma,\pi}$ such that
	$$\pr\brk{\hat\SIGMA_m\neq\hat\SIGMA_{m+1}}=\tilde O(n^{-1})\quad\mbox{and}\quad
		\pr\brk{|\hat\SIGMA_m\triangle\hat\SIGMA_{m+1}|>\sqrt n\ln n}=O(n^{-2}).$$
\end{corollary}
\begin{proof}
The second assertion is immediate from \Cor~\ref{Cor_intContig}.
To prove the first assertion, we need to show that	$\hat\SIGMA_m$, $\hat\SIGMA_{m+1}$ have total variation distance $\tilde O(1/n)$.
To this end, assume that $\|\lambda_\sigma-|\Omega|^{-1}\vecone\|_2=\tilde O(n^{-1/2})$;
the probability mass of $\sigma$ that do not satisfy this condition is negligible under either measure by \Cor~\ref{Cor_intContig}.
We expand
	$$F:\lambda\in\cP(\Omega)\mapsto{
		\sum_{\tau\in\Omega^k}\Erw[\PSI(\tau_1,\ldots,\tau_k)]\prod_{j=1}^k\lambda\bc{\tau_j}}$$
to the second order.
Due to {\bf BAL} the uniform distribution $\bar\lambda$ maximizes 
$\sum_{\tau\in\Omega^k}\Erw[\PSI(\tau_1,\ldots,\tau_k)]\prod_{j=1}^k\lambda\bc{\tau_j}$.
Hence,
	\begin{align}\label{eqD2F}
	F(\lambda)&=F(\bar\lambda)+\frac12\scal{D^2F|_{\bar\lambda}(\lambda-\bar\lambda)}{(\lambda-\bar\lambda)}+O(\|\lambda-\bar\lambda\|_2^3)
		=\xi+O(\|\lambda-\bar\lambda\|_2^2).
	\end{align}
(In fact, since the entropy is strictly concave, condition {\bf BAL} ensures that all eigenvalues of the Hessian $D^2F|_{\bar\lambda}$ on the space
$\{x\in\RR^{\Omega}:x\perp\vecone\}$ are strictly negative.)
Consequently,
 we obtain from \Lem~\ref{Lemma_intFirstMoment} that
in the case $\|\lambda_\sigma-|\Omega|^{-1}\vecone\|_2=\tilde O(n^{-1/2})$,
	\begin{align*}
	\frac{\Erw[\psi_{\G(n,m+1,\gamma,\pi)}(\sigma)]}{\Erw[\psi_{\G(n,m,\gamma,\pi)}(\sigma)]}=
	\sum_{\tau_1,\ldots,\tau_k}\Erw[\PSI(\tau_1,\ldots,\tau_k)]\prod_{j=1}^k\lambda_\sigma(\tau_j)=\exp(\tilde O(1/n))\xi,
	\end{align*}
whence $\hat\SIGMA_m$, $\hat\SIGMA_{m+1}$ have total variation distance $\tilde O(1/n)$.
\end{proof}

\subsubsection{Proof of \Prop~\ref{Lemma_interpolation}.}\label{Sec_Lemma_interpolation}
The proof requires several steps.
The first, summarized in the following proposition, is to derive an expression for the derivative of $\phi_T(t)$.
We write $\bck{\nix}_{T,t}$ for the expectation with respect to the Gibbs measure of $\hat\G_{T,t}$.
Unless specified otherwise $\SIGMA_1,\SIGMA_2,\ldots$ denote independent samples from  $\mu_{\hat\G_{T,t}}$.

\begin{proposition}\label{Prop_deriv}
With $\PSI$ chosen from $p$, $\vec y_1,\ldots,\vec y_k$ chosen uniformly from the set of variable nodes,
and $\vec\mu_1,\ldots,\vec\mu_k$ chosen from $\pi$, all mutually independent and independent of $\hat\G_{T,t}$, let
	\begin{align*}
	\Xi_{t,l}=\Erw\brk{\bck{1-\PSI(\SIGMA(\vec y_1),\ldots,\SIGMA(\vec y_k))}_{T,t}^l}
			&-\sum_{i=1}^k\Erw\brk{\bck{1-\sum_{\tau\in\Omega^k}\PSI(\tau)\vecone\{\tau_i=\SIGMA(\vec y_i)\}\prod_{j\neq i}\vec\mu_j(\tau_j)}_{T,t}^l}\\&
				+(k-1)\Erw\brk{\bc{1-\sum_{\tau\in\Omega^k}\PSI(\tau)\prod_{j=1}^k\vec\mu_j(\tau_j)}^l}.
				\end{align*}
Then  uniformly for all $t\in(0,1)$ and all $T\geq0$,
	$$\frac{\partial}{\partial t}\phi_T(t)=o(1)+\frac d{k\xi}\sum_{l\geq2}\frac{\Xi_{t,l}}{l(l-1)}.$$
\end{proposition}

\noindent
We proceed to prove \Prop~\ref{Prop_deriv}.
Let
	\begin{align*}
	\Delta_t&=\Erw\brk{\ln Z(\hat\G_{T,t}(\vec m_t+1,\vec\gamma_t))}
		-\Erw\brk{\ln Z(\hat\G_{T,t}(\vec m_t,\vec\gamma_t))},\\
	\Delta_t'&=\frac1n\sum_{x\in V}\Erw\brk{\ln Z(\hat\G_{T,t}(\vec m_t,\vec\gamma_t+\vecone_x))}-
			\Erw\brk{\ln Z(\hat\G_{T,t}(\vec m_t,\vec\gamma_t))}.
	\end{align*}

\begin{lemma}\label{Lemma_PoissonDeriv}
We have $\frac1n\frac{\partial}{\partial t}\Erw[\ln Z(\hat\G_{T,t})]=\frac dk\Delta_t-d\Delta_t'.$
\end{lemma}
\begin{proof}
The computation is similar to the one performed in~\cite{PanchenkoTalagrand}.
Let $P_\lambda(j)=\lambda^j\exp(-\lambda)/j!$.
By the construction of the random graph model, the parameter $t$ only enters into the distribution of $\vec m_t,\vec\gamma_t$.
Explicitly, with the sum ranging over all possible outcomes $m,\gamma$,
	\begin{align*}
	\Erw[\ln Z(\hat\G_{T,t})]&=\sum_{m,\gamma}\Erw[\ln Z(\hat\G_{T,t})|\vec m_t=m,\vec\gamma_t=\gamma]
		P_{tdn/k}(m)\prod_{x\in V}P_{(1-t)d}(\gamma_x).
	\end{align*}
We recall that
	\begin{align*}
	\frac\partial{\partial t}P_{tdn/k}(m)&=\frac1{m!}\frac\partial{\partial t}\bcfr{tdn}{k}^m\exp(-tdn/k)
		=\frac{dn}k\brk{\vecone\{m\geq1\}P_{tdn/k}(m-1)-P_{tdn/k}(m)},\\
	\frac\partial{\partial t}P_{(1-t)d}(\gamma_v)&=\frac1{\gamma_v!}\frac\partial{\partial t}\bc{(1-t)d}^{\gamma_v}\exp(-(1-t)d)
		=-d\brk{\vecone\{\gamma_v\geq1\}P_{(1-t)d}(\gamma_v-1)-P_{(1-t)d}(\gamma_v)}.
	\end{align*}
Hence, by the product rule
	\begin{align*}
	\frac1n\frac{\partial}{\partial t}\Erw[\ln Z(\hat\G_{T,t})]&=
		\frac1n\sum_{m,\gamma}\Erw[\ln Z(\hat\G_{T,t})|\vec m_t=m,\vec\gamma_t=\gamma]
			\frac\partial{\partial t}P_{tdn/k}(m)\prod_{v\in[n]}P_{(1-t)d}(\gamma_v)\\
		&=\frac dk\sum_m\brk{\Erw\brk{\ln Z(\hat\G_{T,t})|\vec m_t=m+1}-\Erw\brk{\ln Z(\hat\G_{T,t})|\vec m_t=m}}P_{tdn/k}(m)\\
		&\quad-\frac dn\sum_{x\in V}\sum_{\gamma_x}
			\brk{\Erw\brk{\ln Z(\hat\G_{T,t})|\vec\gamma_{t,x}=\gamma_x+1}-
				\Erw\brk{\ln Z(\hat\G_{T,t})|\vec\gamma_{t,x}=\gamma_{t,x}}}P_{(1-t)d}(\gamma_{t,x})\\
		&=\frac dk\brk{\Erw\brk{\ln Z(\hat\G_{T,t}(\vec m_t+1,\vec\gamma_t)}-\Erw\brk{\ln Z(\hat\G_{T,t}(n,\vec m_t,\vec\gamma_t)}}\\
		&\quad	-\frac dn\sum_x
			\brk{\Erw\brk{\ln Z(\hat\G_{T,t}(\vec m_t,\vec\gamma_t+\vecone_x)}-\Erw\brk{\ln Z(\hat\G_{T,t}(\vec m_t,\vec\gamma_t)}},
	\end{align*}
as claimed.
\end{proof}

\noindent
To calculate $\Delta_t,\Delta_t'$ we
 continue to denote by $\PSI$ a weight function chosen from the prior distribution, independently of everything else.

\begin{lemma}\label{Lemma_Deltat}
We have
	$\displaystyle
	\Delta_t=o(1)-\frac{1-\xi}\xi+
		\frac1{n^k\xi}\sum_{y_1,\ldots,y_k\in V}\sum_{l\geq2}\frac{1}{l(l-1)}
			\Erw\bck{{\prod_{h=1}^{l}1-\PSI(\SIGMA_h(y_1),\ldots,\SIGMA_h(y_k))}}_{T,t}.
	$
\end{lemma}
\begin{proof}
Because the tails of the Poisson distribution decay sub-exponentially and since 
	$$\ln Z(\hat\G_{T,t}(\vec m_t,\vec\gamma _t))=O\bc{n+\vec m_t+\sum_{x\in V}\gamma_{t,x}},$$ we may safely assume that 
	\begin{equation}\label{eqLemma_Deltat1}
	\vec m_t+\sum_{x\in V}\gamma_{t,x}\leq (d+1)n.
	\end{equation}
By \Cor~\ref{Cor_GenCouple} we can couple the two assignments $\hat\SIGMA'=\hat\SIGMA_{n,m,\vec\gamma_t,\vec m_t},\hat\SIGMA''=\hat\SIGMA_{n,m,\vec\gamma_t,\vec m_t+1}$
such that 
	\begin{align}\label{eqLemma_Deltat2}
	\pr\brk{\hat\SIGMA'=\hat\SIGMA''}&=1-\tilde O(n^{-1}),&
		\pr\brk{|\hat\SIGMA'\triangle\hat\SIGMA''|>\sqrt n\ln n}&=O(n^{-2}).
	\end{align}
We are going to extend this to a coupling of $\hat\G_{T,t}(n,\vec m_t,\vec\gamma_t,\pi)$, $\hat\G_{T,t}(n,\vec m_t+1,\vec\gamma_t,\pi)$.
Specifically, given $\hat\SIGMA',\hat\SIGMA''$ we construct a pair $(\G',\G'')$ of factor graphs as follows.
	\begin{description}
	\item[Case 1: $\hat\SIGMA'=\hat\SIGMA''$] then we define
		$\G'$ as the outcome of {\bf INT1--INT4} with 
			$\check\SIGMA=\hat\SIGMA'=\hat\SIGMA''$.
		Further, $\G''$ is obtained from $\G'$ by adding one single $k$-ary constraint node $\vec a$ such that $\partial\vec a$,
		$\psi_{\vec a}$ have distribution
			\begin{align}\label{eqLemma_Deltat3}
			\pr\brk{\partial\vec a=(x_{i_1},\ldots,x_{i_k}),\psi_{\vec a}=\psi}&\propto p(\psi)\psi(\hat\SIGMA'(x_{i_1}),\ldots,
				\hat\SIGMA'(x_{i_k}))
					\qquad(i_1,\ldots,i_k\in[n],\psi\in\Psi).
			\end{align}
	\item[Case 2: $|\hat\SIGMA'\triangle\hat\SIGMA''|\leq\sqrt n\ln n$]
		consider the probability distributions $q',q''$ on $V^k\times\Psi$ defined by
			\begin{align*}
			q'(y_1,\ldots,y_k,\psi)&\propto p(\psi)\psi(\hat\SIGMA'(y_1),\ldots,\hat\SIGMA'(y_{k})),\\
			q''(y_1,\ldots,y_k,\psi)&\propto p(\psi)\psi(\hat\SIGMA''(y_1),\ldots,\hat\SIGMA''(y_{k})).	
			\end{align*}
		Since $|\hat\SIGMA'\triangle\hat\SIGMA''|\leq\sqrt n\ln n$
		these two distributions have total variation distance $\tilde O(n^{-1/2})$.
		Consequently, we can couple  $\G^*(n,\vec m_t,\vec\gamma_t,\pi,\hat\SIGMA')$
		and $\G^*(n,\vec m_t+1,\vec\gamma_t,\pi,\hat\SIGMA'')$
		such that with probability $1-O(n^{-2})$ no more than $\tilde O(\sqrt n)$ constraint nodes
		either have different neighborhoods or different weight functions.
		Let $(\G',\G'')$ be the outcome of this coupling subjected to pinning the same set $\vU$ of variable nodes to $\hat\SIGMA'$, $\hat\SIGMA''$,
		respectively.
	\item[Case 3: $|\hat\SIGMA'\triangle\hat\SIGMA''|>\sqrt n\ln n$]
			choose $\G^*(n,\vec m_t,\vec\gamma_t,\pi,\hat\SIGMA')$ and
					$\G^*(n,\vec m_t+1,\vec\gamma_t,\pi,\hat\SIGMA'')$ independently and obtain $\G',\G''$ by pinning.
	\end{description}
The construction ensures that $(\G',\G'')$ is a coupling of $\hat\G_{T,t}(n,\vec m_t,\vec\gamma_t,\pi)$, $\hat\G_{T,t}(n,\vec m_t+1,\vec\gamma_t,\pi)$.
Hence,
	\begin{align}\label{eqLemma_Deltat4}
	\Erw\brk{\ln Z(\hat\G_{T,t}(n,\vec m_t+1,\vec\gamma_t,\pi))}-\Erw\brk{\ln Z(\hat\G_{T,t}(n,\vec m_t,\vec\gamma_t,\pi))}
		&=\Erw\brk{\ln\frac{Z(\G'')}{Z(\G')}}.
	\end{align}		
Further, (\ref{eqLemma_Deltat1}) and (\ref{eqLemma_Deltat2}) and the construction in case 2 ensure that
	\begin{align}
	\Erw\brk{\ln\frac{Z(\G'')}{Z(\G')}}&=
		\Erw\brk{\ln\frac{Z(\G'')}{Z(\G')} \Bigg | \hat\SIGMA'=\hat\SIGMA''}\pr\brk{\hat\SIGMA'=\hat\SIGMA''}\nonumber\\&
			\qquad+\Erw\brk{\ln\frac{Z(\G'')}{Z(\G')} \Bigg | 1\le |\hat\SIGMA'\triangle\hat\SIGMA''|\leq\sqrt n\ln n } \pr\brk{1\le |\hat\SIGMA'\triangle\hat\SIGMA''|\leq\sqrt n\ln n}
			\nonumber\\&\qquad+
		\Erw\brk{\ln\frac{Z(\G'')}{Z(\G')} \Bigg |\hat\SIGMA'\triangle\hat\SIGMA''|>\sqrt n\ln n }
			\pr\brk{|\hat\SIGMA'\triangle\hat\SIGMA''|>\sqrt n\ln n}\nonumber\\
		&=\Erw\brk{\ln\frac{Z(\G'')}{Z(\G')}\Bigg|\hat\SIGMA'=\hat\SIGMA''}+\tilde O(n^{-1/2}).
		\label{eqLemma_Deltat5}
	\end{align}	
Thus, if we denote by $\vec a$ an additional random factor node drawn from the distribution (\ref{eqLemma_Deltat3}), regardless
whether  or not $\hat\SIGMA'=\hat\SIGMA''$, then (\ref{eqLemma_Deltat2}), (\ref{eqLemma_Deltat4}) and (\ref{eqLemma_Deltat5}) yield
	\begin{align}\nonumber
	\Erw\brk{\ln Z(\hat\G_{T,t}(\vec m_t+1,\vec\gamma_t))}-\Erw\brk{Z(\hat\G_{T,t}(\vec m_t,\vec\gamma_t))}
		&=\Erw\brk{\ln\bck{\psi_{\vec a}(\SIGMA_{\G'})}_{\G'}\bigg|\hat\SIGMA'=\hat\SIGMA''}+\tilde O(n^{-1/2})\\
		&=\Erw\brk{\ln\bck{\psi_{\vec a}(\SIGMA_{\G'})}_{\G'}}+\tilde O(n^{-1/2})\label{eqLemma_Deltat6}.
	\end{align}		

Hence, we are left to compute $\Erw\brk{\ln\bck{\psi_{\vec a}(\SIGMA_{\G'})}_{\G'}}$.
Writing $\SIGMA,\SIGMA_1,\SIGMA_2,\ldots$ for independent samples from $\mu_{\G'}$ and
plugging in the definition (\ref{eqLemma_Deltat3}) of $\vec a$, we find
	\begin{align*}
	\Erw\brk{\ln\bck{\psi_{\vec a}(\SIGMA_{\G'})}_{\G'}}&=
		\frac{\sum_{y_1,\ldots,y_k\in V}\Erw\brk{\PSI(\hat\SIGMA'(y_1),\ldots,\hat\SIGMA'(y_k))
				\ln\bck{\PSI(\SIGMA(y_1),\ldots,\SIGMA(y_k))}_{\G'}}}
			{\sum_{y_1,\ldots,y_k\in V}\Erw\brk{\PSI(\hat\SIGMA'(y_1),\ldots,\hat\SIGMA'(y_k))}}.
	\end{align*}
Since by \Cor~\ref{Cor_intContig} the empirical distribution $\lambda_{\hat\SIGMA'}$ is asymptotically uniform with very high probability,
the denominator in the above expression equals $n^k(\xi+o(1))$ with probability $1-O(n^{-2})$.
Thus,
	\begin{align}\label{eqMont}
	\Erw\brk{\ln\bck{\psi_{\vec a}(\SIGMA_{\G'})}_{\G'}}&=o(1)+
		\frac1{n^k\xi}{\sum_{y_1,\ldots,y_k\in V}\Erw\brk{\PSI(\hat\SIGMA'(y_1),\ldots,\hat\SIGMA'(y_k))
				\ln\bck{\PSI(\SIGMA(y_1),\ldots,\SIGMA(y_k))}_{\G'}}}.
	\end{align}
Further, because all weight functions $\psi\in\Psi$ take values in $(0,2)$, expanding the logarithm gives
	\begin{align*}
	\ln\bck{\PSI(\SIGMA(y_1),\ldots,\SIGMA(y_k))}_{\G'}&=-\sum_{l\geq1}\frac1l\bck{1-\PSI(\SIGMA(y_1),\ldots,\SIGMA(y_k))}_{\G'}^l
		=-\sum_{l\geq1}\frac1l\bck{\prod_{h=1}^l1-\PSI(\SIGMA_h(y_1),\ldots,\SIGMA_h(y_k))}_{\G'};
	\end{align*}
the second equality sign holds because $\SIGMA_1,\SIGMA_2,\ldots$ are mutually independent.
Combining the last two equations, we obtain 
	\begin{align}\nonumber
	\Erw\brk{\ln\bck{\psi_{\vec a}(\SIGMA_{\G'})}_{\G'}}&=
		o(1)-\sum_{l\geq1}\sum_{y_1,\ldots,y_k}
			\Erw\brk{\frac{\PSI(\hat\SIGMA'(y_1),\ldots,\hat\SIGMA'(y_k))}{ln^k\xi}\bck{\prod_{h=1}^l1-\PSI(\SIGMA_h(y_1),\ldots,\SIGMA_h(y_k))}_{\G'}}\\
		&=o(1)+\sum_{l\geq1}\frac1{ln^k\xi}\sum_{y_1,\ldots,y_k}
			\Erw\brk{(1-\PSI(\hat\SIGMA'(y_1),\ldots,\hat\SIGMA'(y_k)))\bck{\prod_{h=1}^l1-\PSI(\SIGMA_h(y_1),\ldots,\SIGMA_h(y_k))}_{\G'}}\nonumber\\
		&\qquad\qquad\qquad-\sum_{l\geq1}\frac1{ln^k\xi}\sum_{y_1,\ldots,y_k}
			\Erw\bck{\prod_{h=1}^l1-\PSI(\SIGMA_h(y_1),\ldots,\SIGMA_h(y_k))}_{\G'}.\label{eqLemma_Deltat10}
	\end{align}
Since \Prop~\ref{Cor_NishimoriTilt} implies that given $\G'$ the assignment $\hat\SIGMA'$ is distributed as a sample from the Gibbs measure $\mu_{\G'}$, we obtain
	\begin{align*}
	\Erw\brk{\bc{1-\PSI(\hat\SIGMA'(y_1),\ldots,\hat\SIGMA'(y_k))}\bck{\prod_{h=1}^l1-\PSI(\SIGMA_h(y_1),\ldots,\SIGMA_h(y_k))}_{\G'}}
		&=\Erw\bck{\prod_{h=1}^{l+1}1-\PSI(\SIGMA_h(y_1),\ldots,\SIGMA_h(y_k))}_{\G'}
	\end{align*}
 for $l\geq1$.
 Moreover, by \Cor~\ref{Cor_intContig} 
	\begin{align*}
	\frac1{n^k}\sum_{y_1,\ldots,y_k}\Erw\bck{1-\PSI(\SIGMA(y_1),\ldots,\SIGMA(y_k))}_{\G'}&=1-
		\sum_{y_1,\ldots,y_k}\frac{\Erw\brk{\PSI(\hat \SIGMA'(y_1),\ldots,\hat \SIGMA'(y_k))}}{n^k}=1-\xi+o(1).
	\end{align*}
Plugging these two into (\ref{eqLemma_Deltat10}) and simplifying, we finally obtain
	\begin{align*}
	\Erw\brk{\ln\bck{\psi_{\vec a}(\SIGMA_{\G'})}_{\G'}}
		&=o(1)-\frac{1-\xi}\xi+
			\sum_{l\geq2}\sum_{y_1,\ldots,y_k}\frac1{l(l-1)n^k\xi}
				\Erw\bck{\prod_{h=1}^l1-\PSI(\SIGMA_h(y_1),\ldots,\SIGMA_h(y_k))}_{\G'}
	\end{align*}
and the assertion follows from (\ref{eqLemma_Deltat6}).
\end{proof}

\noindent
The steps that we just followed from (\ref{eqMont}) onward to calculate $\Erw\,{\ln\bck{\psi_{\vec a}(\SIGMA_{\G'})}_{\G'}}$ are similar to the manipulations
from the interpolation argument of Abbe and Montanari~\cite{Abbe}.
Similar manipulations will be used in the proof of the next two lemmas.

\begin{lemma}\label{Lemma_Deltat'}
With $\MU_1,\MU_2,\ldots$ chosen from $\pi$ mutually independently and independently of everything else,
	\begin{align*}
	\Delta_t'	&=-\frac\xi{1-\xi}+\sum_{l\geq2}\frac{1}{l(l-1)kn\xi}
				\Erw\bck{\sum_{x\in V,i\in[k]}{\prod_{h=1}^l1-\sum_{\tau\in\Omega^k}\PSI(\tau)\vecone\{\tau_i=\SIGMA_h(x)\}
					\prod_{j\neq i}\MU_j(\tau_j)}}_{T,t}.
	\end{align*}
\end{lemma}
\begin{proof}
By \Cor~\ref{Cor_intFirstMoment1}, $\hat\SIGMA_{n,m,\vec\gamma_t,\vec m_t},\hat\SIGMA_{n,m,\vec\gamma_t+\vecone_x,\vec m_t}$
are identically distributed.
Hence, let $\hat\SIGMA=\hat\SIGMA_{n,m,\vec\gamma_t,\vec m_t}$ for brevity and write $\vec x$ for a uniformly random element of $V$.
Starting from $\hat\SIGMA$ we can easily construct a coupling $(\G',\G'')$ of 
$\hat\G_{T,t}(n,\vec m_t,\vec\gamma_t,\pi)$ and $\hat\G_{T,t}(n,\vec m_t,\vec\gamma_t+\vecone_{\vec x},\pi)$.
Namely, let $\G'=\G^*_{T}(n,\vec m_t,\vec\gamma_t,\pi,\hat\SIGMA)$.
Then obtain $\G''$ by choosing $\vec x\in V$ (independently of $\G'$)
and add a unary constraint node $\vec b$ adjacent to $\vec x$ whose weight function is distributed as follows.
Pick an index $\vec i\in[k]$, a weight function $\psi_{\vec b,*}\in\Psi$ and $\hat\MU_1,\ldots,\hat\MU_k$ from the distribution
	\begin{align}\label{eqLemma_Deltat'1}
	\pr\brk{\vec i=i,(\hat\MU_1,\ldots,\hat\MU_k)\in\cA,\psi_{\vec b,*}=\psi}
		&=
	\frac{\sum_{\tau\in\Omega^k}\vecone\{\tau_{i}=\hat\SIGMA(\vec x)\}\psi(\tau)\int_{\cA}\prod_{j\neq i}\hat\mu(\tau_j)\dd\pi^{\tensor k}(\hat\mu_1,\ldots,\hat\mu_k)}
		{\sum_{\tau\in\Omega^k}\sum_{i=1}^k\vecone\{\tau_i=\hat\SIGMA(\vec x)\}
			\Erw[\PSI(\tau)\prod_{j\neq i}\MU_j(\tau_j)]}.
	\end{align}
Then the weight function associated with $\vec b$ is
	$$\psi_{\vec b}(\sigma)=\sum_{\tau\in\Omega^k}\vecone\{\tau_{\vec i}=\sigma\}\psi_{\vec b,*}(\tau)\prod_{j\neq\vec i}\hat\MU_j(\tau_j).$$
\Prop~\ref{Cor_NishimoriTilt} implies that $\G'$ is distributed as $\hat\G_{T,t}(n,\vec m_t,\vec\gamma_t,\pi)$ and that
$\G''$ is distributed as $\hat\G_{T,t}(n,\vec m_t,\vec\gamma_t+\vecone_{\vec x},\pi)$.

Therefore, with $\SIGMA,\SIGMA_1,\ldots$ denoting independent samples from $\mu_{\G'}$,
	\begin{align}\label{eqLemma_Deltat'2}
	\Erw[\ln Z(\hat\G_{T,t}(\vec m_t,\vec\gamma_t+\vecone_{\vec x}))]&-\Erw[\ln Z(\hat\G_{T,t}(\vec m_t,\vec\gamma_t))]=
		\Erw\ln(Z(\G'')/Z(\G'))=\Erw\ln\bck{\psi_{\vec b}(\SIGMA(\vec x))}_{\G'}.
	\end{align}
Because $\int\mu\dd\pi(\mu)$ is the uniform distribution,
assumption {\bf SYM} ensures that the denominator on the r.h.s.\ of (\ref{eqLemma_Deltat'1}) equals $k\xi$.
Therefore,
	\begin{align*}
	\Erw\ln\bck{\psi_{\vec b}(\SIGMA(\vec x))}_{\G'}
		&=\frac1{kn\xi}\sum_{x\in V}\sum_{i=1}^k
			\Erw\brk{\sum_{\tau\in\Omega^k}\vecone\{\tau_i=\hat\SIGMA(x)\}\PSI(\tau)\prod_{j\neq i}\MU_j(\tau_j)
				\ln\bck{\sum_{\sigma\in\Omega^k}\vecone\{\sigma_i=\SIGMA(x)\}\PSI(\sigma)\prod_{j\neq i}\MU_j(\sigma_j)}_{\G'}}.
	\end{align*}
Further, since the weight functions take values in $(0,2)$, expanding the logarithm yields
	\begin{align}
	\Erw\ln\bck{\psi_{\vec b}(\SIGMA(\vec x))}_{\G'}&=
		-\sum_{x\in V}\sum_{i=1}^k\sum_{l\geq1}\frac1{kln\xi}
			\Erw\Bigg[\sum_{\tau\in\Omega^k}\vecone\{\tau_i=\hat\SIGMA(x)\}\PSI(\tau)\prod_{j\neq i}\MU_j(\tau_j)\nonumber\\
			&\qquad\qquad\qquad\qquad\qquad\qquad
				\bck{\prod_{h=1}^l1-\sum_{\sigma\in\Omega^k}\vecone\{\sigma_i=\SIGMA_h(x)\}\PSI(\sigma)\prod_{j\neq i}\MU_j(\sigma_j)}_{\G'}\Bigg]\nonumber\\
		&\hspace{-2cm}=\sum_{x\in V}\sum_{i=1}^k\sum_{l\geq1}\frac1{kln\xi}
			\Erw\Bigg[\bc{1-\sum_{\tau\in\Omega^k}\vecone\{\tau_i=\hat\SIGMA(x)\}\PSI(\tau)\prod_{j\neq i}\MU_j(\tau_j)}
				\bck{\prod_{h=1}^l1-\sum_{\sigma\in\Omega^k}\vecone\{\sigma_i=\SIGMA_h(x)\}\PSI(\sigma)\prod_{j\neq i}\MU_j(\sigma_j)}_{\G'}\nonumber\\
		&\qquad\qquad\qquad\qquad\qquad\qquad\qquad\qquad\qquad
			-\bck{\prod_{h=1}^l1-\sum_{\sigma\in\Omega^k}\vecone\{\sigma_i=\SIGMA_h(x)\}\PSI(\sigma)\prod_{j\neq i}\MU_j(\sigma_j)}_{\G'}\Bigg].
			\label{eqLemma_Deltat'3}
	\end{align}
Since by \Prop~\ref{Cor_NishimoriTilt}
the conditional distribution of $\hat\SIGMA$ given $\G'$ coincides with the Gibbs measure $\mu_{\G'}$, we find
	\begin{align}			\nonumber
	\Erw\brk{\bc{1-\sum_{\tau\in\Omega^k}\vecone\{\tau_i=\hat\SIGMA(x)\}\PSI(\tau)\prod_{j\neq i}\MU_j(\tau_j)}
				\bck{\prod_{h=1}^l1-\sum_{\sigma\in\Omega^k}\vecone\{\sigma_i=\SIGMA_h(x)\}\PSI(\sigma)\prod_{j\neq i}\MU_j(\sigma_j)}_{\G'}}\\
					&\hspace{-6cm}=\Erw\bck{\prod_{h=1}^{l+1}1-\sum_{\sigma\in\Omega^k}\vecone\{\sigma_i=\SIGMA_h(x)\}\PSI(\sigma)\prod_{j\neq i}\MU_j(\sigma_j)}_{\G'}.\label{eqLemma_Deltat'4}
	\end{align}
Moreover, since $\int\mu\dd\pi(\mu)\in\cP(\Omega)$ is the uniform distribution, {\bf SYM} implies
	\begin{align}			\label{eqLemma_Deltat'5}
	\Erw\bck{1-\sum_{\sigma\in\Omega^k}\vecone\{\sigma_i=\SIGMA(x)\}\PSI(\sigma)\prod_{j\neq i}\MU_j(\sigma_j)}_{\G'}&=1-\xi.
	\end{align}
Plugging (\ref{eqLemma_Deltat'4}) and (\ref{eqLemma_Deltat'5}) into (\ref{eqLemma_Deltat'3}), we obtain
	\begin{align*}
		\Erw\ln\bck{\psi_{\vec b}(\SIGMA(\vec x))}_{\G'}&=
-\frac\xi{1-\xi}+\sum_{l\geq2}\frac{1}{l(l-1)kn\xi}
				\Erw\bck{\sum_{x\in V,i\in[k]}{\prod_{h=1}^l1-\sum_{\tau\in\Omega^k}\PSI(\tau)\vecone\{\tau_i=\SIGMA_h(x)\}
					\prod_{j\neq i}\MU_j(\tau_j)}}_{\G'}\end{align*}
and the assertion follows from (\ref{eqLemma_Deltat'2}).
\end{proof}

\begin{lemma}\label{Lemma_Deltat''}
With $\vec\mu_1,\vec\mu_2$ chosen independently from $\pi$ we have
	\begin{align*}
	\Delta_t''&=\frac1{d(k-1)}\frac{\partial}{\partial t}\Gamma_t=
		-\frac\xi{1-\xi}+\frac{1}\xi\sum_{l\geq2}\frac{1}{l(l-1)}
			\Erw\brk{\bc{1-\sum_{\tau\in\Omega^k}\PSI(\tau)\prod_{j=1}^k\MU_j(\tau_j)}^l}
	\end{align*}
\end{lemma}
\begin{proof}
{This follows by expanding the logarithm in the expression that defines $\Gamma_t$.}
\end{proof}

\noindent
\Prop~\ref{Prop_deriv} is now immediate from \Lem s~\ref{Lemma_PoissonDeriv}--\ref{Lemma_Deltat''}.

\begin{proof}[Proof of \Prop~\ref{Lemma_interpolation}]
Let $\rho_{\hat\G_{T,t}}$ be the empirical distribution of the marginals of $\mu_{\hat\G_{T,t},x}$; in symbols,
	\begin{align*}
	\rho_{\hat\G_{T,t}}&=\frac1n\sum_{x\in V}\delta_{\mu_{\hat\G_{T,t},x}}\in\cP^2(\Omega).
	\end{align*}
Write $\NU_1,\NU_2,\ldots$ for independent samples drawn from $\rho_{\hat\G_{T,t}}$ and define
	\begin{align*}
	\Xi_{t,l}'&=
	\Erw\brk{\bc{1-\sum_{\sigma\in\Omega^k}\PSI(\sigma)\prod_{j=1}^k\NU_j(\sigma_j)}^l
		-\sum_{i=1}^k\bc{1-\sum_{\tau\in\Omega^k}\PSI(\tau)\NU_1(\tau_i)\prod_{j\neq i}\vec\mu_j(\tau_j)}^l
			+(k-1)\brk{1-\sum_{\tau\in\Omega^k}\PSI(\tau)\prod_{j=1}^k\vec\mu_j(\tau_j)}^l}.
	\end{align*}
\Lem~\ref{Lemma_lwise} implies that for any $\eps>0$, $l\geq1$ there is $\delta>0$ such that in the case that $\hat\G_{T,t}$ is $\delta$-symmetric 
for any $\psi\in\Psi$, $i\in[k]$ we have
	\begin{align*}
	\abs{\frac1{n^k}\sum_{y_1,\ldots,y_k\in V}\bck{{1-\psi(\SIGMA( y_1),\ldots,\SIGMA( y_k))}}_{\hat\G_{T,t}}^l
		-\Erw\brk{\bc{1-\sum_{\sigma\in\Omega^k}\psi(\sigma)\prod_{j=1}^k\NU_j(\sigma_j)}^l\bigg|\hat\G_{T,t}}}<\eps,\\
	\abs{\frac1{n}\sum_{y\in V}\bck{1-\sum_{\tau\in\Omega^k}\psi(\tau)\vecone\{\tau_i=\SIGMA_h(y)\}\prod_{j\neq i}\vec\mu_j(\tau_j)}_{\hat\G_{T,t}}^l
		-\Erw\brk{\bc{1-\sum_{\tau\in\Omega^k}\psi(\tau)\NU_1(\tau_i)\prod_{j\neq i}\vec\mu_j(\tau_j)}^l\bigg|\hat\G_{T,t}}}<\eps.
	\end{align*}
Since $\hat\G_{T,t}$ is $o_T(1)$-symmetric with probability $1-o_T(1)$ by \Lem~\ref{Lemma_pinning}, 
we therefore conclude that
	\begin{align}\label{eqXiXi'}
	\abs{\Xi_{t,l}-\Xi_{t,l}'}=o_T(1).
	\end{align}
Furthermore, \Lem~\ref{Prop_NishimoriTilt} implies together with \Cor~\ref{Cor_intContig} that
$\int\mu\dd\rho_{\check\G_{T,t}}(\mu)$ is within total variation distance $o(1)$ of the uniform distribution \whp\ 
Therefore, {\bf POS} implies that $\Xi_{t,l}'\geq o(1)$.
Finally, the assertion follows from \Prop~\ref{Prop_deriv} and (\ref{eqXiXi'}), and the fact that the terms $\Xi_{t,l}$ decay rapidly in $l$.
\end{proof}

\subsubsection{Proof of \Prop~\ref{Prop_interpolation}}\label{Sec_Prop_interpolation}
Let us recap what we learned from \Prop~\ref{Lemma_interpolation}.

\begin{lemma}\label{Lemma_recap}
For any distribution $\pi\in\cP_*^2(\Omega)$ we have
	$$\liminf_{n\to\infty}\frac1n\Erw[\ln Z(\hat\G)]\geq\liminf_{n\to\infty}\frac1n\Erw[\ln Z(\hat\G_{0,0})]-\Gamma_1.$$
\end{lemma}
\begin{proof}
Together with the fundamental theorem of calculus
\Prop~\ref{Lemma_interpolation} implies that for any $\eps>0$ there is $T=T(\eps)>0$ (independent of $n$) such that for large enough $n$,
	\begin{align}\label{eqLemma_recap}
	\frac1n\Erw[\ln Z(\hat\G_{T,1})]&\geq\frac1n\Erw[\ln Z(\hat\G_{T,0})]-\Gamma_1-\eps.
	\end{align}
Furthermore, by \Lem~\ref{Prop_NishimoriTilt} $\hat\G_{T,1}$ results from $\hat\G$ simply by attaching a random number of constraint nodes with $\{0,1\}$-valued weight functions.
Therefore, $\Erw[\ln Z(\hat\G_{T,1})]\leq\Erw[\ln Z(\hat\G)]$.
Similarly, by \Lem~\ref{Prop_NishimoriTilt} we can think of $\hat\G_{T,0}$ as being obtained from $\hat\G_{0,0}$ by adding a few constraint nodes with $\{0,1\}$-weights.
The expected number of these constraint nodes does not exceed $T$, which remains fixed as $n\to\infty$, and each connected component of 
$\hat\G_{T,0}$ contains only a single variable node and a $\Po(d)$ number of unary constraint nodes.
Consequently, $\Erw[\ln Z(\hat\G_{T,0})]=\Erw[\ln Z(\hat\G_{0,0})]+o(n)$ and the assertion follows from (\ref{eqLemma_recap}).
\end{proof}

Thus, we are left to calculate $\Erw[\ln Z(\hat\G_{0,0})]$.
That is straightforward because every connected component of $\hat\G_{0,0}$ contains just a single variable node.

\begin{lemma}\label{Lemma_00}
With  independent $\vec\gamma=\Po(d)$, $\PSI_i$ from $p$, $\vec\mu_{ij}$ chosen from $\pi$ and uniform $\vec h_i\in[k]$ we have
	\begin{align*}
	\frac1n\Erw[\ln Z(\hat\G_{0,0})]&=\frac1{|\Omega|}\Erw\brk{\xi^{-\vec\gamma}
			\Lambda\bc{\sum_{\sigma\in\Omega}\prod_{b=1}^{\vec\gamma}\sum_{\tau\in\Omega^k}\vecone\{\tau_{\vec h_b}=\sigma\}\PSI_b(\tau)\prod_{j\neq \vec h_b}\vec\mu_{bj}(\tau_j)}}.
	\end{align*}
\end{lemma}
\begin{proof}
Because the random graph model is symmetric under permutations of the variable nodes,
we can view $\frac1n\Erw[\ln Z(\hat\G_{0,0})]$ as the contribution to $\Erw[\ln Z(\hat\G_{0,0})]$ of the connected component of $x_1$.
The partition function of the component of $x_1$ is nothing but
	\begin{align*}
	z&=\sum_{\sigma\in\Omega}\prod_{j=1}^{\vec\gamma_{x_1}}\psi_{b_{x_1,j}}(\sigma).
	\end{align*}
Furthermore, by construction at $t=0$ the degree $\vec\gamma_{x_1}$ is chosen from the Poisson distribution $\Po(d)$.
Hence, recalling the distribution of the weight functions $\psi_{b_{x_1,j}}$, $j\leq \vec\gamma_{x_1}$
from {\bf G3} in \Sec~\ref{Sec_theIntScheme},
we find
	\begin{align*}
	\frac1n\Erw[\ln Z(\hat\G_{0,0})]&=\Erw\brk{z}=
	\frac1{|\Omega|}\Erw\brk{\xi^{-\vec\gamma_{x_1}}
			\Lambda\bc{\sum_{\sigma\in\Omega}\prod_{j=1}^{\vec\gamma_{x_1}}\sum_{\tau\in\Omega^k}\vecone\{\tau_{\vec h_j}=\sigma\}\PSI_b(\tau)\prod_{i\neq \vec h_j}\vec\mu_{ij}(\tau_i)}},
	\end{align*}
as desired.
\end{proof}

\noindent
Finally, \Prop~\ref{Prop_interpolation} is immediate from \Lem s~\ref{Lemma_recap} and~\ref{Lemma_00}.

\subsection{Proof of \Thm s~\ref{Thm_G} and~\ref{Cor_cond}}\label{Sec_Thm_G}
We derive \Thm s~\ref{Thm_G} and~\ref{Cor_cond} from \Thm~\ref{Thm_stat}.
Recall that the Kullback-Leibler divergence is defined as
	\begin{align*}
	\KL{\hat\G(n,m,p)}{\G(n,m,p)}&=\sum_G\pr\brk{\hat\G(n,m,p)=G}\ln\frac{\pr\brk{\hat\G(n,m,p)=G}}{\pr\brk{\G(n,m,p)=G}},
	\end{align*}
with the sum ranging over all possible factor graphs.
Let us begin with the following humble observation.

\begin{fact}
For any $n,m,p$ we have
	\begin{align}\label{eqThm_G_1}
	\Erw[\ln Z(\hat\G(n,m,p))]&=\ln\Erw[Z(\G(n,m,p))]+\KL{\hat\G(n,m,p)}{\G(n,m,p)}\\
		&\geq\ln\Erw[Z(\G(n,m,p))]-\KL{\G(n,m,p)}{\hat\G(n,m,p)}=\Erw[\ln Z(\G(n,m,p))].
		\nonumber
	\end{align}
\end{fact}
\begin{proof}
Plugging in the definition  (\ref{eqNishi3}) of $\hat\G$ and using (\ref{eqJensen}), we obtain
	\begin{align*}
	\KL{\hat\G(n,m,p)}{\G(n,m,p)}&=\sum_G\frac{Z(G)\pr\brk{\G(n,m,p)=G}}{\Erw[Z(\G(n,m,p))]}\ln\frac{Z(G)\pr\brk{\G(n,m,p)=G}/\Erw[Z(\G(n,m,p))]}{\pr\brk{\G(n,m,p)=G}}\\
		&=\Erw[\ln Z(\hat\G(n,m,p))]-\ln\Erw[Z(\G(n,m,p))],\\
	\KL{\G(n,m,p)}{\hat\G(n,m,p)}&=\sum_G\pr\brk{\G(n,m,p)=G}\ln\frac{\pr\brk{\G(n,m,p)=G}}{Z(G)\pr\brk{\G(n,m,p)=G}/\Erw[Z(\G(n,m,p))]}\\
		&=\ln\Erw[Z(\G(n,m,p))]-\Erw[\ln Z(\G(n,m,p))].
	\end{align*}
The middle inequality follow from the fact that the Kullback-Leibler divergence is non-negative.
\end{proof}

\begin{lemma}\label{Lemma_quietPlanting}
Assume that $m=m(n)$ is such that $\Erw[\ln Z(\G(n,m,p))]=\ln\Erw[Z(\G(n,m,p))]+o(n)$.
Then for any event $\cE$  on graph/assignment pairs,
	\begin{align*}
	\Erw\bck{\vecone\{(\hat\G(n,m,p),\SIGMA)\in\cE\}}_{\hat\G(n,m,p)}&\leq\exp(-\Omega(n))\quad\Rightarrow\quad
		\Erw\bck{\vecone\{(\G(n,m,p),\SIGMA)\in\cE\}}_{\G(n,m,p)}\leq\exp(-\Omega(n)).
	\end{align*}
\end{lemma}
\begin{proof}
The argument is similar to the one behind the ``planting trick'' from~\cite{Barriers}.
Suppose that 
	\begin{equation}\label{eqThm_G_666}
	\Erw\bck{\vecone\{(\hat\G(n,m,p),\SIGMA)\in\cE\}}_{\hat\G(n,m,p)}\leq\exp(-2\eps n)
	\end{equation}
	 for some $\eps>0$.
By \Lem~\ref{Lemma_Azuma} and the assumption $\Erw[\ln Z(\G(n,m,p))]=\ln\Erw[Z(\G(n,m,p))]+o(n)$ there is $\delta=\delta(\eps,\Psi)>0$ such that
for large enough $n$,
	\begin{align}\label{eqThm_G_2}
	\pr\brk{\ln Z(\G(n,m,p))\leq\ln\Erw[Z(\G(n,m,p))]-\eps n}&\leq\exp(-\delta n).
	\end{align}
Consider the event $\cZ=\{\ln Z(\G(n,m,p))\geq\ln\Erw[Z(\G(n,m,p))]-\eps n\}$.
Then (\ref{eqThm_G_2}) implies
	\begin{align}\label{eqThm_G_3}
	\Erw\bck{\vecone\{(\G(n,m,p),\SIGMA)\in\cE\}}_{\G(n,m,p)}\leq\exp(-\delta n)+
		\Erw\brk{\bck{\vecone\{(\G(n,m,p),\SIGMA)\in\cE\}}_{\G(n,m,p)}\,\big|\,\cZ}.
	\end{align}
Further, by (\ref{eqGibbs}) and (\ref{eqNishi3}) and (\ref{eqThm_G_666}), with the sum ranging over all possible factor graphs and assignments,
	\begin{align}
	\Erw\brk{\bck{\vecone\{(\G(n,m,p),\SIGMA)\in\cE\}}_{\G(n,m,p)}\vecone\{\cZ\}}&=
		\sum_{G,\sigma}\vecone\{G\in\cZ\}\vecone\{(G,\sigma)\in\cE\}\pr\brk{\G(n,m,p)=G}\mu_G(\sigma)\nonumber\\
		&=\sum_{G,\sigma}\vecone\{G\in\cZ\}\vecone\{(G,\sigma)\in\cE\}\pr\brk{\G(n,m,p)=G}\frac{\psi_G(\sigma)}{Z(G)}\nonumber\\
		&\leq\exp(\eps n)\sum_{G,\sigma}\vecone\{(G,\sigma)\in\cE\}\frac{\psi_G(\sigma)\pr\brk{\G(n,m,p)=G}}{\Erw[Z(\G(n,m,p))]}\nonumber\\
		&=\exp(\eps n)\Erw\bck{\vecone\{(\hat\G(n,m,p),\SIGMA)\in\cE\}}_{\hat\G(n,m,p)}\leq\exp(-\eps n).\label{eqThm_G_4}
	\end{align}
Finally, the assertion follows from (\ref{eqThm_G_2}), (\ref{eqThm_G_3}) and (\ref{eqThm_G_4}).
\end{proof}

\begin{corollary}\label{Lemma_quietsmm}
We have $$\Erw[\ln Z(\hat\G(n,m,p))]=\ln\Erw[Z(\G(n,m,p))]+o(n)\quad\Leftrightarrow\quad\Erw[\ln Z(\G(n,m,p))]=\ln\Erw[Z(\G(n,m,p))]+o(n).$$
\end{corollary}
\begin{proof}
Assume that $\Erw[\ln Z(\hat\G(n,m,p))]=\ln\Erw[Z(\G(n,m,p))]+o(n)$.
Then there is a sequence $\Omega(1/\ln n)\leq\eps(n)=o(1)$ such that $\Erw[\ln Z(\hat\G(n,m,p))]\leq\ln\Erw[Z(\G(n,m,p))]+n\eps(n)$.
Because $\eps(n)=\Omega(1/\ln n)$, \Lem~\ref{Lemma_Azuma} implies that the event
	$$\cE=\cbc{\ln Z(\G(n,m,p))\leq\ln\Erw[Z(\G(n,m,p))]+2n\eps(n)}$$
satisfies $\pr\brk{\hat\G(n,m,p)\in\cE}=1-o(1)$.
As a consequence, recalling (\ref{eqNishi3}), we conclude that the random variable $\cZ(\G(n,m,p))=Z(\G(n,m,p))\vecone\cbc{\cE}$ satisfies
	\begin{align}
	\Erw\brk{\cZ(\G(n,m,p))}&=\Erw\brk{Z(\G(n,m,p))\vecone\{\cE\}}
		=\Erw[Z(\G(n,m,p))]\pr\brk{\hat\G(n,m,p)\in\cE}
			=(1+o(1))\Erw[Z(\G(n,m,p))].\label{eqThm_G1}
	\end{align}
On the other hand,  the definition of $\cZ(\G(n,m,p))$  guarantees that
	\begin{align}\label{eqThm_G11}
	\Erw[\cZ(\G(n,m,p))^2]=\Erw[Z(\G(n,m,p))^2\vecone\{\cE\}]\leq\exp(4n\eps(n))\Erw[Z(\G(n,m,p))]^2=\exp(o(n))\Erw[Z(\G(n,m,p))]^2.
	\end{align}
Combining (\ref{eqThm_G1}) and (\ref{eqThm_G11}) with the Paley-Zygmund inequality, we obtain
	\begin{align}\nonumber
	\pr\brk{Z(\G(n,m,p))\geq\Erw[Z(\G(n,m,p))]/4}&\geq\pr\brk{\cZ(\G(n,m,p))\geq\Erw[\cZ(\G(n,m,p))]/2}\\
		&\geq\frac{\Erw[\cZ(\G(n,m,p))]^2}{4\Erw[\cZ(\G(n,m,p))^2]}
		\geq\exp(o(n)).
		\label{eqThm_G12}
	\end{align}
Since $\ln Z(\G(n,m,p))$ is tightly concentrated by \Lem~\ref{Lemma_Azuma}, (\ref{eqThm_G12}) implies that
$$\Erw[\ln Z(\G(n,m,p))]=\ln\Erw[Z(\G(n,m,p))]+o(n).$$

Conversely, assume that $\Erw[\ln Z(\hat\G(n,m,p))]=\ln\Erw[Z(\G(n,m,p))]+\Omega(n)$.
Then there is $\delta>0$ such that  for large enough $n$,
 $\Erw[\ln Z(\hat\G(n,m,p))]\geq\ln\Erw[Z(\G(n,m,p))]+\delta n$.
Therefore, by \Lem~\ref{Lemma_Azuma} the event 
	$$\cE=\{G:\Erw[\ln Z(G)]\geq\ln\Erw[Z(\G(n,m,p))]+\delta n/2\}$$
satisfies $\pr\brk{\hat\G(n,m,p)\in\cE}=1-\exp(-\Omega(n))$.
Applying \Lem~\ref{Lemma_quietPlanting} to  $\cE$
and recalling that $\Erw[\ln Z(\G(n,m,p))]\leq\ln\Erw[Z(\G(n,m,p))]$ by Jensen, we conclude that
	$\Erw[\ln Z(\G(n,m,p))]\leq\ln\Erw[Z(\G(n,m,p))]-\Omega(n)$. 
\end{proof}

\noindent
We recall from (\ref{eqJensen}) that for any sequence $m=m(n)=O(n)$, 
	\begin{equation}\label{eqJensenRecap}
	\ln\Erw[Z(\G(n,m,p))]=(n-km)\ln|\Omega|+m\ln\sum_{\sigma\in\Omega^k}\Erw[\PSI(\sigma)]+o(n).
	\end{equation}
Moreover, \Thm~\ref{Thm_stat}, \Prop~\ref{Lemma_Nishi} and \Lem~\ref{lem:FEmutualInfo} imply that
	\begin{align}\label{eqFreeEnergyRecap}
	\lim_{n\to\infty}\frac1n\Erw[\ln Z(\hat\G)]&=	\sup_{\pi\in\cP_*^2(\Omega)}\cB(d,\pi).
	\end{align}

\begin{corollary}\label{Cor_quietPlanting}
Assume that $d>0$ is such that
	\begin{equation}\label{eqCor_quietPlanting1}
	\sup_{\pi\in\cP_*^2(\Omega)}\cB(d,\pi)>
		(1-d)\ln|\Omega|+\frac{d}k\ln\sum_{\sigma\in\Omega^k}\Erw[\PSI(\sigma)].
	\end{equation}
Then
	\begin{align*}
	\limsup_{n\to\infty}\frac1n\Erw[\ln Z(\G)]&<(1-d)\ln|\Omega|+\frac{d}k\ln\sum_{\sigma\in\Omega^k}\Erw[\PSI(\sigma)].
	\end{align*}
\end{corollary}
\begin{proof}
If (\ref{eqCor_quietPlanting1}) holds, then (\ref{eqFreeEnergyRecap}) shows that there is $\delta>0$ such that for large enough $n$,
	$$\frac1n\Erw[\ln Z(\hat\G)]\geq(1-d)\ln|\Omega|+\frac dk\ln\sum_{\sigma\in\Omega^k}\Erw[\PSI(\sigma)]+2\delta.$$
Hence, there exists a sequence $m=m(n)= dn/k+O(\sqrt n)$ such that for large $n$,
	$$\frac1n\Erw[\ln Z(\hat\G(n,m,p))]\geq(1-d)\ln|\Omega|+\frac dk\ln\sum_{\sigma\in\Omega^k}\Erw[\PSI(\sigma)]+\delta.$$
Consequently, (\ref{eqThm_G_1}), (\ref{eqJensenRecap}) and \Cor~\ref{Lemma_quietsmm} imply that
$\Erw[\ln Z(\G(n,m,p))]\leq\ln\Erw[Z(\G(n,m,p))]-\Omega(n)$
and the assertion follows from \Lem~\ref{Lemma_Azuma}.
\end{proof}

\begin{proof}[Proof of \Thm~\ref{Cor_cond}]
Assume that $d<\dinf$. Then
	\begin{align}\label{eqThm_G0}
	\sup_{\pi\in\cP_*^2(\Omega)}\cB(d,\pi)\leq
		(1-d)\ln|\Omega|+\frac{d}k\ln\sum_{\sigma\in\Omega^k}\Erw[\PSI(\sigma)]
	\end{align}
and 
(\ref{eqJensenRecap}) and (\ref{eqFreeEnergyRecap}) yield
	\begin{align*}
	\frac1n\Erw[\ln Z(\hat\G)]&=o(1)+	\sup_{\pi\in\cP_*^2(\Omega)}\cB(d,\pi)
			\leq(1-d)\ln|\Omega|+\frac{d}k\ln\sum_{\sigma\in\Omega^k}\Erw[\PSI(\sigma)]+o(1).
	\end{align*}
Hence, (\ref{eqThm_G_1}) and (\ref{eqJensenRecap}) imply that there exists $m=m(n)=dn/k+O(\sqrt n)$ such that
	 $$\Erw[\ln Z(\hat\G(n,m,p))]=\ln\Erw[Z(\G(n,m,p))]+o(n).$$
Therefore, \Cor~\ref{Lemma_quietsmm} shows that
	$\Erw[\ln Z(\G(n,m,p))]=\ln\Erw[Z(\G(n,m,p))]+o(n)$.
Consequently, (\ref{eqJensenRecap}) and \Lem~\ref{Lemma_Azuma} yield
	$\Erw[\ln Z(\G)]=(1-d)\ln|\Omega|+\frac{d}k\ln\sum_{\sigma\in\Omega^k}\Erw[\PSI(\sigma)]+o(1)$.

Conversely, suppose that $d>d_{\mathrm{cond}}$.
Then there exist $d'<d$ and $\delta>0$ such that
	$$\sup_{\pi\in\cP_*^2(\Omega)}\cB(d',\pi)>(1-d')\ln|\Omega|+\frac{d'}k\ln\sum_{\sigma\in\Omega^k}\Erw[\PSI(\sigma)]+\delta.$$
Therefore, letting $\vec m'=\Po(d'n/k)$, we obtain from \Thm~\ref{Thm_stat} and (\ref{eqJensenRecap})
	\begin{align*}
	\frac1n\Erw[\ln Z(\hat\G(n,\vec m',p))]&>(1-d')\ln|\Omega|+\frac{d'}k\ln\sum_{\sigma\in\Omega^k}\Erw[\PSI(\sigma)]+\delta.
	\end{align*}
Thus, \Lem~\ref{Lemma_Azuma}, (\ref{eqThm_G_1}) and (\ref{eqJensenRecap}) imply that the event
	$$\cE'=\cbc{G:\ln Z(G)
	\leq(1-d')\ln|\Omega|+\frac{d'}k\ln\sum_{\sigma\in\Omega^k}\Erw[\PSI(\sigma)]+\delta/2
		}$$
satisfies
	\begin{align}\label{eqThm_G_50}
	\pr\brk{\hat\G(n,\vec m',p)\in\cE'}&=\exp(-\Omega(n)),&\pr\brk{\G(n,\vec m',p)\in\cE'}&=1-\exp(-\Omega(n)).
	\end{align}
Now, for a factor graph $G$ let $G'$ be the random factor graph obtained from $G$ by removing each constraint node with probability $1-d'/d$ independently.
Moreover, consider the event $\cE=\cbc{G:\pr\brk{G'\in\cE'}\geq1/2},$
where, of course, the probability is over the coin tosses of the removal process only.
Then the distribution of $\G(n,\vec m,p)'$ coincides with the distribution of $\G(n,\vec m',p)$.
Furthermore, \Prop~\ref{Lemma_Nishi} implies that $\hat\G(n,\vec m,p)'$ and $\hat\G(n,\vec m',p)$ are mutually contiguous.
Therefore, (\ref{eqThm_G_50}) entails that
	\begin{align*}
	\pr\brk{\hat\G(n,\vec m,p)\in\cE}&\leq\exp(-\Omega(n))\qquad\mbox{while}\qquad\pr\brk{\G(n,\vec m,p)\in\cE}=1-\exp(-\Omega(n)).
	\end{align*}
Consequently, \Lem~\ref{Lemma_quietPlanting} yields $\Erw[\ln Z(\G(n,\vec m,p))]\leq\Erw[\ln Z(\hat\G(n,\vec m,p))]-\Omega(n)$, whence
 the assertion follows from \Cor~\ref{Lemma_quietsmm} and (\ref{eqThm_G_1}).
\end{proof}

\noindent
Finally, to derive \Thm~\ref{Thm_G} from \Thm~\ref{Cor_cond} we need the following lemma.

\begin{lemma}\label{Lemma_eqMIFE}
Under {\bf SYM} and {\bf BAL} we have
 $$\KL{\G^*, \SIGMA^*}{\G,\SIGMA_{\G}}=o(n)\ \Leftrightarrow\ 
 	\frac1n\Erw[\ln Z(\G)]=(1-d)\ln|\Omega|+\frac dk\ln\sum_{\sigma\in\Omega^k}\Erw[\PSI(\sigma)]+o(1).$$
\end{lemma}
\begin{proof}
We have
	\begin{align}
	\KL{\G^*,\SIGMA^*}{\G,\SIGMA}&=\sum_{G,\sigma}\pr\brk{\G^*=G,\SIGMA^*=\sigma}
			\ln\frac{\pr\brk{\G^*=G,\SIGMA^*=\sigma}}{\pr\brk{\G=G,\SIGMA=\sigma}}\nonumber\\
		&=\KL{\G^*,\SIGMA^*}{\hat\G,\hat\SIGMA}+
			\sum_{G,\sigma}\pr\brk{\G^*=G,\SIGMA^*=\sigma}
						\ln\frac{\pr\brk{\hat\G=G}\pr\brk{\hat\SIGMA=\sigma|\hat\G=G}}
							{\pr\brk{\G=G}\mu_G(\sigma)}\nonumber\\
		&=\KL{\G^*,\SIGMA^*}{\hat\G,\hat\SIGMA}+
						\sum_{G,\sigma}\pr\brk{\G^*=G,\SIGMA^*=\sigma}
						\ln\frac{Z(G)\pr\brk{\G=G}\mu_G(\sigma)}{\Erw[Z(\G)]\pr\brk{\G=G}\mu_G(\sigma)}&[\mbox{by (\ref{eqNishi3})}]\nonumber\\
		&=\KL{\G^*,\SIGMA^*}{\hat\G,\hat\SIGMA}+\Erw[\ln Z(\G^*)]-\Erw[\ln\Erw[Z(\G)|\vec m]].
			\label{eqSec_eqMIFE1}
	\end{align}
Further, because $\SIGMA^*,\hat\SIGMA$ are asymptotically balanced with overwhelming probability by \Lem~\ref{Lemma_contig},
	\begin{align*}
	\KL{\G^*,\SIGMA^*}{\hat\G,\hat\SIGMA}&=\sum_\sigma\pr\brk{\SIGMA^*=\sigma}\sum_H\pr\brk{\G^*=H|\SIGMA^*=\sigma}
			\ln\frac{\pr\brk{\G^*=H|\SIGMA^*=\sigma}\pr\brk{\SIGMA^*=\sigma}}
				{\pr\brk{\G^*=H|\SIGMA^*=\sigma}\pr\brk{\hat\SIGMA=\sigma}}\\
		&=\sum_\sigma\pr\brk{\SIGMA^*=\sigma}
			\ln\frac{\pr\brk{\SIGMA^*=\sigma}}
				{\pr\brk{\hat\SIGMA=\sigma}}=\KL{\SIGMA^*}{\hat\SIGMA}=o(n).
	\end{align*}
Hence, (\ref{eqSec_eqMIFE1}) yields
	\begin{align}			\label{eqSec_eqMIFE2}
	\KL{\G^*,\SIGMA^*}{\G,\SIGMA}=o(n)\Leftrightarrow\Erw[\ln Z(\G^*)]=\Erw[\ln\Erw[Z(\G)|\vec m]]+o(n).
	\end{align}
Further, by \Prop~\ref{Lemma_Nishi} and \Lem~\ref{Lemma_Azuma} we have $\Erw[\ln Z(\G^*)]=\Erw[\ln Z(\hat\G)]+o(n)$.
Thus, the assertion follows from  (\ref{eqJensenRecap}), (\ref{eqSec_eqMIFE2}) and \Cor~\ref{Lemma_quietsmm}.
\end{proof}

\begin{proof}[Proof of \Thm~\ref{Thm_G}]
The theorem is immediate from \Thm~\ref{Cor_cond} and \Lem~\ref{Lemma_eqMIFE}.
\end{proof}

\subsection{Proof of Theorem~\ref{Cor_stat}}\label{Sec_prop:betheUB}
Here we prove that under the assumptions {\bf SYM}, {\bf BAL} and {\bf POS}, 
\begin{equation*}
\sup_{\pi\in\cP^2_*(\Omega)}\cB(d,\pi)=	\sup_{\pi\in\cPfix}\cB(d,\pi),
\end{equation*}
where $$\cPfix=\{\pi\in\cPcent(\Omega):\cT_d(\pi)=\pi\}.$$
Since $\cPfix \subseteq \cP^2_*(\Omega)$, we have immediately $\sup_{\pi\in\cP^2_*(\Omega)}\cB(d,\pi)\geq	\sup_{\pi\in\cPfix}\cB(d,\pi)$. The other direction  follows from the following bound
\begin{equation}
\label{eq:limsupFP}
\limsup_{n\to\infty}\frac1n \Erw\ln Z(\hat\G)\leq \sup_{\pi\in\cPfix}\cB(d,\pi),
\end{equation}
since Proposition~\ref{Prop_interpolation} gives
\begin{align*}
\sup_{\pi\in\cP^2_*(\Omega)}\cB(d,\pi) &\le \liminf_{n\to\infty}\frac1n \Erw\ln Z(\hat\G) \le  \limsup_{n\to\infty}\frac1n \Erw\ln Z(\hat\G).
\end{align*}

To show~\eqref{eq:limsupFP}, we show that the random factor graph $\G^*_T(n, \vec m(n), p, \SIGMA^*_n)$ (from Definition~\ref{Def_Peg}) and its empirical marginal distribution $\rho_{ \G^*_T}$  satisfy an approximate distributional Belief Propagation fixed point property.
\begin{lemma}
\label{lem:W1fb}
For $n$ large enough,
\begin{equation}
\label{eq:W1fpbound}
\Erw[ W_1( \cT_d(\rho_{ \G^*_T}),\rho_{ \G^*_T}  ) ] = o_T(1).
\end{equation}
\end{lemma}

We prove \Lem~\ref{lem:W1fb} below, but first we derive \eqref{eq:limsupFP} from it. 
We first define a set of approximate distributional BP fixed points.  Let $\cPfixE{\eps} $ be the set of all $\pi \in \cP^2(\Omega)$ so that
\begin{description}
\item[FIX1] $W_1( \cT(d, \pi), \pi) <\eps$.
\item[FIX2] $\| \int \mu \, d\pi (\mu)  - \vec1 /|\Omega| \|_{TV} < \eps  $.
\end{description}

Recall the random factor graph $\tilde \G$ defined by {\bf CPL1} and {\bf CPL2} in Section~\ref{sec:UBproofOutline}, and $\Delta_T(n)=\Erw[\ln Z(\G^*_T(n+1,\vec m(n+1),p,\SIGMA_{n+1}^*)]-
        \Erw[\ln Z(\G^*_T(n,\vec m(n),p,\SIGMA_n^*)]$ from Lemma~\ref{Lemma_ASS}.  Lemmas~\ref{lem:W1fb} and~\ref{Lemma_reweight} and Claim~\ref{Claim_CPLTV} show that for any $\eps>0$, with probability $1 -o_T(1)$, $\rho_{\tilde \G} \in \cPfixE{\eps}$, and so Claims~\ref{Claim_eqASS3} and~\ref{Claim_eqASS2} give that for any $\eps>0$,
$$\limsup_{n \to \infty} \frac{1}{n}\Erw[ \ln Z(\hat \G)] \le \limsup_{T\to\infty}\limsup_{n \to \infty}\Delta_T(n) \le \sup_{\pi\in\cPfixE{\eps} }\cB(d, \pi).$$
Now we take $\eps \to 0$ and must show that 
\begin{equation}
\label{eq:epsSup}
\limsup_{\eps \to 0} \sup_{\pi\in\cPfixE{\eps} }\cB(d, \pi) \le \sup_{\pi\in\cPfix }\cB(d, \pi).
\end{equation}
Let $(\eps_k, \pi_k)$ be a sequence so that $\eps_k \to 0$, $\pi_k \in \cPfixE{\eps_k}$, and 
\[ \lim_{k \to \infty}  \cB(d,\pi_k) = \limsup_{\eps \to 0}  \sup_{\pi\in\cPfixE{\eps}}  \cB(d,\pi). \]
Since the space $\cP^2(\Omega)$ is compact under the weak topology, there is a convergent subsequence $\pi_{k_j}$ with 
\[ \lim_{j \to \infty} W_1 \left ( \pi_{k_j}, \pi_{\infty} \right) =0 ,\]
for some $\pi_\infty \in \cP^2(\Omega)$. 
Now from \Lem~\ref{lem:BPcontinuous}, $\cB(d,\nix )$ and $\cT_d(\nix)$ are continuous in the $W_1$ metric, and so  we have  
\begin{align*}
&\pi_\infty \in \cPfix \quad \text{and} \quad \cB(d, \pi_{\infty}) = \limsup_{\eps \to 0}  \sup_{\pi\in\cPfixE{\eps}}  \cB(d,\pi), 
\end{align*}
which gives \eqref{eq:epsSup} and in turn~\eqref{eq:limsupFP}. 
 
Before turning to the proof Lemma~\ref{lem:W1fb}, we introduce an additional tool, based on~\cite[Lemma~3.1]{Will}, that shows that the empirical distribution of an $\eps$-symmetric factor graph is stable under a bounded number of perturbations.
\begin{lemma}
 \label{lem:Stable}
For every finite set $\Omega$, finite set $\Psi$ of $k$-ary constraint functions $\psi: \Omega^k \to (0,2)$, $\eps>0$ and $K>0$, there exists $\del>0, n_0 >0$ so that the following is true.
Let $G_0$ be a factor graph on $n > n_0$ variable nodes $V_0$ taking values in $\Omega$, with a set $F_0$ of $m_1$ constraint functions from the set $\Psi$ and $m_2$ `hard' fields of the form $\vec 1 \{ \sigma(x_i) = \omega_i\}$ for arbitrary values $\omega^*_i \in \Omega$.  Let $G_1$ be formed by adding a set $V_1$ of at most $K$ new variable nodes, each attached to at most $K$ new constraint nodes, with the other attached variables chosen arbitrarily from $V_0$, and constraint functions chosen from the set $\Psi$. Then if $G_0$ is $(\del,2)$-symmetric,
\[W_1( \rho(G_0), \rho(G_1) ) < \eps. \]
\end{lemma}

The proof of \Lem~\ref{lem:Stable} requires the following `Regularity Lemma' for probability measures from~\cite{Victor}. For $\sigma \in \Omega^n$, and $U \subseteq V$, let $\sigma[ \nix | U] \in \cP(\Omega)$ be the measure defined by
\begin{align*}
\sigma[ \omega | U] &= \frac{1}{|U|} \sum_{u \in U} \vec 1 \{ \sigma(u) = \omega \} .
\end{align*}
We say a measure $\mu \in \cP(\Omega^n)$ is $\eps$-regular with respect to $U \subseteq V$ if for every $U' \subset U$, $|U'| \ge \eps |U|$, 
\begin{align*}
 \bck{ \SIGMA[ \nix |U'] - \SIGMA[ \nix | U]   }_{\mu} < \eps.
\end{align*}

 We say a measure $\mu$ on $\Omega^n$ is $\eps$-regular with respect to a partition $\vec V$ of $V$ if there is a set $J \in [ \# \vec V]$ such that $\sum_{j \in J} | V_j| > (1-\eps)n$ and $\mu$ is $\eps$-regular with respect to $V_j$ for all $j \in J$.  For $S \subseteq \Omega^n$, let $\mu [ \nix | S]$ be the measure defined by
  \[ \mu[\sigma|S] = \frac{ \vec 1\{ \sigma \in S\} }{\mu (S) }. \]

  \begin{theorem}[\cite{Victor}, Theorem 2.1]
\label{thm:regularityPart}
Given any $\eps>0$ and $\Omega$, there exists $N(\eps, \Omega)$ so that for any $n > N$ and $\mu \in \cP(\Omega^n)$ the following is true. There exists a partition $\vec V$ of $[n]$ and a partition $\vec S$ of $\Omega^n$ so that $\# \vec S + \# \vec V \le N$ and 
 there is a subset $I \subset[\#\vS]$ such that the following conditions hold.
\begin{description}
\item[REG1]  $\mu(S_j)>0$ for all $i \in I$, and $\sum_{i \in I} \mu (S_i) \ge 1- \eps$.
\item[REG2] For all $i \in I$ and $j\in [ \# \vec V]$, and all $\sigma,\sigma'\in S_i$ we have
	$\TV{\sigma[\nix|V_i]-\sigma'[\nix|V_i]}<\eps$.
\item[REG3] For all $i \in I$, $\mu[ \nix | S_i]$ is $\eps$-regular with respect to $\vec V$.
\item[REG4] $\mu$ is $\eps$-regular with respect to $\vec V$.
\end{description}
\end{theorem}

\begin{proof}[Proof of \Lem~\ref{lem:Stable}]
The proof follows along the lines of that of Lemma~3.1 of~\cite{Will}, but here we must take into account the hard external fields of $G_0$.  
Recall $V_0, F_0$ are the set of variable nodes and constraint nodes of $G_0$, and let $U_0$ be the indices of  variable nodes with hard fields in $V_0$.  Let $V_1, F_1$ be the set of variable and constraint nodes respectively added to $G_0$ to form $G_1$. Let $V= V_0 \cup V_1$ and $F = F_0 \cup F_1$. 
 
Let $\Sigma_0 = \{ \sigma \in \Omega^{V_0} : \sigma(x_j) = \omega^*_j \, \forall j \in U_0 \}$.  Then we claim that there exists $M = M(K, \Psi)>0$ so that for all $\sigma \in \Sigma_0$ and all $\tau \in \Omega^{V_1}$,
\begin{equation}
\label{eq:ratioBound}
\frac{1}{M}  \le  \frac{\mu_{G_0} (\sigma)   }{ \sum_{\tau \in \Omega^{V_1}} \mu_{G_1} (\sigma,\tau)   }  \le  M  .
\end{equation}
For all $\sigma \notin \Sigma_0$, both $\mu_{G_0} (\sigma) , \mu_{G_1} (\sigma,\tau) $ are $0$ on account of the hard fields.

For $\sigma\in \Sigma_0$ and $\tau  \in \Omega^{V_1}$, we write:
\begin{align*}
\mu_{G_0} (\sigma)  &= \frac{ \prod_{a \in F_0} \psi_{a} (\sigma(\partial a))   }{\sum_{\sigma' \in \Sigma_0}   \prod_{a \in F_0} \psi_{a} (\sigma'(\partial a))  }
\end{align*}
and
\begin{align*}
\mu_{G_1} (\sigma,\tau)  &= \frac{ \prod_{a \in F_1} \psi_{a} ((\sigma,\tau)(\partial a))  \cdot   \prod_{a \in F_0} \psi_{a} (\sigma(\partial a))   }{\sum_{\sigma' \in \Sigma_1 } \sum_{\tau' \in \Omega^{V_1}}  \prod_{a \in F_1} \psi_{a} ((\sigma',\tau')(\partial a))  \cdot   \prod_{a \in F_0} \psi_{a} (\sigma'(\partial a))  }.
\end{align*}
Now because for some $\eta>0$, $\eta  <\psi(\sigma) <2 $ for all $\sigma \in \Omega^k$ and $\psi \in \Psi$, we have
\begin{align*}
 \eta^{K^2} &\le\prod_{a \in F_1} \psi_{a} ((\sigma,\tau)(\partial a))     \le 2^{K^2} ,
\end{align*}
and for all $\sigma' \in \Sigma_0$,
\begin{align*}
|\Omega|^k \eta^{K^2} \le \sum_{\tau' \in \Omega^{V_1}}  \prod_{a \in F_1} \psi_{a} ((\sigma',\tau')(\partial a))  \le |\Omega|^k 2^{K^2} .
\end{align*}
Taking $M = (2/\eta)^{ K^2} |\Omega|^{K}$ proves the claim.

Now consider the measure $\tilde \mu$ that $G_1$ induces on $V_0$. That is, for $\sigma \in \Omega^{V_0}$,
\[ \tilde \mu ( \sigma) = \sum_{\tau \in \Omega^{V_1}} \mu_{G_1} ((\sigma,\tau)).    \]
Note that for $x \in V_0$, $\mu_{G_1, x} = \tilde \mu_x$. We will show that for every $\eps>0$, there is $\del>0$ small enough and $n_0 >0$ large enough so that if $\mu_{G_0}$ is $\del$-symmetric and $|V_0| =n \ge n_0$, then
\begin{align}
\label{eqtvlembound}
\sum_{x \in V_0} \| \mu_{G_0,x} - \tilde \mu_{x}  \|_{TV} \le \eps n.
\end{align}

Let $\vec V, \vec S$ be partitions of $V_0$ and $\Omega^{V_0}$ guaranteed by Theorem~\ref{thm:regularityPart} so that $\tilde \mu$ is $\eps'$-homogeneous with respect to $\vec V, \vec S$, and let $N = N(\eps')$ be such that $\# \vec V + \# \vec S\le N$.

Let  $J$ be the set of all $j\in[\# \vec S]$ so that $\tilde \mu(S_j)\ge \eps'/N$ and $\tilde \mu[\cdot|S_j]$ is $\eps$-regular with respect to $\vec V$.
Then {\bf REG1} and {\bf REG3} ensure that
	\begin{equation}\label{eqExtraJ}
	\sum_{j\not\in J}\tilde \mu(S_j)<2\eps'.
	\end{equation}

Now we claim that~\eqref{eq:ratioBound} and~\eqref{eqExtraJ} imply that $\mu_{G_0}[\cdot|S_j]$ is $M^2 \eps'$-regular with respect to $\vec V$ for all $j\in J$.  Let $V_i$ be such that $\tilde \mu$ is $\eps'$-regular on $V_i$ and let $U\subset V_i$ be such that $|U|\geq\eps'|V_i|$. Then
	\begin{align*}
	\bck{\TV{\SIGMA[\nix|V_i]-\SIGMA[\nix|U]}}_{\mu_{G_0}[\nix|S_j]}
		&=\sum_{\sigma\in\Omega^{V_0}}\mu_{G_0}(\sigma|S_j)\TV{\sigma[\nix|V_i]-\sigma[\nix|U]} \\
		&\le M^2 \bck{\TV{\SIGMA[\nix|V_i]-\SIGMA[\nix|U]}}_{\tilde \mu[\nix|S_j]}<M^2 \eps',
	\end{align*}
and so $\mu_{G_0}[\nix|S_j]$ is $M^2 \eps'$-regular. 

Next, using {\bf REG2} we have
	\begin{align}\label{eqrslem1}
	\sum_{i\in[\#\vec V]}\frac{|V_i|}{n}\bck{\TV{\SIGMA[\nix|V_i]-\bck{\TAU[\nix|V_i]}_{\mu_{G_0}[\nix|S_j]}}}_{\mu_{G_0}[\nix|S_j]}<3\eps'.
	\end{align}
	 for any $j\in J$.
\cite[\Lem~2.4]{Will}, the  $M^2 \eps'$-regularity of $\mu_{G_0}[\nix|S_j]$, and 
(\ref{eqrslem1}) imply that $S_j$ is an $(\eps'',2)$-state of $\mu_{G_0}$ for every $j\in J$, provided that $\eps'=\eps'(\eps'')$ was chosen small enough.
The bound~\eqref{eq:ratioBound} implies that $\mu_{G_0}(S_j)\geq\eps'/(M^2N)$ for all $j\in J$.
Therefore, if we choose $\del$ small enough, \Cor~2.3 of~\cite{Will} and the $\del$-symmetry of $\mu_{G_0}$ give that for each $j\in J$,
	\begin{align}
	\label{eqcv3}
	\sum_{x\in V}\TV{\mu_{G_0,x}-\mu_{G_0,x}[\nix|S_j]}&<\eps n/4,
	\end{align}
provided $\eps''=\eps''(\eps)$ is chosen small enough and $n$ is large enough.
Further, by \cite[\Lem~2.5]{Will} and $M^2\eps'$-regularity,
	\begin{align*}
	\sum_{i\in[\#\vec V]}\sum_{x\in V_i}\TV{\mu_{G_0, x}[\nix|S_j]-\sigma[\nix|V_i]}&<\eps n/4\qquad\mbox{for all }j\in J,\ \sigma\in S_j,
	\end{align*}
and by (\ref{eqcv3}),
	\begin{align}\label{eqcv4}
	\sum_{i\in[\#\vec V]}\sum_{x\in V_i}\TV{\mu_{G_0, x}-\sigma[\nix|V_i]}&<2\eps n/4\qquad\mbox{for all }j\in J,\ \sigma\in S_j.
	\end{align}
Similarly, 
	\begin{align}
	\label{eqcv4a}
	\sum_{i\in[\#\vec V]}\sum_{x\in V_i}\TV{\tilde \mu_{x}[\nix|S_j]-\sigma[\nix|V_i]}&<\eps'''n\qquad\mbox{for all }j\in J,\ \sigma\in S_j.
	\end{align}
Combining (\ref{eqcv4}) and (\ref{eqcv4a}) and using the triangle inequality, we obtain
	\begin{align*}
	\sum_{x\in V_0}\TV{\mu_{G_0, x}-\tilde \mu_{x}[\nix|S_j]}
		&\leq 3 \eps n/4\qquad\mbox{for all }j\in J.
	\end{align*}
Therefore,
	\begin{align*}
	\sum_{x\in V_0}\TV{\mu_{G_0, x}-\tilde \mu_{x}}&\leq
		2\eps n+\sum_{j\in J}\sum_{i\in[\#\vec V]}\sum_{x\in V_i}\tilde \mu(S_j)\TV{\mu_{G_0, x}-\tilde \mu_{ x}[\nix|S_j]}<3\eps n,
	\end{align*}
which proves~\eqref{eqtvlembound}.

Now consider sampling a variable node $x$ uniformly from $V_0$ and $V$ and outputting $\mu_{G_0,x}$ and $\mu_{G_1,x}$ respectively.  The distributions of $\mu_{G_0,x}$ and $\mu_{G_1,x}$  are exactly $\rho(G_0)$ and $\rho(G_1)$.  Since the probability we choose $x \in V_1$ in the second experiment is $O(1/n)$ we can couple the choice of $x$ to coincide with probability $1-O(1/n)$.   On the event they coincide the expected total variation distance between $\mu_{G_0,x}$ and $\mu_{G_1,x} = \tilde \mu_x$ is at most $\eps$ by~\eqref{eqtvlembound}, and so $W_1(\rho(G_0),\rho(G_1)) \le \eps - o(1)$, completing the proof of \Lem~\ref{lem:Stable}. 
\end{proof}

With this tool we now prove \Lem~\ref{lem:W1fb}.
\begin{proof}[Proof of \Lem~\ref{lem:W1fb}]
Let $\G^*_T=\G^*_{T}(n, \vec m(n), p, \SIGMA^*_n)$ and $\rho_{\G^*_T}$ be its empirical marginal distribution. We must show that for $n$ large enough,
\[ \Erw[ W_1( \cT_d(\rho_{ \G^*_T}),\rho_{ \G^*_T}  ) ] = o_T(1). \]

More precisely we will show that for any $\eps>0$, there is $T$ large enough so that 
\begin{equation}
\label{eq:FPepsbound}
 \Erw[ W_1( \cT_d(\rho_{\G^*_{n,T}}),\rho_{\G^*_{n,T}}  ) ] <\eps .
\end{equation} 

Fix $\eps>0$.  For $L=L(\eps)$ large enough, we will couple the factor graph $\G^*_T=\G^*_{T}(n, \vec m(n), p, \SIGMA^*_n)$ on $n$ variable nodes with a factor graph $\G^\prime$ on $n+L$ variable nodes as follows. Form $\G^*_{T}(n, \vec m(n), p, \SIGMA^*_n)$ as usual by choosing $\vec m \sim \Po(dn/k)$, $\vec \theta$ uniformly from $[0,T]$, and a ground truth $\SIGMA^*_n$ uniformly at random from $\Omega^n$.  Then add $\vec m$ random constraint nodes with weight functions from $\Psi$ and pin each variable node independently with probability $\vec \theta/n$. To obtain $\G^\prime$ we add $L$ additional variable nodes $x_{n+1}, \dots x_{n+L}$, extending $\SIGMA^*_n$ to $\SIGMA^*_{n+L}$ by choosing $\SIGMA_{n+L}^*(x_{n+1}), \dots \SIGMA_{n+L}^*(x_{n+L})$ uniformly at random, then we add $\Po(d)$ constraint nodes with weight functions from $\Psi$ adjacent to each new variable node $x_{n+1}, \dots x_{n+L}$ with respect to $\SIGMA_{n+L}^*$,
and finally pin each new variable node independently with probability $\vec \theta/n$.
 
Up to total variation distance $o(1)$, the distribution of $\G^\prime$ with the $L$ distinguished variable nodes $x_{n+1}, \dots x_{n+L}$ is identical to the distribution of $\G^\prime$ with $L$ uniformly chosen distinguished variable nodes from $x_1, \dots x_{n+L}$. Let $\rho_L $ denote the empirical marginal distribution of $x_{n+1}, \dots x_{n+L}$, that is
\[ \rho_L = \frac{1}{L} \sum_{j=1}^L \delta_{ \mu_{\G^\prime , x_{n+j}}} .  \]
  By \Prop~\ref{prop:sampling}, for $L=L(\eps)$ chosen large enough we have 
\begin{equation}
\label{eq:W1ineq1}
 \Erw [ W_1( \rho_L, \rho_{\G^\prime}  ) ] <\eps/3. 
 \end{equation}
 Next we claim that the empirical marginal distributions of $\G^*_T$ and $\G^\prime$ are close:  for $n, T$ large enough, 
 \begin{equation}
 \label{eq:fixEMclose}
 \Erw [ W_1( \rho_{\G^*_T}, \rho_{\G^\prime}  ) ] <\eps/3.
 \end{equation}
 To prove this we use \Lem~\ref{lem:Stable}. Take $K>L$ large enough so that with probability at least $1-\eps/10$, each  variable node $x_{n+1},\dots x_{n+L}$ in $\G^\prime$ is joined to at most $K$ constraint nodes. With probability $1-o(1)$, none of these $L$ variable nodes are pinned, and no two are joined to the same constraint node.   Since $\mu_{\G^*}$ is $o_T(1)$-symmetric with probability $1- o_T(1)$, we apply \Lem~\ref{lem:Stable} with $G_0= \G^*_T$ and $G_1 = \G^\prime$ to obtain~\eqref{eq:fixEMclose}.
 
 Now it remains to show that 
 \begin{equation}
 \label{eq:W1bnound3}
 \Erw [ W_1 (  \cT_d(\rho_{\G^*_T}),\rho_{L}      )  ] < \eps/3.
 \end{equation}

The Gibbs measure $\mu_{\G^*_T}$ is $o_T(1)$-symmetric with probability $1-o_T(1)$,  and so by \Prop~\ref{prop:sampling} and  repeated applications of \Lem~\ref{lem:Stable} and the triangle inequality, it suffices to show that 
\begin{equation}
\label{eq:W1expbound}
 \Erw [ W_1 (  \cT_d(\rho_{\G^*_T}),\vec \mu_{\G^*_{n+1,T},x_{n+1}}      )  ] < \eps/4 
 \end{equation}
where $\vec \mu _{\G^*_{n+1,T},x_{n+1}}$ is the distribution of the marginal of $x_{n+1}$ over the randomness in adding a single variable node $x_{n+1}$ to $\G^*_{T}$ with a uniformly chosen $\SIGMA_{n+1}(x_{n+1})$, and attaching $\Po(d)$ random constraint nodes from $\Psi$ to  it.  We may assume that $x_{n+1}$ is not pinned, as this occurs with probability $O(1/n)$.

With $\vec \gamma \sim \Po(d)$, let $b_1\ldots,b_{\vec\gamma}\in\partial x_{n+1}$ be the factor nodes adjoining $x_{n+1}$. With probability $1-o_T(1)$, $\mu_{\G^*_T}$ is $o_T(1)$-symmetric, and so the random set $Y=\bigcup_{i=1}^{\vec\gamma}\partial b_i$ of variable nodes satisfies
	$\|\mu_{ \G^*_T,Y}-\bigotimes_{y\in Y}\mu_{ \G^*_T,y}\|_{\mathrm{TV}}=o_T(1)$
	with probability $1- o_T(1)$, again using the contiguity of $Y$ with a uniformly chosen set, as in \eqref{eqcPsi}, \eqref{eqeventY}. 
Under this condition we can compute
	\begin{align}
	\label{eq:marginalCalc}
	\vec \mu_{\G^*_{n+1,T},x_{n+1}} (\omega) &=o_T(1)+ \frac{ \prod_{i=1}^{\vec\gamma}\sum_{\tau\in\Omega^{\partial b_i}}
			\vecone\{\tau(x_{n+1})=\omega\}\psi_{b_i}(\tau)\prod_{y\in\partial b_i\setminus\{x_{n+1}\}}
				\mu_{\G^*_T,y}(\tau(y))} 
				{ \sum_{\sigma\in\Omega}\prod_{i=1}^{\vec\gamma}\sum_{\tau\in\Omega^{\partial b_i}}
			\vecone\{\tau(x_{n+1})=\sigma\}\psi_{b_i}(\tau)\prod_{y\in\partial b_i\setminus\{x_{n+1}\}}
				\mu_{\G^*_T,y}(\tau(y))}   
				\hspace{-3mm} \\
				&=o_T(1)+ \frac{ \prod_{i=1}^{\vec\gamma} \vec \mu_{b_i} (\omega)   }{ \sum_{\sigma\in\Omega}\prod_{i=1}^{\vec\gamma} \vec \mu_{b_i} (\sigma)   },
	\end{align}
	where 
	\begin{align*}
\vec \mu_{b_i} (\omega) &=  \sum_{\tau \in \Omega^k} \vecone\{\tau_{\vec h_i}=\omega\}\psi_{b_i}(\tau)\prod_{j \ne h_i}
				\mu_{\G^*_T,\vec y_{ij}}(\tau_j),
\end{align*}
 $\vec h_i$ is the position at which $x_{n+1}$ is attached to the constraint node $b_i$, and $\vec y_{ij}$ is the variable node attached to constraint node $b_i$ at position $j$. 
As before, the neighborhoods $\partial b_i$ and weight functions $\psi_{b_i}$ are chosen according to 
the teacher-student scheme  with respect to $\SIGMA_{n+1}^*$, and so by assumption {\bf SYM}  and \Lem~\ref{Lemma_reweight}, we have 
	\begin{align}
	\label{eq:marginalCalc2}
	\pr\brk{\partial b_i=(y_1,\ldots,y_k),\hat{\vec\psi}=\psi}\propto o(1)+
		\vecone\{y_{\vec h_i}=x_{n+1}\}\Erw[\PSI(\SIGMA_n^*(y_1),\ldots,\SIGMA_n^*(y_k))],
	\end{align}
where $\vec h_1, \dots $ are independent and uniform on $[k]$. Conditioned on their spins, the variables in $\partial b_i$ are uniformly chosen and independent, and so their marginals are independent samples from the corresponding empirical distributions $\rho_{\G^*_T,\SIGMA^*_n,\sigma}$.
Combining the definition of $ \cT_d(\cdot)$, the weak continuity of $ \cT_d(\cdot)$, (\ref{eq:marginalCalc}), (\ref{eq:marginalCalc2}), and \Lem~\ref{Lemma_reweight2}, we obtain (\ref{eq:W1expbound}) and thus~\eqref{eq:W1bnound3}. 

The bound \eqref{eq:FPepsbound} follows from \eqref{eq:W1ineq1}, \eqref{eq:fixEMclose}, \eqref{eq:W1bnound3}, and the triangle inequality. 
\end{proof}

\subsection{Proof of \Lem~\ref{Lemma_pinning}}\label{Sec_pinning}

\noindent
As a first step we establish the following lemma.

\begin{lemma}\label{Lemma_pin}
Let $\Omega\neq\emptyset$ be a finite set, let $n>0$ be an integer and let $\mu\in\cP(\Omega^n)$.
Given $\theta_1,\ldots,\theta_n\in(0,1)$, consider the following experiment.
\begin{enumerate}
\item choose $\vU\subset[n]$ by including each $i\in\vU$ with probability $\theta_i$ independently.
\item independently choose $\SIGMA\in\Omega^n$ from $\mu$.
\end{enumerate}
Then for any $i,j\in[n]$, $i\neq j$, we have
	\begin{align*}
	\Erw_{\vU}[I(\SIGMA_i,\SIGMA_j|(\SIGMA_u)_{u\in\vU})]&=
		(1-\theta_i)(1-\theta_j)\frac{\partial^2}{\partial\theta_i\partial\theta_j}\Erw_{\vU}[H(\SIGMA|(\SIGMA_u)_{u\in\vU})].
	\end{align*}
\end{lemma}

\noindent
\Lem~\ref{Lemma_pin} and \Cor~\ref{Cor_pin} below are generalized version of \cite[\Lem~3.1]{Andrea}.
The proofs are based on very similar calculations, parts of which go back to~\cite{Macris,MMRU,MMU}.
We proceed to prove \Lem~\ref{Lemma_pin}.
We begin with the following claim.

\begin{claim}\label{Claim_pin1}
We have
	$\frac{\partial}{\partial\theta_i}\Erw_{\vU}[H(\SIGMA|(\SIGMA_u)_{u\in\vU})]=-\Erw_{\vU}[H(\SIGMA_i|(\SIGMA_u)_{u\in\vU})|i\not\in\vU].$
\end{claim}
\begin{proof}
By the chain rule, for any $i\in[n]$ we have
	\begin{align*}
	\Erw_{\vU}[H(\SIGMA|(\SIGMA_u)_{u\in\vU})]&=\Erw_{\vU}[H(\SIGMA_i|(\SIGMA_u)_{u\in\vU})+
		H(\SIGMA|(\SIGMA_u)_{u\in\vU\cup\{i\}})].
	\end{align*}
Hence,
	\begin{align*}
	\frac{\partial}{\partial\theta_i}\Erw_{\vU}[H(\SIGMA|(\SIGMA_u)_{u\in\vU})]
		&=\frac{\partial}{\partial\theta_i}\Erw_{\vU}[H(\SIGMA_i|(\SIGMA_u)_{u\in\vU})]+
			\frac{\partial}{\partial\theta_i}\Erw_{\vU}[H(\SIGMA|(\SIGMA_u)_{u\in\vU\cup\{i\}})].
	\end{align*}
We claim that
	\[ \frac{\partial}{\partial \theta_i}  \Erw_{\vU}[H( \SIGMA   | ( \SIGMA_u)_{u\in\vU\cup\{i\}})] = 0. \]
To show this define for $U\subset[n]$ and $j\in[n]$
	\begin{align*}
	p(U)&=\pr\brk{\vU=U}=\prod_{i=1}^n\theta_i^{\vecone\{i\in U\}}(1-\theta_i)^{\vecone\{i\not\in U\}},&
	p_j(U)&=\pr\brk{\vU\setminus\cbc j=U\setminus\cbc j}=\prod_{i\neq j}\theta_i^{\vecone\{i\in U\}}(1-\theta_i)^{\vecone\{i\not\in U\}}.
	\end{align*}
Then
	\begin{align*}
	\frac{\partial}{\partial \theta_i}\Erw_{\vU}[ H(  \SIGMA   | ( \SIGMA_u)_{u\in\vU\cup\{i\}})]
		&=\sum_{U\subset[n]}\brk{\frac{\partial}{\partial \theta_i}p(U)}
			\sum_{\sigma\in\Omega^n}\mu(\sigma)
			H( \SIGMA   | ( \SIGMA_u)_{u\in U\cup\{i\}}=( \sigma_u)_{u\in U\cup\{i\}})\\
		&=\sum_{\sigma\in\Omega^n}\mu(\sigma)\bigg[\sum_{U\subset[n]:i\in U}p_i(U)
			H( \SIGMA   |( \SIGMA_u)_{u\in U\cup\{i\}}=( \sigma_u)_{u\in U\cup\{i\}})\\
			&\qquad\qquad\qquad\qquad -\sum_{U\subset[n]:i\not\in U}p_i(U)
				H( \SIGMA   | ( \SIGMA_u)_{u\in U\cup\{i\}}=( \sigma_u)_{u\in U\cup\{i\}})\bigg]=0.
	\end{align*}
Moreover,
	\begin{align*}
	\frac{\partial}{\partial\theta_i}\Erw_{\vU}[H(\SIGMA_i|(\SIGMA_u)_{u\in\vU})]&=
		\sum_{U\subset[n]}\brk{\frac{\partial}{\partial\theta_i}p(U)}\sum_{\sigma}
			\mu(\sigma)H(\SIGMA_i|(\SIGMA_u)_{u\in U}=(\sigma_u)_{u\in U})\\
		&=\sum_{U\subset[n]:i\not\in U}\brk{\frac{\partial}{\partial\theta_i}p(U)}\sum_{\sigma}
			\mu(\sigma)H(\SIGMA_i|(\SIGMA_u)_{u\in U}=(\sigma_u)_{u\in U})
	\end{align*}
because $H(\SIGMA_i|(\SIGMA_u)_{u\in U}=(\sigma_u)_{u\in U})=0$ if $i\in U$.
Hence,
	\begin{align*}
	\frac{\partial}{\partial\theta_i}\Erw_{\vU}[H(\SIGMA_i|(\SIGMA_u)_{u\in\vU})]&=
		-\sum_{U\subset[n]\setminus\cbc i}p_i(U)
			\sum_{\sigma}\mu(\sigma)H(\SIGMA_i|(\SIGMA_u)_{u\in U}=(\sigma_u)_{u\in U})\\
		&=-\Erw_{\vU}[H(\SIGMA_i|(\SIGMA_u)_{u\in\vU})|i\not\in\vU],
	\end{align*}	
as claimed.
\end{proof}

\begin{claim}\label{Claim_pin2}
If $i\neq j$, then
	$\frac{\partial^2}{\partial\theta_i\partial\theta_j}
		\Erw_{\vU}[H(\SIGMA|(\SIGMA_u)_{u\in\vU})]=\Erw_{\vU}[I(\SIGMA_i,\SIGMA_j|(\SIGMA_u)_{u\in\vU})|i,j\not\in\vU].$
\end{claim}
\begin{proof}
By Claim~\ref{Claim_pin1}
	\begin{align*}
	\frac{\partial}{\partial\theta_i}\Erw_{\vU}[H(\SIGMA|(\SIGMA_u)_{u\in\vU})]&=
		-\Erw_{\vU}[H(\SIGMA_i|(\SIGMA_u)_{u\in\vU})|i\not\in\vU]\\&=
		-\sum_{U\subset[n]\setminus\cbc i}p_i(U)
			\sum_{\sigma}\mu(\sigma)H(\SIGMA_i|(\SIGMA_u)_{u\in U}=(\sigma_u)_{u\in U}).
	\end{align*}	
Hence,
	\begin{align*}
	\frac{\partial^2}{\partial\theta_i\partial\theta_j}\Erw_{\vU}[H(\SIGMA_i|(\SIGMA_u)_{u\in\vU})]&=
		-\sum_{U\subset[n]\setminus\cbc i}\brk{\frac{\partial}{\partial\theta_j}p_i(U)}
			\sum_{\sigma}\mu(\sigma)H(\SIGMA_i|(\SIGMA_u)_{u\in U}=(\sigma_u)_{u\in U}).
	\end{align*}	
Letting
$$p_{ij}(U)=\pr\brk{\vU\setminus\{i,j\}=U\setminus\{i,j\}}=\prod_{h\neq i,j}\theta_h^{\vecone\{h\in U\}}(1-\theta_h)^{\vecone\{h\not\in U\}},$$
we get
	\begin{align*}
	\frac{\partial^2}{\partial\theta_i\partial\theta_j}\Erw_{\vU}[H(\SIGMA_i|(\SIGMA_u)_{u\in\vU})]&=
		\sum_{U\subset[n]\setminus\cbc{i,j}}p_{ij}(U)
		\sum_{\sigma}\mu(\sigma)H(\SIGMA_i|(\SIGMA_u)_{u\in U}=(\sigma_u)_{u\in U})\\
		&\qquad-\sum_{U\subset[n]\setminus\cbc{i},j\in U}p_{ij}(U)\sum_{\sigma}\mu(\sigma)H(\SIGMA_i|(\SIGMA_u)_{u\in U}=(\sigma_u)_{u\in U})\\
		&=\sum_{U'\subset[n]\setminus\cbc{i,j}}p_{ij}(U')
			\sum_{\sigma}\mu(\sigma)\bigg[
				H(\SIGMA_i|(\SIGMA_u)_{u\in U'}=(\sigma_u)_{u\in U'})\\
		&\qquad\qquad\qquad\qquad\qquad\qquad\qquad\qquad
						-H(\SIGMA_i|(\SIGMA_u)_{u\in U'\cup\{j\}}=(\sigma_u)_{u\in U'\cup\{j\}})\bigg]\\
		&=\sum_{U'\subset[n]\setminus\cbc{i,j}}p_{ij}(U')\sum_{\sigma}\mu(\sigma)
			I(\SIGMA_i,\SIGMA_j|(\SIGMA_u)_{u\in U'}=(\sigma_u)_{u\in U'}).
	\end{align*}
The last line follows from the general formula $I(X,Y)=H(X)-H(X|Y)$.
\end{proof}

\begin{proof}[Proof of \Lem~\ref{Lemma_pin}]
The mutual information $I(\SIGMA_i,\SIGMA_j|(\SIGMA_u)_{u\in U})$ vanishes if $i\in U$ or $j\in U$.
Therefore, Claim~\ref{Claim_pin2} yields
	\begin{align*}
	\Erw_{\vU}\brk{I(\SIGMA_i,\SIGMA_j|(\SIGMA_u)_{u\in\vU})}
		&=(1-\theta_i)(1-\theta_j)\Erw_{\vU}\brk{I(\SIGMA_i,\SIGMA_j|(\SIGMA_u)_{u\in\vU})|i,j\not\in\vU}\\
	&=(1-\theta_i)(1-\theta_j)\frac{\partial^2}{\partial\theta_i\partial\theta_j}
		\Erw_{\vU}[H(\SIGMA|(\SIGMA_u)_{u\in\vU})],
	\end{align*}
as desired.
\end{proof}

\begin{corollary}\label{Cor_pin}
Suppose  in the experiment from \Lem~\ref{Lemma_pin} we set $\theta_i=\theta$ for all $i\in[n]$.
Then
	\begin{align*}
	\sum_{i,j=1}^n\int_0^t\Erw_{\vU}[I(\SIGMA_i,\SIGMA_j|(\SIGMA_u)_{u\in\vU})]d\theta&\leq n\ln|\Omega|\qquad\mbox{for all }0<t<1.
	\end{align*}
\end{corollary}
\begin{proof}
By the chain rule and \Lem~\ref{Lemma_pin}, for $\theta\in(0,1)$,
	\begin{align*}
	\sum_{i,j=1}^n\Erw_{\vU}[I(\SIGMA_i,\SIGMA_j|(\SIGMA_u)_{u\in\vU})]&\leq
	\sum_{i,j=1}^n\frac{\Erw_{\vU}[I(\SIGMA_i,\SIGMA_j|(\SIGMA_u)_{u\in\vU})]}{(1-\theta_i)(1-\theta_j)}\\
	&	=\sum_{i,j=1}^n\frac{\partial^2}{\partial\theta_i\partial\theta_j}\Erw_{\vU}[H(\SIGMA|(\SIGMA_u)_{u\in\vU})]=
		\frac{\partial^2}{\partial\theta^2}\Erw_{\vU}[H(\SIGMA|(\SIGMA_u)_{u\in\vU})].
	\end{align*}
Hence,
	\begin{align*}
	\int_0^t\sum_{i,j=1}^n\Erw_{\vU}[I(\SIGMA_i,\SIGMA_j|(\SIGMA_u)_{u\in\vU})]d\theta
		&=\int_0^t\frac{\partial^2}{\partial\theta^2}\Erw_{\vU}[H(\SIGMA|(\SIGMA_u)_{u\in\vU})]
			=\frac{\partial}{\partial\theta}\Erw_{\vU}[H(\SIGMA|(\SIGMA_u)_{u\in\vU})]\bigg|_{\theta=0}^{\theta=t}.
	\end{align*}
Once more by the chain rule and Claim~\ref{Claim_pin1},
	\begin{align*}
	\frac{\partial}{\partial\theta}\Erw_{\vU}[H(\SIGMA|(\SIGMA_u)_{u\in\vU})]\bigg|_{\theta=0}^{\theta=t}
		&=\sum_{i=1}^n\Erw_{\vU}[H(\SIGMA_i|(\SIGMA_u)_{u\in\vU})|i\not\in\vU]\bigg|_{\theta=0}-
			\Erw_{\vU}[H(\SIGMA_i|(\SIGMA_u)_{u\in\vU})|i\not\in\vU]\bigg|_{\theta=t}\leq n\ln|\Omega|,
	\end{align*}
whence the assertion follows.
\end{proof}

\begin{corollary}\label{Lemma_pinKL}
For the random measure $\check\MU$ from \Lem~\ref{Lemma_pinning} we have
	\begin{align*}
	\sum_{i,j=1}^n\Erw\brk{\KL{\cMU_{ij}}{\cMU_{i}\tensor\cMU_{j}}}&\leq\frac{n^2\ln|\Omega|}{T}.
	\end{align*}
\end{corollary}
\begin{proof}
We claim that
	\begin{align*}
	\Erw_{\vU,\cSIGMA}\brk{\KL{\cMU_{ij}}{{\cMU}_{i}\tensor\cMU_{j}}}&=
		\Erw_{\vU}\brk{I(\cSIGMA_i,\cSIGMA_j|(\cSIGMA_v)_{v\in\vU})}.
	\end{align*}
Indeed, 
since $\check\SIGMA$ is chosen from $\mu$,
given $\vU$ such that $i,j\not\in\vU$ we have
	\begin{align*}
	I(\cSIGMA_i,\cSIGMA_j|(\cSIGMA_v)_{v\in\vU})&=\sum_{\check\sigma\in\Omega^n}\mu(\check\sigma)
		\sum_{\sigma_i,\sigma_j\in\Omega}\mu(\SIGMA_i=\sigma_i,\SIGMA_j=\sigma_j|\forall u\in\vU:\SIGMA_u=\check\sigma_u)\\
		&\qquad\qquad\qquad\qquad\qquad	\ln\frac{\mu(\SIGMA_i=\sigma_i,\SIGMA_j=\sigma_j|
					\forall u\in\vU:\SIGMA_u=\check\sigma_u)}{
				\mu(\SIGMA_i=\sigma_i|\forall u\in\vU:\SIGMA_u=\check\sigma_u)
				\mu(\SIGMA_j=\sigma_j|\forall u\in\vU:\SIGMA_u=\hat\sigma_u)}\\
		&=\Erw\brk{\KL{\cMU_{ij}}{\cMU_{i}\tensor\cMU_{j}}\big|\vU}.
	\end{align*}
Moreover, both the mutual information and the Kullback-Leibler divergence vanish if $i\in\vU$ or $j\in\vU$.
Therefore, \Cor~\ref{Cor_pin} implies
	\begin{align*}
	\Erw\brk{\KL{\cMU_{ij}}{\cMU_{i}\tensor\cMU_{j}}}
		&=\frac{n}{T}\int_0^{T/n}\Erw[I(\cSIGMA_i,\cSIGMA_j|(\cSIGMA_u)_{u\in\vU})]d\theta
			\leq\frac{n^2\ln|\Omega|}{T},
	\end{align*}
as desired.
\end{proof}

\begin{proof}[Proof of \Lem~\ref{Lemma_pinning}]
By \Lem~\ref{Lemma_pinKL} and Markov's inequality for large enough $T=T(\eps,\Omega)$ we get
	\begin{align*}
	\pr\brk{\abs{\cbc{(i,j)\in[n]\times[n]:\KL{\cMU_{ij}}{\cMU_{i}\tensor\cMU_{j}}>\eps^2}}<\eps n^2}>1-\eps.
	\end{align*}
Therefore, the assertion follows from Pinsker's inequality~(\ref{eqPinsker}).
\end{proof}

\section{Applications}\label{Sec_applications}

\noindent
In this section we derive the results stated in \Sec~\ref{Sec_intro} from those in \Sec~\ref{Sec_general}.
We begin with the proof of \Thm~\ref{Thm_Potts} in \Sec~\ref{Sec_Potts}.
\Sec~\ref{Sec_sbm} contains the proof of \Thm~\ref{Thm_SBM}, parts of which we will reuse in \Sec~\ref{Sec_graphcol} to prove \Thm~\ref{Thm_col}.
Then in \Sec~\ref{Sec_thm:noisyXor} we prove \Thm~\ref{thm:noisyXor}.
Finally, \Sec~\ref{Sec_further} deals with a few further examples.

\subsection{Proof of \Thm~\ref{Thm_Potts}}\label{Sec_Potts}
The Potts antiferromagnet can easily be cast as a random factor graph model.
Indeed, for $q\ge 2$, let $\Omega=[q]$ be the set of spins and set $c_\beta=1-\exp(-\beta)$.
There is just a single weight function of arity two, namely
	\begin{equation}\label{eqPottsPsi}
	\psi_{\beta}:\Omega^2\to(0,1],\ (\sigma,\tau)\mapsto1-c_\beta\vecone\{\sigma=\tau\}.
	\end{equation}
Thus, $\Psi=\{\psi_{\beta}\}$ and $p_\beta(\psi_{\beta})=1$.
With $\vec m=\vec m(d,n)=\Po(dn/2)$ let $\G=\G(n,\vec m,p_\beta)$ be the resulting random factor graph model.

\begin{lemma}\label{Lemma_simple}
Let $\mathfrak S$ be the event that every constraint node is adjacent to two distinct variable nodes and
that for all $1\leq i<i'\leq\vec m$ the set of neighbors of $a_i$ is distinct from the set of neighbors of $a_{i'}$.
For any $d>0$ there is $\zeta(d)>0$ such that for all $q\geq2$, $\beta>0$ we have $\pr\brk{\mathfrak S}\geq\zeta(d)+o(1)$.
\end{lemma}
\begin{proof}
Given $\vec m$,
the number $X_1(\G)$ of constraint nodes that hit the same variable node twice has mean $(1+o(1))\vec m/n$ and a standard argument
shows that $X_1(\G)$ is asymptotically Poisson.
Similarly, the number $X_2(\G)$ of pairs of constraint nodes that have the same neighbors has mean $(1+o(1))2\vec m^2/n^2$.
Since $\vec m=\Po(dn/2)$, a standard argument shows that $(X_1(\G),X_2(\G))$ is within total variation distance $o(1)$ of a pair of independent Poisson variables
with means $d/2$ and $d^2/2$.
Hence, $\pr\brk{\mathfrak S}\geq\exp(-d/2-d^2/2+o(1))$.
\end{proof}

\noindent
We remember that $\GG(n,d/n)$ denotes the \Erdos-\Renyi\ random graph.

\begin{corollary}\label{Cor_simple}
For all $d>0$, $\beta>0$ we have
	$\Erw\brk{\ln Z_\beta(\GG(n,d/n))}=\Erw\brk{\ln Z(\G)}+o(n).$
\end{corollary}
\begin{proof}
The number of edges of the random graph $\GG(n,d/n)$ has distribution $\Bin(\bink n2,d/n)$, which
is at total variation distance $o(1)$ from the Poisson distribution $\Po(dn/2)$.
Therefore,
	\begin{align}\label{eqCor_simple}
	\Erw\brk{\ln Z_\beta(\GG(n,d/n))}&=\Erw\brk{\ln Z(\G)|\mathfrak S}+o(n).
	\end{align}
Further, since $\pr\brk{\mathfrak S}=\Omega(1)$ by \Lem~\ref{Lemma_simple} and since $\ln Z(\G)$ is tightly concentrated by \Lem~\ref{Lemma_Azuma},
we see that $\Erw\brk{\ln Z(\G)|\mathfrak S}=\Erw\brk{\ln Z(\G)}+o(n)$.
Hence, the assertion follows from (\ref{eqCor_simple}).
\end{proof}

\noindent
Thus, we can prove \Thm~\ref{Thm_Potts} by applying \Cor~\ref{Cor_cond} to $\G$.
We just need to verify the assumptions {\bf BAL}, {\bf SYM} and {\bf POS}.

\begin{lemma}\label{Lemma_PottsAssumptions}
The Potts antiferromagnet satisfies the assumptions {\bf BAL}, {\bf SYM} and {\bf POS} for all $q\geq2,\beta\geq0$.
\end{lemma}
\begin{proof}
Condition {\bf SYM} is immediate from the symmetry amongst the colors.
Then
	\begin{align*}
	\sum_{\sigma,\tau\in\Omega}\psi(\sigma,\tau)\mu(\sigma)\mu(\tau)=1-c_\beta\sum_{\sigma\in\Omega}\mu(\sigma)^2
	\qquad\mbox{for any $\mu\in\cP(\Omega)$}.
	\end{align*}
{\bf BAL} follows because the uniform distribution is the (unique) minimizer of $\sum_{\sigma\in\Omega}\mu(\sigma)^2$.
With respect to {\bf POS}, fix $\pi,\pi'\in\cP_*^2(\Omega)$.
Plugging in the single weight function $\psi=\psi_{c_\beta}$ and simplifying, we see that the condition comes down to
	\begin{align*}\nonumber
	0&\leq \Erw\bigg[\bigg(\sum_{\sigma_1,\sigma_2\in\Omega}\vecone\{\sigma_1=\sigma_2\}\prod_{j=1}^2\vec\mu_j^{(\pi)}(\sigma_j)\bigg)^l
		+\bigg(\sum_{\sigma_1,\sigma_2\in\Omega}\vecone\{\sigma_1=\sigma_2\}\prod_{j=1}^2\vec\mu_j^{(\pi')}(\sigma_j)\bigg)^l\\
		&\qquad\qquad-2\bigg(\sum_{\sigma_1,\sigma_2\in\Omega}\vecone\{\sigma_1=\sigma_2\}\vec\mu_1^{(\pi)}(\sigma_1)
			\vec\mu_2^{(\pi')}(\sigma_2)\bigg)^l\bigg].
	\end{align*}
Since $\MU_1^{(\pi)},\MU_2^{(\pi)},\MU_1^{(\pi')},\MU_2^{(\pi')}$ are mutually independent,
the expression on the right hand side can be rewritten as
	\begin{align*}
	&\sum_{\sigma_1,\ldots,\sigma_l\in\Omega}\Erw\brk{\bc{\prod_{j=1}^l\MU_1^{(\pi)}(\sigma_j)}\bc{\prod_{j=1}^l{\MU_2^{(\pi)}(\sigma_j)}}		-2\bc{\prod_{j=1}^l\MU_1^{(\pi)}(\sigma_j)}\bc{\prod_{j=1}^l{\MU_2^{(\pi')}(\sigma_j)}}
		+\bc{\prod_{j=1}^l\MU_1^{(\pi')}(\sigma_j)}\bc{\prod_{j=1}^l{\MU_2^{(\pi')}(\sigma_j)}}}\\
	&\qquad\qquad\qquad=\sum_{\sigma_1,\ldots,\sigma_l\in\Omega}\bc{\Erw\brk{{\prod_{j=1}^l\MU_1^{(\pi)}(\sigma_j)}}-
		\Erw\brk{{\prod_{j=1}^l\MU_1^{(\pi')}(\sigma_j)}}}^2.
	\end{align*}
Clearly the last expression is non-negative, whence {\bf POS} follows.
\end{proof}

\begin{proof}[Proof of \Thm~\ref{Thm_Potts}]
A straightforward calculation reveals that in the case of the Potts model the formula 
from \Thm~\ref{eqSBM} boils down to the expression $\cB_{\mathrm{Potts}}(q,d,1-\exp(-\beta))$ from (\ref{eqSBM}).
Therefore, the assertion follows from Corollaries~\ref{Cor_cond} and~\ref{Cor_simple}.
\end{proof}

Finally, we supply the argument to support the statement following (\ref{eqSBM_3}).
\begin{lemma}\label{Lemma_monotonicityFix}
With $\dc(\beta)$ from (\ref{eqSBM_3}) we have
	\begin{align*}
	\lim_{n\to\infty}-\frac1n\Erw[\ln Z_\beta(\GG(n,d))]&=-\ln q-d\ln(1-(1-\exp(-\beta))/q)/2&
		&\mbox{iff}&d&\leq\dc(\beta).
	\end{align*}
\end{lemma}
\begin{proof}
The definition of $\dc(\beta)$ matches the definition of the information-theoretic threshold from (\ref{eqThm_G_infTh}) for the Potts model.
Therefore, the assertion follows from \Thm~\ref{Cor_cond}, \Lem~\ref{Lemma_PottsAssumptions}
and the continuity of the function $d\mapsto\lim_{n\to\infty}-\frac1n\Erw[\ln Z_\beta(\GG(n,d))]$, which is established in~\cite{bayati}.
\end{proof}

\subsection{Proof of \Thm~\ref{Thm_SBM}}\label{Sec_sbm}

To derive \Thm~\ref{Thm_SBM} from \Thm~\ref{Thm_stat} a bit of work is required because 
the total number of edges that are present in the stochastic block model contains a small bit of information about the ground truth.
Specifically, the total number of edges contains a hint as to how ``balanced'' the ground truth $\SIGMA^*$ is.
Yet we will show that the \disso\ stochastic block model is mutually contiguous with the planted Potts antiferromagnet.
We tacitly condition on the event $\mathfrak S$ that neither graph features multiple edges; this has a negligible effect on the mutual information as the number of multiple edges is well known to be Poisson with a constant mean (cf.\ \Lem~\ref{Lemma_simple}).

\begin{lemma}
\label{lem:sbmcontig}
The random graphs $\Gsbm^*(\SIGMA^*)$ and $\Gpotts^*(\SIGMA^*)$ are mutually contiguous for all $q\geq2$, $d>0$, $\beta>0$.
\end{lemma}
\begin{proof}
We identify $\Gsbm^*(\SIGMA^*)$ with a factor graph model in the obvious way by identifying the
 edges of the original graph correspond to the constraint nodes of the factor graph.
Let $G$ be any possible outcome of $\Gsbm^*(\SIGMA^*)$.
Let  $m(G, \SIGMA^*)$ the number of monochromatic edges under $\SIGMA^*$, and $M(\SIGMA^*)$ the number of monochromatic pairs of vertices under $\SIGMA^*$.  Then
\begin{align*}
\Pr[ \Gsbm^*(\SIGMA^*) = G] &= e^{-\beta m(G,\SIGMA^*)} \left (\frac{d}{n ( q-1 + e^{-\beta}  )} \right)^{|E|}    \left( 1-  \frac{d}{n ( q-1 + e^{-\beta}  )}  \right )^{\binom{n}{2} - M(\SIGMA^*) -|E|+ m(G,\SIGMA^*)}  \\
& \qquad\qquad\qquad\qquad\qquad\qquad\qquad\qquad\cdot \left(  1-  \frac{d e^{-\beta}}{n ( q-1 + e^{-\beta}  )} \right) ^{M(\SIGMA^*)- m(G,\SIGMA^*)  }.
\end{align*}
For the planted Potts model, each edge is added independently with probability of the form $\Pr[\Po(\lambda)\ge 1]$ where $\lambda = \Theta(1/n)$ and depends whether the edge is monochromatic under $\SIGMA^*$:  
\begin{align*}
&\lambda_{\mathrm{in}} = \frac{ dqn e^{-\beta}  }{2 ((e^{-\beta}-1) M(\SIGMA^*) + \binom{n}{2}   )    }  
&\lambda_{\mathrm{out}}= \frac{ dqn  }{2 ((e^{-\beta}-1) M(\SIGMA^*) + \binom{n}{2}   )    }  
\end{align*}
and we can write
\begin{align*}
\Pr[ \Gpotts^*(\SIGMA^*) = G] &= e^{-\beta m(G, \SIGMA^*)}  (\lambda_{\mathrm{out}} + O(n^{-2})) ^{|E|}  \\
& \cdot \left(1- \lambda_{\mathrm{out}}  +O(n^{-2})   \right)^{\binom{n}{2} - M(\SIGMA^*) -|E|+ m(G,\SIGMA^*)}  \cdot \left(1- \lambda_{\mathrm{in}}  +O(n^{-2})   \right) ^{M(\SIGMA^*)- m(G,\SIGMA^*)  } 
\end{align*}

Now suppose for some large $C$, $\left | M(\SIGMA^*) - \frac{n^2}{2q} \right | \le C n$, then 
\begin{align*}
\frac{ dqn  }{2 ((e^{-\beta}-1) M(\SIGMA^*) + \binom{n}{2}   )    } &= \frac{ dqn  }{2 ((e^{-\beta}-1)n^2/2q + n^2/2 + O(Cn)   )    } 
= \frac{ d}{n( q - 1 +e^{-\beta})  } (1 +O(C/n)) ,
\end{align*}
and so  
\begin{align*}
\Pr[ \Gpotts^*(\SIGMA^*) = G] &= e^{-\beta m(G, \SIGMA^*)}  \left( \frac{ d}{n( q - 1 +e^{-\beta})  } (1 +O(C/n)) \right) ^{|E|}  \\
& \cdot \left(1- \frac{ d}{n( q - 1 +e^{-\beta})  } (1 +O(C/n))  \right)^{\binom{n}{2} - M(\SIGMA^*) -|E|+ m(G,\SIGMA^*)}  \\
&\cdot \left( 1-  \frac{ d e^{-\beta}}{n( q - 1 +e^{-\beta})  } (1 +O(C/n))  \right) ^{M(\SIGMA^*)- m(G,\SIGMA^*)  }.
\end{align*}
And so if we have $\left | M(\SIGMA^*) - \frac{n^2}{2q} \right | \le C n$, $|E| \le Cn$ and $m(G,\SIGMA^*) \le Cn$,   then for some $C'$,
\[ \frac{1}{C'} \le  \frac{ \Pr[ \Gsbm^*(\SIGMA^*) = G | \SIGMA^*] }{\Pr[ \Gpotts^*(\SIGMA^*) = G | \SIGMA^*]   }  \le C'  . \]
Moreover, these conditions all occur with probability tending to $1$ as $C \to \infty$, which proves mutual contiguity.  
\end{proof}

\noindent
We also recall from \Prop~\ref{Lemma_Nishi} that $\Gpotts^*(\SIGMA^*)$ and $\hGpotts$ are mutually contiguous.
Write $\rho(\sigma,\tau)$ for the $q\times q$-overlap matrix of two colorings $\sigma,\tau$, defined by
	$$\rho_{ij}(\sigma,\tau)=\frac1n|\sigma^{-1}(i)\cap\tau^{-1}(j)|.$$
Accordingly we write $\rho(\sigma_1,\ldots,\sigma_l)\in\cP(\Omega^l)$ for the $l$-wise overlaps, i.e.,
	$$\rho_{i_1,\ldots,i_l}(\sigma_1,\ldots,\sigma_l)=\frac1n\abs{\bigcap_{j=1}^l\sigma_j^{-1}(i_j)}.$$
Let $\bar\rho\in\cP(\Omega^l)$ be the uniform distribution (for any $l$).
The following proposition marks the main step toward deriving \Thm~\ref{Thm_SBM} from \Thm~\ref{Thm_G}.
In the following we write $\hat\G=\hGpotts$ and $\G^*=\Gpotts^*$ for brevity.

\begin{proposition}\label{Prop_ssc}
With $d_{\mathrm{inf}}(q,\beta)$ as in \Thm~\ref{Thm_SBM} the following is true.
\begin{enumerate}
\item For all $d<d_{\mathrm{inf}}(q,\beta)$ we have
	\begin{align}\label{eqOverlapCondition}
		\Erw\bck{\|\rho(\SIGMA_1,\SIGMA_2)-\bar\rho\|_2}_{\hat\G}=o(1).
	\end{align}
\item For every $d_{\mathrm{inf}}(q,\beta)<d{\leq((q-c_\beta)/c_\beta)^2}$ there is $\eps>0$ such that
	\begin{align}\label{eqOverlapCondition2}
		\Erw \bck{\|\rho(\SIGMA_1,\SIGMA_2)-\bar\rho\|_2}_{\hat\G}>\eps.
	\end{align}
\end{enumerate}
\end{proposition}

\noindent
To prove \Prop~\ref{Prop_ssc} we need a few preparations.

\begin{lemma}
\label{lem:overlap1}
  Fix $\beta$ and suppose that for some $d>0$, the average overlap is non-trivial.  That is, for some $\eps >0$,
\begin{equation}
\label{eq:pottsmono1}
 \pr\brk{\bck{\norm{\rho(\SIGMA,\TAU)-\bar\rho}_2}_{\hat\G(n,\vec m_d(n),p_\beta)} > \eps} > \eps.
\end{equation}
Then there exists $\del>0$ so that for all $d<d' <d+\delta$, the overlap is non-trivial as well, i.e.,
\begin{align*}
\pr\brk{\bck{\norm{\rho(\SIGMA,\TAU)-\bar\rho}_2}_{\hat\G(n,\vec m_{d'}(n),p_\beta)} > \del} > \del.
\end{align*}
\end{lemma}

\noindent
Call a vector $\sigma \in \Omega^n$ {\em nearly balanced} if for all $\omega \in \Omega$, $\left|  |\sigma^{-1}(\omega) | - n/|\Omega| \right| < n^{3/5}$. To prove \Lem~\ref{lem:overlap1} we need the following fact.

\begin{lemma}
\label{lem:overlap2}
For all $\eps> 0$ there exists $\delta>0$ so that for large enough $n$ for any probability measure $\mu\in\cP(\Omega^n)$ the following is true.
If
\begin{equation}
\label{eq:ExpOverlp}
\bck{\norm{\rho(\SIGMA,\TAU)-\bar\rho}_2}_{\mu} <\delta
\end{equation}
then for any nearly balanced vector $\tilde \sigma \in \Omega^n$,
\begin{equation}
\label{eq:fixV}
\bck{\norm{\rho(\SIGMA,\tilde \sigma)-\bar\rho}_2}_{\mu} <\eps,
\end{equation}
and for any vector $\tau \in \Omega^n$,
\begin{equation}
\label{eq:fixV2}
\bck{A(\SIGMA,\tau)}_{\mu} <\eps.
\end{equation}
\end{lemma}
\begin{proof}
Given $\eps>0$ choose a small enough $\eta=\eta(\eps,\Omega)>0$ and a smaller $\delta=\delta(\eta,\Omega)>0$ and assume $n=n(\delta)$ is sufficiently large.
{By~\cite[\Cor~2.2 and \Prop~2.5]{Victor}} there exists $K=K(\eta,\Omega)>0$ and pairwise disjoint $S_0,\ldots,S_K\subset\Omega^n$ such that
	\begin{enumerate}[(i)]
	\item  $\mu[\nix|S_i]$ is $\eta$-symmetric for all $i\in[K]$,
	\item $\sum_{i\in[K]}\mu(S_i)\geq1-\eta$ and 
	\item $\mu(S_i)\geq\eta/K$ for all $i\in[K]$.
	\end{enumerate}
Let us write $\bck{\nix}_i=\bck{\nix}_{\mu[\nix|S_i]}$ for the average w.r.t.\ the conditional distribution $\mu[\nix|S_i]$.
Due to (iii) we can choose $\delta$ small enough so that (\ref{eq:ExpOverlp}) implies
	\begin{equation}\label{eq:ExpOverlp1}
	\bck{\norm{\rho(\SIGMA,\TAU)-\bar\rho}_2^2}_{i} <\sqrt\delta\qquad\mbox{for all $i\in[K]$}.
	\end{equation}
Further, define a random variable $R_{st}(v)=\vecone\{\SIGMA(v)=s,\TAU(v)=t\}$.
Then
	\begin{align*}
	\bck{\norm{\rho(\SIGMA,\TAU)-\bar\rho}_2^2}_{i}&=\sum_{s,t\in[q]}\bck{(\rho_{st}(\SIGMA,\TAU)-q^{-2})^2}_{i}
		=\sum_{s,t\in[q]}\bck{\bc{\frac1{n}\sum_{v\in[n]}R_{st}(v)-q^{-2}}^2}_{i}\\
		&=\sum_{s,t\in[q]}\brk{\frac1{n^2}\sum_{v,w\in[n]}\bck{R_{st}(v)R_{st}(w)}_i
			-\frac{2q^{-2}}n\sum_{v\in[n]}\bck{R_{st}(v)}_i+q^{-4}}.
	\end{align*}
Hence, (\ref{eq:ExpOverlp1}) and (i) imply that for all $i\in[K]$,
	\begin{align*}
	\sqrt\delta
		&\geq O(\eta)+\brk{\sum_{s,t\in[q]}\bc{\frac1n\sum_{v\in[n]}\bck{R_{st}(v)}_i
			-q^{-2}}^2}
		=O(\eta)+\frac1{n^2}\brk{\sum_{s,t\in[q]}\bc{\frac1n\sum_{v\in[n]}\mu_v(s|S_i)\mu_v(t|S_i)
			-q^{-2}}^2}.
	\end{align*}
Consequently, for all $s,t\in\Omega$ we have	$\abs{q^{-2}-\frac1n\sum_{v\in[n]}\mu_v(s|S_i)\mu_v(t|S_i)}\leq O(\sqrt\eta).$
Therefore, for all $s\in\Omega$
	\begin{align}\label{eq:ExpOverlp2}
	\abs{q^{-1}-\frac1n\sum_{v\in[n]}\mu_v(s|S_i)}&\leq O(\sqrt\eta),&
	\abs{q^{-2}-\frac1n\sum_{v\in[n]}\mu_v(s|S_i)^2}&\leq O(\sqrt\eta).
	\end{align}
Since a sum of squares is minimized by a uniform distribution, (\ref{eq:ExpOverlp2}) implies that for all $i\in[K]$,
	\begin{equation}\label{eq:ExpOverlp3}
	\frac1n\sum_{v\in[n]}\TV{\mu_v(\nix|S_i)-q^{-1}\vecone}\leq\eta^{1/8}.
	\end{equation}
Together with (ii) and {\cite[\Lem~2.8]{Victor}} equation (\ref{eq:ExpOverlp3}) implies that $\mu$ is $\eps^3$-symmetric and 
	\begin{equation}\label{eq:ExpOverlp4}
	\frac1n\sum_{v\in[n]}\TV{\mu_v-q^{-1}\vecone}\leq\eps^3.
	\end{equation}

To prove~\eqref{eq:fixV}, let $U=\tilde\sigma^{-1}(i)$ for some $i\in[q]$.
Since $\tilde\sigma$ is nearly balanced, we have $|U|\geq n/(2q)$.
For $s\in[q]$ let $X_s$ be the number of $u\in U$ such that $\SIGMA(u)=s$.
Then (\ref{eq:ExpOverlp4}) implies that $\bck{X_s}_\mu=(q^{-1}+O(\eps^3))|U|$.
Moreover, because $\mu$ is $\eps^3$-symmetric we have
	\begin{align*}
	\bck{X_s^2}_\mu&=\sum_{u,v\in U}\bck{\vecone\{\SIGMA(u)=s\}\vecone\{\SIGMA(v)=s\}}_\mu
		=|U|^2(q^{-2}+O(\eps^3)).
	\end{align*}
Therefore, Chebyshev's inequality implies that $\bck{\vecone\{|X_s-q^{-1}|U|>\eps|U|\}}_\mu=O(\eps)$.
Hence, $\bck{\|\rho(\SIGMA,\tilde\sigma) -\bar\rho\|_2}_\mu=O(\eps)$, giving~\eqref{eq:fixV}.

Proving~\eqref{eq:fixV2} is similar. Let $\kappa \in S_q$ be a fixed permutation. Let $U_i=\tau^{-1}(i)$.  Summing over all $i \in [q]$, either $|\tau^{-1} (i)| < \eps n$ or  as above we have $\bck{\vecone\{|X_{\kappa(i)}-q^{-1}|U|>\eps|U|\}}_\mu=O(\eps)$, and so
\begin{align*}
\bck{ \frac{q}{(q-1)n} \sum_{x \in V} ( \vec 1 \{ \tau (x) = \kappa ( \SIGMA(x))    -1/q ) }_{\mu} &=  O(\eps).
\end{align*}
Then summing over all $\kappa \in S_q$ gives~\eqref{eq:fixV2}.
\end{proof}

We now make a connection between the normalized agreement with the planted partition and the overlap.  
\begin{lemma}
\label{lem:algOvr}
Suppose $\Erw\bck{\norm{\rho(\SIGMA,\TAU)-\bar\rho}_2}_{\hat\G} > \eps$.  Then there is an algorithm that given  $ \G^*(\hat\SIGMA)$  outputs a nearly balanced  $\tau (\G^*(\hat\SIGMA))$ so that 
\begin{equation}\label{eqlem:algOvr_1} \Erw [ A (\hat \SIGMA, \tau(\G^*(\hat\SIGMA))] > \frac{\eps^3}{8q^3}.\end{equation}
\end{lemma}
\begin{proof}
By \Prop~\ref{Lemma_Nishi} $(\hat\G,\SIGMA_{\hat\G})$ and $(\G^*(\hat\SIGMA),\hat\SIGMA)$ are identically distributed.
Given $\hat\G$, the ``obvious'' (deterministic) algorithm is to output a coloring $\tau=\tau(\hat\G)$ that maximizes
	$\bck{A ( \SIGMA_{\hat\G}, \tau)}_{\hat\G}$, with ties broken arbitrarily.
To establish that this algorithm delivers (\ref{eqlem:algOvr_1}) it suffices to show that
	\begin{equation}\label{eqlem:algOvr_0}
	\Erw\bck{A(\SIGMA_{\hat\G},\TAU_{\hat\G})}_{\hat\G}> \frac{\eps^2}{8q^3}.
	\end{equation}

To show (\ref{eqlem:algOvr_0}) observe that if $\Erw\bck{\norm{\rho(\SIGMA,\TAU)-\bar\rho}_2}_{\hat\G} > \eps$ then 
	\begin{equation}\label{eqlem:algOvr_2}
	\pr\brk{\bck{\norm{\rho(\SIGMA,\TAU)-\bar\rho}_2}_{\hat\G}>\eps}>\eps.
	\end{equation}
Further, assuming that $\hat\G$ is such that $\bck{\norm{\rho(\SIGMA,\TAU)-\bar\rho}_2}_{\hat\G}>\eps$, we obtain
	\begin{equation}\label{eqlem:algOvr_3}
	\bck{\vecone\{\norm{\rho(\SIGMA,\TAU)-\bar\rho}_2>\eps\}}_{\hat\G}>\eps.
	\end{equation}
In addition,  since by \Lem~\ref{Lemma_PottsAssumptions} the Potts model satisfied {\bf BAL},
\Lem~\ref{Lemma_contig} shows that $\SIGMA=\SIGMA_{\hat\G},\TAU=\TAU_{\hat\G}$ are nearly balanced \whp\
and we 
are going to show momentarily that
	\begin{equation}\label{eqlem:algOvr_4}
	\mbox{ $\SIGMA,\TAU$ are nearly balanced and $\norm{\rho(\SIGMA,\TAU)-\bar\rho}_2>\eps$}\ \Rightarrow\ 
		A (\SIGMA, \TAU) \ge \frac{\eps}{4q^3}
	\end{equation}
so that (\ref{eqlem:algOvr_0}) follows from (\ref{eqlem:algOvr_2}) and (\ref{eqlem:algOvr_3}).

Thus, we are left to prove (\ref{eqlem:algOvr_4}).
Consider the $q \times q$ matrix $M$ where $M_{ij} = \rho_{ij} (\SIGMA, \TAU) -1/q^2$.
Then  all row and column sums are $O(n^{-1/3})$ since $\SIGMA$ and $\TAU$ are nearly balanced.
The condition $\norm{\rho(\SIGMA,\TAU)-\bar\rho}_2 > \eps/2$ implies that $\sum_{i,j} M_{ij}^2  > \eps^2/4$.
If so, then $\sum_{i,j} |M_{ij}| \ge \eps/2$, and so $\sum_{i,j} (M_{i,j})_+ \ge \eps/4$.
This implies that there is some entry $M_{ij}$ with $M_{ij} \ge \eps/4q^2$.
Let $M'$ be the $(q-1)\times(q-1)$ matrix obtained by removing row $i$ and column $j$ from $M$.
We claim there is some permutation $\kappa' \in S_{q-1}$ so that $\sum_{i'} M'_{i', \kappa'(i')} \ge 0 $.
This is because the nearly $0$ row and column sums mean that the sum of all entries of $M'$ is $M_{ij} +o(1) \ge \eps/4q^2$.
If we pick a random permutation $\kappa'$, then in expectation the sum $\sum_{i'} M'_{i', \kappa'(i')}  \ge \eps/2q^2 $ and so there exists some $\kappa'$ with a non-negative sum.
Adjoining $\kappa '$ with $i \mapsto j$ gives a permutation $\kappa \in S_q$ so that 
\[ \sum_{i} \rho_{i \kappa(i)} -1/q^2 > \eps/4q^2 . \]
Now
\begin{align*}
A (\SIGMA, \TAU) &= \max_{\kappa \in S_q}  \frac{q}{(q-1)n} \sum_{x \in V} ( \vec 1 \{\SIGMA (x) = \kappa ( \TAU(x))    -1/q ) 
=\max_{\kappa \in S_q}  - \frac{1}{q-1} +  \frac{q}{q-1} \sum_{i} \rho_{i \kappa(i)}( \SIGMA, \TAU)  \\
&= \frac{1}{q-1} \max_{\kappa \in S_q} \sum_{i} (\rho_{i \kappa(i)}(\SIGMA, \TAU) -1/q^2) 
\ge \frac{\eps}{4q^3},
\end{align*}
as desired.
\end{proof}

\begin{proof}[Proof of \Lem~\ref{lem:overlap1}]
Pick a small enough $\eta=\eta(d,\eps)>0$ and a smaller $\delta=\delta(\eta)>0$.
Let $d<d'<d+\delta$.
We claim that $\hat\SIGMA=\hat\SIGMA_{n,\vec m_{d},p_\beta}$ and $\hat\SIGMA{}'=\hat\SIGMA_{n,\vec m_{d'},p_\beta}$ have total variation distance less than $\eta$.
Indeed, for any coloring $\sigma$ and any $m,m'$ we find
	\begin{align*}
	\ln\frac{\Erw[\psi_{\G(n,m',p_\beta)}(\sigma)]}{\Erw[\psi_{\G(n,m,p_\beta)}(\sigma)]}&
		=(m'-m)\ln\bc{1-c_\beta\sum_{\omega\in\Omega}\lambda_\sigma(\omega)^2}.
	\end{align*}
Hence, if 
$\sigma$ is nearly balanced, then there is a constant $C=C(q)>0$ such that
	\begin{align*}
	\abs{\ln\frac{\Erw[\psi_{\G(n,m',p_\beta)}(\sigma)]}{\Erw[\psi_{\G(n,m,p_\beta)}(\sigma)]}-(m'-m)\ln\bc{1-c_\beta/q}}
		&\leq C(m'-m)\sum_{\omega\in\Omega}(\lambda_\sigma(\omega)-1/q)^2.
	\end{align*}
Therefore, the desired bound on the total variation distance follows from (\ref{eqNishi1}).
In effect, we can couple $\hat\SIGMA$, $\hat\SIGMA{}'$ such that both coincide with probability at least $1-\eta$.
If indeed $\hat\SIGMA=\hat\SIGMA{}'$, then we obtain $\G''$ from $\G'=\G^*(\hat\SIGMA)$ by adding a random number
$\Delta=\Po((d'-d)n/k)$ of further constraint nodes according to (\ref{eqTeacher})
and otherwise $\G''$ contains $\vec m_{d'}$ random constraint nodes chosen independently of the constraint nodes of $\G'$
so that $\G''$ is distributed as $\G^*(n,\vec m_{d'},p_\beta,\hat\SIGMA{}')$.
Thus, we have got a coupling of $\G'$ and $\G''$ such that with probability at least $1-\eta$ the former is obtained from the latter
by omitting $\Delta$ random constraint nodes.

Using  \Prop~\ref{Lemma_Nishi}, \Lem~\ref{Lemma_contig} and \Lem~\ref{lem:algOvr},
	\eqref{eq:pottsmono1} implies that there is an algorithm that given $\G'$, finds a nearly balanced partition $\tau(\G')$
	with $A(\tau(\G',\hat\SIGMA))>\eta$ with probability at least $3\eta$.
Hence, by applying this algorithm to the factor graph obtained from $\G''$ by deleting $\Delta$ random constraint nodes
we conclude that with probability at least $\eta$ we can identify a nearly balanced $\tau'(\G'')$ such that $A(\tau'(\G''),\hat\SIGMA))>\eta$.
Consequently, \Prop~\ref{Lemma_Nishi} yields
	\begin{align*}
	\Erw\bck{A(\tau'(\hat\G(n,\vec m_{d'},p_\beta)),\SIGMA)}_{\hat\G(n,\vec m_{d'},p_\beta)}\geq\eta^2.
	\end{align*}
Thus, 
\Lem~\ref{lem:overlap2} shows that two samples from $\mu_{\hat \G}$ must have non-trivial expected overlap. 
\end{proof}

\begin{lemma}\label{Lemma_quiet}
For all $\beta,d,q$ we have $\Erw[\ln Z(\hat\G)]\geq n\ln q+dn\ln(1-c_\beta/q)/2+o(n)$.
\end{lemma}
\begin{proof}
Since $\Erw[Z(\G(n,m,p_\beta))]=n\ln q+dn\ln(1-c_\beta/q)/2+o(n)$, the 
assertion follows from (\ref{eqNishi3}).
\end{proof}

\begin{lemma}\label{Lemma_dd}
For all $d>0$ we have $\frac1n\frac{\partial}{\partial d}\Erw\ln Z(\hat\G)\geq\ln(1-c_\beta/q)+o(1)$ and
if (\ref{eqOverlapCondition}) is violated, then
	$\frac1n\frac{\partial}{\partial d}\Erw\ln Z(\hat\G)\geq\ln(1-c_\beta/q)+\Omega(1).$
\end{lemma}
\begin{proof}
The same calculation as in \Lem~\ref{Lemma_PoissonDeriv} shows that
	$$\frac1n\frac{\partial}{\partial d}\Erw\ln Z(\hat\G)=
		\Erw[\ln Z(\hat\G(n,\vec m+1,p_\beta))]-\Erw[\ln Z(\hat\G(n,\vec m,p_\beta))].$$
Furthermore,
with $\hat\SIGMA=\hat\SIGMA_{n,\vec m,p_\beta}$ and $\hat\SIGMA'=\hat\SIGMA_{n,\vec m+1,p_\beta}$,
by \Prop~\ref{Lemma_Nishi} we can identify  $\hat\G(n,\vec m+1,p_\beta)$
with $\G^*(n,\vec m+1,p_\beta,\hat\SIGMA')$
and $\hat\G(n,\vec m,p_\beta)$ with $\G^*(n,\vec m,p_\beta,\hat\SIGMA)$.
Moreover, \Cor~\ref{Cor_GenCouple} shows that we can couple $\hat\SIGMA,\hat\SIGMA'$ such that both coincide with probability $1-O(1/n)$
and such that $|\hat\SIGMA\triangle\hat\SIGMA'|=\tilde O(n^{-1/2})$ with probability $1-O(n^{-2})$.
Further, as in the proof of \Lem~\ref{Lemma_Deltat} this coupling extends to a coupling of $\G'=\G^*(n,\vec m+1,p_\beta,\hat\SIGMA')$
and $\G''=\G^*(n,\vec m,p_\beta,\hat\SIGMA')$ such that in the case $\hat\SIGMA=\hat\SIGMA'$
we obtain $\G''$ from $\G'$ by adding one additional random constraint node $\vec e$ chosen from (\ref{eqTeacher}) and such that
$\Erw[\ln(Z(\G'')/Z(\G'))|\hat\SIGMA\neq\hat\SIGMA']=\tilde O(n^{1/2})$.
Hence, letting $\bck\nix=\bck\nix_{\G'}$, we find
	\begin{align}\label{eqLemma_dd1}
	\frac1n\frac{\partial}{\partial d}\Erw\ln Z(\hat\G)=
		\Erw[\ln(Z(\G'')/Z(\G'))|\hat\SIGMA=\hat\SIGMA']+\tilde O(n^{-1/2})
			=\Erw\ln\bck{\psi_{\vec e}(\SIGMA)}_{\G'}+o(1).
	\end{align}
Further, writing $\vec v,\vec w$ for the two variable nodes adjacent to $\vec e$ and expanding the logarithm, we obtain
	\begin{align}\nonumber
	\ln\bck{\psi_{\vec e}(\SIGMA)}_{\G'}&=\ln(1-\bck{c_\beta\vecone\{\SIGMA(\vec v)=\SIGMA(\vec w)\}}_{\G'}
		=-\sum_{l=1}^\infty\frac{c_\beta^l}l\bck{\vecone\{\SIGMA(\vec v)=\SIGMA(\vec w)\}}_{\G'}^l\\
		&=-\sum_{l=1}^\infty\frac{c_\beta^l}l\bck{\prod_{j=1}^l\vecone\{\SIGMA_j(\vec v)=\SIGMA_j(\vec w)\}}_{\G'}.
			\label{eqLemma_dd2}
	\end{align}
Since $\vec v,\vec w$ are chosen from (\ref{eqTeacher}), (\ref{eqLemma_dd1}) and (\ref{eqLemma_dd2}) yield
	\begin{align*}
	\frac1n\frac{\partial}{\partial d}\Erw\ln Z(\hat\G)&=
		o(1)-\sum_{v,w\in V}\sum_{l\geq1}\frac{c_\beta^l}l
			\Erw\brk{\frac{1-c_\beta\vecone\{\hat\SIGMA(v)=\hat\SIGMA(w)\}}{\sum_{s,t\in V}1-c_\beta\vecone\{\hat\SIGMA(s)=\hat\SIGMA(t)\}}
				\bck{\prod_{j=1}^l\vecone\{\SIGMA_j(v)=\SIGMA_j( w)\}}_{\G'}}.
	\end{align*}
Hence, \Cor~\ref{Cor_intContig} and \Prop~\ref{Lemma_Nishi} yield
	\begin{align*}
	\frac1n\frac{\partial}{\partial d}\Erw\ln Z(\hat\G)&=
		o(1)-\frac1{1-c_\beta/q}\sum_{v,w\in V}\sum_{l\geq1}\frac{c_\beta^l}{ln^2}
			\Erw\brk{\bc{1-c_\beta\vecone\{\hat\SIGMA(v)=\hat\SIGMA(w)\}}
				\bck{\prod_{j=1}^l\vecone\{\SIGMA_j(v)=\SIGMA_j( w)\}}_{\G'}}\\
			&=			o(1)-\frac1{1-c_\beta/q}\sum_{v,w\in V}\sum_{l\geq1}\frac{c_\beta^l}{ln^2}\brk{
								\Erw\bck{\prod_{j=1}^l\vecone\{\SIGMA_j(v)=\SIGMA_j( w)\}}_{\G'}
										-c_\beta\Erw\bck{\prod_{j=1}^{l+1}\vecone\{\SIGMA_j(v)=\SIGMA_j( w)\}}_{\G'}}\\
			&=o(1)-\frac{c_\beta}{q-c_\beta}+\sum_{l\geq2}\frac{c_\beta^l}{l(l-1)}
				\Erw\bck{\frac1{n^2}\sum_{v,w\in V}\prod_{j=1}^l\vecone\{\SIGMA_j(v)=\SIGMA_j( w)\}}_{\G'}.
	\end{align*}
The last expression can be rewritten nicely in terms of $l$-wise overlaps: we obtain
	\begin{align}			\label{eqLemma_dd3}
	\frac1n\frac{\partial}{\partial d}\Erw\ln Z(\hat\G)&=
		o(1)-\frac{c_\beta}{q-c_\beta}+\frac1{1-c_\beta/q}\sum_{l\geq2}\frac{c_\beta^l}{l(l-1)}
				\Erw\bck{\|\rho(\SIGMA_1,\ldots,\SIGMA_l)\|_2^2}_{\G'}.
	\end{align}
Since $\|\rho(\sigma_1,\ldots,\sigma_l)\|_2^2\geq q^{-l}$ for all $\sigma_1,\ldots,\sigma_l$, (\ref{eqLemma_dd3}) yields the first assertion.
Moreover, if $\Erw\bck{\|\rho(\SIGMA_1,\SIGMA_2)-\bar\rho\|_2}_{\hat\G}$ is bounded away from $0$, then
$\Erw\bck{\|\rho(\SIGMA_1,\SIGMA_2)\|_2^2}_{\G'}$ is bounded away from $q^{-2}$ and the second assertion follows.
\end{proof}

\begin{lemma}
If $\beta,d,k$ are such that $\Erw\ln Z(\hat\G)=n\ln q+dn\ln(1-c_\beta/q)/2+o(n)$, then the same holds for all $d'<d$.
\end{lemma}
\begin{proof}
This is immediate from \Lem s~\ref{Lemma_quiet} and \ref{Lemma_dd}.
\end{proof}

\begin{proof}[Proof of \Prop~\ref{Prop_ssc}]
If (\ref{eqOverlapCondition}) is violated, then \Lem~\ref{Lemma_dd} shows that
$\frac1n\frac{\partial}{\partial d}\Erw\ln Z(\hat\G)>\ln(1-c_\beta/q)+\Omega(1)$.
Moreover, by \Lem~\ref{lem:overlap1} the set of all $d$ for which (\ref{eqOverlapCondition}) is violated contains
	 an interval $(d_0,d_0+\delta)$.
Therefore, if (\ref{eqOverlapCondition}) is violated for some $d_0<d_{\mathrm{inf}}(q,\beta)$, then  
 \Lem~\ref{Lemma_quiet} gives
	\begin{align*}
	\Erw[\ln Z(\hat\G(n,\vec m(d_1)))]&=\Erw[\ln Z(\hat\G(n,\vec m(d_0)))]+\int_{d_0}^{d_1}\frac{\partial}{\partial d}\Erw\ln Z(\hat\G)\dd d
		=n\ln q+\frac{d_1n}2\ln(1-c_\beta/q)+\Omega(n),
	\end{align*}
in contradiction to \Cor~\ref{Cor_cond}, \Lem~\ref{Lemma_PottsAssumptions} and the definition of $d_{\mathrm{inf}}(q,\beta)$.
Thus the first assertion follows.

With respect to the second assertion, pick $\eps=\eps(q,d)$ small enough and assume that
	\begin{equation}\label{eqProp_ssc_part2_1}
	\pr\brk{\bck{\|\rho(\SIGMA_1,\SIGMA_2)-\bar\rho\|_2}_{\hat\G}<\eps}>\eps.
	\end{equation}
Then a second moment argument shows that $\Erw\ln Z(\G)\sim\ln\Erw[Z(\G)]$, because $d\leq((q-c_\beta)/c_\beta)^2$.
Indeed, define
	$\cZ(\G)=Z(G)\vecone\{\bck{\|\rho(\SIGMA_1,\SIGMA_2)-\bar\rho\|_2}_{\hat\G}<\eps\}.$
Then (\ref{eqNishi3}) and (\ref{eqProp_ssc_part2_1}) imply that $\Erw[\cZ(\G)]=\Omega(\Erw[Z(\G)])$.
Further, for a given overlap matrix $\rho$ let
	$$Z_\rho^\tensor(\G)=Z(\G)^2\bck{\vecone\{\rho(\SIGMA_1,\SIGMA_2)=\rho\}}_{\G}.$$
Summing over the discrete set of possible overlaps for a given $n$, we obtain from the definition of $\cZ(\G)$ that
	\begin{align}\label{eqProp_ssc_part2_2}
	\Erw[\cZ(\G)^2]&\leq O(1)\sum_{\rho:\|\rho-\bar\rho\|_2<\eps}\Erw[Z_\rho^\tensor(\G)]
		\leq\sum_{\rho:\|\rho-\bar\rho\|_2<\eps}\exp\bc{o(n)+n(H(\rho)+d\ln(1-2/k+c_\beta\|\rho\|_2^2)/2)};
	\end{align}
the last formula follows from a simple inclusion/exclusion argument (cf.~\cite[\Prop~6]{Nor}).
Moreover, expanding the exponent to the second order, we see that
for $d\leq((q-c_\beta)/c_\beta)^2$ the maximizer is just $\bar\rho$.
Consequently, (\ref{eqProp_ssc_part2_2}) implies that $\Erw[\cZ(\G)^2]=\exp(o(n))\Erw[\cZ(\G)]^2$.
Hence, by the Paley-Zygmund inequality, for any fixed $\eps>0$ we have
	$$\pr\brk{Z(\G)\geq\exp(-\eps n)\Erw[Z(\G)]}\geq\pr\brk{\cZ(\G)\geq\exp(-\eps n/2)\Erw[\cZ(\G)]}=\exp(o(n)).$$
Taking $\eps\to0$ sufficiently slowly as $n\to\infty$ and applying \Lem~\ref{Lemma_Azuma} twice,
	we thus get $\Erw[\ln Z(\G)]=\ln \Erw[Z(\G)]+o(n)$.
Therefore, 
another application of \Lem~\ref{Lemma_Azuma}
and \Cor~\ref{Lemma_quietsmm} yields $\Erw\ln Z(\hat\G)\sim\ln\Erw[Z(\G)]$.
But this contradicts the assumption $d_{\mathrm{inf}}(q,\beta)<d$.
\end{proof}

\begin{proof}[Proof of \Thm~\ref{Thm_SBM}]
The theorem follows from \Lem~\ref{lem:sbmcontig}, \Lem~\ref{lem:overlap2}, \Prop~\ref{Prop_ssc}, and \Lem~\ref{lem:algOvr}.
By \Lem~\ref{lem:sbmcontig} it is enough to prove the theorem for the planted Potts model.
First suppose $d< d_{\mathrm{inf}}(q,\beta)$.   Then by \Prop~\ref{Prop_ssc}, we have 
$\Erw\bck{\|\rho(\SIGMA_1,\SIGMA_2)-\bar\rho\|_2}_{\hat\G}=o(1).$
 \Lem~\ref{lem:overlap2}, \eqref{eq:fixV2}, then says that for any $\tau = \tau (\hat \G)$, 
$\bck{A(\SIGMA,\tau)}_{\hat \G} = o(1),$
which by \Prop~\ref{Lemma_Nishi} implies $\bck{A(\hat \SIGMA,\tau)}_{\hat \G} = o(1)$.
 
For the second part of \Thm~\ref{Thm_SBM}, suppose that $d> d_{\mathrm{inf}}(q,\beta)$. We can assume $d\leq((q-c_\beta)/c_\beta)^2$  since if $d>((q-c_\beta)/c_\beta)^2$, the algorithm of Abbe and Sandon~\cite{abbe2015detection} succeeds \whp\
With $d_{\mathrm{inf}}(q,\beta)< d \leq((q-c_\beta)/c_\beta)^2$, \Prop~\ref{Prop_ssc} says that there is some $\eps>0$ so that
		$\Erw \bck{\|\rho(\SIGMA_1,\SIGMA_2)-\bar\rho\|_2}_{\hat\G}>\eps.$
Then for some $\del>0$, the first part of \Lem~\ref{lem:algOvr} implies that there is an algorithm that returns $\tau = \tau (\hat \G)$ so that 
$\Erw[A( \hat \SIGMA, \tau(\hat \G))] > \del,$ completing the proof.
\end{proof}

\subsection{Proof of \Thm~\ref{Thm_col}}\label{Sec_graphcol}
To derive \Thm~\ref{Thm_col} about the graph coloring problem from \Thm~\ref{Thm_G} some care is required 
because we need to accommodate the `hard' constraint that no single edge be monochromatic.
Indeed, if we cast graph coloring as a factor graph model, then the weight functions are $\{0,1\}$-valued.
As in \Sec~\ref{Sec_Potts} we work with the Potts antiferromagnet to circumvent this problem.
Thus, let $\Omega=[q]$ for some $q\geq3$ and let $c_\beta$, $\psi_\beta$ be as in \Sec~\ref{Sec_Potts}.
Let $m_d(d)=m_d(n)=\lceil dn/2\rceil$ and $\vec m_d=\vec m_d(n)=\Po(dn/2)$.
{\Lem~\ref{Lemma_simple} shows that the event $\mathfrak S$ occurs with a non-vanishing probability and throughout this section
we always tacitly condition on $\mathfrak S$.}
Moreover, $\G(n,m,p_\infty)$ denotes the factor graph model where $c_\beta=1$, i.e., the weight function (\ref{eqPottsPsi}) is $\{0,1\}$-valued.
If $Z(\G(n,m,p_\infty))>0$, then we define the Gibbs measure via (\ref{eqGibbs}); otherwise we let $\mu_{\G(n,m,p_\infty)}$ be the uniform
distribution on $\Omega^n$.
Of course none of the results from \Sec~\ref{Sec_rss} apply to $\beta=\infty$ directly.
But the plan is to apply \Thm~\ref{Thm_stat} to the Potts antiferromagnet and take $\beta\to\infty$.
To carry this out we need to apply a few known facts about the random graph coloring problem.

\begin{lemma}[{\cite{Barriers}}]\label{Lemma_Ehud}
For any $q\geq3$ and any $\zeta>0$ the property
	$$\cA_{q,\zeta}=\{Z(\G(n,\vec m_d,p_\infty))\geq\zeta^n\}$$
has a non-uniform sharp threshold.
That is, there exists a sequence $(u_{q,\zeta}(n))_n$ such that for any $\eps>0$,
	\begin{align*}
	\lim_{n\to\infty}\pr\brk{\G(n,\vec m_{u_{q,\zeta}(n)-\eps}(n),p_\infty))\in\cA_{q,\zeta}}&=1\quad\mbox{and}\quad
	\lim_{n\to\infty}\pr\brk{\G(n,\vec m_{u_{q,\zeta}(n)+\eps}(n),p_\infty))\in\cA_{q,\zeta}}=0.
	\end{align*}
\end{lemma}

\begin{lemma}\label{Lemma_NotQuietEnough}
If $d>0$, $\delta>0$ are such that for a strictly increasing sequence $(n_l)_l$ we have
	\begin{equation}\label{eqLemma_NotQuietEnough}
	\liminf_{l\to\infty}\frac1{n_l}\Erw\ln Z(\hat\G(n_l,\vec m_d(n_l),p_\beta))>\ln q+\frac d2\ln(1-c_\beta/q)+\delta,
	\end{equation}
 for all large enough $\beta>0$, then 
	\begin{equation}\label{eqLemma_NotQuietEnough2}
	\limsup_{l\to\infty}\Erw[Z(\G(n_l,\vec m_d(n_l),p_\infty))^{1/n_l}]<q(1-1/q)^{d/2}.
	\end{equation}
\end{lemma}
\begin{proof}
By \Prop~\ref{Lemma_Nishi} and \Lem~\ref{Lemma_Azuma}, (\ref{eqLemma_NotQuietEnough}) implies 
	\begin{equation}\label{eqLemma_NotQuietEnough3}
	\liminf_{l\to\infty}\frac1{n_l}\Erw\ln Z(\G^*(n_l,\vec m_d(n_l),p_\beta,\SIGMA^*))>\ln q+\frac d2\ln(1-c_\beta/q)+\delta.
	\end{equation}
Further, we claim that (\ref{eqLemma_NotQuietEnough3}) implies that for large enough $\beta$
	\begin{equation}\label{eqLemma_NotQuietEnough4}
	\liminf_{l\to\infty}\frac1{n_l}\Erw\ln Z_\beta(\G^*(n_l,\vec m_d(n_l),p_\infty,\SIGMA^*))>\ln q+\frac d2\ln(1-c_\beta/q)+\delta/2,
	\end{equation}
where $Z_\beta(G)=\sum_{\sigma}\prod_{a\in F(G)}\psi_\beta(\sigma(\partial a))$.
In words, we generate a random graph with the weight distribution $p_\infty$ but evaluate the free energy at inverse temperature $\beta$.
To get from (\ref{eqLemma_NotQuietEnough3}) to (\ref{eqLemma_NotQuietEnough4}), we simply observe that by (\ref{eqTeacher}) the factor graphs
$\G^*(n_l,\vec m_d(n_l),p_\infty,\SIGMA^*)$ and $\G^*(n_l,\vec m_d(n_l),p_\beta,\SIGMA^*)$ can be coupled such that
they differ in at most $2\exp(-\beta)dn/2$ constraint nodes with probability $1-O(n^{-2})$.
Since altering a single constraint node shifts the free energy at inverse temperature $\beta$ by no more than $\beta$ in absolute value, we obtain (\ref{eqLemma_NotQuietEnough4}).

By comparison, the first moment bound (\ref{eqJensen}) implies that
	\begin{equation}\label{eqLemma_NotQuietEnough5}
	\limsup_{l\to\infty}\frac1{n_l}\Erw\ln Z_\beta(\G(n_l,\vec m_d(n_l),p_\infty))\leq\ln q+\frac d2\ln(1-c_\beta/q).
	\end{equation}
Furthermore, by Azuma's inequality both $\ln Z_\beta(\G(n_l,\vec m_d(n_l),p_\infty))$ and
$\ln Z_\beta(\G^*(n_l,\vec m_d(n_l),p_\infty,\SIGMA^*))$ are tightly concentrated.
Therefore, there exists $\beta>0$ such that 
	\begin{align*}
	\pr\brk{n_l^{-1}\ln Z_\beta(\G^*(n_l,\vec m_d(n_l),p_\infty,\SIGMA^*))\leq \ln q+\frac d2\ln(1-c_\beta/q)+\delta/2}&\leq\exp(-\Omega(n)),\\
	\pr\brk{n_l^{-1}\ln Z_\beta(\G(n_l,\vec m_d(n_l),p_\infty))\geq \ln q+\frac d2\ln(1-c_\beta/q)+\delta/2}&\leq\exp(-\Omega(n)),
	\end{align*}
and thus the assertion follows from \cite[\Lem~6.2]{Cond}.
\end{proof}

\noindent
Call $\sigma:V\to\Omega$ {\em balanced} if $|\sigma^{-1}(\omega)|\in\{\lceil n/|\Omega|\rceil,\lfloor n/|\Omega|\rfloor\}$
	for all $\omega\in\Omega$.
Let $\mathfrak B(n,\Omega)$ be the set of all balanced  $\sigma$.
Further, for a factor graph $G$ define the ``balanced'' partition function as 
	\begin{align*}
	\tilde Z(G)=\sum_{\tilde\sigma\in\mathfrak B(n,q)}\psi_{G}(\tilde\sigma)
	\end{align*}
and let $\tilde\mu_{G}(\nix)=\mu_G(\nix|\mathfrak B(n,\Omega))$ be the corresponding ``balanced'' Gibbs measure.
Furthermore, let us write $\tilde\SIGMA=\tilde\SIGMA_{n,\Omega}$ for a uniformly random element of $\mathfrak B(n,\Omega)$.
Finally, let $\tilde\G(n,m,p_\beta)$ be the balanced version of the factor graph distribution (\ref{eqNishi3}), i.e.,
	\begin{equation}\label{eqBalNishi}
	\pr\brk{\tilde\G=G}=\tilde Z(G)\pr\brk{\G=G}/{\Erw[\tilde Z(\G)]}\qquad\mbox{ for every possible $G$.}
	\end{equation}
The proof of \Prop~\ref{Lemma_Nishi} extends to balanced assignments, which shows that
$\tilde\G$ enjoys the Nishimori property; this was actually already observed (with different terminology) in~\cite{Barriers}.
Formally, we have

\begin{fact}\label{Fact_balancedNishimori}
The pairs $(\tilde\SIGMA,\G^*(n,m,p_\infty,\tilde\SIGMA)$ and 
$(\SIGMA_{\tilde\G(n,m,p_\infty)},\tilde\G(n,m,p_\infty))$ are identically distributed.
\end{fact}

We recall that for two color assignments $\sigma,\tau:V\to\Omega$
the overlap is $\rho(\sigma,\tau)=(\rho_{ij}(\sigma,\tau))_{i,j\in\Omega}$, where
	$$\rho_{ij}(\sigma,\tau)=n^{-1}|\sigma^{-1}(i)\cap\tau^{-1}(j)|.$$
Thus, $\rho(\sigma,\tau)\in\cP(\Omega\times\Omega)$.
For $\rho\in\cP(\Omega\times\Omega)$ let $\|\rho\|_2^2=\sum_{i,j\in\Omega}\rho_{ij}^2$ and
 write $\bar\rho$ for the uniform distribution.

\begin{lemma}[{\cite[Proposition 5.6]{Silent}}]\label{Lemma_balancedSMM}
For any $q\geq3$ there exist $\eps>0$ such that for every $0<d<(q-1)^2$ there is $n_0>0$ such that for all $n>n_0$ and all 
 and all $m\leq dn/2$ the following is true.
Let
	$$\tilde Z^{\tensor}(\G(n,m,p_\infty)=
		\abs{\cbc{(\sigma,\tau)\in\cB(n,[q])\times\cB(n,[q]):
			\|\rho(\sigma,\tau)-\bar\rho\|_2<\eps\mbox{ and }\sigma,\tau\mbox{ are $q$-colorings of }\G(n,m,p_\infty)}}.$$
Then $\Erw[\tilde Z^{\tensor}(\G(n,m,p_\infty)]\leq\eps^{-1}\Erw[\tilde Z(\G(n,m,p_\infty)]^2$.
\end{lemma}

\begin{corollary}\label{Lemma_col2}
For any $q\geq3$,  $0<d<(q-1)^2$ is such there exist $\delta>0$, $n_0>0$ such that for all $n>n_0$ the following is true.
Suppose that $m\leq dn/2$ is such that
	\begin{equation}\label{eqLemma_col2}
	\pr\brk{\bck{\norm{\rho(\SIGMA,\TAU)-\bar\rho}_2}_{\tilde\mu_{\tilde\G(n,m,p_\infty)}}<\delta}\geq2/3.
	\end{equation}
Then $$\pr\brk{Z(\G(n, m,p_\infty))\geq q^n(1-1/q)^{dn/2}\exp(-\ln^2n)}>\delta.$$
\end{corollary}
\begin{proof}
Let $\eps>0$ be the number promised by \Lem~\ref{Lemma_balancedSMM} and pick $\delta=\delta(\eps,q)>0$ small enough.
Define
	$$\cZ(G)=\tilde Z(G)\vecone\cbc{
			\bck{\norm{\rho(\SIGMA,\TAU)-\bar\rho}_2}_{\tilde\mu_{G}}<\delta}.$$
(Thus, $\cZ(G)=0$ if $\tilde Z(G)=0$.)
Combining (\ref{eqBalNishi}) and (\ref{eqLemma_col2}), we obtain
	\begin{equation}\label{eqLemma_col2_2}
	\Erw[\cZ(\G(n,m,p_\infty))]\geq \Erw[\tilde Z(\G(n,m,p_\infty))]/10.
	\end{equation}
Moreover, by construction $\cZ$ satisfies
	$\cZ(\G(n,m,p_\infty))^2\leq2\tilde Z^{\tensor}(\G(n,m,p_\infty)),$
provided $\delta$ is small enough.
Hence, by \Lem~\ref{Lemma_balancedSMM}
	\begin{equation}\label{eqLemma_col2_22}
	\Erw[\cZ(\G(n,m,p_\infty))^2]\leq\frac4\eps\Erw[\tilde Z(\G(n,m,p_\infty))]^2.
	\end{equation}
Combining (\ref{eqLemma_col2_2}) and (\ref{eqLemma_col2_22}) and applying the Paley-Zygmund inequality, we find
	\begin{equation}\label{eqLemma_col2_23}
	\pr\brk{\cZ(\G(n,m,p_\infty))\geq\Erw[\tilde Z(\G(n,m,p_\infty))]/8}\geq\frac{\Erw[\cZ(\G(n,m,p_\infty))]^2}{2\Erw[\cZ(\G(n,m,p_\infty))^2]}
		\geq\frac{\eps^2}{128}.
	\end{equation}
Since a standard calculation shows that 
	$\Erw[\tilde Z(\G(n,m,p_\infty))]\geq n^{-q^2}q^n(1-1/q)^m$ (cf.~\cite[\Sec~3]{AchNaor}) and $m\leq dn/2$,
	(\ref{eqLemma_col2_23}) yields
	\begin{equation}\label{eqLemma_col2_24}
	\pr\brk{Z(\G(n,m,p_\infty))\geq n^{-q^2}q^n(1-1/q)^{dn/2}/8}\geq\frac{\eps^2}{128},
	\end{equation}
as desired.
\end{proof}

\noindent
The following statement is a weak converse of \Cor~\ref{Lemma_col2}.

\begin{lemma}\label{Lemma_overlap1}
For any $\eps>0$ and any $0<d'<d''\leq100(q-1)^2$ there is $\delta>0$ such that the following is true.
Assume that $(n_l)_l$ is a subsequence such that
	\begin{align}\label{eqThm_col1_1}
	\liminf_{l\to\infty}\max_{d'n_l/2\leq m\leq d''n_l/2}
	\pr\brk{\bck{\norm{\rho(\SIGMA,\TAU)-\bar\rho}_2}_{\tilde\G(n_l,m,p_\infty)}<\eps}<1.
	\end{align}
Then
	\begin{align}\label{eqThm_col1_2}
	\limsup_{l\to\infty}\frac1{n_l}\Erw\brk{\ln Z(\tilde\G(n_l,m_{d''}(n_l),p_\infty)}>\ln q+\frac {d''}2\ln(1-1/q)+\delta.
	\end{align}
\end{lemma}
\begin{proof}
Fact~\ref{Fact_balancedNishimori} shows that the Nishimori property extends to the balanced graph coloring problem.
Thus, we obtain $\tilde\G(n,m,p_\infty)$ by first choosing $\tilde\SIGMA\in\mathfrak B(n,[q])$ uniformly and then
generating $\G^*(n,m,p_\infty,\tilde\SIGMA)$.
In effect, we can couple $\tilde\G(n,m,p_\infty)$ and $\tilde\G(n,m+1,p_\infty)$
such that the first is obtained by generating $\G'=\G^*(n,m,p_\infty,\tilde\SIGMA)$ and the second, denoted $\G''$, results by
adding one single random constraint node $\vec e$ incident to a random pair of variable nodes with distinct colors under $\tilde\SIGMA$.
Hence, with $\bck\nix=\bck\nix_{\tilde\mu_{\tilde\G'}}$, we obtain
	\begin{align*}
	\Erw\ln\frac{Z(\G'')}{Z(\G')}&=\Erw\ln\bck{\psi_{\vec e}(\SIGMA)}
		=o(1)+\frac1{n^2(1-1/q)}\sum_{v,w}\Erw\brk{(1-\vecone\{\tilde\SIGMA(v)=\tilde\SIGMA(w)\})
				\ln(1-\bck{\vecone\{\SIGMA(v)=\SIGMA(w)\}})}\\
		&=-\frac1{n^2(1-1/q)}\sum_{v,w}\sum_{l\geq1}\frac{1}l
			\Erw\brk{(1-\vecone\{\tilde\SIGMA(v)=\tilde\SIGMA(w)\})
				\bck{\prod_{j=1}^l\vecone\{\SIGMA_j(v)=\SIGMA_j(w)\}}}.
	\end{align*}
Since by the Nishimori property we can identify $\tilde\SIGMA$ with a sample from the Gibbs measure, we obtain
	\begin{align}\nonumber
	\Erw\ln\frac{Z(\G'')}{Z(\G')}&=o(1)
		-\frac1{n^2(1-1/q)}\sum_{v,w}\sum_{l\geq1}\frac{1}l
			\Erw\brk{\bck{\prod_{j=1}^l\vecone\{\SIGMA_j(v)=\SIGMA_j(w)\}}-
				\bck{\prod_{j=1}^{l+1}\vecone\{\SIGMA_j(v)=\SIGMA_j(w)\}}}\\
		&=-\frac1{q-1}+\sum_{v,w}\sum_{l\geq2}\frac{q}{l(l-1)n^2(q-1)}\Erw\bck{\prod_{j=1}^l\vecone\{\SIGMA_j(v)=\SIGMA_j(w)\}}
			+o(1).\label{eqLemma_overlap1a}
	\end{align}
Write $\rho(\sigma_1,\ldots,\sigma_l)\in\cP(\Omega^l)$ for the $l$-wise overlap; that is, 
	$\rho_{i_1,\ldots,i_l}(\sigma_1,\ldots,\sigma_l)=\frac1n\abs{\bigcap_{j=1}^l\sigma_j^{-1}(i_j)}.$
Then (\ref{eqLemma_overlap1a}) yields
	\begin{align}\label{eqLemma_overlap1b}
	\Erw\ln\frac{Z(\G'')}{Z(\G')}&=o(1)-\frac1{q-1}+\sum_{l\geq2}\frac q{l(l-1)(q-1)}
		\Erw\bck{\|\rho(\SIGMA_1,\ldots,\SIGMA_l)\|_2^2}.
	\end{align}
Hence, if we let $\xi_l=\Erw\bck{\|\rho(\SIGMA_1,\ldots,\SIGMA_l)\|_2^2}-q^{-l}\geq0$, then (\ref{eqLemma_overlap1b}) becomes
	\begin{align}\label{eqLemma_overlap1c}
	\Erw[\ln\tilde\G(n,m+1,p_\infty)]-\Erw[\ln\tilde\G(n,m,p_\infty)]&=\Erw\ln\frac{Z(\G'')}{Z(\G')}=
		o(1)+\ln(1-1/q)+\sum_{l\geq2}\frac{q\xi_l}{l(l-1)(q-1)}.
	\end{align}
Moreover, (\ref{eqBalNishi}) implies that
	\begin{align}\label{eqLemma_overlap1d}
	\Erw[\ln\tilde\G(n,m,p_\infty)]&\geq\ln q+\frac mn\ln(1-1/q)+o(n).
	\end{align}
Finally, since (\ref{eqThm_col1_1}) guarantees that $\xi_2$ is bounded away from $0$,
	(\ref{eqLemma_overlap1c}) and (\ref{eqLemma_overlap1d}) imply (\ref{eqThm_col1_2}).
\end{proof}

\noindent
The following observation shows that we can extend (\ref{eqThm_col1_2}) to sufficiently large but finite $\beta$.

\begin{lemma}\label{Lemma_betaToInfty}
Assume that $d>0$ is such that for some $\delta>0$ and some subsequence $(n_l)_l$ we have
	\begin{align}\label{eqThm_col1_2b}
	\limsup_{l\to\infty}\frac1{n_l}\Erw\brk{\ln Z(\tilde\G(n_l,\vec m(d,n_l),\infty))}>\ln q+\frac{d}2\ln(1-1/q)+2\delta.
	\end{align}
Then for all large enough $\beta$ we have
	\begin{align}\label{eqThm_col1_2a}
	\limsup_{l\to\infty}\frac1{n_l}\Erw\brk{\ln Z(\hat\G(n_l,\vec m(d,n_l),c_\beta))}>\ln q+\frac{d}2\ln(1-1/q)+\delta.
	\end{align}
\end{lemma}
\begin{proof}
By \Cor~\ref{Cor_intContig}   for any $d,\beta$
the distribution of $\hat\SIGMA$ and the uniform distribution on balanced assignments can be coupled such that the distance is $\tilde O(\sqrt n)$
with probability $1-O(n^{-2})$.
Hence, we can couple $\tilde\G(n,\vec m(d),1)$ and  $\hat\G(n,\vec m(d),c_\beta)$
such that they differ on no more than $\exp(-\beta)dn$ constraint nodes with probability $1-O(n^{-2})$.
Since altering a constraint node affects $\ln Z(\hat\G(n,\vec m(d),c_\beta))$ by no more than $\beta$ in absolute value,
we can choose $\beta=\beta(\delta)$ large enough so that (\ref{eqThm_col1_2b}) implies (\ref{eqThm_col1_2a}).
\end{proof}

\noindent
With the notation from (\ref{eqSBM}) define
	\begin{align*}
	\cB_{\mathrm{Potts}}(\pi;q,d,c)&=\Erw\brk{\frac{(1-c/q)^{-\vec\gamma}}q
			\Lambda\bc{\sum_{\sigma=1}^q\prod_{i=1}^{\vec\gamma}1-c\vec\mu_{i}^{(\pi)}(\sigma)}
		-\frac d{2(1-c/q)}\Lambda\bc{1-\sum_{\tau=1}^qc\vec\mu_1^{(\pi)}(\tau)\vec\mu_2^{(\pi)}(\tau)}}.
	\end{align*}
In the case of the Potts antiferromagnet, $\cB(d,\pi)$ from \Thm~\ref{Thm_stat}
specializes to $\Bsbm(\pi;q,d,c_\beta)$.

\begin{lemma}\label{Lemma_col1}
For all $\pi\in\cP_*^2(\Omega)$ we have $\Bsbm(\pi;q,d,1)=\lim_{\beta\to\infty}\Bsbm(\pi;d,q,c_\beta).$
\end{lemma}
\begin{proof}
This follows from the dominated convergence theorem
because $\Lambda$ is bounded and continuous on $[0,1]$.
\end{proof}

\begin{lemma}\label{Cor_col1}
If $d<\dc$, then $\Bsbm(q,d,1)=\ln q+\frac d2\ln(1-1/q)$.
\end{lemma}
\begin{proof}
The lower bound is attained at the distribution $\pi=\delta_{q^{-1}\vecone}$, i.e., the atom sitting on the uniform distribution on $\Omega$.
The upper bound is immediate from the definition (\ref{eqcol}) of $\dc$.
\end{proof}

\noindent
In order to derive an upper bound on $\dc$ we use the following observation.

\begin{lemma}\label{Lemma_AbbeSandon}
For any $d_1 > (q-1)^2$ there exists $\del >0$ such that for all $d\geq d_1$ the following is true. 
\Whp\ there is an assignment 
 $ \tau_{\tilde\G(n,\vec m_d,p_\infty)}$ such that 
	$$\bck{A(\SIGMA,  \tau_{\tilde\G(n,\vec m_d,p_\infty)})}_{\tilde\mu_{\tilde\G(n,\vec m_d,p_\infty)}} > \del.$$
\end{lemma}
\begin{proof}
We begin by observing that it suffices to prove the statement for $d=d_1$.
By the Nishimori property for balanced colorings from Fact~\ref{Fact_balancedNishimori}, $\tilde\G(n,\vec m_d,p_\infty)$ is distributed as
$\G'=\G^*(n,\vec m_d,p_\infty,\tilde\SIGMA)$.
Furthermore, if we obtain $\G''$ from $\G'$
by deleting each constraint node with probability $1-d_1/d$ independently, then $\G''$ is distributed as $\G^*(n,\vec m_{d_1},p_\infty,\tilde\SIGMA)$.
Hence, setting $\tau_{\G'}=\tau_{\G''}$, we see that $\bck{A(\SIGMA,  \tau_{\G'})}_{\tilde\mu_{\G'}} > \del$ \whp\

Thus, assume that $d=d_1$ and fix some $(q-1)^2<d'<d$.
The algorithm of Abbe and Sandon~\cite{abbe2015detection} delivers the following:
	\begin{equation}\label{eqAbbeSandon}
	\parbox{14cm}{for some $\del'>0$ 
	\whp\ the algorithm returns
	$\tau_{\G^*(n,\vec m_{d'},p_{\infty},\SIGMA^*)}$ such that $A( \SIGMA^*,  \tau_{\G^*(n,\vec m_{d'},p_{\infty},\SIGMA^*)}) > \del'.$}
	\end{equation}
We are going to use this algorithm to achieve the same for the balanced planted coloring model.

Given an instance of $\G_0=\tilde \G(n, \vec m_d, p_\infty,\tilde \SIGMA)$, delete a uniformly random set of
$\eps n$ vertices to form the graph $\G_1$ for some suitable $\eps=\eps(d,d',\delta')>0$ such that $n_1=(1-\eps)n$ is an integer.
Let $\SIGMA_1$ be $\tilde \SIGMA$ restricted to the vertices that remain after deletion.
Then $\G_1$ is distributed as $\G(n_1, \vec m_{(1-\eps +O(\eps^2)) d}, p_\infty,\SIGMA')$.
Hence, by choosing an appropriate $\eps$ we can ensure that $\G_1$ and 
$\G(n_1, \vec m_{d'}, p_\infty,\SIGMA')$ have total variation distance $o(1)$.
Moreover, $\SIGMA_1$ and the uniformly random map $\SIGMA^*_{n_1}$ are mutually contiguous.
Hence, so are $\G_1$ and $\G(n_1, \vec m_{d'}, p_\infty,\SIGMA^*)$.
Thus, (\ref{eqAbbeSandon}) applies to $\G_1$ and
we extend the assignment produced by that algorithm to an assignment of $n$ vertices by assigning colors at random to the $\eps n$ deleted vertices.
Consequently,  
 choosing $d-d'$ and thus $\eps$ sufficiently small,
we deduce from (\ref{eqAbbeSandon}) that there is an algorithm such that
	\begin{equation}\label{eqAbbeSandon'}
	\parbox{14cm}{for some $\del'>0$ 
	\whp\ the algorithm returns
	$\tau'_{\G^*(n,\vec m_{d},p_{\infty},\tilde\SIGMA)}$ such that $A( \tilde\SIGMA,  \tau'_{\G^*(n,\vec m_{d},p_{\infty},\tilde\SIGMA)}) > \del'.$}
	\end{equation}
Since $\tau'_{\G^*(n,\vec m_{d},p_{\infty},\tilde\SIGMA)}$ depends on the graph 
$\G^*(n,\vec m_{d},p_{\infty},\tilde\SIGMA)$ only, the assertion follows from (\ref{eqAbbeSandon'}) and the Nishimori property.
\end{proof}

\begin{corollary}\label{Lemma_sbm}
We have $\dc\leq(q-1)^2$ for all $q\geq3$.
\end{corollary}
\begin{proof}
Combining \Lem~\ref{Lemma_AbbeSandon} with \Lem~\ref{lem:overlap2} and \Lem~\ref{Lemma_overlap1}, we conclude that
for every $d>(q-1)^2$ there is $\delta>0$ such that
	$$\limsup_{n\to\infty}\frac1{n}\Erw\brk{\ln Z(\tilde\G(n,\vec m_{d},p_\infty)}>\ln q+\frac {d}2\ln(1-1/q)+\delta.$$
Therefore, \Lem~\ref{Lemma_betaToInfty} shows that for (\ref{eqThm_col1_2a}) holds for some subsequence $(n_l)$ for all large enough $\beta$.
Consequently, \Thm~\ref{Thm_stat}, \Lem~\ref{lem:FEmutualInfo} and \Lem~\ref{Lemma_PottsAssumptions} yield
$\cB_{\mathrm{Potts}}(q,d,c_\beta)>\ln q+\frac {d}2\ln(1-1/q)+\delta$ for all large enough $\beta$.
Hence, \Lem~\ref{Lemma_col1} shows that $\dc\leq d$.
\end{proof}

\begin{remark}
For $q\geq5$ the upper bound $\dc\leq(q-1)^2$ actually follows from a simple first moment argument.
\end{remark}

\noindent
As a final preparation we need the following elementary observation.

\begin{lemma}
\label{Lemma_col3}
Assume that $d>0$, $\eta>0$ are such that for some strictly increasing sequence $(n_l)_{l\geq1}$ there is a sequence $m(n_l)$ such that
	$$\lim_{l\to\infty}\pr\brk{Z(\G(n_l,m(n_l),p_\infty))\geq q^{n_l}(1-1/q)^{m(n_l)}\exp(-\eta n_l)}=0.$$
Then 
	$$\lim_{l\to\infty}\max_{m(n_l)\leq m\leq n+m(n_l)}\pr\brk{Z(\G(n_l,m,p_\infty))\geq q^{n_l}(1-1/q)^{m}\exp(-\eta n_l/2)}=0.$$
\end{lemma}
\begin{proof}
We use two-round exposure.
Thus, for $m>m(n_l)$ we think of $\G(n_l,m,1)$ as being obtained from $\G(n_l,m(n_l),\infty)$ by adding $m-m(n_l)$ random constraint nodes.
Then for each $q$-coloring $\sigma$ of $\G(n_l,m(n_l),\infty)$ we have
	\begin{align*}
	\pr\brk{\sigma\mbox{ is a $q$-coloring of }\G(n_l,m,p_\infty)|\sigma\mbox{ is a $q$-coloring of }\G(n_l,m(n_l),p_\infty)}&\leq(1-1/q)^{m-m(n_l)+o(n)}.
	\end{align*}
Therefore,
	\begin{align*}
	\Erw[Z(\G(n_l,m,p_\infty))|\G(n_l,m(n_l),p_\infty)]\leq Z(\G(n_l,m(n_l),p_\infty))(1-1/q)^{m-m(n_l)+o(n)}
	\end{align*}
and the assertion follows from Markov's inequality.
\end{proof}

\begin{proof}[Proof of \Thm~\ref{Thm_col}]
From \Lem~\ref{Lemma_sbm} we know that $\dc\leq(q-1)^2$.
Hence, assume for contradiction that $d_1<\dc\le(q-1)^2$ but
	$$\liminf_{n\to\infty}\Erw\sqrt[n]{Z(\G(n,\vec m_{d_1},p_\infty))}<q(1-1/q)^{d_1/2}.$$
Then there exist a subsequence $(n_l)_l$ and $\eta>0$ such that 
	\begin{align}\label{eqThm_col_1}
	\lim_{l\to\infty}\Erw[Z(\G(n_l,\vec m_{d_1}(n_l),p_\infty))^{1/n_l}]=q(1-1/q)^{d_1/2}\exp(-3\eta).
	\end{align}
Set $\zeta=q(1-1/q)^{d_1/2}\exp(-2\eta)$ and let $(u(n))_n$ be the sharp threshold sequence from \Lem~\ref{Lemma_Ehud}.
Then (\ref{eqThm_col_1}) implies that $\limsup_{l\to\infty}u(n_l)\leq d_1.$
Hence, there exists $d_1<d_2<\dc\leq(q-1)^2$ such that 
	\begin{align*}
	\lim_{l\to\infty}\pr[Z(\G(n_l,\vec m_{d_2}(n_l),p_\infty))^{1/n_l}\geq q(1-1/q)^{d_2/2}\exp(-\eta)]&=0.
	\end{align*}
Consequently, if we fix $d_2<d_3<d_4<\dc$ with $d_4-d_2$ sufficiently small, then \Lem~\ref{Lemma_col3} yields
	\begin{align*}
	\lim_{l\to\infty}
		\max_{d_3n_l/2<m<d_4n_l/2}
		\pr[Z(\G(n_l,m,p_\infty))\geq q^{n_l}(1-1/q)^{m}\exp(-\eta/2)]&=0.
	\end{align*}
{Therefore, \Cor~\ref{Lemma_col2} shows that 
for any fixed $d_3<d_5<d_6<d_4$ there is $\eps>0$ such that}
	\begin{align*}
	\liminf_{l\to\infty}\max_{d_5n_l/2\leq m\leq d_6n_l/2}
	\pr\brk{\bck{\norm{\rho(\SIGMA,\TAU)-\bar\rho}_2}_{\tilde\G(n_l,m,p_\infty)}<\eps}<1.
	\end{align*}
Hence, \Lem~\ref{Lemma_overlap1} yields
	\begin{align*}
	\limsup_{l\to\infty}\frac1{n_l}\Erw\brk{\ln Z(\tilde\G(n_l,m_{d_6}(n_l),p_\infty))}>\ln q+\frac {d_6}2\ln(1-1/q)+\delta.
	\end{align*}
Further, applying \Lem~\ref{Lemma_betaToInfty} we obtain
	\begin{align*}
	\limsup_{l\to\infty}\frac1{n_l}\Erw\brk{\ln Z(\hat\G(n_l,\vec m_{d_6}(n_l),p_\beta))}>\ln q+\frac{d_6}2\ln(1-1/q)+\delta
		\qquad\mbox{for all large enough $\beta$.}
	\end{align*}
Since \Lem~\ref{Lemma_PottsAssumptions} shows that the Potts antiferromagnet satisfies the assumptions of
\Thm~\ref{Thm_stat}, we conclude
	$$\Bsbm(q,d_6,c_\beta)\geq\limsup_{l\to\infty}\frac1{n_l}\Erw\brk{\ln Z(\hat\G(n_l,\vec m_{d_6}(n_l),p_\beta))}\geq\ln q+\frac{d_6}2\ln(1-1/q)+\delta$$
for all large enough $\beta$.
Finally, \Lem~\ref{Lemma_col1} shows that then
	$$\Bsbm(q,d_6,1)>\ln q+\frac{d_6}2\ln(1-1/q),$$
which contradicts the fact that $d_6<\dc$.

Conversely, assume that $d$ is such that $\Bsbm(\pi;q,d,1)>\ln q+\frac d2\ln(1-1/q)$ for $\pi\in\cP_*^2(\Omega)$.
Then \Lem~\ref{Lemma_col1} implies that there is $\delta>0$ such that
	$\Bsbm(\pi;q,d,c_\beta)>\ln q+\frac d2\ln(1-c_\beta/q)+\delta$ for all large enough $\beta$.
Therefore, \Lem~\ref{Lemma_PottsAssumptions} and \Thm~\ref{Thm_stat} imply that for all large enough $\beta$ and $n>n_0(\beta)$,
	 $$\frac1n\Erw\ln Z(\hat\G(n,m(d),p_\beta))>\ln q+\frac d2\ln(1-c_\beta/q)+\delta/2.$$ 
Consequently, 
\Lem~\ref{Lemma_NotQuietEnough} yields $\limsup_{n\to\infty}\Erw\sqrt[n]{Z(\G(n,\vec m(d),p_\infty))}<q(1-1/q)^{d/2}$.
\end{proof}

\subsection{Proof of \Thm~\ref{thm:noisyXor}}\label{Sec_thm:noisyXor}
Here we prove \Thm~\ref{thm:noisyXor} on LDGM codes. We will apply Theorem~\ref{Thm_stat} as follows.  Let $\Omega =\{ \pm 1\}$, $\Psi = \{ \psi_1, \psi_{-1} \}$ with 
\begin{align*}
\psi_J (\sigma) &=  1+  (1-2 \eta) J  \cdot \prod_{i=1}^k \sigma_i  
\end{align*}
for $\sigma \in \Omega^k$,  $J \in \{ \pm 1\}$.    
The prior is uniform: $p(\psi_1) = p(\psi_{-1}) = 1/2$.  In particular, the distribution on $\Psi$ conditioned on the planted assignment is exactly as in the description of the LDGM codes:
\begin{align*}
\Pr[ \psi_a = \psi_1 | \SIGMA(\partial  a) = (\sigma_1, \dots \sigma_k)] &= \frac{ 1+  (1-2 \eta)  \cdot \prod_{i=1}^k \sigma_i } { 1+  (1-2 \eta)  \cdot \prod_{i=1}^k \sigma_i + 1-  (1-2 \eta)  \cdot \prod_{i=1}^k \sigma_i  }\\
&= \begin{cases} 
      1-\eta &\text{ if }\quad \prod_{i=1}^k \sigma_i =1 \\
      \eta  &\text{ if }\quad \prod_{i=1}^k \sigma_i =- 1.
   \end{cases}
\end{align*}

Recall that $\xi =  |\Omega|^{-k} \sum_{\tau \in \Omega^k} \Erw [ \vec \psi(\tau) ] $, so in this setting we have $\xi = \Erw [ \vec \psi( \vec 1 ) ] = 1$. 
We also compute
\begin{align*}
\frac{ d}{k \xi |\Omega|^k}  \sum_{\tau \in \Omega^k} \Erw [ \PSI(\tau) \ln \PSI(\tau) ] &=  \frac{d}{k}  \frac{1}{2} \left [2 (1- \eta) \ln (2-2 \eta)   +2 \eta \ln (2 \eta)   \right  ] \\
&= \frac{d}{k} [ \ln 2 + \eta \ln \eta + (1-\eta) \ln (1-\eta) ].
\end{align*}

Now a distribution $\pi' \in \cP^2_*(\{\pm1 \})$ corresponds exactly to a distribution $\pi \in \cP_0([-1,1])$ via the map $\vec\theta_{j}^{(\pi)} = 2\vec\mu_j^{(\pi')}(1)-1$. So the Bethe formula becomes: 
\begin{align}
\nonumber
	\cB(d,\pi')&=
	\Erw\brk{\frac{\xi^{-\vec\gamma}}{|\Omega|}
			\Lambda\bc{\sum_{\sigma\in\Omega}\prod_{i=1}^{\vec\gamma}\sum_{\tau\in\Omega^k}\vecone\{\tau_{\vec h_i}=\sigma\}\PSI_b(\tau)\prod_{j\neq {\vec h_i}}\vec\mu_{ki+j}^{(\pi')}(\tau_j)}
	-\frac{d(k-1)}{k\xi}\Lambda\bc{\sum_{\tau\in\Omega^k}\PSI(\tau)\prod_{j=1}^k\vec\mu_j^{(\pi')}(\tau_j)}} \\
	\nonumber
	&= \Erw \left [ \frac{1}{2}
			\Lambda\bc{\sum_{\sigma\in \{\pm 1\}}\prod_{i=1}^{\vec\gamma} \left( 1+ \sum_{\tau \in \{\pm 1\}^{k-1}}(1- 2\eta)\vec J_b  \sigma  \prod_{j=1}^{k-1} \tau_j \vec\mu_{ki+j}^{(\pi')}(\tau_j) \right )}  \right.\\
			& \qquad \left.
	-\frac{d(k-1)}{k}\Lambda\bc{ 1+ \sum_{\tau\in\{\pm 1\}^k}  (1-2 \eta) \vec J  \cdot \prod_{j=1}^k \tau_j \vec\mu_j^{(\pi')}(\tau_j)} \right ]\nonumber\\
	&= \Erw\brk{\frac{1}{2}
			\Lambda\bc{\sum_{\sigma\in \{\pm 1\}}\prod_{i=1}^{\vec\gamma} \left( 1+ (1- 2\eta) \sigma \vec J_b    \prod_{j=1}^{k-1} \vec\theta_{ki+j}^{(\pi)}\right )}
	-\frac{d(k-1)}{k}\Lambda\bc{ 1+   (1-2 \eta) \vec J  \cdot \prod_{j=1}^k \vec \theta_j^{(\pi)} }} \label{eq:ldgmbethe}.
	\end{align}

Now we check the three conditions {\bf SYM}, {\bf BAL}, and {\bf POS}.  Both {\bf SYM} and {\bf BAL} are immediate since the function $\tau \mapsto \Erw [ \vec \psi(\tau) ]$ is constant over all $\tau \in \{\pm 1\}^k$. 

Now recall the {\bf POS} condition:
	\begin{align*}\nonumber
	\Erw\bigg[\bigg(1-\sum_{\sigma\in\Omega^k}\PSI(\sigma)\prod_{j=1}^k\vec\mu_j^{(\pi)}(\sigma_j)\bigg)^l&
		+(k-1)\bigg(1-\sum_{\sigma\in\Omega^k}\PSI(\sigma)\prod_{j=1}^k\vec\mu_j^{(\pi')}(\sigma_j)\bigg)^l\\
		&-\sum_{i=1}^k\bigg(1-\sum_{\sigma\in\Omega^k}\PSI(\sigma)\vec\mu_i^{(\pi)}(\sigma_i)
			\prod_{j\neq i}\vec\mu_j^{(\pi')}(\sigma_j)\bigg)^l\bigg]
		\geq0.
			\end{align*}
Let $\vec J\in\{\pm1\}$ be chosen uniformly.
Then $\PSI=\psi_{\vec J}$ and
\begin{align*}
	\bigg(1-\sum_{\sigma\in\Omega^k}\PSI(\sigma)\prod_{j=1}^k\vec\mu_j^{(\pi)}(\sigma_j)\bigg)^l&=
		((1-2\eta)\vec J)^l\prod_{j=1}^k\bc{\sum_\sigma\sigma\vec\mu_j^{(\pi)}}^l,\\
	\bigg(1-\sum_{\sigma\in\Omega^k}\PSI(\sigma)\prod_{j=1}^k\vec\mu_j^{(\pi')}(\sigma_j)\bigg)^l&=
		((1-2\eta)\vec J)^l\prod_{j=1}^k\bc{\sum_\sigma\sigma\vec\mu_j^{(\pi')}}^l,\\
	\bigg(1-\sum_{\sigma\in\Omega^k}\PSI(\sigma)\vec\mu_i^{(\pi)}(\sigma_i)
			\prod_{j\neq i}\vec\mu_j^{(\pi')}(\sigma_j)\bigg)^l&=
			((1-2\eta)\vec J)^l\bc{\sum_\sigma\sigma\vec\mu_i^{(\pi)}}^l		
				\prod_{j\neq i}^k\bc{\sum_\sigma\sigma\vec\mu_j^{(\pi')}}^l.
	\end{align*}
Hence, if we let
	\begin{align*}
	X&=\Erw\brk{\bc{\sum_\sigma\sigma\vec\mu_1^{(\pi)}}^l},&
	Y&=\Erw\brk{\bc{\sum_\sigma\sigma\vec\mu_1^{(\pi')}}^l},
	\end{align*}
then {\bf POS} becomes
	\begin{align*}
	\Erw \left [((1-2\eta)\vec J)^l \right]\bc{X^k+(k-1)Y^k-kXY^{k-1}}&\geq0.
	\end{align*}
Crucially, if $l$ is odd then $\Erw \left [((1-2\eta)\vec J)^l \right]=0$.
Moreover, if $l$ is even than $X,Y\geq0$.
Since
	$$X^k+(k-1)Y^{k}-kXY^{k-1}\geq0\qquad\mbox{if }X,Y\geq0$$
the assertion follows.

Now  with 	
	\begin{align*}
\cI(k,d, \eta)&=\sup_{\pi\in\cP_0([-1,1])} \Erw\brk{ \frac{1}{2} 
			\Lambda\bc{\sum_{\sigma\in \{ \pm1 \}}\prod_{b=1}^{\vec\gamma}\left( 1+    \sigma \vec J_b (1-2 \eta)  \prod_{j =1}^{k-1}    \vec\theta_{kb+j}^{(\pi)} \right ) }-
		\frac{d(k-1)}{k}\Lambda\bc{ 1+  \vec J (1-2 \eta)    \prod_{j=1}^k \vec\theta_j^{(\pi)}}},
\end{align*}
Theorem~\ref{Thm_stat} and \eqref{eq:ldgmbethe} give
\begin{align*}
\lim_{n \to \infty} \frac{1}{n}I (\SIGMA^*, \G^*)  =
	 \left(1 + d/{k} \right)  \ln 2 +\eta \ln \eta +(1-\eta)\ln(1-\eta)- \cI(k,d, \eta),
\end{align*}
completing the proof of \Thm~\ref{thm:noisyXor}.

\subsection{Further examples}\label{Sec_further}
Finally, we compile just a few further examples of well known models that satisfy the conditions {\bf SYM}, {\bf BAL} and {\bf POS}.
The first one is a hypergraph version of the Potts antiferromagnet related to the hypergraph $q$-coloring problem.

\begin{lemma}\label{Lemma_PottsHypergraph}
Let $\Omega=[q]$ for some $q\geq2$, let $k\geq2$, $\beta>0$ and let $\Psi=\{\psi\}$ where
	$$\psi:\sigma\in\Omega^k\mapsto\exp(-\beta\vecone\{\sigma_1=\cdots=\sigma_k\}).$$
Then  {\bf BAL}, {\bf SYM} and {\bf POS} hold.
\end{lemma}
\begin{proof}
As in the Potts antiferromagnet {\bf SYM} is immediate from the symmetry amongst the colors.
Let $c_\beta=1-\exp(-\beta)$.
Then $$\psi(\sigma)=1-c_\beta\sum_{\tau\in\Omega}\prod_{i=1}^k\vecone\{\sigma_i=\tau\}.$$
Hence, for any $\mu\in\cP(\Omega)$ we have
	\begin{align*}
	\sum_{\sigma\in\Omega^k}\psi(\sigma)\prod_{i=1}^k\mu(\sigma_i)=1-c_\beta\sum_{\sigma\in\Omega}\mu(\sigma)^k.
	\end{align*}
Thus, {\bf BAL} follows from the convexity of $x\in[0,1]\mapsto x^k$.
Moving on to {\bf POS}, we fix $\pi,\pi'\in\cP_*^2(\Omega)$.
In the present case the condition boils down to
	\begin{align*}\nonumber
	0&\leq \Erw\bigg[\bigg(\sum_{\sigma\in\Omega}\prod_{j=1}^k\vec\mu_j^{(\pi)}(\sigma)\bigg)^l
		+(k-1)\bigg(\sum_{\sigma\in\Omega}\prod_{j=1}^k\vec\mu_j^{(\pi')}(\sigma)\bigg)^l
		-k\bigg(\sum_{\sigma\in\Omega}\vec\mu_1^{(\pi)}(\sigma)
			\prod_{j=1}^{k-1}\vec\mu_j^{(\pi')}(\sigma)\bigg)^l\bigg].
	\end{align*}
Using the mutual independence of $\MU_1^{(\pi)},\MU_1^{(\pi')},\ldots$, the expression simplifies to
	\begin{align*}
	&\sum_{\sigma_1,\ldots,\sigma_l\in\Omega}\Erw\brk{\prod_{j=1}^l\MU_1^{(\pi)}(\sigma_j)}^k	
		-k\Erw\brk{\prod_{j=1}^l\MU_1^{(\pi)}(\sigma_j)}\Erw\brk{\prod_{j=1}^l{\MU_1^{(\pi')}(\sigma_j)}}^{k-1}
		+(k-1)\Erw\brk{\prod_{j=1}^l\MU_1^{(\pi')}(\sigma_j)}^k.
	\end{align*}
Clearly the last expression is non-negative (because $x^k-kxy^{k-1}+(k-1)y^k\geq0$ for all $x,y\geq0$), whence {\bf POS} follows.
\end{proof}

As a second example we consider the random $k$-SAT model at inverse temperature $\beta>0$.
We represent the Boolean values by $\pm1$ rather than $0,1$ to simplify the calculations.
Moreover, the vector $J$ represents the signs with which the literals appear in a given clause.

\begin{lemma}\label{Lemma_PottsHypergraph}
Let $\Omega=\{\pm1\}$, $k\geq2$, $\beta>0$ and let $\Psi=\{\psi_J:J\in\{\pm1\}^k\}$ where
	$$\psi_J:\sigma\in\Omega^k\mapsto1-(1-\exp(-\beta))\prod_{i=1}^k\frac{1+J_i\sigma_i}2.$$
Let $p$ be the uniform distribution on $\Psi$.
Then  {\bf BAL}, {\bf SYM} and {\bf POS} hold.
\end{lemma}
\begin{proof}
Let $c_\beta=1-\exp(-\beta)$.
The assumption {\bf SYM} is satisfied because for any $i\in[k]$, $\tau=\pm1$ we have
	$$2^{-k}\sum_{J\in\{\pm1\}^k}\sum_{\sigma\in\Omega^k:\sigma_i=\tau_i}\psi_J(\sigma)=2^{k}-c_\beta.$$
Moreover, {\bf BAL} holds because
	$$\mu\in\cP(\Omega)\mapsto2^{-k}\sum_{J\in\{\pm1\}^k}\sum_{\sigma\in\Omega^k}\psi_J(\sigma)\prod_{j=1}^k\mu(\sigma_j)
		=1-c_\beta 2^{-k}$$
is a constant function.
To check {\bf POS}, we follow similar steps as in the interpolation argument from~\cite{PanchenkoTalagrand}.
Fix $\pi,\pi'$.
We need to show that
	\begin{align*}\nonumber
	0&\leq 2^{-k}c_\beta^l\sum_{J\in\{\pm1\}^k}\Erw\bigg[\bigg(\sum_{\sigma\in\Omega^k}\prod_{j=1}^k(1+J_j\sigma_j)\vec\mu_j^{(\pi)}(\sigma_j)\bigg)^l
		+(k-1)\bigg(\sum_{\sigma\in\Omega^k}\prod_{j=1}^k(1+J_j\sigma_j)\vec\mu_j^{(\pi')}(\sigma_j)\bigg)^l\\
		&\qquad\qquad\qquad\qquad\qquad-\sum_{i=1}^k\bigg(\sum_{\sigma\in\Omega^k}(1+J_i\sigma_i)\vec\mu_i^{(\pi)}(\sigma_i)
			\prod_{j\in[k]\setminus\{ i\}}(1+J_j\sigma_j)\vec\mu_j^{(\pi')}(\sigma_j)\bigg)^l\bigg]\\
		&=c_\beta^l\sum_{J\in\{\pm1\}^k}\Erw\bigg[\bigg(\prod_{j=1}^k\vec\mu_j^{(\pi)}(J_j)\bigg)^l
			+(k-1)\bigg(\prod_{j=1}^k\vec\mu_j^{(\pi')}(J_j)\bigg)^l
				-\sum_{i=1}^k\bigg(\vec\mu_i^{(\pi)}(J_i)
			\prod_{j\in[k]\setminus\{ i\}}\vec\mu_j^{(\pi')}(J_j)\bigg)^l\bigg].
	\end{align*}
Since $\vec\mu_1^{(\pi)},\vec\mu_1^{(\pi')},\ldots$ are independent, the last expectation simplifies to
	\begin{align*}
	\Erw\brk{\sum_{J\in\{\pm1\}}\vec\mu_1^{(\pi)}(J)^l}^k
			+(k-1)\Erw\brk{\sum_{J\in\{\pm1\}}\vec\mu_1^{(\pi')}(J)^l}^k
				-k\Erw\brk{\sum_{J\in\{\pm1\}}\vec\mu_1^{(\pi)}(J)^l}
					\Erw\brk{\sum_{J\in\{\pm1\}}\vec\mu_1^{(\pi')}(J)^l}^{k-1}.
	\end{align*}
The last expression is non-negative because $x^k-kxy^{k-1}+(k-1)y^k\geq0$ for all $x,y\geq0$.
\end{proof}

Finally, let us check the conditions for the random $k$-NAESAT model at inverse temperature $\beta>0$.
Again we represent the Boolean values by $\pm1$ and the literal signs by  a vector $J$.

\begin{lemma}\label{Lemma_NAE}
Let $\Omega=\{\pm1\}$, $k\geq2$, $\beta>0$ and let $\Psi=\{\psi_J:J\in\{\pm1\}^k\}$ where
	$$\psi_J:\sigma\in\Omega^k\mapsto1-(1-\exp(-\beta))\prod_{i=1}^k\frac{1+J_i\sigma_i}2-(1-\exp(-\beta))\prod_{i=1}^k\frac{1-J_i\sigma_i}2.$$
Let $p$ be the uniform distribution on $\Psi$.
Then  {\bf BAL}, {\bf SYM} and {\bf POS} hold.
\end{lemma}
\begin{proof}
Let $c_\beta=1-\exp(-\beta)$.
{\bf SYM} holds because for any $i\in[k]$, $\tau=\pm1$ we have
	$$2^{-k}\sum_{J\in\{\pm1\}^k}\sum_{\sigma\in\Omega^k:\sigma_i=\tau_i}\psi_J(\sigma)=2^{k}-2c_\beta$$
and {\bf BAL} holds because
	$$\mu\in\cP(\Omega)\mapsto2^{-k}\sum_{J\in\{\pm1\}^k}\sum_{\sigma\in\Omega^k}\psi_J(\sigma)\prod_{j=1}^k\mu(\sigma_j)
		=1-c_\beta 2^{1-k}$$
is a constant.
To check {\bf POS}, fix $\pi,\pi'$.
Then {\bf POS} comes down to
	\begin{align*}\nonumber
	0&\leq\sum_{J\in\{\pm1\}^k}\Erw\bigg[\bigg(\prod_{j=1}^k\vec\mu_j^{(\pi)}(J_j)+\prod_{j=1}^k\vec\mu_j^{(\pi)}(-J_j)\bigg)^l
		+(k-1)\bigg(\prod_{j=1}^k\vec\mu_j^{(\pi')}(J_j)+\prod_{j=1}^k\vec\mu_j^{(\pi')}(-J_j)\bigg)^l\\
		&\qquad\qquad\qquad\qquad\qquad-\sum_{i=1}^k\bigg(\vec\mu_i^{(\pi)}(J_i)
			\prod_{j\in[k]\setminus\{ i\}}\vec\mu_j^{(\pi')}(J_j)
				+\vec\mu_i^{(\pi)}(-J_i)\prod_{j\in[k]\setminus\{ i\}}\vec\mu_j^{(\pi')}(-J_j)\bigg)^l\bigg]\\
		&=\sum_{J\in\{\pm1\}^k}\sum_{s_1,\ldots,s_l\in\{\pm1\}}
			\Erw\brk{\prod_{h=1}^l\prod_{j=1}^k\vec\mu_j^{(\pi)}(s_hJ_j)
				+(k-1)\prod_{h=1}^l\prod_{j=1}^k\vec\mu_j^{(\pi')}(s_hJ_j)
					-\sum_{i=1}^k\prod_{h=1}^l\vec\mu_i^{(\pi)}(s_hJ_i)\prod_{j\neq i}\vec\mu_j^{(\pi')}(s_hJ_j)}.
	\end{align*}
Due to the independence of the $\vec\mu_1^{(\pi)},\vec\mu_1^{(\pi')},\ldots$, the last expression boils down to
	\begin{align*}
	&\sum_{s_1,\ldots,s_l\in\{\pm1\}}\Erw\brk{\sum_{J\in\{\pm1\}}\prod_{h=1}^l\vec\mu_1^{(\pi)}(s_hJ)}^k
				+(k-1)\Erw\brk{\sum_{J\in\{\pm1\}}\prod_{h=1}^l\vec\mu_1^{(\pi')}(s_hJ)}^k\\
	&\qquad\qquad\qquad				-k\Erw\brk{\sum_{J\in\{\pm1\}}\prod_{h=1}^l\vec\mu_1^{(\pi)}(s_hJ)}
						\Erw\brk{\sum_{J\in\{\pm1\}}\prod_{h=1}^l\vec\mu_1^{(\pi')}(s_hJ)}^{k-1},
	\end{align*}
which is non-negative because $x^k-kxy^{k-1}+(k-1)y^k\geq0$ for all $x,y\geq0$.
\end{proof}

\smallskip\noindent{\bf Acknowledgements.} We thank Joe Neeman, Guilhem Semerjian, and Ryan O'Donnell for helpful comments and Charilaos Efthymiou and Nor Jaafari for pointing out a number of typos.


\begin{thebibliography}{29}

\bibitem{Abbe}
E.\ Abbe, A.\ Montanari: Conditional random fields, planted constraint satisfaction and entropy concentration.
	Theory of Computing {\bf11} (2015) 413--443.

\bibitem{abbe2015detection}
E.\ Abbe, C.\ Sandon: Detection in the stochastic block model with multiple clusters: proof of the achievability conjectures, acyclic BP, and the information-computation gap.
arXiv:1512.09080 (2015).

\bibitem{Barriers}
D.~Achlioptas, A.~Coja-Oghlan:
Algorithmic barriers from phase transitions.
Proc.~49th FOCS (2008) 793--802.

\bibitem{Dimitris}
D.~Achlioptas, H.~Hassani, N.~Macris, R.~Urbanke:
Bounds for random constraint satisfaction problems via spatial coupling.
Proc.\ 27th SODA (2016) 469--479.

\bibitem{nae}
D.~Achlioptas, C.~Moore:
Random $k$-SAT: two moments suffice to cross a sharp threshold.
SIAM Journal on Computing {\bf 36} (2006) 740--762.

\bibitem{AchMooreHyp2}
D.\ Achlioptas, C.\ Moore: On the 2-colorability of random hypergraphs.
Proc.\ 6th RANDOM (2002) 78--90.

\bibitem{AchMoore3col}
D.\ Achlioptas, C.\ Moore: Almost all graphs of degree 4 are 3-colorable.
Journal of Computer and System Sciences {\bf 67} (2003) 441--471.

\bibitem{AchNaor}
D.~Achlioptas, A.~Naor:
The two possible values of the chromatic number of a random graph.
Annals of Mathematics {\bf 162} (2005) 1333--1349.

\bibitem{ANP}
D.\ Achlioptas, A.\ Naor, and Y.\ Peres: Rigorous location of phase transitions in hard optimization problems. Nature \textbf{435} (2005) 759--764.

\bibitem{Aizenman}
M.\ Aizenman, R.\ Sims, S.\ Starr: An extended variational principle for the SK spin-glass model.
Phys.\ Rev.\ B {\bf68}  (2003) 214403.

\bibitem{alekhnovich2003more}
M.\ Alekhnovich: More on average case vs approximation complexity.
44th Annual IEEE Symposium on Foundations of Computer Science (2003) 298--307.

\bibitem{Allen}
S.\ Allen, R.\ O'Donnell: Conditioning and covariance on caterpillars. IEEE Information Theory Workshop (ITW) (2015) 1--5.

\bibitem{AlonKahale}
N.\ Alon, N.\ Kahale: A spectral technique for coloring random 3-colorable graphs.
SIAM J.\ Comput.\ {\bf 26} (1997) 1733--1748

\bibitem{AlonKriv}
N.~Alon, M.~Krivelevich: The concentration of the chromatic number of random graphs.
Combinatorica {\bf 17} (1997) 303--313

\bibitem{AKS}
N.\ Alon, M.\ Krivelevich, B.\ Sudakov: Finding a large hidden clique in a random graph. Random Structures \&\ Algorithms {\bf13} (1998) 457--466.

\bibitem{applebaum2010public}
B.\ Applebaum, B.\ Barak, A.\ Wigderson: Public-key cryptography from different assumptions.
Proceedings of the forty-second ACM symposium on Theory of computing (2010) 171--180.

\bibitem{Banks}
J.\ Banks, C.\ Moore, J.\ Neeman, P.\ Netrapalli:
Information-theoretic thresholds for community detection in sparse networks.
Proc.\ 29th COLT (2016) 383--416.

\bibitem{Victor}
V.\ Bapst, A.\ Coja-Oghlan: Harnessing the Bethe free energy.
Random Structures and Algorithms {\bf 49} (2016) 694--741.

\bibitem{Silent}
V.\ Bapst, A.\ Coja-Oghlan, C.\ Efthymiou: Planting colourings silently. 
Combinatorics, Probability and Computing {\bf 26} (2017) 338--366. 

\bibitem{Cond}
V.\ Bapst, A.\ Coja-Oghlan, S.\ Hetterich, F.\ Rassmann, D.\ Vilenchik:
The condensation phase transition in random graph coloring.
Communications in Mathematical Physics {\bf341} (2016) 543--606.

\bibitem{bayati}
M.~Bayati, D.~Gamarnik, P.~Tetali:
Combinatorial approach to the interpolation method and scaling limits in sparse random graphs.
Annals of Probability {\bf41} (2013) 4080--4115.

\bibitem{billingsley}
P.\ Billingsley:  Convergence of probability measures. 
Second edition. John Wiley \& Sons (1999).

\bibitem{BBColor}
B.~Bollob\'as: The chromatic number of random graphs.
Combinatorica {\bf8} (1988) 49--55

\bibitem{BJR}
B.~Bollob\'as, S.\ Janson, O.\ Riordan:
The phase transition in inhomogeneous random graphs. 
Random Structures \& Algorithms {\bf 31} (2007) 3--122.

\bibitem{Boppana}
R.\ Boppana:
Eigenvalues and graph bisection: An average-case analysis. 
28th Annual IEEE Symposium on Foundations of Computer Science (1987)  280--285.

\bibitem{BLM}
C.\ Bordenave, M.\ Lelarge, L.\ Massouli\'e:
Non-backtracking spectrum of random graphs: community detection and non-regular Ramanujan graphs.
Proc.\ 56th FOCS (2015) 1347--1357.

\bibitem{ChengMcEliece}
J.\ Cheng, R.\ McEliece,
Some high-rate near capacity codecs for the Gaussian channel.
Proc.\ 34th ALLERTON (1996).

\bibitem{Adaptive}
A.~Coja-Oghlan: Graph partitioning via adaptive spectral techniques. Combinatorics, Probability and Computing {\bf19} (2010) 227--284.

\bibitem{Covers}
A.~Coja-Oghlan:  Upper-bounding the $k$-colorability threshold by counting covers.
Electronic Journal of Combinatorics  {\bf 20} (2013) \#P32.

\bibitem{Nor}
A.\ Coja-Oghlan, N.\ Jaafari: On the Potts model on random graphs.
Electronic Journal of Combinatorics {\bf23} (2016) P4.3.

\bibitem{KostaSAT}
A.\ Coja-Oghlan, K.\ Panagiotou: The asymptotic $k$-SAT threshold.
Advances in Mathematics {\bf 288} (2016) 985--1068.

\bibitem{Will}
A.\ Coja-Oghlan, W.\ Perkins: Belief Propagation on replica symmetric random factor graph models.
Proc.\ 20th RANDOM (2016) 27:1--27:15.

\bibitem{COPS}
A.\ Coja-Oghlan, W.\ Perkins, K.\ Skubch:
Limits of discrete distributions and Gibbs measures on random graphs.
European Journal of Combinatorics {\bf 66} (2017) 37--59.

\bibitem{CDGS}
P.\ Contucci, S.\ Dommers, C.\ Giardina, S.\ Starr:
Antiferromagnetic Potts model on the Erd\H os-R\'enyi random graph.
Communications in Mathematical Physics {\bf323} (2013) 517--554.

\bibitem{Decelle}
A.\ Decelle, F.\ Krzakala, C.\ Moore, L.\ Zdeborov\'a:
Asymptotic analysis of the stochastic block model for modular networks and its algorithmic applications.
Phys.\ Rev.\ E {\bf 84} (2011) 066106.

\bibitem{DM}
A.\ Dembo, A.\ Montanari:
Gibbs measures and phase transitions on sparse random graphs.
Braz.\ J.\ Probab.\ Stat.\ {\bf24} (2010) 137--211.

\bibitem{demboPotts}
A.\ Dembo, A.\ Montanari, A.\ Sly, N.\ Sun:
The replica symmetric solution for Potts models on d-regular graphs.
 Communications in Mathematical Physics {\bf 327} (2014) 551--575.

\bibitem{dembo}
 A.\ Dembo, A.\ Montanari, N.\ Sun:
  Factor models on locally tree-like graphs.
 Annals of Probability {\bf 41} (2013) 4162--4213.
 
\bibitem{deshpande15}
 Y.\ Deshpande, E.\ Abbe, A.\ Montanari: 
 Asymptotic mutual information for the two-groups stochastic block model.
Information and Inference {\bf 6} (2017) 125--170.

\bibitem{DSS1}
J.~Ding, A.~Sly, N.~Sun: Satisfiability threshold for random regular NAE-SAT.
Communications in Mathematical Physics {\bf 341} (2016) 435--489.

\bibitem{DSS2}
J.~Ding, A.~Sly, N.~Sun: Maximum independent sets on random regular graphs.
Acta Mathematica {\bf 217} (2016) 263--340.

\bibitem{DSS3}
J.~Ding, A.~Sly, N.~Sun: Proof of the satisfiability conjecture for large $k$.
Proc.\ 47th STOC (2015) 59--68.

\bibitem{ER60} P.~Erd\H os, A.~R\'enyi,
On the evolution of random graphs.
Magyar Tud.\ Akad.\ Mat.\ Kutat\'o Int.\ K\"ozl~\textbf{5} (1960) 17--61. 

\bibitem{Feige}
U.\ Feige: 
Relations between average case complexity and approximation complexity. 
Proceedings of the Thirty-fourth Annual ACM on Symposium on Theory of Computing (2002) 534--543.

\bibitem{Feldman2015}
V.\ Feldman, W.\ Perkins, S.\ Vempala: On the complexity of random satisfiability problems with planted solutions.
Proceedings of the Forty-Seventh Annual ACM on Symposium on Theory of Computing (2015) 77--86.

\bibitem{FranzLeone}
S.\ Franz, M.\ Leone: Replica bounds for optimization problems and diluted spin systems.
 J.\ Stat.\ Phys.\ {\bf111} (2003) 535--564.

\bibitem{GMU}
A.~Giurgiu, N.~Macris, R.~Urbanke: Spatial coupling as a proof technique and three applications.
IEEE Transactions on Information Theory {\bf62} (2016) 5281--5295.

\bibitem{Greenhill}
M.\ Dyer, A.\ Frieze, C.\ Greenhill: 
On the chromatic number of a random hypergraph. 
Journal of Combinatorial Theory, Series B, {\bf 113} (2015) 68--122.

\bibitem{Guedon}
O.\ Gu\'{e}don, R.\ Vershynin: (2015). 
Community detection in sparse networks via Grothendieck's inequality. 
Probability Theory and Related Fields {\bf 165} (2015) 1025--1049.

\bibitem{Guerra}
F.~Guerra: Broken replica symmetry bounds in the mean field spin glass model. Comm.\ Math.\ Phys.\ {\bf 233} (2003) 1--12.

\bibitem{Holland}
P.\ Holland, K.\ Laskey, S.\ Leinhardt:
Stochastic blockmodels: First steps.
Social networks, {\bf 5} (1983) 109--137.

\bibitem{KabashimaSaad}
Y.\ Kabashima, D.\ Saad: Statistical mechanics of error correcting codes. Europhys.\ Lett.\ {\bf 45}  (1999) 97--103.  

\bibitem{KrivSud}
M.\ Krivelevich and B.\ Sudakov: 
The chromatic numbers of random hypergraphs. 
Random structures and algorithms, {\bf 12} (1998) 381--403.

\bibitem{quiet}
F.\ Krzakala and L.\ Zdeborov\'a:
Hiding quiet solutions in random constraint satisfaction problems.
Phys.\ Rev.\ Lett.\ {\bf102} (2009) 238701.



\bibitem{pnas}
F.~Krzakala, A.~Montanari, F.~Ricci-Tersenghi, G.~Semerjian, L.~Zdeborov\'a:
Gibbs states and the set of solutions of random constraint satisfaction problems.
Proc.~National Academy of Sciences {\bf104} (2007) 10318--10323.

\bibitem{KXZ}
F.\ Krzakala, J.\ Xu, L.\ Zdeborov\'a:
Mutual Information in Rank-One Matrix Estimation.
IEEE Information Theory Workshop (ITW) (2016) 71--75.

\bibitem{KRU}
S.~Kudekar, T.~Richardson, R.~Urbanke: Spatially coupled ensembles universally achieve capacity under belief propagation.
IEEE Transactions on Information Theory {\bf 59} (2013) 7761--7813.

\bibitem{kumar2010phase}
K.\ Kumar, P.\ Pakzad, A.\ Salavati, A.\ Shokrollahi: Phase transitions for mutual information.
IEEE 6th International Symposium on Turbo Codes \& Iterative Information Processing (2010) 137--141.

\bibitem{Lelarge}
M.~Lelarge, L.~Miolane: Fundamental limits of symmetric low-rank matrix estimation. 
Conference on Learning Theory (COLT) (2017) 1297--1301.

\bibitem{LuczakColor}
T.~{\L}uczak: The chromatic number of random graphs.
Combinatorica {\bf11} (1991) 45--54

\bibitem{McSherry}
F.\ McSherry:
Spectral partitioning of random graphs.
42nd Annual IEEE Symposium on Foundations of Computer Science  (2001) 529--537.

\bibitem{Macris}
N.\ Macris: Griffith-Kelly-Sherman correlation inequalities: a useful tool in the theory of error correcting codes.
IEEE Transactions on Information Theory {\bf 53} (2007) 664--683.

\bibitem{massoulie2014community}
L.\ Massouli{\'e}: Community detection thresholds and the weak Ramanujan property.
Proceedings of the 46th Annual ACM Symposium on Theory of Computing (2014) 694--703.

\bibitem{MMRU}
C.\ M\'easson, A.\ Montanari, T.\ Richardson, R.\ Urbanke: The generalized area theorem and some of its consequences.
IEEE Transactions on Information Theory {\bf 55} (2009) 4793--4821.

\bibitem{MMU}
C.\ M\'easson, A.\ Montanari, R.\ Urbanke: Maxwell construction: the hidden bridge between iterative and maximum a posteriori decoding.
IEEE Transactions on Information Theory {\bf 54} (2008) 5277--5307.

\bibitem{MM}
M.~M\'ezard, A.~Montanari:
Information, physics and computation.
Oxford University Press~2009.

\bibitem{MP1}
M.\ M\'ezard, G.\ Parisi:
The Bethe lattice spin glass revisited.
Eur.\ Phys.\ J.\ B {\bf20} (2001) 217--233.

\bibitem{MP2}
M.\ M\'ezard, G.\ Parisi:
The cavity method at zero temperature.
Journal of Statistical Physics {\bf111} (2003)  1--34.

\bibitem{mezard1990spin}
M.~M\'ezard, G.~Parisi, M.~Virasoro:
Spin glass theory and beyond.
World Scientific~1987.

\bibitem{MPZ}
M.~M\'ezard, G.~Parisi, R.~Zecchina:
Analytic and algorithmic solution of random satisfiability problems.
Science {\bf 297} (2002) 812--815.

\bibitem{Molloy}
M.~Molloy: The freezing threshold for $k$-colourings of a random graph.
Proc.\ 43rd STOC (2012) 921--930.

\bibitem{Monasson}
R.\ Monasson:
Optimization problems and replica symmetry breaking in finite connectivity spin glasses.
Journal of Physics A: Mathematical and General {\bf 31} (1998) 513.


\bibitem{MontanariBounds}
A.\ Montanari: Tight bounds for LDPC and LDGM codes under MAP decoding.
 IEEE Transactions on Information Theory {\bf51} (2005) 3221-3246.

\bibitem{Andrea}
A.\ Montanari: Estimating random variables from random sparse observations.
European Transactions on Telecommunications {\bf 19} (2008) 385--403.

\bibitem{montanari2011reconstruction}
A.~Montanari, R.~Restrepo, P.~Tetali:
Reconstruction and clustering in random constraint satisfaction problems.
SIAM Journal on Discrete Mathematics {\bf 25} (2011) 771--808.

\bibitem{MontanariSen}
A.\ Montanari, S.\ Sen:
Semidefinite programs on sparse random graphs and their application to community detection. 
Proceedings of the 48th Annual ACM Symposium on Theory of Computing (2016) 814--827.

\bibitem{mossel2013proof}
E.\ Mossel, J.\ Neeman, A.\ Sly: A proof of the block model threshold conjecture.
Combinatorica, in press.

\bibitem{Mossel}
E.\ Mossel, J.\ Neeman, A.\ Sly: Reconstruction and estimation in the planted partition model.
Probability Theory and Related Fields (2014) 1--31.

\bibitem{Ralph}
R.\ Neininger, L.\ R\"{u}schendorf:
A general limit theorem for recursive algorithms and combinatorial structures.
The Annals of Applied Probability {\bf 14} (2004) 378--418. 

\bibitem{PanchenkoBook}
D.\ Panchenko:
The Sherrington-Kirkpatrick model.
Springer 2013.

\bibitem{Panchenko}
D.\ Panchenko:
Spin glass models from the point of view of spin distributions.
Annals of Probability {\bf 41}  (2013) 1315--1361.

\bibitem{PanchenkoTalagrand}
D.\ Panchenko, M.\ Talagrand:
Bounds for diluted mean-fields spin glass models.
Probab.\ Theory Relat.\ Fields {\bf130} (2004) 319--336.

\bibitem{rachev}
S.\ Rachev:  Probability metrics and the stability of stochastic models.
 John Wiley \& Sons {\bf 269} (1991).	

\bibitem{Raghavendra}
P.\ Raghavendra, N.\ Tan: Approximating CSPs with global cardinality constraints using SDP hierarchies. Proc.\ 23rd SODA (2012) 373--387.

\bibitem{SSZ}
A.~Sly, N.~Sun, Y.~Zhang: The number of solutions for random regular NAE-SAT.
57th Annual IEEE Symposium on Foundations of Computer Science (2016) 724--731.

\bibitem{Talagrand}
M.\ Talagrand:The Parisi formula. Ann.\ Math.\ {\bf 163} (2006) 221--263.

\bibitem{villani}
C.\ Villani: Optimal transport: old and new.
Springer-Verlag Berlin Heidelberg {\bf 338} (2009).

\bibitem{LF}
L.\ Zdeborov\'a, F.\ Krzakala:
Statistical physics of inference: thresholds and algorithms.
Advances in Physics {\bf 65} (2016) 453--552.

\bibitem{LenkaFlorent}
L.~Zdeborov\'a, F.~Krzakala: Phase transition in the coloring of random graphs.
Phys.\ Rev.\ E {\bf76} (2007) 031131.


\end{thebibliography}
\end{document}